\documentclass[11pt,fullpage,letterpaper]{article} 
  \usepackage[margin=0.9in]{geometry}
  \usepackage{listings} 
\usepackage{mathrsfs} 
\usepackage{kpfonts,comment}
\usepackage[T1]{fontenc}

\usepackage{kz_style}

\title{
\LARGE  \bf  
Principled Learning-to-Communicate with \\
Quasi-Classical Information Structures}  
\vspace{38pt}
\author{ 
\vspace{18pt}
Xiangyu Liu$^\dag$ \and \qquad\qquad~~~  Haoyi You$^\dag$ \and \qquad\qquad  Kaiqing Zhang\thanks{The authors are ordered alphabetically, and are affiliated with the University of Maryland, College Park, MD, USA, 20742. Emails: 
        {\tt\small \{xyliu999,~yuriiyou,~kaiqing\}@umd.edu}.}
        }
\date{}

\usepackage{setspace}
\renewcommand{\thesection}{\Roman{section}}

\usepackage{titlesec}

\titleformat{\section}
  {\centering\normalfont\Large\bfseries}
  {\thesection.}{0.5em}{}

\begin{document} 
\maketitle
\begin{abstract}
Learning-to-communicate (LTC) in partially observable environments has received increasing attention in deep multi-agent reinforcement learning, where the control and communication strategies are {jointly}  learned. Meanwhile, the impact of communication on decision-making has been extensively studied in control theory. In this paper, we seek to formalize and better understand LTC by bridging these two lines of work, through the lens of \emph{information structures} (ISs). To this end, we formalize LTC in decentralized partially observable Markov decision processes  (Dec-POMDPs) under the common-information-based framework from decentralized stochastic control, and classify LTC problems based on the ISs before (additional)  information sharing. We first show that non-classical LTCs are computationally intractable in general, and thus focus on quasi-classical (QC) LTCs. We then propose a series of conditions for QC LTCs, 
under which LTC preserves  the QC IS after information  sharing, whereas violating them can 
cause computational hardness in general. 
Further, we develop provable planning and learning algorithms for QC LTCs, and establish quasi-polynomial time and sample complexities for several QC LTC examples that satisfy the above conditions. 
Along the way, we also establish new results on a  relationship between (strictly) QC IS and the condition of having strategy-independent common-information-based beliefs (SI-CIBs), as well as on solving Dec-POMDPs without computationally
intractable oracles but beyond those with SI-CIBs, 
which may be of independent interest. 
\end{abstract}

\section{Introduction}

The learning-to-communicate (LTC) problem has emerged and gained traction in the area of (deep) multi-agent reinforcement learning (MARL) \cite{2016learningtocommunicate,learningpropogation,jiang2018learning}. Unlike classical MARL, which aims to learn \emph{control} strategies that maximize the expected accumulated rewards, LTC seeks to \emph{jointly} optimize over both the \emph{control} and the \emph{communication} strategies of all the agents, as a way to mitigate the challenges due to the agents' {partial observability} of the environment.
Despite the promising empirical successes, theoretical understanding of LTC 
remains largely underexplored. 

On the other hand, in control theory, a rich literature has investigated the role of \emph{communication} in decentralized/networked control \cite{tatikonda2004control,nair2007feedback,xiao2005joint,yuksel2013jointly}, inspiring us to rigorously examine LTC from such a principled perspective. 
Most of these studies, however,  focused on linear systems, and did not explore the computational or sample complexity guarantees when the system model is not 
(fully) known. More recently, a few studies   \cite{ashutoshcommunicate1,ashutoshcommunicate2} explored the settings with general discrete (non-linear)  spaces, under special communication protocols and state transition dynamics (see \S\ref{sec:related_work} for more detailed discussions).

More broadly, the design of communication strategies dictates the \emph{information structure} (IS) of the control system, which characterizes \emph{who knows what and when} \cite{witsenhausen1971separation}. IS and its impact on the \emph{optimization tractability}, especially for linear systems, have been extensively studied in (decentralized) stochastic control, see  \cite{aditya2012information,yuksel2023stochastic} for comprehensive overviews. In this work, we seek a more principled understanding of LTC through the lens of information structures, with a focus on the computational and sample complexities of the problem. 

Specifically, we formalize LTC in the general framework of decentralized partially observable Markov decision processes ({Dec-POMDPs}) \cite{DecPOMDP-NEXP}, 
as in the empirical studies  \cite{2016learningtocommunicate,learningpropogation,jiang2018learning}.  To achieve finite-time and sample complexity  guarantees, we resort to the recent development in \cite{liu2023tractable} on partially observable MARL, based on the common-information-based (CIB) framework \cite{ashutosh2013team,ashutosh2013game} from decentralized stochastic control, to model the communication and information sharing protocols among agents. We detail  our contributions as follows.

\paragraph{Contributions.} (\romannumeral 1) We formalize learning-to-communicate in  Dec-POMDPs under the common-information-based framework \cite{ashutosh2013team,ashutosh2013game,liu2023tractable}, allowing the sharing of \emph{historical} information,  and the modeling of communication costs.  (\romannumeral 2) We classify LTCs through the lens of \emph{information structures},  according to the ISs before (additional) information sharing, i.e., under the \emph{baseline}  sharing. 
We then show that LTCs with \emph{non-classical} \cite{aditya2012information} IS of the baseline sharing can be computationally intractable in general. (\romannumeral 3) Given the hardness, we thus focus on \emph{quasi-classical} (QC) LTCs, and propose a series of conditions under which LTC preserves the QC IS after (additional)  information sharing, whereas violating them can cause computational hardness in general. (\romannumeral 4) We propose both planning and learning algorithms for QC LTCs, by converting them to Dec-POMDPs with \emph{strategy-independent common-information-based beliefs} (SI-CIBs) \cite{ashutosh2013game}, a condition shown to be critical for tractable computation and learning \cite{liu2023tractable}. (\romannumeral 5) Quasi-polynomial time and sample complexities of the algorithms are established for several QC LTC examples that satisfy the conditions in (\romannumeral 3). Along the way, we also establish new results on 
a relationship between \emph{(strictly) quasi-classical} ((s)QC) ISs and the SI-CIB  condition in the framework of \cite{ashutosh2013game}, 
as well as  solving general Dec-POMDPs without computationally intractable oracles but  beyond those with SI-CIBs, which thus advances the results in \cite{liu2023tractable}. 
These results may be of independent interest for studying general Dec-POMDPs. We conclude with some experimental results to validate the implementability and effectiveness of our algorithms.  

\subsection{Related Work}\label{sec:related_work}

\paragraph{Communication-control joint optimization.} The joint design of control and communication strategies has been well-studied in the control theory  literature \cite{xiao2005joint,zhang2006communication,matni2015communication,peng2013event,yuksel2013jointly,ashutoshcommunicate1,ashutoshcommunicate2}. However, even with model knowledge, the computational complexity (and associated necessary conditions) of solving many of these models remains elusive, let alone the sample complexity when it comes to learning. Moreover, these models mostly have special structures, e.g., with linear systems \cite{bansal1989simultaneous,yuksel2013jointly} and sometimes specific, fixed communication strategies (e.g., the event-triggered ones) \cite{xiao2005joint,zhang2006communication,peng2013event,yuksel2013jointly}, or sharing only instantaneous observations 
\cite{ashutoshcommunicate1,ashutoshcommunicate2}. 
 
\paragraph{Information sharing and information structures.} 
Information structure has been extensively studied to characterize \emph{who knows what and when} in (decentralized) stochastic  control  \cite{aditya2012information,yuksel2023stochastic}. Our paper aims to formally understand LTC through the lens of information structures.    
The common-information-based approach to formalize {information sharing} in \cite{ashutosh2013team,ashutosh2013game} serves as the basis for our work. In comparison, these results did not \emph{optimize} or \emph{learn} to share the information, and 
focused on the \emph{structural results}, without concrete computational (and sample) complexity analysis.

\paragraph{Partially observable MARL theory.} 
Planning and learning in partially observable MARL 
are known to be hard  in general \cite{papadimitriou1987complexity,lusena2001nonapproximability,jin2020sample,DecPOMDP-NEXP}. Recently, \cite{liu2022sample,altabaarole} developed polynomial-sample complexity algorithms for partially observable stochastic games, but with computationally intractable oracles; \cite{liu2023tractable} developed quasi-polynomial-time and sample algorithms for such models, leveraging information sharing among the agents. In contrast, our paper focuses on \emph{optimizing/learning to share}, together with control strategy optimization/learning.

\section{Preliminaries}
\label{problem formulation}

\noindent\textbf{Notation.} We use $\mathbb{N},\mathbb{Q},\mathbb{R}$ to denote the sets of all the natural, rational, and real numbers, respectively.  For an integer $m>0$, we denote $[m]:=\{1,2,\cdots,m\}$. 
For a finite set $\cX$, we use $|\cX|$ to denote the cardinality of $\cX$, and use $\Delta(\cX)$ to denote the probability simplex over $\cX$. 
By a slight abuse of notation, we use small letters to denote both the random variables and their realizations, when it is clear from the context. In particular, this simplifies the notation as 
terms like $\PP(s_h,p_h\given c_h)$ may be either viewed as the numerical values with realizations $s_h,p_h,c_h$, or as a distribution-valued random variable $\PP(s_h=\cdot,p_h=\cdot\given c_h)$. 
For a random variable $x$, we use $\sigma(x)$ to denote the sigma-algebra generated by $x$. For $\sigma$-algebras $\mathscr{F}_1$ on the space $\cX_1$ and $\mathscr{F}_2$ on the space $\cX_2$, we denote by $\mathscr{F}_1\otimes \mathscr{F}_2$ the product $\sigma$-algebra on the space $\cX_1\times \cX_2$. We use $\mathds{1}[]$ to denote the indicator function. Unless otherwise noted, the set $\{\}$ considered is ordered, such that the elements in the set are indexed.  

\subsection{Learning-to-Communicate Formulation}
\label{sec:ltc_formulation} 

For  $n>1$  agents, a (cooperative) \textit{learning-to-communicate} problem   
can be described by the components in a tuple 
$
\cL=\la H,\cS,\{\cA_{i,h}\}_{i\in[n],h\in[H]},\{\cO_{i,h}\}_{i\in[n],h\in[H]},\{\cM_{i,h}\}_{i\in[n],h\in[H]},\TT,\OO,\mu_1,\{\cR_h\}_{h\in[H]},\{\cK_h\}_{h\in[H]}\ra$, where  $H$ denotes the length of each episode, and other components are introduced as follows.

\subsubsection{Decision-making components} 
We use  
$\cS$ to denote the state space, and $\cA_{i,h}$ to denote the {control action} space of agent $i$ at timestep  $h\in[H]$.  We denote by $s_h\in \cS$ the state and by $a_{i,h}$ the control action of agent $i$ at timestep $h$. We use $a_h:=(a_{1,h},\cdots,a_{n,h})\in \cA_h:=\prod_{i\in[n]}\cA_{i,h}$ 
to denote the joint control action of all the $n$ agents at timestep $h$. 
We denote by $\TT=\{\TT_h\}_{h\in[H]}$ the collection of state transition kernels, where 
$s_{h+1}\sim\TT_h(\cdot\given s_h,a_h)\in \Delta(\cS)$ at timestep $h$. 
We use $\mu_1\in\Delta(\cS)$ to denote the initial state distribution.  We denote by $\cO_{i,h}$ the observation space and by $o_{i,h}\in \cO_{i,h}$ the observation of agent $i$ at timestep  $h$. We use
$o_h:=(o_{1,h},o_{2,h},\cdots,o_{n,h})\in\cO_h:=\prod_{i\in[n]}\cO_{i,h}$ to denote the joint observation of all the $n$ agents at timestep $h$. 
We use $\OO=\{\OO_h\}_{h\in[H]}$ to denote the collection of emission functions, where $o_h\sim\OO_h(\cdot\given s_h)\in \Delta(\cO_h)$ at timestep $h$ and state $s_h\in\cS$. 
Also,  for each $s_h\in \cS$, we denote by $\OO_{i,h}(\cdot\given s_h)$ the emission for agent $i$, the marginal distribution of $o_{i,h}$ given $\OO_h(\cdot\given s_h)$. 
At each timestep $h$, agents will receive a common reward $r_h=\cR_h(s_h,a_h)$, where 
$\cR_h:\cS\times \cA_h\rightarrow [0,1]$ denotes the reward function shared by the agents.

\subsubsection{Communication components} In addition to reward-driven decision-making, agents also need to decide and learn \emph{(what) to communicate with others}.  
At timestep $h$, agents share part of their information $z_h\in\cZ_h$ with other agents, 
where $\cZ_h$ denotes the collection of all possible shared information at timestep $h$. 
Here we consider a general setting where the shared information $z_h$ may contain two parts, a  \emph{baseline-sharing} part $z^b_h$ that comes from some existing sharing protocol among agents, and an  \emph{additional-sharing} part $z_{i,h}^a$ for each agent $i$ that comes from explicit communication \emph{to be decided/learned},  
with the joint additional-sharing information $z^a_h:=\cup_{i=1}^n z_{i,h}^a$. This general setting covers those  considered in most empirical studies on LTC \cite{2016learningtocommunicate,learningpropogation,jiang2018learning}, with no baseline sharing. We kept the baseline sharing since our focus is on the \emph{finite-time} and \emph{sample}  tractability of LTC, for which a certain amount of information sharing is known to be 
necessary 
\cite{liu2023tractable}. Note that $z_h=z^b_h\cup z^a_h$  and $z^b_h\cap z^a_h=\emptyset$. The shared information is part of the historical observations and (both \emph{control} and \emph{communication}) actions. 
We denote by $\cZ^b_h, \cZ^a_h$, and $\cZ_{i,h}^a$ the collections of all possible $z_h^b$, $z_h^a$, and $z^a_{i,h}$ at each timestep $h$, respectively.

At timestep $h$,  the \emph{common information} among all the agents is thus defined as the union of all the \emph{shared information} so far:   $c_{h^-}=\cup_{t=1}^{h-1} z_t\cup z_h^b$, and $c_{h^+}=\cup_{t=1}^hz_t$, where $c_{h^-}$ and $c_{h^+}$ denote the (accumulated) common information \emph{before} and \emph{after}  additional sharing, respectively.  
The \emph{private information} of agent $i$ at timestep $h$ \emph{before} and \emph{after}  additional sharing are denoted by $p_{i,h^-}$ and $p_{i,h^+}$, respectively, where $p_{i,h^-}\subseteq\{o_{i,1},a_{i,1},\cdots,a_{i,h-1},o_{i,h}\}\backslash c_{h^-}, p_{i,h^+}\subseteq\{o_{i,1},a_{i,1}, \cdots,a_{i,h-1},o_{i,h}\}\backslash c_{h^+}$. We denote by $p_{h^-}:=(p_{1,h^-},\cdots,p_{n,h^-})$ and $p_{h^+}:=(p_{1,h^+},\cdots,p_{n,h^+})$ the joint private information {before} and {after} additional sharing at timestep $h$, respectively. 
{We then denote by $\tau_{i,h^-}:=p_{i,h^-}\cup c_{h^-}$ and $\tau_{i,h^+}:=p_{i,h^+}\cup c_{h^+}$ the \emph{information available} to agent $i$ at timestep $h$, before and after additional sharing, respectively, with $\tau_{h^-}:=p_{h^-}\cup c_{h^-}, \tau_{h^+}:=p_{h^+}\cup c_{h^+}$ denoting the associated joint information. We use $\cC_{h^-},\cC_{h^+},\cP_{i,h^-},\cP_{i,h^+},\cP_{h^-},\cP_{h^+},\cT_{i,h^-},\cT_{i,h^+},\cT_{h^-},\cT_{h^+}$ to denote, respectively,  the corresponding collections of all possible $c_{h^-},c_{h^+},p_{i,h^-},p_{i,h^+},p_{h^-},p_{h^+},\tau_{i,h^-},\tau_{i,h^+},\tau_{h^-},\tau_{h^+}$.} 

We use $m_{i,h}$ 
to denote the \emph{communication action} of agent $i$ at timestep $h$, and it will determine what information $z_{i,h}^a$ she will share, through the way to be specified later. We denote by $\cM_{i,h}$ the space of $m_{i,h}$, and by $m_h:=(m_{1,h},\cdots,m_{n,h})\in \cM_h:=\prod_{i=1}^n\cM_{i,h}$  the {joint communication action} of all the agents. 
We denote by $\cK_h:\cZ^a_h\rightarrow [0,1]$ the \emph{communication  cost} function and  by $\kappa_{h}=\cK_{h}(z^a_{h})$ the incurred communication cost at timestep $h$,  due to additional sharing.

\subsubsection{System evolution} 
The system evolves by 
alternating between the communication and the control steps as follows. 

\paragraph{\emph{\textbf {Communication step:}}} At each timestep $h$, each agent $i$ observes $o_{i,h}$ and may share part of her private information via baseline sharing, 
receives the baseline sharing of information from others, and forms $p_{i,h^-}$ and $c_{h^-}$. Then, each agent $i$ chooses her communication action, which determines the additional sharing of information, receives the additional-sharing of information from others, forms $p_{i,h^+}$ and $c_{h^+}$,  and incurs some communication cost $\kappa_h$.  Formally, the evolution of information is depicted as follows, which, unless otherwise noted, will be assumed throughout the paper. We follow the convention that any quantity at $h=0$ is empty/null.

\begin{assumption}[\emph{Information evolution}]
\label{assumption:information evolution}
For each $h\in [H]$, 
\begin{enumerate}[(a)]
     \item {({Baseline sharing}).} 
     $z_{h}^b=\chi_{h}(p_{(h-1)^+},a_{h-1},o_{h})$ for some fixed  transformation $\chi_{h}$;
     \item {({Additional sharing}).} For each agent  $i\in[n], z_{i,h}^a=\phi_{i,h}(p_{i,h^-},m_{i,h})$ for some function $\phi_{i,h}$, given communication action  $m_{i,h}$, and $m_{i,h}\in z_{i,h}^a$; 
    and the joint sharing  $z_h^a$ is thus generated by $z^a_h=\phi_h(p_{h^-},m_h)$, for some function $\phi_h$; 
     \item {({Private information before  sharing}).} For each agent $i\in[n]$, 
     $p_{i,h^-}=\xi_{i,h}(p_{i,(h-1)^+},a_{i,h-1},o_{i,h})$ for some fixed transformation  $\xi_{i,h}$, and the joint private information thus evolves as   $p_{h^-}=\xi_{h}(p_{(h-1)^+},a_{h-1},o_{h})$ for some fixed transformation $\xi_{h}$;
    \item {({Private information after   sharing}).} For each agent $i\in[n]$, 
    $p_{i,h^+}=p_{i,h^-}\backslash z_{i,h}^a$;
    \item {($(\tau_{i,h^-},\tau_{i,h^+})$-inclusion).} For each agent $i\in[n]$, $\tau_{i,h^-}\subseteq \tau_{i,h^+}\subseteq \tau_{i,(h+1)^-}$, and $o_{i,h}\in \tau_{i,h^-}$.
\end{enumerate}
\end{assumption}
 Note that \emph{fixed transformations} (e.g., $\chi_h$ and $\xi_{i,h}$ above) are not affected by the \emph{realized values} of the random variables, but dictate some \emph{pre-defined} transformation of the input random variables.  See \cite{ashutosh2013team,ashutosh2013game}, and \cite{liu2023tractable}  for common examples of baseline sharing that admit such fixed transformations when there is no additional sharing, and examples in \S \ref{sec: examples of QC}  on how they can be  extended to the LTC setting. It should not be confused with some general \emph{function} (e.g., $\phi_{i,h}$ above), which may depend on the \emph{realized values} of the input random variables. (a) and (c) on baseline sharing follow from those in \cite{ashutosh2013game,liu2023tractable}; (b) and (d) on additional sharing dictate how the communication action affects the sharing based on private information.  For example, a common choice of $(\cM_{i,h},\phi_{i,h})$ is that  $\cM_{i,h}= \{0,1\}^{\max\limits_{p_{i,h^-}\in \cP_{i,h^-}}|p_{i,h^-}|}$, and 
for any $p_{i,h^-}\in \cP_{i,h^-}$ and $m_{i,h}\in \cM_{i,h}$, $\phi_{i,h}(p_{i,h^-},m_{i,h})$ consists of the $k$-th element (with $k\in [|p_{i,h^-}|]$) of $p_{i,h^-}$  if and only if the $k$-th element of $m_{i,h}$ is 1, while other elements are $0$. 
As $m_{i,h}$ (dictating what to share) will be known given $z_{i,h}^a$ (what has been shared), $m_{i,h}$ is thus also modeled as being shared, i.e., $m_{i,h}\in z_{i,h}^a$. This is also consistent with the models in \cite{ashutoshcommunicate1,ashutoshcommunicate2} on control/communication joint optimization. (e) means that the agent has full memory of the information she had in the past and at present. We emphasize that this is closely related, but different from the common notion of \emph{perfect recall} \cite{kuhn1953extensive},  
where the agent has to recall all her own \emph{past actions}. Condition (e), in contrast, relaxes the memorization of the actions, but includes the instantaneous observation $o_{i,h}$. This condition is satisfied by all the models and examples  in \cite{aditya2012information,ashutosh2013team,ashutosh2013game,liu2023tractable}.  See also \S \ref{sec: examples of QC} for more examples  that satisfy this assumption.  Note that $o_{i,h}\in \tau_{i,h^-}$ has been noted necessary in order to have \emph{closed-loop} ISs in the literature \cite{yuksel2023stochastic}, which are the focus of the present paper.

Meanwhile, for both the baseline and the additional sharing protocols, we follow the model in the series of studies on partial history/information  sharing \cite{ashutosh2013team,ashutosh2013game,liu2023tractable,ashutoshcommunicate1,ashutoshcommunicate2} that, if an agent shares, she will share the information with \emph{all other}  agents {as \emph{common information}}. Additionally, we follow the convention from the literature on information structures \cite{aditya2012information,yuksel2023stochastic}, by incorporating the $\sigma$-algebra of the random variables. These conventions lead to the following regularity assumption on information sharing.
\begin{assumption}
\label{assumption:sigma_include}
$\forall i_1,i_2\in[n],h_1,h_2\in[H],i_1\neq i_2,h_1<h_2$, 
if $\sigma(o_{i_1,h_1})\subseteq \sigma(\tau_{i_2,h_2^-})$, then $\sigma(o_{i_1,h_1})\subseteq \sigma(c_{h_2^-})$, and if $\sigma(a_{i_1,h_1})\subseteq \sigma(\tau_{i_2,h_2^-})$, then $\sigma(a_{i_1,h_1})\subseteq \sigma(c_{h_2^-})$;  
if $\sigma(o_{i_1,h_1})\subseteq \sigma(\tau_{i_2,h_2^+})$, then $\sigma(o_{i_1,h_1})\subseteq \sigma(c_{h_2^+})$, and if $\sigma(a_{i_1,h_1})\subseteq \sigma(\tau_{i_2,h_2^+})$, then $\sigma(a_{i_1,h_1})\subseteq \sigma(c_{h_2^+})$.
\end{assumption}
Assumptions \ref{assumption:information evolution}-\ref{assumption:sigma_include} will be made throughout the paper.

\paragraph{\emph{\textbf {Decision-making   step:}}} After the communication, each agent $i$ chooses her control action $a_{i,h}$, receives a reward $r_h$, and the joint action $a_h$ drives the 
state  to $s_{h+1}\sim \TT_h(\cdot\given s_h,a_h)$.

\subsubsection{Strategies and solution concept}
\label{sec:strategy_solution}

At timestep $h$, each agent $i$ has two strategies, a \emph{control} strategy and a \emph{communication} strategy. We define a control strategy as $
g^a_{i,h}:\cT_{i,h^+}\rightarrow \cA_{i,h}$ and a communication strategy as $g^m_{i,h}:\cT_{i,h^-}\rightarrow \cM_{i,h}$. See \Cref{lemma:no_lose_optimality} for a formal argument on the use of such \emph{deterministic}  strategies without loss of optimality. 
We denote by $g^a_h=(g^a_{1,h},\cdots,g^a_{n,h})$ the joint control strategy and by $g^m_h=(g^m_{1,h},\cdots,g^m_{n,h})$ the joint communication strategy.
 We denote by $\cG^a_{i,h}, \cG^m_{i,h},\cG^a_h,\cG^m_h$ the corresponding spaces of $g^a_{i,h}, g^m_{i,h},g^a_h,g^m_h$.
 
 The objective of the agents in  LTC is to maximize the expected accumulated sum of the reward and the negative communication cost from timestep $h=1$ to $H$:
\begin{equation*}
     J_{\cL}(g^a_{1:H},g^m_{1:H}):=\EE_{\cL}\left[\sum_{h=1}^H (r_h-\kappa_h)\bigggiven g^a_{1:H},g^m_{1:H}\right],
\end{equation*} 
where the expectation $\EE_{\cL}$ is taken over all the randomness in the system evolution, given the strategies $(g^a_{1:H},g^m_{1:H})$.
 With this objective,  for any $\epsilon\geq 0$, we can define the solution concept of an  \emph{$\epsilon$-team optimum} for $\cL$ as follows.
\begin{definition}[$\epsilon$-team optimum] 
    We call a joint 
     strategy $(g^a_{1:H},g^m_{1:H})$ an $\epsilon$-team-optimal strategy of the LTC problem $\cL$ if 
     \begin{align*}
         \max_{\tilde{g}_{1:H}^a\in \cG_{1:H}^a, \tilde{g}_{1:H}^m\in\cG_{1:H}^m}J_\cL(\tilde{g}_{1:H}^a, \tilde{g}_{1:H}^m)-J_\cL(g_{1:H}^a, g_{1:H}^m)\le \epsilon.
     \end{align*}
\end{definition}
 \subsection{Information Structures of LTC} 
In decentralized stochastic control, the notion of information structure \cite{witsenhausen1975intrinsic,aditya2012information} captures \emph{who knows what and when} as the system evolves. In LTC, as the {additional sharing} via communication will also affect the IS and is \emph{not}  determined \emph{beforehand}, when we discuss the \emph{IS of an LTC problem}, we will refer to that of the problem \emph{with only baseline sharing}. 
In particular, an LTC $\cL$ without additional sharing is essentially a  Dec-POMDP (with potential baseline information sharing), as defined in \S \ref{sec:Dec-POMDP definition} for completeness. We formally define such a Dec-POMDP \emph{induced} by $\cL$  as follows. 

\begin{definition}[Dec-POMDP (with information sharing) induced by LTC]\label{def:LTC_induced_Dec-POMDP}
For an LTC problem $\cL=\la H,\cS,\{\cA_{i,h}\}_{i\in[n],h\in[H]},\{\cO_{i,h}\}_{i\in[n],h\in[H]},\{\cM_{i,h}\}_{i\in[n],h\in[H]},\TT,\OO,\mu_1,\{\cR_h\}_{h\in[H]},\{\cK_h\}_
{h\in[H]}\ra$, we call a Dec-POMDP (with information sharing)  $\overline{\cD}_\cL$  \emph{the Dec-POMDP (with information sharing)  
 induced by $\cL$}, if the agents share information only following the baseline sharing protocol of $\cL$, i.e., without additional sharing, which can be characterized by the tuple  $\overline{\cD}_\cL:=\la H,\cS,\{\cA_{i,h}\}_{i\in[n],h\in[H]},\{\cO_{i,h}\}_{i\in[n],h\in[H]},\TT,\OO,\mu_1,\{\cR_{h}\}_{h\in[H]}\ra$, together with the  baseline sharing protocol in Assumption  \ref{assumption:information evolution}. We may refer to it as the \emph{Dec-POMDP induced by LTC} or the \emph{induced Dec-POMDP}  for short. 
\end{definition}

In \S\ref{sec:ltc_formulation}, 
we introduced LTC in the \textit{state-space model}. In contrast, information structure is oftentimes more conveniently discussed under the equivalent framework of \textit{intrinsic models} \cite{witsenhausen1975intrinsic} 
(see the instantiation for Dec-POMDPs in \S \ref{sec:Dec-POMDP definition} for completeness). In an intrinsic model, each agent only \emph{acts once} throughout the system evolution, and the same agent in the state-space model at different timesteps is now treated as \emph{different agents}. There are thus $n\times H$ agents in total. 
Formally, for completeness, we extend the intrinsic-model-based reformulation to LTCs in \S \ref{sec:Dec-POMDP definition}. 
 
(Strictly) quasi-classical 
ISs are important subclasses of ISs, which were first introduced for decentralized stochastic control  \cite{witsenhausen1975intrinsic,mahajan2010measure,yuksel2023stochastic} (see  the instantiation for Dec-POMDPs in  \S \ref{sec:Dec-POMDP definition}). 
An IS that is not QC is called \emph{non-classical} \cite{aditya2012information,yuksel2023stochastic}. We extend such a categorization to LTC problems with different ISs as follows.
 \begin{definition}[(Strictly) quasi-classical LTC]\label{def:LTC_QC}
    We call an LTC $\cL$ \emph{(strictly) quasi-classical} if the Dec-POMDP induced by $\cL$ (see Definition \ref{def:LTC_induced_Dec-POMDP}) is \emph{(strictly) quasi-classical.} Namely,
    each agent in the intrinsic model of $\overline{\cD}_\cL$ knows the information (and the actions) of the agents who influence her, either directly or indirectly. 
\end{definition}

Similarly, an LTC $\cL$ that is not QC is called \emph{non-classical}.  See \S \ref{sec: examples of QC} for examples of QC and sQC LTCs.  Note that the categorization above  is defined based on the ISs \emph{before} additional sharing, as an inherent property of the LTC problem, since additional sharing is the solution \emph{to be} decided/learned. 
We focus on finding such a solution next.

\section{Hardness and Structural Assumptions}\label{section 3}

It is known that computing an  (approximate) team-optimal strategy in Dec-POMDPs, which are LTCs \emph{without} information-sharing, is \texttt{NEXP-hard} \cite{DecPOMDP-NEXP} in general. 
The hardness cannot be fully circumvented even when agents are allowed to share information:  even if agents share all the information, the LTC problem becomes a Partially Observable Markov Decision Process   (POMDP), which is known to be \texttt{PSPACE-hard}  \cite{papadimitriou1987complexity,lusena2001nonapproximability}. Hence, additional assumptions are necessary to make  LTCs computationally more tractable. We introduce several such  assumptions and their justifications below, whose  proofs  can be found in \S \ref{sec: proof details sec 3}.  

Recently, \cite{noah2022gamma} showed that \textit{observable} POMDPs \cite{even2007value}, a class of POMDPs with relatively \emph{informative}  observations, 
admit  \emph{quasi-polynomial time} algorithms to solve.  Such a condition was then extended to Dec-POMDPs with information sharing in \cite{liu2023tractable}, which also developed 
quasi-polynomial time and sample complexity algorithms. 

As solving LTCs is at least as hard as solving the Dec-POMDPs 
considered in \cite{liu2023tractable}, 
we first also make such an observability assumption on the \emph{joint} emission function as in \cite{liu2023tractable}, 
to potentially avoid computationally intractable oracles. 

\begin{assumption}[$\gamma$-observability \cite{even2007value,noah2022gamma,liu2023tractable}]
\label{gamma observability}
There exists a  $\gamma>0$ such that $\forall h\in[H]$, the emission $\OO_h$ satisfies 
that 
$\forall b_1,b_2\in \Delta(\cS)$, $
   \big\|{\OO_h^\top b_1-\OO_h^\top b_2}\big\|_1\ge \gamma\big\|{b_1-b_2}\big\|_1.$ 
\end{assumption} 

However, we show next that Assumption \ref{gamma observability} is not enough when it comes to LTC, if the baseline sharing IS is not favorable, and in particular, \emph{non-classical} \cite{aditya2012information}. The hardness persists even under a few additional assumptions to be introduced later 
that will make LTC  more tractable.

\begin{lemma}[Non-classical LTCs are hard]
\label{lemma: nonqc_hardness}
For non-classical LTCs under Assumptions \ref{gamma observability},  \ref{limited communication strategy}, \ref{useless action}, and \ref{weak gamma observability},  finding
     an $\frac{\epsilon}{H}$-team-optimal strategy is \texttt{PSPACE-hard}. 
\end{lemma}

 Note that the hardness comes from the intuition that, when communication costs are high, the additional sharing from LTC will be limited, preventing the upgrade of the IS from a non-classical one to a (quasi-)classical one, which is hard with only the \emph{joint}  observability of the emission (see Assumption \ref{gamma observability}), even along with  several other assumptions.

By Lemma \ref{lemma: nonqc_hardness}, 
we will hence focus on \emph{quasi-classical} LTCs hereafter. 
Indeed, QC is also known to be critical for efficiently solving \emph{continuous-space} and \emph{linear} decentralized control \cite{ho1972team,lamperski2015optimal}.
{However, quasi-classicality may not be sufficient for LTC problems, since 
the additional sharing may \emph{break} the QC IS,  and introduce computational hardness, as argued below. 

Firstly, the breaking of QC IS may result from the \emph{communication strategies}. Specifically, the  communication strategy space in 
\S\ref{sec:strategy_solution} allows the dependence on agents' \emph{private information}, {which} introduces incentives for \emph{signaling} \cite{aditya2012information} and  can also  cause computational hardness, as shown next.
\begin{lemma}[QC LTCs with full-history-dependent communication strategies are hard]\label{lemma: limited communication}
For QC LTCs under Assumption   \ref{gamma observability}, together with Assumptions \ref{useless action},  and \ref{weak gamma observability},
computing a team-optimal strategy in the general space  of 
$(\cG_{1:H}^a,\cG_{1:H}^m)$ with $\cG_{i,h}^m:=\{g_{i,h}^m: \cT_{i,h^-}\rightarrow \cM_{i,h}\}$
is \texttt{NP-hard}.
\end{lemma}

The hardness in Lemma \ref{lemma: limited communication} originates from the fact that when depending on the private/local  information, determining the communication action can be cast as a \emph{Team Decision problem} (TDP) \cite{teamdecisionhardness}, which is known to be hard. This will be the case even when the instantaneous observations are relatively observable (see Assumptions \ref{gamma observability} and \ref{weak gamma observability}). 

To avoid this hardness, we thus focus on communication strategies that only condition on the \emph{common information}.
Intuitively, this assumption is not unreasonable, as it means that \emph{which historical information to share} is determined by \emph{what has been shared} (in the common information). Note that this does not lose generality in the sense that the private information $p_{i,h^-}$ \emph{can still be shared}. It only means that the communication action is not {determined} based on $p_{i,h^-}$, and the additional sharing is still dictated by $z_{i,h}^a=\phi_{i,h}(p_{i,h^-},m_{i,h})$  (see Assumption \ref{assumption:information evolution}), depending on $p_{i,h^-}$.

\begin{assumption}[Common-information-based communication strategy]
    \label{limited communication strategy}
    The communication strategies  take \emph{common information} as input, with the following form: 
\begin{equation}
    \forall i\in[n],h\in[H],\quad  g^m_{i,h}:\cC_{h^-}\rightarrow \cM_{i,h}.
    \label{equ: limited communication strategy}
\end{equation}
\normalsize
\end{assumption}}

Secondly, the breaking of QC IS may result from the \emph{control strategies}:
if some agent did \emph{not} influence others in the baseline sharing (and thus these other agents did \emph{not} have to access the agent's information, while still satisfying QC), while she starts to influence others by \emph{sharing} her (\emph{useless}) \emph{control} actions, this will make her \emph{control strategies} relevant. 
We make the following two assumptions to avoid the related  pessimistic cases, each followed by a computational hardness result when (only) the condition is missing.  

Specifically, sometimes 
the action of some agents may not influence the \emph{state transition}.  
However, if they were deemed \emph{non-influential},  
but shared via additional sharing, then the QC IS may break for LTC. We thus make the following assumption. 

\begin{assumption}[Control-useless  action is not used]\label{useless action}
    $\forall i\in[n],h\in[H]$, suppose agent $i$'s action $a_{i,h}$ does not influence  the state $s_{h+1}$, namely,  $\forall s_h\in\cS, a_h\in\cA_h, a'_{i,h}\in \cA_{i,h},a'_{i,h}\neq a_{i,h}, \TT_h(\cdot\given s_h,a_h)=\TT_h(\cdot\given s_h,(a'_{i,h},a_{-i,h}))$. Then,  $\forall h'> h$, the random variable $a_{i,h}\notin \tau_{h'^-}$ and $a_{i,h}\notin \tau_{h'^+}$.
\end{assumption}

\begin{lemma}[QC LTCs  without Assumption \ref{useless action} are hard]\label{lemma: useless action}For QC LTCs under Assumptions \ref{gamma observability}, \ref{limited communication strategy}, and \ref{weak gamma observability}, finding a team-optimal strategy  is still \texttt{NP-hard}.
\end{lemma}

{Note that other than the justification above based on computational hardness, Assumption \ref{useless action} has been \emph{implicitly} made in the  IS examples in the literature when there are \emph{uncontrolled} state dynamics, see e.g., \cite{ashutosh2013game,liu2023tractable}. Moreover, we emphasize that for common non-degenerate cases where actions \emph{do}  affect the state transition, this assumption becomes unnecessary.}

{Other than {not influencing} state transition, an action may also be \emph{non-influential}  if the emission functions of other agents are \emph{degenerate}: they cannot \emph{sense} the influence from previous agents'  actions. We thus make the following assumption on the  emissions, followed by a justification result.} 

\begin{assumption}[Other agents' emissions are non-degenerate]\label{weak gamma observability} 
    $\forall h\in[H], i\in[n] $, $\OO_{-i,h}$ satisfies that $
        \forall b_1,b_2\in \Delta(\cS)$ such that $b_1\neq b_2,~\OO_{-i,h}^\top b_1\neq \OO_{-i,h}^\top b_2$, where $\OO_{-i,h}$ denotes the joint emission except agent $i$ at timestep $h$. 
\end{assumption} 
\begin{lemma}[QC LTCs without Assumption \ref{weak gamma observability} are hard]\label{lemma: weak gamma observability}
For QC LTCs under Assumptions \ref{gamma observability}, \ref{limited communication strategy}, and \ref{useless action}, finding an $\frac{\epsilon}{H}$-team-optimal strategy is still \texttt{PSPACE-hard}.
\end{lemma}

We have justified the above assumptions by showing that missing one of them may cause computational intractability of LTCs in general. Hence, Assumptions \ref{limited communication strategy}, \ref{useless action},  and \ref{weak gamma observability} will be made hereafter, unless otherwise noted. More  importantly, as we will show later,   as another justification, 
LTCs under Assumptions  \ref{limited communication strategy}, \ref{useless action},  and \ref{weak gamma observability} can indeed \emph{preserve} the QC/sQC information structure \emph{after}  additional sharing, making it possible for the overall LTC problem to be computationally more tractable.  
More examples that satisfy these assumptions can also be found in \S \ref{sec: examples of QC}.

\section{Solving QC LTC Problems Provably}\label{sec:positive_results}
 
We now study how to solve QC LTC  provably,  via either \emph{planning} (with model knowledge) or \emph{learning} (without model knowledge). The pipeline of our solution is shown in \Cref{fig: LTC process}, and proofs of the results can be found in \S \ref{proof details sec 4}. 
\begin{figure}[!htp]
    \centering
    \includegraphics[width=1\textwidth] {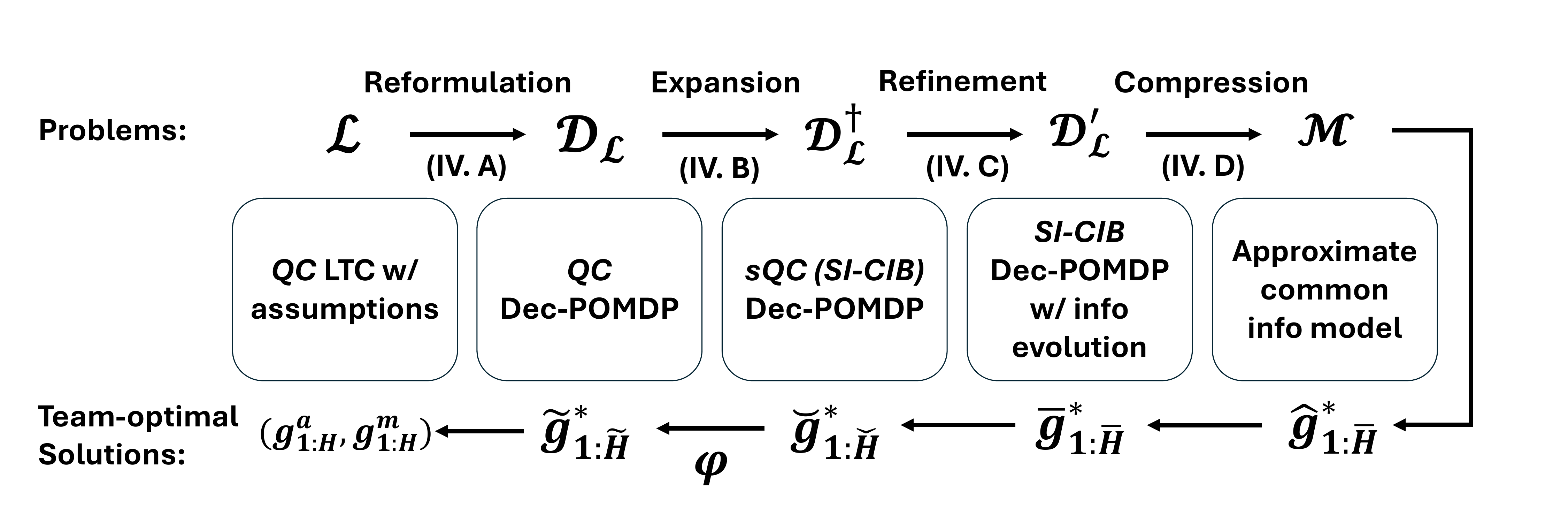}
     \captionsetup{font=small}
     \caption{Illustrating the subroutines 
    for solving the LTC problems.}
    \label{fig: LTC process}
\end{figure}

\subsection{An Equivalent Dec-POMDP}
\label{formulation from LTC to Dec-POMDP}

Given any LTC $\cL$, we can define a Dec-POMDP $\cD_\cL$ characterized  by  $\la\tilde{H},\tilde{\cS},\{\tilde{\cA}_{i,h}\}_{i\in[n],h\in[\tilde{H}]},\{\tilde{\cO}_{i,h}\}_{i\in[n],h\in[\tilde{H}]},\{\tilde{\TT}_h\}_{h\in[\tilde{H}]},\{\tilde{\OO}_h\}_{h\in[\tilde{H}]},\tilde{\mu}_1,
\{\tilde{\cR}_h\}_{h\in[\tilde H]}\ra$,  such that these two are equivalent (under the assumptions in \S\ref{section 3}): 
\begin{flalign}
&\tilde{H}=2H,~~
\tilde{\cS}=\cS,~~
\tilde{s}_{2h-1}=\tilde{s}_{2h}=s_h,~~
\tilde{\cA}_{i,2h-1}=\cM_{i,h},~~
\tilde{\cA}_{i,2h}=\cA_{i,h},~~
\tilde{\cO}_{i,2h-1}=\cO_{i,h},~~
\tilde{\cO}_{i,2h}=\{\emptyset\},~~
\tilde{\mu}_1=\mu_1,\notag\\
&\tilde{\OO}_{2h-1}=\OO_h,~~
\tilde{\TT}_{2h-1}(\tilde{s}_{2h}\given \tilde{s}_{2h-1},\tilde{a}_{2h-1})=\mathds{1}[\tilde{s}_{2h}=\tilde{s}_{2h-1}],~~
\tilde{\TT}_{2h}(\tilde s_{2h+1}\given\tilde s_{2h}, \tilde a_{2h})=\TT_h(\tilde s_{2h+1}\given\tilde s_{2h}, \tilde a_{2h}),\label{LTC to Dec-POMDP}\\
&\tilde{\cR}_{2h-1}=-\cK_{h},~~
\tilde{\cR}_{2h}=\cR_{h},~~
\tilde{p}_{i,2h-1}=\emptyset,~~
\tilde{p}_{i,2h}=p_{i,h^+},~~
\tilde{c}_{2h-1}=c_{h^-},~~
\tilde{c}_{2h}=c_{h^+},~~
\tilde{z}_{2h-1}=z^b_{h},~~
\tilde{z}_{2h}=z^a_{h},\notag 
\end{flalign}
{for all} $(i,h)\in[n]\times[H]$,   $s_h\in \cS,a_{i,h}\in\cA_{i,h},m_{i,h}\in \cM_{i,h},p_{i,h^-}\in\cP_{i,h^-},p_{i,h^+}\in\cP_{i,h^+},c_{h^-}\in\cC_{h^-},c_{h^+}\in\cC_{h^+},\tau_{i,h^-}\in\cT_{i,h^-},\tau_{i,h^+}\in\cT_{i,h^+}$. Note that we follow the convention of $\tilde{\tau}_{i,h}:=\tilde{p}_{i,h}\cup\tilde{c}_h$ for any $h\in[\tilde H]$, and at the odd timestep $2t-1$ for any $t\in[H]$, 
we have $\tilde{p}_{i,2t-1}=\emptyset$ under  Assumption \ref{limited communication strategy}, i.e., in $\cD_\cL$, 
each agent only uses the common information so far for decision-making at timestep $2h-1$. 
Correspondingly, for any $h\in[\tilde{H}],i\in[n]$, we denote by $\tilde{g}_{i,h}, \tilde{g}_{h}$ the agent $i$'s strategy and the joint strategy,  respectively, and denote by $\tilde{\cG}_{i,h},\tilde{\cG}_{h}$ their associated spaces.
Moreover, to unify the presentation, we define that $\forall h\in[H]$,    $\tilde{r}_{2h-1}=\tilde \cR_{2h-1}(\tilde {s}_{2h-1},\tilde{a}_{2h-1},{p}_{2h-1}):=-\cK_h(\phi_h({p}_{2h-1},\tilde{a}_{2h-1}))$, with a slight abuse of notation, where ${p}_{2h-1}:=p_{h^-}$ can be viewed as part of the underlying state. 
Similarly, we define $\tilde{r}_{2h}=\tilde \cR_{2h}(\tilde {s}_{2h},\tilde{a}_{2h},{p}_{2h}):=\cR_h(\tilde {s}_{2h},\tilde{a}_{2h})$, where ${p}_{2h}:=p_{h^+}$. 
Hence, the objective of $\cD_\cL$ is defined as $J_{\cD_\cL}(\tilde{g}_{1:\tilde{H}})=\EE_{\cD_\cL}[\sum_{h=1}^{\tilde{H}}\tilde{r}_h\given \tilde{g}_{1:\tilde{H}}]$.

Essentially, this reformulation splits the $H$-step  control and communication decision-making procedure into a $2H$-step one. 
A similar splitting of the timesteps was also used in \cite{ashutoshcommunicate1,ashutoshcommunicate2}. In comparison, we consider a more general setting, where the state is not decoupled,  and agents are allowed to share the observations and actions at the \emph{previous} timesteps, due to the generality of our  LTC formulation. The equivalence between $\cL$ and $\cD_{\cL}$ is more formally stated as follows.

\begin{proposition}[Equivalence between $\cL$ and $\cD_\cL$]
    \label{prop: equivalence of LTC and Dec-POMDP}
    Let $\cD_\cL$ be the reformulated Dec-POMDP from $\cL$ {satisfying  Assumption \ref{limited communication strategy}}, then the solutions of the two problems are equivalent, in the sense that $\forall g^m_{1:H}\in\cG^m_{1:H}, g^a_{1:H}\in\cG^a_{1:H},i\in[n]$, let $\tilde{g}_{1:\tilde{H}}=(g^m_1,g^a_1,\cdots,g^m_H,g^a_H)$, then $J_{\cD_\cL}(\tilde{g}_{1:\tilde{H}})=J_\cL(g^a_{1:H},g^m_{1:H})$. Also, $\forall \tilde{g}_{1:\tilde{H}}\in \tilde{\cG}_{1:\tilde{H}}$, let $g^m_{1:H}=(\tilde{g}_1, \tilde{g}_3, \cdots,\tilde{g}_{\tilde{H}-1}), g^a_{1:H}=(\tilde{g}_2,\tilde{g}_4,\cdots,\tilde{g}_{\tilde{H}})$, then $J_\cL(g^a_{1:H},g^m_{1:H})=J_{\cD_\cL}(\tilde{g}_{1:\tilde{H}})$.
\end{proposition}

Moreover, the Dec-POMDP $\cD_\cL$ preserves the QC information structure of $\cL$, as shown next.

\begin{theorem}[Preserving  (s)QC]
    \label{reformulate QC}
    If $\cL$ is (s)QC, then the reformulated Dec-POMDP $\cD_\cL$  is also (s)QC.
\end{theorem}

Proofs of \Cref{prop: equivalence of LTC and Dec-POMDP} and \Cref{reformulate QC} can be found in \S\ref{sec:equivalence_appendix} and \S\ref{sec:reformulate_QC_appendix}, respectively. 
By Proposition \ref{prop: equivalence of LTC and Dec-POMDP}, it suffices to solve the reformulated $\cD_{\cL}$ that is QC/sQC, which is our focus next.

\subsection{Strict Expansion of $\cD_\cL$}
\label{formulate QC to SQC}

However, being QC does not necessarily imply $\cD_\cL$ can be solved  \emph{without}  computationally intractable oracles. Note that this is different from the continuous-space, linear quadratic case, where QC problems can be reformulated and solved efficiently \cite{ho1972team,lamperski2015optimal}. With nonlinear,  discrete spaces, the recent results in  \cite{liu2023tractable} established a concrete \emph{quasi-polynomial}-time complexity for planning, 
under the \emph{strategy independence}  assumption \cite{ashutosh2013game} on the common-information-based beliefs \cite{ashutosh2013team,ashutosh2013game}. This SI-CIB assumption was shown critical for \emph{computational tractability} \cite{liu2023tractable}: it eliminates the need to  \emph{enumerate} the past strategies 
in dynamic programming, which would otherwise be prohibitively large. Thus, we may need to connect QC IS to the SI-CIB condition for better computational tractability. 

Interestingly, under certain conditions, one can connect these two conditions for the reformulated Dec-POMDP $\cD_\cL$. 
As the first step, we will \emph{expand} the QC $\cD_{\cL}$ by adding the \emph{actions} of the agents who {influence} the later agents in the intrinsic model of $\cD_{\cL}$ to the shared information. We denote the strictly expanded Dec-POMDP as $\cD_\cL^\dag$. 
We replace the $~\tilde{\phantom{x}}~$ notation in $\cD_\cL$ by the 
$~\Breve{\phantom{x}}~$  notation in $\cD_\cL^\dag$. All the elements  remain the same, except the set of common information $\Breve{c}_h$: 
for any $h\in[\tilde{H}]$
\begin{equation}
\begin{aligned}
&\Breve{c}_h=\tilde{c}_h\cup\left\{\tilde{a}_{j,t}\biggiven j\in[n],t<h,\sigma(\tilde{\tau}_{j,t})\subseteq \sigma(\tilde{c}_h), \tilde{a}_{j,t}
\text{~influences }\tilde{s}_{t+1}\right\}
\end{aligned}
    \label{construction:QC to sQC}
\end{equation}
and we follow the convention to define $\Breve{\tau}_{i,h}:=\Breve{p}_{i,h}\cup  \Breve{c}_h$ and $\Breve{z}_h=\Breve{c}_h\backslash \Breve{c}_{h-1}$. 
It is not hard to verify the following. 

\begin{lemma}
    \label{lemma:QC to sQC} 
       If {$\cD_{\cL}$}  is QC, then $\cD_\cL^\dag$ is sQC.
\end{lemma}

In contrast to the reformulation in \S\ref{formulation from LTC to Dec-POMDP}, the expansion here cannot guarantee the \emph{equivalence} between $\cD_\cL$ and $\cD_\cL^\dag$: the strategy spaces of $\cD_\cL^\dag$ are larger than those of $\cD_\cL$, as each agent can now access more information, i.e.,  $\tilde{\tau}_{i,h}\subseteq\Breve{\tau}_{i,h}$.   Fortunately, the team-optimal value and strategy of both Dec-POMDPs are related, 
as shown in the following theorem. 

\begin{theorem}
    \label{theorem: sQC to QC}
    Let $\cD_\cL$ be the QC Dec-POMDP reformulated from a QC LTC $\cL$,  and $\cD_\cL^\dag$ be the sQC expansion of $\cD_\cL$. Then, for any $\epsilon$-team-optimal strategy $\Breve{g}^\ast_{1:\Breve{H}}$ of $\cD_\cL^\dag$, there exists a function $\varphi$ such that $\tilde{g}_{1:\tilde{H}}^\ast=\varphi(\Breve{g}^\ast_{1:\Breve{H}},\cD_\cL)$ 
    is an $\epsilon$-team-optimal strategy of $\cD_\cL$, with $J_{\cD_\cL}(\tilde{g}^\ast_{1:\tilde{H}})=J_{\cD_\cL^\dag}(\Breve{g}^\ast_{1:\Breve{H}})$.
\end{theorem} 

Theorem \ref{theorem: sQC to QC} shows that one can solve the QC $\cD_{\cL}$ by first solving the sQC expansion $\cD_{\cL}^\dag$, and then using an oracle $\varphi$, as tabulated in \Cref{algorithm varphi}, to translate the solution back as a solution in the strategy spaces of $\cD_{\cL}$, without loss of optimality. {Importantly, we also show  in Algorithm \ref{algorithm Implement varphi}
how to implement such a $\varphi$ function efficiently}.

As shown below, a benefit of obtaining an \emph{sQC} $\cD_{\cL}^\dag$ is that it also has \emph{SI-CIBs}, making it possible to be solved without computationally intractable oracles as in \cite{liu2023tractable}.

\begin{theorem}
    \label{theorem: SI=QC}
    Let $\cD_\cL^\dag$ be an sQC Dec-POMDP generated from $\cL$  after reformulation and strict expansion, then $\cD_\cL^\dag$ has \emph{strategy-independent common-information-based beliefs}   \cite{ashutosh2013game,liu2023tractable}. More formally, for any $h\in[\Breve{H}]$, any two different joint strategies $\Breve{g}_{1:h-1}$ and $\Breve{g}_{1:h-1}'$, and any common information $\Breve{c}_h$ that  can be reached under both $\Breve{g}_{1:h-1}$ and $\Breve{g}_{1:h-1}'$, for any joint private information $\Breve{p}_h\in\Breve{\cP}_h$ and state $\Breve{s}_h\in\Breve{\cS}$, we have
    \begin{equation}
        \PP_h^{\cD_\cL^\dag}(\Breve{s}_h,\Breve{p}_h\given \Breve{c}_h,\Breve{g}_{1:h-1})=\PP_h^{\cD_\cL^\dag}(\Breve{s}_h,\Breve{p}_h\given \Breve{c}_h,\Breve{g}'_{1:h-1}).
    \end{equation}
\end{theorem}

\subsection{Refinement of $\cD_\cL^\dag$}
\label{sec: refinement}

Despite having SI-CIBs, $\cD_\cL^\dag$ is still not eligible for applying the results in \cite{liu2023tractable}: the information evolution rules of $\cD_\cL^\dag$ break those in \cite{ashutosh2013game,liu2023tractable}. 
Specifically, due to Assumption \ref{limited communication strategy}, we set $\tilde{\tau}_{i,2t-1}=\tilde{c}_{2t-1}, \tilde{p}_{i,2t-1}=\emptyset, \forall t\in[H],i\in[n]$ in $\cD_\cL$, which violates Assumption 1 in \cite{ashutosh2013game,liu2023tractable}. 
 To address this issue, 
we propose to further \emph{refine} $\cD_\cL^\dag$ to obtain a Dec-POMDP $\cD_\cL'$, which satisfies the information evolution rules.  We replace the $~\Breve{\phantom{x}}~$ notation in $\cD_\cL^\dag$ by the $~\overline{\phantom{x}}~$ notation in $\cD_\cL'$. The elements in $\cD_\cL'$ remain the same as those in $\cD_\cL^\dag$, except that the private information at odd steps is now refined as: for any $t\in[H],i\in[n], \overline{p}_{i,2t-1}:=p_{i,t^-}$, 
and we define $\overline{\tau}_{i,2t-1}:=\overline{p}_{i,2t-1}\cup \overline{c}_{2t-1}$ for any $t\in[H]$.
Moreover, we define the reward functions as  $\overline{r}_{2t-1}=\overline \cR_{2t-1}(\overline {s}_{2t-1},\overline{a}_{2t-1},\overline{p}_{2t-1}):=-\cK_t(\phi_t(\overline{p}_{2t-1},\overline{a}_{2t-1}))$, and $\overline{r}_{2t}=\overline \cR_{2t}(\overline {s}_{2t},\overline{a}_{2t},\overline{p}_{2t}):=\cR_t(\overline {s}_{2t},\overline{a}_{2t})$,  for any  $t\in[H]$. 
The new Dec-POMDP $\cD_\cL'$ is not equivalent to $\cD_\cL^\dag$ in general, since it enlarges the strategy space at odd timesteps. However, if we define new strategy spaces in $\cD_{\cL}'$ as $\overline{\cG}_{i, 2t-1}:\overline{\cC}_{2t-1}\rightarrow \overline{\cA}_{i,2t-1}, \overline{\cG}_{i,2t}: \overline{\cT}_{i,2t}\rightarrow \overline{\cA}_{i,2t}$ for each $t\in[H],i\in[n]$, and thus define $\overline{\cG}_{1:\overline{H}}$ to be the associated joint strategy space, then solving $\cD_\cL^\dag$ is equivalent to finding a \emph{best-in-class} team-optimal strategy of $\cD_\cL'$ within  $\overline{\cG}_{1:\overline{H}}$, as shown below. 

\begin{theorem} Let $\cD_\cL^\dag$ be an sQC Dec-POMDP generated from the strict expansion of a QC $\cD_\cL$, which is generated from the reformulation of a QC $\cL$, 
    and let $\cD_\cL'$ be the refinement of $\cD_\cL^\dag$ as introduced above.  
Then, finding an  optimal strategy in $\cD_\cL^\dag$ is equivalent to finding an optimal strategy of $\cD_\cL'$ in the space $\overline{\cG}_{1:\overline{H}}$, and $\cD_\cL'$ satisfies the following information evolution rules: for each $h\in[\overline{H}]$: 
    \begin{align*}
        \overline{c}_{h}=\overline{c}_{h-1}\cup \overline{z}_{h},~~ \overline{z}_{h}=\overline{\chi}_{h}(\overline{p}_{h-1},\overline{a}_{h-1},\overline{o}_{h})\qquad 
        \text{for each }i\in[n],~~ \overline{p}_{i,h}=\overline{\xi}_{i,h}(\overline{p}_{i,h-1},\overline{a}_{i,h-1},\overline{o}_{i,h}),
    \end{align*}
    with some functions $\{\overline{\chi}_{h}\}_{h\in[\overline{H}]}, \{\overline{\xi}_{i,h}\}_{i\in[n],h\in[\overline{H}]}$. Furthermore, $\cD'_\cL$ has SI-CIBs with respect to the strategy space  $\overline{\cG}_{1:\overline{H}}$, i.e., for any $h\in[\overline{H}], \overline{s}_h\in\overline{\cS}, \overline{p}_h\in\overline{\cP}_h, \overline{c}_h\in \overline{\cC}_h, \overline{g}_{1:h-1},\overline{g}_{1:h-1}'\in \overline{\cG}_{1:h-1}$ such that $\overline{c}_h$ is reachable under both $\overline{g}_{1:h-1}$ and $\overline{g}_{1:h-1}'$, it holds that 
    \begin{equation}
        \PP_h^{\cD_\cL'}(\overline{s}_h,\overline{p}_h\given \overline{c}_h,\overline{g}_{1:h-1})=\PP_h^{\cD_\cL'}(\overline{s}_h,\overline{p}_h\given \overline{c}_h,\overline{g}'_{1:h-1}).
    \end{equation}
    \label{theorem: refinement}
\end{theorem} 
\vspace{-10mm}

\subsection{Planning in QC LTC with Finite-Time Complexity}
\label{Solution for sQC Dec-POMDP}

Now we focus on solving the Dec-POMDP $\cD_\cL'$ that has SI-CIBs {without computationally intractable oracles},  
building upon our results in  \cite{liu2023tractable}.   
Given a Dec-POMDP $\cD_\cL'$ with SI-CIBs, \cite{liu2023tractable} proposed to construct an $(\epsilon_r, \epsilon_z)$-\emph{expected} approximate common-information model  $\cM$  (as defined in  
\S \ref{proof details sec 4})  through \emph{finite memory} truncation of the \emph{common information},  when the joint emission of $\cD_\cL'$ is $\gamma$-observable. Here, $\epsilon_r$ and $\epsilon_z$ denote the approximation errors for rewards and incremental common information,  respectively, for which we defer a detailed introduction to \S \ref{proof details sec 4}. The effectiveness of finite-memory truncation of  history has also been established for (single-agent) POMDPs in \cite{noah2022gamma,noahlearning,kara2023convergence,kara2022near}.

However, the Dec-POMDP $\cD_\cL'$ obtained from LTC  has two key differences from the general ones 
considered in \cite{liu2023tractable}. 
First, $\cD_\cL'$ does not satisfy the $\gamma$-observability assumption \emph{throughout} the whole $\overline{H}=2H$ timesteps. Fortunately, since the emissions at odd steps are still $\gamma$-observable, while those at even steps are unimportant as the states remain \emph{unchanged} from the previous  step, 
similar results of \emph{belief contraction} and the near-optimality of finite-memory truncation of common information as in \cite{liu2023tractable} can still be obtained (i.e.,  Lemma \ref{lemma: AIS belief close}).  
the reward functions in $\cD_\cL'$ can now depend on the \emph{private information}  $\overline{p}_h$, in addition to  the action $\overline{a}_h$ and the state $\overline{s}_h$. 
Thanks to the existence of some \emph{consistent} approximate common-information-based beliefs  $\{\PP^{\cM,c}_h(\overline{s}_h,\overline{p}_h\given \hat{c}_h)\}_{h\in[\overline{H}]}$ (see Definition \ref{def:consistency}),    
which provide the \emph{joint} probability of $\overline{s}_h$ and $\overline{p}_h$ given the approximate common information $\hat{c}_h$ compressed from $\overline{c}_h$,  we can still properly evaluate the rewards  in the algorithms of \cite{liu2023tractable}.    
Hence, we can leverage 
the approaches  
in \cite{liu2023tractable} to develop a planning algorithm for QC LTC, which 
approximates the optimal strategy $\overline{g}_{1:\overline{H}}^\ast$ by backward induction over the space of $\hat{c}_{h}$. 
See Algorithm \ref{main algorithm} for a detailed introduction to the planning algorithm.  

Note that in each step of the backward induction (Line 6 of \Cref{algorithm under AIS}), a \emph{Team Decision problem}  \cite{teamdecisionhardness} needs to be solved for each $\hat{c}_h$, which is known to be \texttt{NP-hard} in general \cite{teamdecisionhardness}:
\begin{equation}
\left(\hat{g}_{1,h}^\ast(\cdot\given \hat{c}_h,\cdot),\cdots,\hat{g}_{n,h}^\ast(\cdot\given \hat{c}_h,\cdot)\right)\leftarrow \argmax_{\gamma_{h}}Q^{\ast,\cM}_{h}(\hat{c}_h,\gamma_{h}),
\label{equ:one_step_opt1}
\end{equation}
\normalsize 
where the  $Q$-value function and the prescription  $\gamma_h$ are 
defined in \S \ref{proof details sec 4}. Hence, to obtain the overall computational tractability, we make the following \emph{one-step} tractability assumption, as in \cite{liu2023tractable}. 

\begin{assumption}[One-step tractability of $\cM$] 
\label{assu: one_step_tract}
   The one-step Team Decision problems induced by $\cM$  (i.e., 
     Line \ref{line:one_step_opt} of \Cref{algorithm under AIS})  can be solved in polynomial time\footnote{By  \emph{polynomial time}, we here mean that the time-complexity depends polynomially on the LTC parameters $|\cS|,|\cO_h|,|\cA_h|,|\cM_h|,H$.} for all $h=2t$ with $\in[H]$.
\end{assumption}

Several remarks are in order regarding this assumption. First, it can be viewed as a  \emph{minimal} assumption when it comes to computational tractability: even with $H=1$ and no LTC, one-step TDP requires additional structures in order to be solved efficiently. Second, since the Dec-POMDP here is reformulated from an LTC problem  under Assumption \ref{limited communication strategy}, it suffices to only assume one-step tractability for the \emph{control} (i.e., even) steps, 
since at odd steps, the strategies do not use private information, and the one-step TDP can thus be solved in polynomial time by searching the maximizer given each $\hat{c}_h$. 
Third, even without Assumption \ref{assu: one_step_tract}, the SI-CIB property of $\cD_\cL'$ and thus the derivation of \emph{fixed, tractable size} dynamic programs to solve $\cL$ efficiently still hold. Without such efforts, intractably many TDPs may need to be solved, leaving it less hopeful for computational tractability (even under Assumption \ref{assu: one_step_tract}). Finally, such an assumption is satisfied for several classes of Dec-POMDPs with information sharing, see \S\ref{sec:one_step_examples}  for more examples. With this assumption, we can obtain a concrete finite-time complexity guarantee for planning in LTCs as follows. Proof of the theorem can be found in \S\ref{sec:planning_QC_LTC}. 

\begin{theorem}
    \label{theorem: planning}
     Given any QC LTC problem $\cL$ satisfying Assumptions \ref{gamma observability}, \ref{limited communication strategy}, \ref{useless action}, and \ref{weak gamma observability}, we can construct a Dec-POMDP problem $\cD_\cL'$ with SI-CIBs such that for any $\epsilon>0$, any $\epsilon$-team-optimal strategy for $\cD_\cL'$ can be converted into an $\epsilon$-team-optimal strategy for $\cL$, and the following holds. Fix $\epsilon_r, \epsilon_z>0$, and given any $(\epsilon_r, \epsilon_z)$-expected approximate common-information model $\cM$ (see Definition \ref{definition: AIS}) for $\cD_\cL'$ that satisfies Assumption \ref{assu: one_step_tract}, there exists an algorithm  that can compute a  $(2\overline{H}\epsilon_r + \overline{H}^2\epsilon_z)$-team-optimal strategy for the original LTC problem $\cL$ with time complexity $\max_{h\in[\overline{H}]}|\hat{\cC}_h|\cdot \texttt{poly}(|\overline{\cS}|, \max_{h\in[\overline{H}]}|\overline{\cA}_h|, \max_{h\in[\overline{H}]}|\overline{\cP}_h|, \overline{H})$. In particular, for  any fixed $\epsilon>0$, if $\cL$ has a baseline sharing protocol as one of the examples in  \S\ref{sec: examples of QC}, then 
     Algorithm \ref{main algorithm}  can find an $\epsilon$-team-optimal strategy for $\cL$ in quasi-polynomial time. 
\end{theorem}

As sufficient conditions to ensure the construction of such an $\cM$ that satisfies Assumption \ref{assu: one_step_tract}, 
as part of the definition, some examples in \S\ref{sec: examples of QC} may need additional structural assumptions on the transition dynamics, emission, and reward/cost functions, while some do not. See \S\ref{sec: examples of QC} and the
proof of \Cref{thm: full planning} for more detailed discussions.

\subsection{LTC with Finite-Time and Sample Complexities}\label{sec:LTC_learning_results}

{Based on the planning result,  we are now ready to solve the \emph{learning} problem 
with both time and sample complexity guarantees.} 
In particular, we can treat the samples from $\cL$ as the samples from  $\cD_\cL'$: 
the \emph{reformulation} step (\S\ref{formulation from LTC to Dec-POMDP}) does not change the system dynamics, but only maps the information to different random variables; the \emph{expansion} step (\S\ref{formulate QC to SQC}) only requires agents to share more actions with each other, without changing the input and output of the environment; the \emph{refinement} step (\S\ref{sec: refinement}) only recovers the private information the agents had in the original $\cL$. 
This way, we can utilize similar algorithmic ideas in \cite{liu2023tractable} to develop a learning algorithm for LTC problems. See \S \ref{proof details sec 4} for more details of the provable LTC algorithms  adapted from \cite{liu2023tractable}. The algorithm has the following finite-time and sample complexity guarantees. Proof of the theorem can be found in \S\ref{main_results_append_learning}. 

\begin{theorem}\label{thm:learning}
Given any QC LTC problem $\cL$ satisfying Assumptions  \ref{gamma observability}, \ref{limited communication strategy}, \ref{useless action},  and \ref{weak gamma observability}, we can construct a Dec-POMDP problem $\cD_\cL'$ with SI-CIBs. 
Moreover, given any compression functions of common information, there exists an LTC algorithm (\Cref{main learning algorithm}) that learns in $\cD_\cL'$, such that if the learned expected approximate common-information models in the algorithm (Line \ref{line: hat M} in \Cref{main learning algorithm}) satisfy Assumption \ref{assu: one_step_tract}, then an approximate team-optimal strategy for $\cL$ can be learned with high probability, with time and sample complexities polynomial in the parameters of the models.   Specifically, if $\cL$ has a baseline sharing protocol as one of the examples in \S \ref{sec: examples of QC}, then 
 an $\epsilon$-team-optimal strategy for $\cL$ can be learned with high probability,  
with both quasi-polynomial time and sample complexities.
    \label{theorem: learning}
\end{theorem}

Again, as sufficient conditions to ensure the learned models above to satisfy Assumption \ref{assu: one_step_tract}, some examples in \S \ref{sec: examples of QC} were defined with additional structural assumptions on the transition,  emission, and reward/cost functions, while some need not.  See \S \ref{sec: examples of QC} and the proof of \Cref{theorem: full learning} for more detailed discussions.

\section{Solving General QC Dec-POMDPs}
\label{general Dec-POMDP}

In \S\ref{sec:positive_results}, we developed a pipeline for solving a special class of QC Dec-POMDPs generated by LTCs, by transforming them into those \emph{with SI-CIBs}. 
In fact, the pipeline can also be extended to solving general QC Dec-POMDPs,  which thus advances the results in \cite{liu2023tractable} that can only address Dec-POMDPs {with SI-CIBs}, a result of independent interest.  
Without much confusion given the context, we will adopt the notation for LTCs to studying  general Dec-POMDPs: we set $h^+=h^-=h$, remove the additional sharing protocol, and add $\overline{\phantom{x}}$ for all the notation in the Dec-POMDP, following that in  $\cD_{\cL}'$ in \S\ref{sec:positive_results}. 
{We extend the results in \S\ref{sec:positive_results} to general Dec-POMDPs as follows.

\begin{theorem}
    \label{theorem: SI=sQC general}
    Consider a Dec-POMDP $\cD$ satisfying Assumption \ref{assumption:information evolution} (e). If $\cD$ is sQC and satisfies Assumptions  \ref{assumption:sigma_include}, \ref{useless action}, and \ref{weak gamma observability}, then it has SI-CIBs. Meanwhile, if $\cD$ has SI-CIBs and perfect recall, then it is sQC (up to null sets).  
\end{theorem}

Perfect recall \cite{kuhn1953extensive} here  means that the agents will never forget their own past information and actions (as formally defined in \S \ref{proof details sec 5} for completeness).} 
Note that Assumption  \ref{assumption:information evolution} (e) is similar to,  but different from,  perfect recall: it is implied by the latter with $\overline{o}_{i,h}\in \overline{\tau}_{i,h}$. Also, Assumptions \ref{assumption:sigma_include},  \ref{useless action}, and  \ref{weak gamma observability}  were originally made for LTCs, and here we meant to impose them for Dec-POMDPs with $h^+=h^-=h$. 
 Finally, by \emph{sQC up to null sets}, we meant that if agent $(i_1,h_1)$ influences agent $(i_2,h_2)$ in the intrinsic model of the Dec-POMDP $\cD$, then under any strategy $\overline{g}_{1:\overline{H}}$, $\sigma(\overline{\tau}_{i_1,h_1})\subseteq\sigma(\overline{\tau}_{i_2,h_2})$ and $\sigma(\overline{a}_{i_1,h_1})\subseteq\sigma(\overline{\tau}_{i_2,h_2})$  except the null sets generated by $\overline{g}_{1:\overline{H}}$.  
Given Theorem \ref{theorem: SI=sQC general} and the results in \S\ref{sec:positive_results}, we can illustrate the relationship between LTCs and Dec-POMDPs with  different assumptions and information structures in \Cref{fig:venn_diagram}, which may be of independent interest. 

\begin{figure}[!t]
    \vspace{7pt}
    \centering
    \begin{subfigure}{0.39\linewidth}
        \centering
        \includegraphics[width=0.8\linewidth]{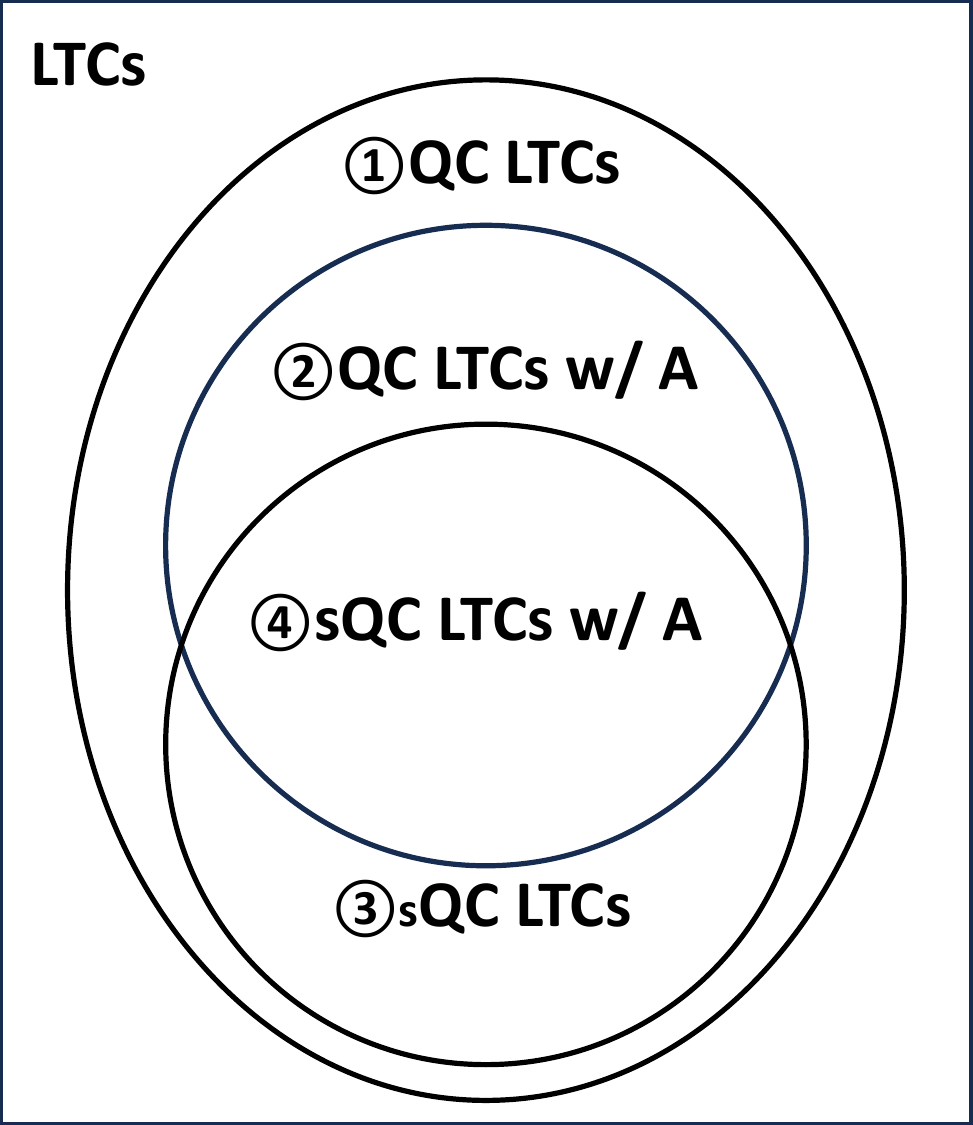}
        \caption{}
        \label{fig:Venn_LTC}
    \end{subfigure}
    \hspace{16pt}
    \begin{subfigure}{0.39\linewidth}
        \centering
        \includegraphics[width=0.8\linewidth]{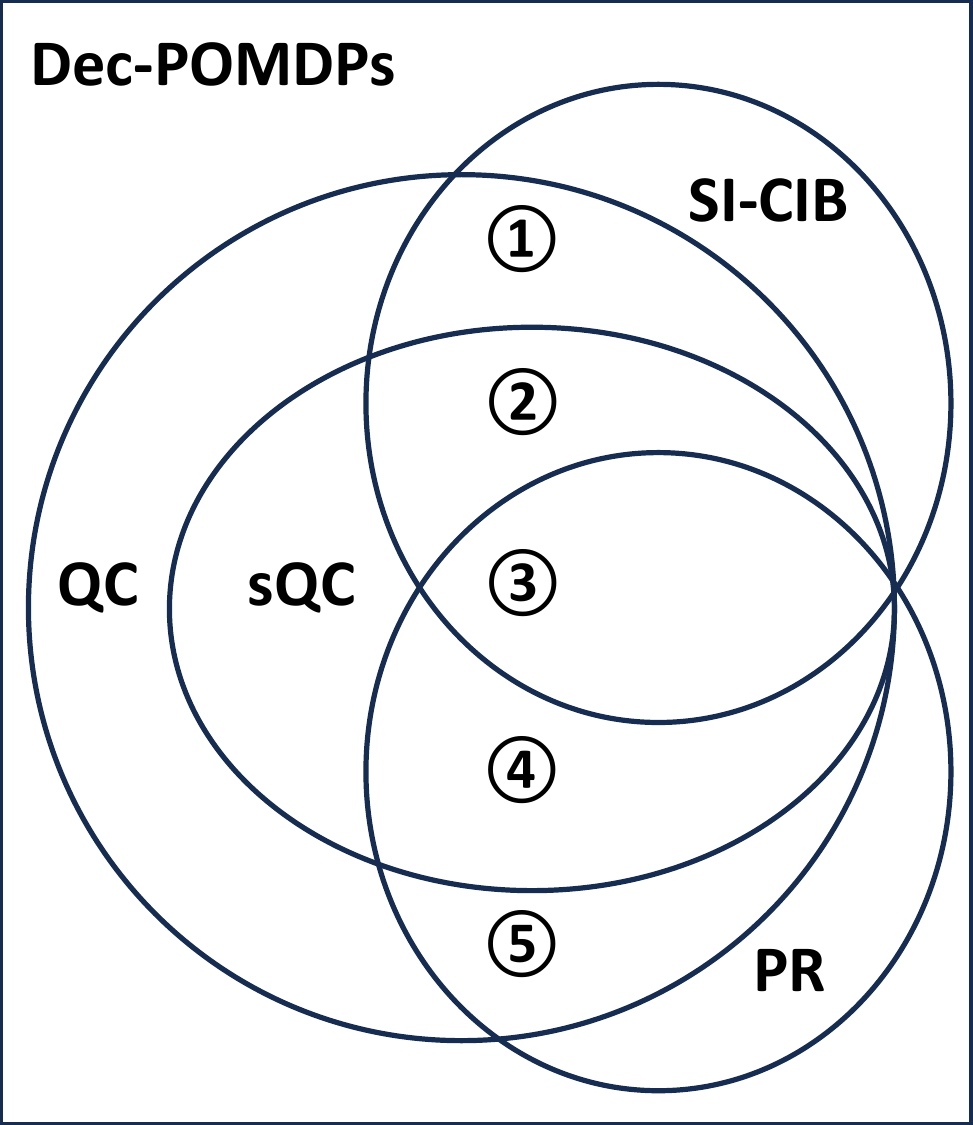}
        \caption{}
        \label{fig:Venn_DecPOMDP}
    \end{subfigure}
\caption{
    (a) Venn diagram of  LTCs  with different ISs:  \ding{172} QC LTCs. \ding{173} QC LTCs satisfying Assumptions \ref{limited communication strategy}, \ref{useless action}, and \ref{weak gamma observability}. \ding{174} sQC LTCs. \ding{175} sQC LTCs satisfying Assumptions \ref{limited communication strategy}, \ref{useless action}, and \ref{weak gamma observability}, 
    which can be converted to Dec-POMDPs with SI-CIBs; (b) Venn diagram of general Dec-POMDPs with different ISs. \emph{PR} denotes \emph{perfect recall}. We construct examples for each area of \ding{172}-\ding{176} in \S \ref{sec:examples_venn_diag}.} 
    \label{fig:venn_diagram}
\end{figure}

By \Cref{theorem: SI=sQC general}, one may start with a QC Dec-POMDP (with information sharing) $\cD$ that does \emph{not} necessarily have SI-CIBs, and then expand it (as in \S\ref{formulate QC to SQC}) to obtain an sQC Dec-POMDP $\cD'$. If $\cD$ satisfies Assumptions \ref{assumption:sigma_include}, \ref{useless action}, and \ref{weak gamma observability}, then $\cD'$ has SI-CIBs and can then be solved with finite-time and sample complexity guarantees as in  \cite{liu2023tractable}. 
For instance, if $\cD$ has an information-sharing protocol as either  \textbf{Example 7} or \textbf{Example 8} in \S \ref{sec: examples of QC},   
one can verify that it does not have  SI-CIBs. However,  
with the above assumptions,  
an approximate team-optimal strategy of $\cD$ can still be learned with quasi-polynomial sample complexity, and without computationally intractable oracles.

\section{Experimental  Results}
\label{experiments}

In this section, we  demonstrate both the practical implementability and performance of our LTC  algorithms via numerical experiments, and conduct ablation studies for LTC problems with different communication costs and horizons. 
\paragraph{Environment setup.} We conduct our experiments on two popular and modest-scale partially observable benchmarks, Dectiger \cite{dectiger} and Grid3x3 \cite{grid3x3}. We train the agents in each LTC problem in the two environments with $20$ different random seeds and different communication cost functions, and execute them in problems with horizons  $\{4,6,8,10\}$.  To fit the LTC setting considered in our paper, we regularize the rewards to $[0,1]$, and set the baseline sharing protocol as the one-step delayed information sharing \cite{ashutosh2013game,liu2023tractable}. As for the communication costs, we set $\cK_{h}(z_h^a)=\alpha\cdot  |z_h^a|$, and set $\alpha\in\{0.01,0.05,0.1\}$ for  ablation studies. Also, we consider 2 baseline cases under the same environment with the information structure of one-step delayed information sharing and fully sharing, respectively. The former can be viewed as an LTC problem with extremely high communication cost, thus no additional sharing; the latter  corresponds to an LTC problem with no communication cost.

\begin{table}[!t]
\centering
\resizebox{0.8\textwidth}{!}{
    \begin{tabular}{|l|l|l|l|l|l|}
    \toprule
    Horizon/Cost & No Sharing & Cost=0.1 & Cost=0.05 & Cost=0.01 & Fully Sharing \\
    \midrule
    H=4 w/ cost & 1.32$\pm$0.025 & 1.33$\pm$0.044 & 1.44$\pm$0.034 & 1.54$\pm$0.013 & 1.57$\pm$0.004 \\
    \midrule
    H=4 w/o cost & -     & 1.36$\pm$0.032 & 1.48$\pm$0.034 & 1.59$\pm$0.002 & - \\
    \midrule
    H=6 w/ cost & 1.95$\pm$0.009 & 1.97$\pm$0.07 & 2.08$\pm$0.068 & 2.26$\pm$0.012 & 2.29$\pm$0.002 \\
    \midrule
    H=6 w/o cost & -     & 2.01$\pm$0.047 & 2.14$\pm$0.072 & 2.27$\pm$0.011 & - \\
    \midrule
    H=8 w/ cost & 2.56$\pm$0.041 & 2.64$\pm$0.078 & 2.74$\pm$0.118 & 2.96$\pm$0.021 & 3.0$\pm$0.002 \\
    \midrule
    H=8 w/o cost & -     & 2.7$\pm$0.044 & 2.83$\pm$0.117 & 2.98$\pm$0.02 & - \\
    \midrule
    H=10 w/ cost & 3.31$\pm$0.024 & 3.37$\pm$0.135 & 3.51$\pm$0.153 & 3.69$\pm$0.029 & 3.87$\pm$0.007 \\
    \midrule
    H=10 w/o cost & -     & 3.46$\pm$0.069 & 3.63$\pm$0.152 & 3.71$\pm$0.026 & - \\
    \bottomrule
    \end{tabular}%
    }
  \caption{Experimental results for Dectiger with  different horizon lengths and cost functions.}
  \label{tab:dectiger}%
\end{table}
\begin{table}[!t]
\centering
\resizebox{0.8\textwidth}{!}{
\begin{tabular}{|l|l|l|l|l|l|}
    \toprule
    Horizon/Cost & No Sharing & Cost=0.1 & Cost=0.05 & Cost=0.01 & Fully Sharing \\
    \midrule
    H=4 w/ cost & 0.14$\pm$0.003 & 0.14$\pm$0.019 & 0.15$\pm$0.002 & 0.26$\pm$0.028 & 0.48$\pm$0.023 \\
    \midrule
    H=4 w/o cost & -     & 0.14$\pm$0.019 & 0.21$\pm$0.007 & 0.33$\pm$0.023 & - \\
    \midrule
    H=6 w/ cost & 0.33$\pm$0.02 & 0.32$\pm$0.025 & 0.4$\pm$0.009 & 0.48$\pm$0.059 & 0.38$\pm$0.075 \\
    \midrule
    H=6 w/o cost & -     & 0.32$\pm$0.025 & 0.54$\pm$0.02 & 0.62$\pm$0.075 & - \\
    \midrule
    H=8 w/ cost & 0.52$\pm$0.084 & 0.52$\pm$0.051 & 0.58$\pm$0.072 & 0.67$\pm$0.031 & 0.4$\pm$0.022 \\
    \midrule
    H=8 w/o cost & -     & 0.52$\pm$0.051 & 0.72$\pm$0.035 & 0.82$\pm$0.074 & - \\
    \midrule
    H=10 w/ cost & 0.73$\pm$0.02 & 0.73$\pm$0.037 & 0.9$\pm$0.169 & 1.03$\pm$0.019 & 0.15$\pm$0.188 \\
    \midrule
    H=10 w/o cost & -     & 0.73$\pm$0.037 & 1.08$\pm$0.14 & 1.25$\pm$0.062 & - \\
    \bottomrule
    \end{tabular}%
  }
  \caption{Experimental results for Grid3x3  with  different horizon lengths and cost functions.}
  \label{tab:grid3x3}%
\end{table}
\paragraph{Results and analysis.}
The results of different horizons and communication  costs over 20 random seeds are shown in \Cref{tab:dectiger} and \Cref{tab:grid3x3}. 
Additionally, the attained values are presented in \Cref{fig:bar}, and the learning curves are shown in \Cref{fig:line}. The results show that  communication is beneficial for agents to obtain higher values with better sample efficiency. Also, lower communication costs can encourage agents to share more information,  and thereby achieve a better joint strategy in terms of the values attained by the team.

\begin{figure}[!t]
    \centering
        \begin{subfigure}{0.4\linewidth}
            \centering
            \includegraphics[width=\linewidth]{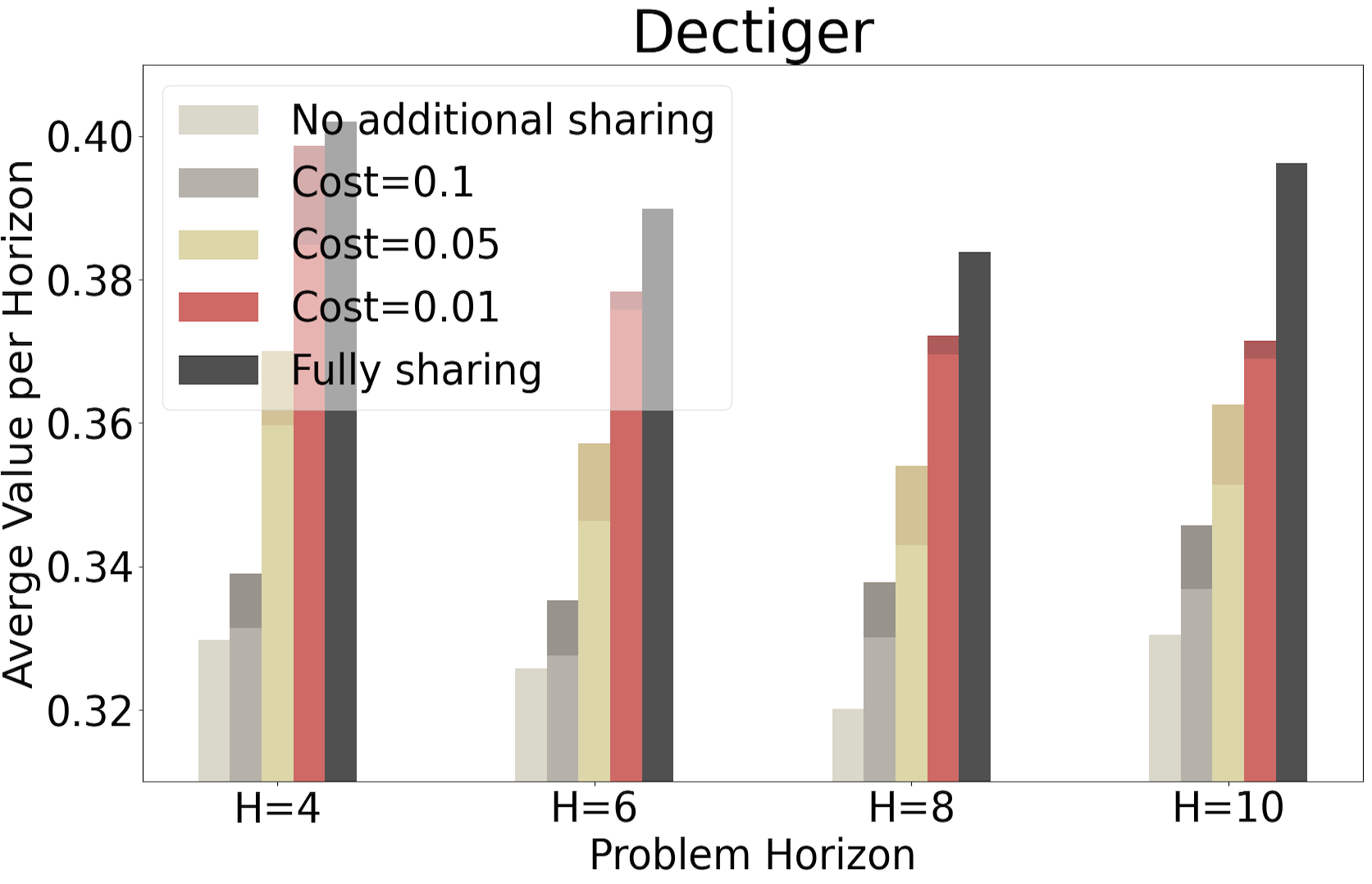}
        \end{subfigure}
        \hspace{10mm}
        \begin{subfigure}{0.4\linewidth}
            \centering
            \includegraphics[width=\linewidth]{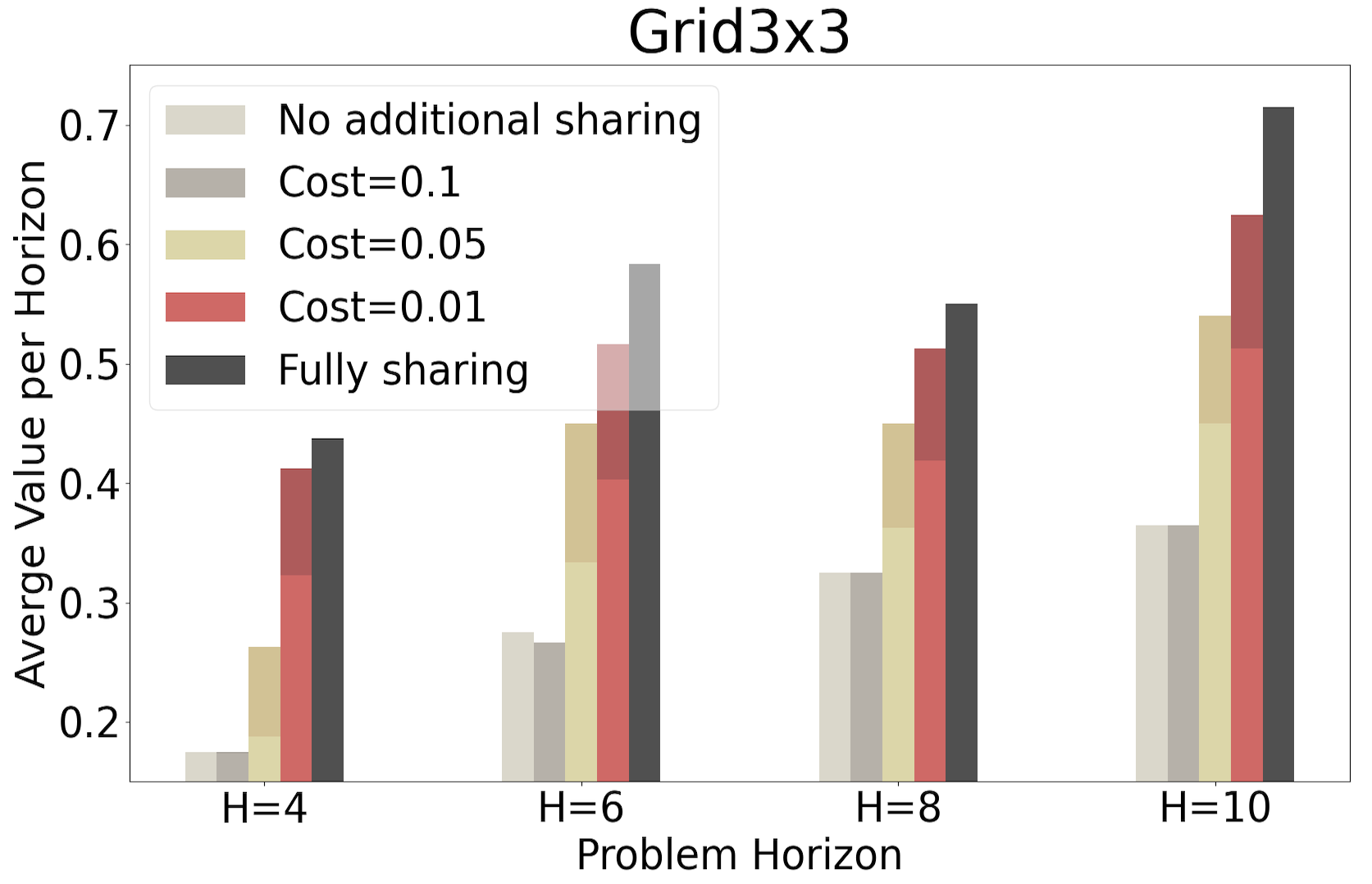}
        \end{subfigure}
    \caption{The time-average values achieved under different communication costs and horizons. 
    For each bar, the dark portion and the light portion correspond to the values associated with the communication cost  and the overall objective (reward minus cost) of the agents, respectively; the full bar corresponds to the values associated with the reward.}
    \label{fig:bar}
\end{figure}
\begin{figure}[!t]
    \centering
        \begin{subfigure}{0.4\linewidth}
            \centering
            \includegraphics[width=\linewidth]{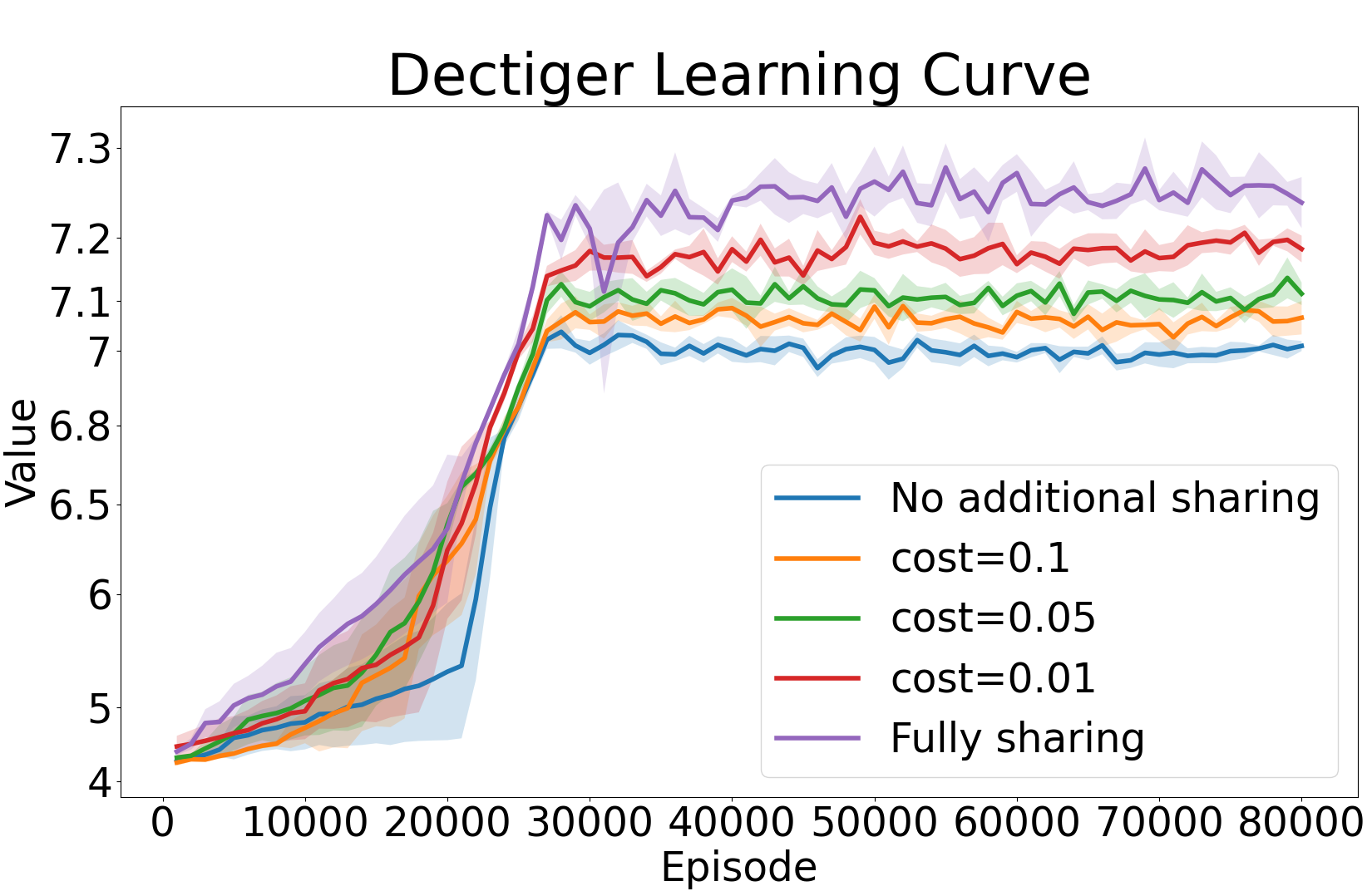}
        \end{subfigure}
        \hspace{10mm}
        \begin{subfigure}{0.4\linewidth}
            \centering
            \includegraphics[width=\linewidth]{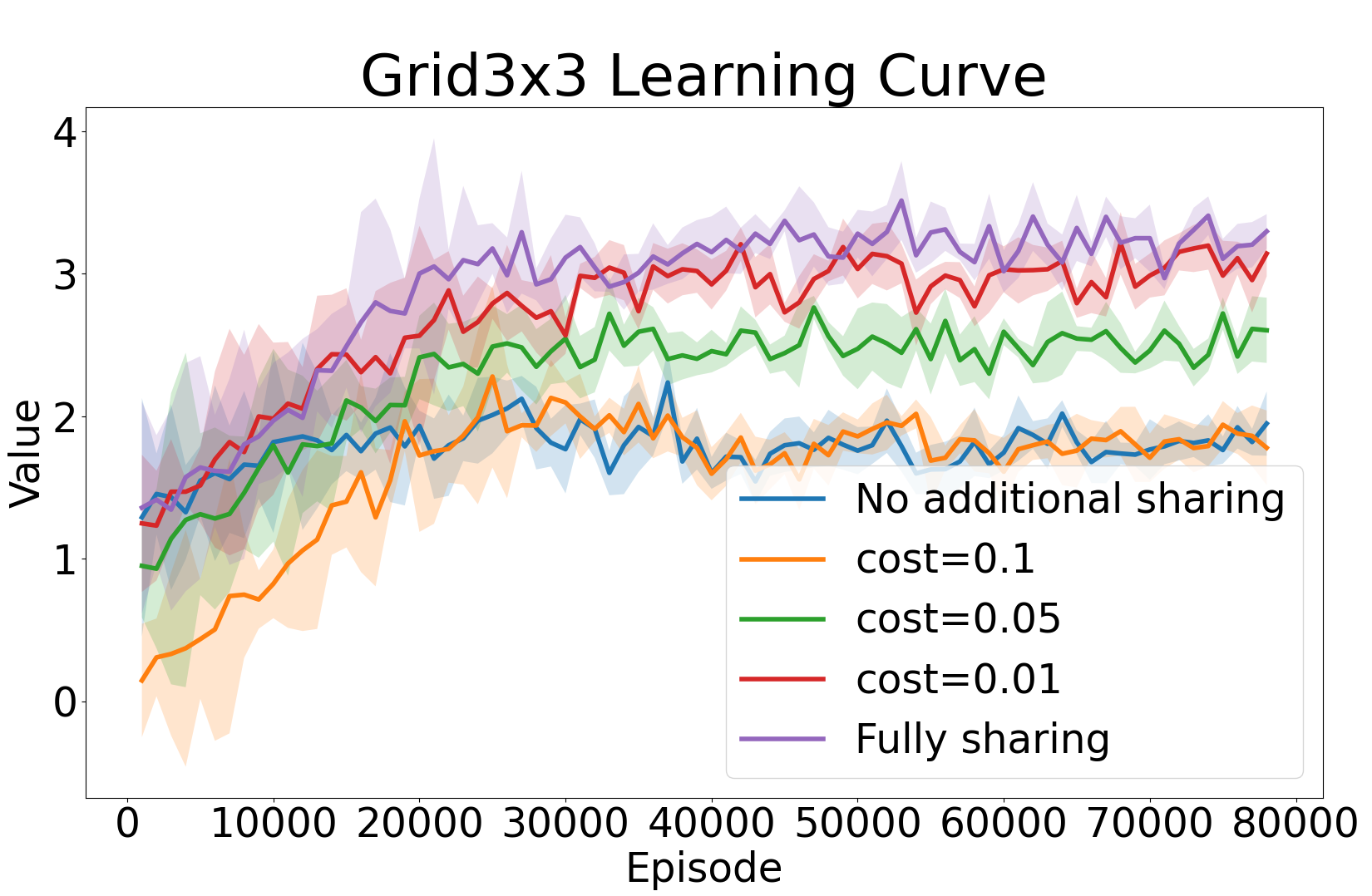}
        \end{subfigure}
    \caption{%
    Learning curves, i.e., the values associated with the overall objective (reward minus cost) achieved during learning, under different communication costs.
    }
    \label{fig:line}
\end{figure}

\section{Concluding Remarks}

We formalized the learning-to-communicate problem under the common-information-based Dec-POMDP framework, and proposed a few structural assumptions for LTCs with quasi-classical information structures, under which LTCs preserve the QC IS after information sharing, whereas violating them can 
cause computational hardness in general. We then developed provable planning and learning algorithms for QC LTCs.  
Along the way, we also established a relationship between 
the strictly quasi-classical information structure and the condition of having strategy-independent common-information-based beliefs, as well as solving general Dec-POMDPs without computationally intractable oracles, beyond those with SI-CIBs.   
Our work has opened up many future directions, including the formulation, together with the development of provable planning and learning algorithms, of LTC in \emph{non-cooperative} (game-theoretic) settings, and the relaxation of (some of) the structural assumptions when it comes to equilibrium computation.

\section*{Acknowledgment}
The authors thank the anonymous reviewers of IEEE CDC 2025 for the  feedback, and Aditya Mahajan for the valuable discussion. The authors also acknowledge  
the support from the Army Research Office grant   W911NF-24-1-0085 and the NSF CAREER Award 2443704. K.Z. additionally  acknowledges the support from the AFOSR YIP Award FA9550-25-1-0258, an AI
Safety Research Award from Coefficient Giving, and a JP Morgan Faculty Research Award.  

\bibliographystyle{IEEEtran}
\bibliography{IEEEabrv,reference}

\clearpage
\onecolumn
\appendix

~\\
\centerline{{\fontsize{14}{14}\selectfont \textbf{Appendices}}}

\section{Examples of QC LTC Problems}
\label{sec: examples of QC}
 
In this section, we introduce $8$ examples of QC LTC problems, with 4 of them being extended from the information structures of the baseline sharing protocols  considered in the literature \cite{ashutosh2013game,liu2023tractable}. It can be shown that LTC with any of these 8 examples as baseline sharing is QC.
\begin{itemize}
    \item \textbf{Example 1: One-step delayed information sharing.} At timestep $h\in[H]$, agents will share all the action-observation history in the private information until timestep $h-1$. Namely, for any $h\in[H],i\in[n], c_{h^-}=c_{(h-1)^+}\cup\{o_{h-1},a_{h-1}\}$ and $p_{i,h^-}=\{o_{i,h}\}$.
    \item \textbf{Example 2: State controlled by one controller with asymmetric delayed information sharing.}  The state transition dynamics are controlled by only one agent (without loss of generality, agent 1), i.e., $\forall h\in[H], \TT_h(\cdot\given s_h,a_{1,h},a_{-1,h})=\TT_h(\cdot\given s_h,a_{1,h},a_{-1,h}')$ for all $s_h\in\cS,a_{1,h}\in\cA_{1,h},a_{-1,h}, a_{-1,h}'\in\cA_{-1,h}$. Agent 1 will share all of her information immediately, while others will share their information with a delay of $d\ge 1$ timesteps\footnote{Throughout this paper, we view the delay $d$ as a \emph{constant}, although our final bounds in \S\ref{Solution for sQC Dec-POMDP} and \S\ref{sec:LTC_learning_results} also apply for $d=\text{poly}\log H$. See the proofs in \S\ref{proof details sec 4}  for more discussions.} 
    in the baseline sharing. Namely, for any $h\in[H],i\neq 1$, $c_{h^-}=c_{(h-1)^+}\cup\{a_{1,h-1},o_{1,h},o_{-1,h-d}\},p_{1,h^-}=\emptyset, p_{i,h^-}=(p_{i,(h-1)^+}\cup\{o_{i,h}\})\backslash\{o_{i,h-d}\}$.  
    \item \textbf{Example 3: Information sharing with one-directional-one-step-delay.} For convenience, we assume there are 2 agents, and this example can be readily generalized to the multi-agent case. In this case, agent 1 will share the information  immediately, while agent 2 will share information with one-step delay. Namely,  $c_{1^-}=\{o_{1,1}\}, p_{1,1^-}=\emptyset, p_{2,1^-}=\{o_{2,1}\}$; for any $h\ge 2, c_{h^-}=c_{(h-1)^+}\cup\{o_{1,h},o_{2,h-1},a_{h-1}\}, p_{1,h^-}=\emptyset,p_{2,h^-}=\{o_{2,h}\}$.
    \item \textbf{Example 4: Uncontrolled state process.} The state transition dynamics do not depend on the actions of agents, i.e., $\forall h\in[H], \TT_h(\cdot\given s_h,a_h)=\TT_h(\cdot\given s_h,a_h')$ for any $s_h\in\cS,a_h, a_h'\in\cA_h$. All agents will share their information with a delay of $d\ge 1$ timesteps. For any $h\in[H],i\in [n]$, $c_{h^-}=c_{(h-1)^+}\cup\{o_{h-d}\}, p_{i,h^-}=(p_{i,(h-1)^+}\cup\{o_{i,h}\})\backslash\{o_{i,h-d}\}$.
    \item \textbf{Example 5: One-step delayed observation sharing.} At timestep $h\in[H]$, each agent has access to the observations of all agents until timestep $h-1$ and her present observation. Namely, for any $h\in[H], i\in[n], c_{h^-}=c_{(h-1)^+}\cup\{o_{h-1}\}$ and $p_{i,h^-}=\{o_{i,h}\}$.
    \item \textbf{Example 6: One-step delayed observation and two-step delayed action sharing.} At timestep $h\in[H]$, each agent will share the observation history until timestep $h-1$ and action  history until timestep $h-2$ from the private information. Namely, for any $h\in[H], i\in[n], c_{h^-}=c_{(h-1)^+}\cup\{o_{h-1},a_{h-2}\}, p_{i,h^-}=\{o_{i,h},a_{i,h-1}\}$. 
    \item  \textbf{Example 7: State controlled by one controller with asymmetric delayed observation sharing.} The state transition dynamics are controlled by only one agent (i.e., the same as {\bf Example 2}). Agent 1 will share all of her observations immediately, while others will share their observations with a delay of $d\ge 1$ timesteps. Namely, for any $h\in[H], i\neq 1, c_{h^-}=c_{(h-1)^+}\cup \{o_{1,h},o_{-1,h-d}\}, p_{1,h^-}=\emptyset, p_{i,h^-}=(p_{i,(h-1)^+}\cup \{o_{i,h}\})\backslash\{o_{i,h-d}\}$.
    \item \textbf{Example 8:  State controlled by one controller with asymmetric delayed observation and two-step delayed action sharing.}  The state transition dynamics are controlled by only one agent (i.e.,   the same as {\bf Example 2}). At timestep $h\in[H]$, agent 1 will share all of her observations immediately and her action history until timestep $h-2$, while others will share their observations with a delay of $d\ge 1$ timesteps. Namely, for any $h\in[H], i\neq 1, c_{h^-}=c_{(h-1)^+}\cup\{o_{1,h}, a_{1,h-2},o_{-1,h-d}\}, p_{1,h^-}=\{a_{1,h-1}\}, p_{i,h^-}=(p_{i,(h-1)^+}\cup\{o_{i,h}\})\backslash\{o_{i,h-d}\}$. 
    \item For computational consideration  (i.e., when it comes to \S\ref{sec:positive_results}),
    we may need additional conditions as part of the examples' definitions: 
    \begin{itemize}
        \item {\bf Examples 2, 4, 7, 8:} for any $h\in[H]$, the reward function $\cR_h$ has an additive form, i.e., $\cR_h(s_h,a_h)=\sum_{i=1}^n\cR_{i,h}(s_h,a_{i,h})$.
        \item {\bf Examples 1, 5, 6:} either one of the following conditions holds:  
        \begin{enumerate}
            \item For any $h\in[H]$, the state $s_h$ can be  controlled by one agent $ct(h)$, i.e., $\TT_h:\cS\times\cA_{ct(h),h}\rightarrow \Delta(\cS)$. Furthermore, the reward function has an additive form, i.e., $\cR_h(s_h,a_h)=\sum_{i=1}^n\cR_{i,h}(s_h,a_{i,h})$, and the increment of the common information satisfies $z_{h+1}^b=\chi_{h+1}(p_{h^+}, a_{ct(h),h}, o_{h+1})$; \label{cond:tract_1}
            \item For any $h\in[H]$, the state $s_h$ can be partitioned into $n$ local states as $s_h=(s_{1,h},s_{2,h},\cdots, s_{n,h})$, and  the transition kernel and observation emission have the factorized  forms of $\TT_h(s_{h+1}\given s_h, a_h)=\prod_{i=1}^n \TT_{i,h}(s_{i,h+1}\given s_{i,h}, a_{i,h}), \OO_h(o_h\given s_h)=\prod_{i=1}^n \OO_{i,h}(o_{i,h}\given s_{i,h})$. Furthermore, the communication cost and reward functions can be  decoupled as $\cK_h(z_h^a)=\sum_{i=1}^n \cK_{i,h}(z_{i,h}^a), \cR_h(s_h,a_h)=\sum_{i=1}^n\cR_{i,h}(s_{i,h},a_{i,h})$; \label{cond:tract_2}
            \item This condition is only applicable to {\bf Examples 1} and  {\bf 5}. For any $h\in[H]$, the emission $\OO_{h}$ satisfies that: $\forall i<j\le n, o_{j,h}\in \cO_{j,h}, \exists o_{i,h}\in\cO_{i,h}$ such that $\forall s_h\in\cS, \OO_{\{i,j\},h}(o_{i,h}, o_{j,h}\given s_h)=\OO_{j,h}(o_{j,h}\given s_h)$, where  $\OO_{\{i,j\},h}$ is the marginalized distribution of $\OO_h$ with respect to agents $i$ and $j$. \label{cond:tract_3}
        \end{enumerate}
        \item {\bf Example 3:} we do not need any additional condition. 
    \end{itemize}
\end{itemize}

\begin{remark}
    The additional conditions  on $\TT_h,\OO_h,\cK_h,\cR_h$ are \emph{sufficient conditions} to ensure the one-step tractability of the backward induction when solving the LTC problem (see Assumption \ref{assu: one_step_tract}). Note that the (s)QC  information structures of these examples are unchanged even without these additional structural conditions.  See the proofs in \S\ref{proof details sec 4} for more details. 
\end{remark}
In fact, the first 4 examples are all sQC LTC problems, while the other 4 examples are QC but may not be sQC problems, as shown in the following lemma.    

\begin{lemma}
    Given an LTC problem $\cL$. If the baseline sharing of $\cL$ is one of the first 4 examples above, then $\cL$ is sQC. If the baseline sharing of $\cL$ is one of the last 4 examples above, then $\cL$ is QC but may not be sQC.
\end{lemma}
\begin{proof}
    Let $\overline{\cD}_{\cL}$ denote  the Dec-POMDP induced by $\cL$ (see \Cref{def:LTC_induced_Dec-POMDP}). We prove this lemma case by case. For convenience, we use $~\check{}~$ for  the notation of the variables/quantities in $\overline{\cD}_{\cL}$.
    \begin{itemize}
        \item \textbf{Example 1:} The information in $\overline{\cD}_\cL$ evolves as $\forall h\in[H], i\in[n], \check c_h=\{\check o_{1:h-1},\check a_{1:h-1}\}$ and $\check p_{i,h}=\{\check o_{i,h}\}$. Then, for any $i_1,i_2\in[n],h_1,h_2\in [H], h_1<h_2$,  $\check \tau_{i_1,h_1}=\{\check o_{1:h_1-1},\check a_{1:h_1-1},\check o_{i_1,h_1}\}\subseteq \check c_{h_1+1}\subseteq \check c_{h_2}\subseteq \check \tau_{i_2,h_2}$, and $\check a_{i_1,h_1}\subseteq \check c_{h_1+1}\subseteq \check c_{h_2}\subseteq \check \tau_{i_2,h_2}$. Therefore, we have $\sigma(\check \tau_{i_1,h_1})\subseteq \sigma(\check \tau_{i_2,h_2})$, and thus $\cL$ is sQC.
        \item \textbf{Example 2:} The information in $\overline{\cD}_\cL$ evolves as $\forall h\in[H], i\neq 1, \check c_h=\{\check a_{1,1:h-1},\check o_{1,1:h-1},\check o_{-1,1:h-d}\}, \check p_{1,h}=\emptyset, \check p_{i,h}=\{o_{i,h-d+1:h}\}$. Then, for any $i_1,i_2\in[n],h_1,h_2\in [H], h_1<h_2$. If $i_1\neq 1$, then agent $(i_1,h_1)$ will not influence agent $(i_2,h_2)$. If $i_1=1$, then $\check \tau_{i_1,h_1}=\{\check o_{1, 1:h_1},\check a_{1, 1:h_1-1},\check o_{-1,1:h_1-d}\}\subseteq \check c_{h_1+1}\subseteq \check c_{h_2}\subseteq \check \tau_{i_2,h_2}$, and $\check a_{i_1,h_1}\subseteq \check c_{h_1+1}\subseteq \check c_{h_2}\subseteq \check \tau_{i_2,h_2}$. Therefore, we have $\sigma(\check \tau_{i_1,h_1})\subseteq \sigma(\check \tau_{i_2,h_2})$ if agent $(i_1,h_1)$ influences agent $(i_2,h_2)$, and thus $\cL$ is sQC.
        \item \textbf{Example 3:} The information in $\overline{\cD}_\cL$ evolves as $\forall h\in[H], \check c_h=\{\check o_{1:h-1},\check a_{1:h-1},\check o_{1,h}\}$ and $\check p_{1,h}=\emptyset, \check p_{2,h}=\{\check o_{2,h}\}$. Then, for any $i_1,i_2\in[n],h_1,h_2\in [H], h_1<h_2$, $\check a_{i_1,h_1}\subseteq \check c_{h_1+1}\subseteq \check c_{h_2}\subseteq \check \tau_{i_2,h_2}$. If $i_1=1$, then $\check \tau_{i_1,h_1}=\{\check o_{1:h_1-1},\check a_{1:h_1-1},\check o_{1,h_1}\}\subseteq \check c_{h_1+1}\subseteq \check c_{h_2}\subseteq \check \tau_{i_2,h_2}$. If $i_1=2$, then $\check \tau_{i_1,h_1}=\{\check o_{1:h_1},\check a_{1:h_1-1}\}\subseteq \check c_{h_1+1}\subseteq \check c_{h_2}\subseteq \check \tau_{i_2,h_2}$. Therefore, we have $\sigma(\check \tau_{i_1,h_1})\subseteq \sigma(\check \tau_{i_2,h_2})$, and thus $\cL$ is sQC.
        \item \textbf{Example 4:} Since in $\overline{\cD}_\cL$, for any $i_1,i_2\in[n],h_1,h_2\in [H]$, agent $(i_1,h_1)$ does not influence agent $(i_2,h_2)$, then $\cL$ is sQC.
        \item  \textbf{Example 5:} The information in $\overline{\cD}_\cL$ evolves as $\forall h\in[H], i\in[n], \check c_h=\{\check o_{1:h-1}\}$ and $\check p_{i,h}=\{\check o_{i,h}\}$. Then, for any $i_1,i_2\in[n],h_1,h_2\in [H], h_1<h_2$,  $\check \tau_{i_1,h_1}=\{\check o_{1:h_1-1}, \check o_{i_1,h_1}\}\subseteq \check c_{h_1+1}\subseteq \check c_{h_2}\subseteq \check \tau_{i_2,h_2}$. Hence, $\cL$ is QC. However, for any $i_1, i_2\in[n], h_1<h_2\le H$, agent $(i_1,h_1)$ may influence agent $(i_2,h_2)$ but $\sigma(\check a_{i_1,h_1})\nsubseteq\sigma(\check \tau_{i_2, h_2})$. Hence, $\cL$ may not be sQC.
        \item \textbf{Example 6:} The information in $\overline{\cD}_\cL$ evolves as $\forall h\in[H], i\in[n], \check c_h=\{\check o_{1:h-1}, \check a_{1:h-2}\}$ and $\check p_{i,h}=\{\check o_{i,h}, \check a_{i,h-1}\}$. Then, for any $i_1,i_2\in[n],h_1,h_2\in [H], h_1<h_2$,  $\check \tau_{i_1,h_1}=\{\check o_{1:h_1-1},\check a_{1:h_1-2},\check o_{i_1,h_1},\check a_{i_1,h_1-1}\}\subseteq \check c_{h_1+1}\subseteq \check c_{h_2}\subseteq \check \tau_{i_2,h_2}$, and $\check a_{i_1,h_1}\subseteq \check c_{h_1+1}\subseteq \check c_{h_2}\subseteq \check \tau_{i_2,h_2}$. Hence, $\cL$ is QC.  However, for any $i_1\neq i_2\in[n], h_1\in[H-1]$,  agent $(i_1,h_1)$ may influence agent $(i_2,h_1+1)$ but $\sigma(\check a_{i_1,h_1})\nsubseteq\sigma(\check \tau_{i_2,h_1+1})$. Hence, $\cL$ may not be sQC.
        \item \textbf{Example 7:} The information in $\overline{\cD}_\cL$ evolves as $\forall h\in[H], i\neq 1, \check c_h=\{\check o_{1,1:h-1},\check o_{-1,1:h-d}\}, \check p_{1,h}=\emptyset, \check p_{i,h}=\{o_{i,h-d+1:h}\}$. Then, for any $i_1,i_2\in[n],h_1,h_2\in [H], h_1<h_2$. If $i_1\neq 1$, then agent $(i_1,h_1)$ will not influence agent $(i_2,h_2)$. If $i_1=1$, then $\check \tau_{i_1,h_1}=\{\check o_{1, 1:h_1}, \check o_{-1,1:h_1-d}\}\subseteq \check c_{h_1+1}\subseteq \check c_{h_2}\subseteq \check \tau_{i_2,h_2}$. Therefore, we have $\sigma(\check \tau_{i_1,h_1})\subseteq \sigma(\check \tau_{i_2,h_2})$ if agent $(i_1,h_1)$ influences agent $(i_2,h_2)$. Hence, $\cL$ is QC. However, for any $i_1\in [n], h_1<h_2\le H$, agent $(1,h_1)$ may influence agent $(i_1,h_2)$ but $\sigma(\check a_{1,h_1})\nsubseteq\sigma(\check \tau_{i_1, h_2})$. Hence, $\cL$ may not be sQC.
        \item \textbf{Example 8:} The information in $\overline{\cD}_\cL$ evolves as $\forall h\in[H], i\neq 1, \check c_h=\{\check o_{1,1:h-1},\check a_{1,1:h-2}, \check o_{-1,1:h-d}\}, \check p_{1,h}=\{\check a_{1,h-1}\}, \check p_{i,h}=\{\check{o}_{i,h-d+1:h}\}$. Then, for any $i_1,i_2\in[n],h_1,h_2\in [H], h_1<h_2$. If $i_1\neq 1$, then agent $(i_1,h_1)$ will not influence agent $(i_2,h_2)$. If $i_1=1$, then $\check \tau_{i_1,h_1}=\{\check o_{1,1:h_1}, \check a_{1,h_1-1}, \check o_{-1,1:h_1-d}\}\subseteq \check c_{h_1+1}\subseteq \check c_{h_2}\subseteq \check \tau_{i_2,h_2}$. Therefore, we have $\sigma(\check \tau_{i_1,h_1})\subseteq \sigma(\check \tau_{i_2,h_2})$ if agent $(i_1,h_1)$ influences agent $(i_2,h_2)$. Hence, $\cL$ is QC. However, for any $i_1\neq 1,h_1\in[H-1]$, agent $(1,h_1)$ may influence agent $(i_1,h_1+1)$ but $\sigma(\check a_{1,h_1})\nsubseteq\sigma(\check \tau_{i_1,h_1+1})$. Hence, $\cL$ may not be sQC.
    \end{itemize}
    This completes the proof. 
\end{proof}

\section{Deferred Details of  \S\ref{section 3}}
\label{sec: proof details sec 3}

\subsection{Supporting Lemmas}\label{sec:supporting_lemmas}

We start by proving several supporting lemmas that will be used in the later proofs in this section. 

\begin{lemma} \label{lemma: influence no-to-any}
    Given any QC LTC $\cL$, its induced Dec-POMDP    $\overline{\cD}_{\cL}$,  and any $i_1,i_2\in[n],h_1,h_2\in[H]$. If agent $(i_1,h_1)$ influences agent $(i_2,h_2)$ in the intrinsic model of $\overline{\cD}_\cL$, then for the random variables $\tau_{i_1,h_1^-},\tau_{i_2,h_2^-}$ in $\cL$, we have $\sigma(\tau_{i_1,h_1^-})\subseteq \sigma(\tau_{i_2,h_2^-})$. 
    Moreover, if $\cL$ is sQC, then for random variables $a_{i_1,h_1},\tau_{i_2,h_2^-}$ in $\cL$, we have $\sigma(a_{i_1,h_1})\subseteq \sigma(\tau_{i_2,h_2^-})$.
\end{lemma}
\begin{proof}
    We denote by $\check{\tau}_{i_1,h_1},\check{\tau}_{i_2,h_2}$ the information of agent $(i_1,h_1),(i_2,h_2)$ in the problem $\overline{\cD}_\cL$. From the definition of $\overline{\cD}_\cL$ being QC, if agent $(i_1,h_1)$ influences agent $(i_2,h_2)$, then $\sigma(\check{\tau}_{i_1,h_1})\subseteq \sigma(\check{\tau}_{i_2,h_2})$. Since for any $h\in[H], i\in[n], \check{\tau}_{i,h}$ is the information of agent $(i,h)$ without additional sharing, then we know that $\tau_{i,h^-}\backslash \check{\tau}_{i,h}\subseteq \cup_{t=1}^{h-1} z_t^a, \tau_{i,h^+}\backslash \check{\tau}_{i,h}\subseteq \cup_{t=1}^{h} z_t^a$. Therefore, we know that $\sigma(\tau_{i_1,h_1^-}\backslash\check{\tau}_{i_1,h_1})\subseteq \sigma(\cup_{t=1}^{h-1} z_t^a)\subseteq \sigma(c_{h_1^-})\subseteq \sigma(c_{h_2^-})\subseteq\sigma(\tau_{i_2,h_2^-})$. Also, we know  that $\sigma(\check{\tau}_{i_1,h_1})\subseteq \sigma(\check{\tau}_{i_2,h_2})\subseteq \sigma(\tau_{i_2,h_2^-})$. Thus, we can conclude that $\sigma(\tau_{i_1,h_1^-})\subseteq \sigma(\tau_{i_2,h_2^-})$. Moreover, if $\cL$ is sQC, then from the definition of $\overline{\cD}_\cL$ being sQC and agent $(i_1,h_1)$ influences agent $(i_2,h_2)$ in $\overline{\cD}_\cL$, it holds that $\sigma(a_{i_1,h_1})\subseteq\sigma(\check{\tau}_{i_2,h_2})\subseteq \sigma(\tau_{i_2,h_2^-})$.
\end{proof} 
\begin{lemma}
\label{lemma: action influence}
    Let $\cL$ be a QC LTC problem satisfying Assumptions  \ref{useless action}  and \ref{weak gamma observability}, and $\cD_\cL$ be the reformulated Dec-POMDP. Then, for any  $i_1,i_2\in[n],t_1,t_2\in[H]$, if agent $(i_1,2t_1)$ influences agent $(i_2,2t_2)$ in $\cD_\cL$, then $\sigma(\tau_{i_1,t_1^-})\subseteq\sigma(\tau_{i_2,t_2^-})$ in $\cL$.  Moreover, if $\cL$ is sQC, then $\sigma(a_{i_1,t_1})\subseteq \sigma(\tau_{i_2,t_2^-})$.
\end{lemma}
\begin{proof}
    We prove this case by case as follows:  \\
     1) If $a_{i_1,t_1}$ influences the underlying state $s_{t_1+1}$, then from Assumption \ref{weak gamma observability}, agent $(i_1,t_1)$ influences $o_{-i_1,t_1+1}$, so there must exist $i_3\neq i_1$, such that agent $(i_1,t_1)$ influences $o_{i_3,t_1+1}$. From part (e) of Assumption \ref{assumption:information evolution} and $t_1<t_2$ (since otherwise agent $(i_1,2t_1)$ cannot influence agent $(i_2,2t_2)$ in $\cD_\cL$), we know that $o_{i_3,t_1+1}\in\tau_{i_3,(t_1+1)^-}\subseteq \tau_{i_3,t_2^-}$ even under no additional sharing, and then we get agent $(i_1,t_1)$ influences agent $(i_3,t_2)$ in $\overline{\cD}_\cL$ (the Dec-POMDP induced by $\cL$). From Lemma \ref{lemma: influence no-to-any}, it holds that $\sigma(\tau_{i_1,t_1^-})\subseteq \sigma(\tau_{i_3,t_2^-})$.  From Assumption \ref{assumption:sigma_include} and $i_3\neq i_1$, we know that $\sigma(\tau_{i_1,t_1^-})\subseteq \sigma(c_{t_2^-})\subseteq\sigma(\tau_{i_2,t_2^-}).$ If $\cL$ is sQC, by Lemma \ref{lemma: influence no-to-any},  we have $\sigma(a_{i_1,t_1})\subseteq \sigma(\tau_{i_3,t_2^-})$ , and then $\sigma(a_{i_1,t_1})\subseteq \sigma(c_{t_2^-})\subseteq\sigma(\tau_{i_2,t_2^-})$ from Assumption \ref{assumption:sigma_include}. \\
    2) If $a_{i_1,t_1}$ does not influence $s_{t_1+1}$, then from Assumption \ref{useless action}, 
    for any $t>t_1, a_{i_1,t_1}\notin \tau_{t^-}$ and $a_{i_1,t_1}\notin \tau_{t^+}$, and then agent $(i_1,2t_1)$ does not influence $\tilde{s}_{2t_1+1}$ and $\tilde{o}_{2t_1+1}$ in $\cD_\cL$. Thus $\tilde{a}_{i_1,2t_1}=a_{i_1,t_1}$ does not influence $\tilde{\tau}_{i,2t_1+1},\forall i\in[n]$, and then it does not influence $\tilde{a}_{i,2t_1+1},\forall i\in[n]$.  And hence, it does not influence $\tilde{\tau}_{i,2t_1+2}$ and $\tilde{a}_{i,2t_1+2},\forall i\in[n]$, either. By recursion, we know that agent $(i_1,2t_1)$ does not influence agent $(i_2,2t_2)$, which leads to a contradiction to the premise of the lemma.  
    This completes the proof. 
\end{proof}

\subsection{Proof of \Cref{lemma: nonqc_hardness}}

\begin{proof}
We first have the following proposition on the hardness of solving POMDPs.  
\begin{proposition}
    \label{prop: epsilon-additive}
    For any $\epsilon\in[0,\frac{1}{4})$, computing an  $\epsilon$-additive optimal strategy in POMDPs with rewards bounded in [0,1/2] is \texttt{PSPACE-hard}.
\end{proposition}
\begin{proof}[Proof of \Cref{prop: epsilon-additive}]
    In the proof of \cite[Theorem 4.11]{lusena2001nonapproximability}, given any $\epsilon\in[0,1)$, it constructed POMDPs from the problems of Stochastic Satisfiability (SSAT) of the Quantified Boolean Formulae. The constructed instances in \cite{lusena2001nonapproximability} satisfy that the reward values lie in $\{0,2\}$, and it was then proved that finding an $\epsilon$-relative approximately optimal solution in such POMDPs is \texttt{PSPACE-hard}. Also, one can verify that finding an $\epsilon$-additive approximately optimal solution in such POMDPs is \texttt{PSPACE-hard}.

    Then, for any $\epsilon\in[0,\frac{1}{4})$, let $\epsilon_1=4\epsilon\in[0,1)$, and leverage the construction in \cite[Theorem 4.11]{lusena2001nonapproximability} with $\epsilon_1$, but scaling the reward values by $\frac{1}{4}$ such that rewards are bounded in $[0,\frac{1}{2}]$. Then, finding an $\epsilon$-additive approximately optimal solution in such POMDPs (after scaling) is \texttt{PSPACE-hard}. 
\end{proof}
Now we proceed with the proof of Lemma \ref{lemma: nonqc_hardness} based on Proposition \ref{prop: epsilon-additive}. Given any POMDP $\cP=(\cS^{\cP},\cA^{\cP},\cO^{\cP},\{\OO_h^{\cP}\}_{h\in[H^\cP]},\{\TT_h^{\cP}\}_{h\in[H^\cP]},\{\cR_h^{\cP}\}_{h\in[H^\cP]}, \mu_1^{\cP})$ with rewards bounded in $[0,1/2]$,  we can construct an LTC $\cL$ with any $\alpha\in(0,1)$ as follows:
\begin{itemize}
    \item Number of agents: $n=3$; length of episode: $H=H^\cP$.
    \item Underlying state space: $\cS=\cS^\cP\times[2]$. For any $s\in\cS$, we can split $s=(s^1,s^2)$, where $s^1\in\cS^\cP,s^2\in[2]$. Initial state distribution: $\forall s\in\cS, \mu_1(s)=\mu_1^\cP(s^1)/2$.
    \item Control action space: For any $h\in [H],\cA_{1,h}=\cA^\cP, \cA_{2,h}=[2],\cA_{3,h}=\{\emptyset\}$. 
    \item Transition: For any $h\in[H], s_h,s_{h+1}\in\cS, a_h\in\cA_h, \TT_h(s_{h+1}\given s_h,a_h)=\TT_h^\cP(s_{h+1}^1\given s_h^1,a_{1,h})\mathds{1}[s_{h+1}^2=a_{2,h}]$.
    \item Observation space:  For any $h\in[H],  \cO_{1,h}=\cO^\cP_h,\cO_{2,h}=\cO_{3,h}=\cS$.
    \item Emission: For any $h\in[H], o_h\in\cO_h,s_h\in\cS,\OO_h(o_h\given s_h)=\OO_h^\cP(o_{1,h}\given s_h^1)\OO_h'(o_{2,h}\given s_h)\OO_h'(o_{3,h}\given s_h)$, where $\OO_h'$ is defined as
    \begin{equation*}
        \forall o\in \cS, s_h\in\cS,\qquad  \OO_h'(o\given s_h)=\left\{\begin{aligned}
            &\alpha+\frac{1-\alpha}{|\cS|}\quad~&\text{if } o=s_h\\&\frac{1-\alpha}{|\cS|}\quad~&\text{o.w.}
        \end{aligned}\right.. 
    \end{equation*}
    \item Baseline sharing: null. 
    \item Communication action space: For any $h\in[H],\cM_{1,h}=\cM_{2,h}=\{0,1\}^{2h-1}, \cM_{3,h}=\{0,1\}^h$. For any $i\in[2], p_{i,h^-}\in \cP_{i,h^-},\phi_{i,h}(p_{i,h^-},m_{i,h})=\{o_{i,k}\given \forall k\le h,\text{~the~} (2k-1)\text{-th digit of $m_{i,h}$ is 1 and } o_{i,k}\in p_{i,h^-}\}\cup\{a_{i,k}\given \forall k\le h-1, \text{ the } 2k\text{-th digit of $m_{i,h}$ is 1 and } a_{i,k}\in p_{i,h^-}\}\cup\{m_{i,h}\}.$ For agent $3,~ p_{3,h^-}\in\cP_{3,h^-},~ \phi_{3,h}(p_{3,h^-},m_{3,h})=\{o_{3,k}\given \forall k\le h, \text{ the }k\text{-th digit of $m_{3,h}$  is 1 and } o_{3,k}\in p_{3,h^-}\}\cup\{m_{3,h}\}$.  
    \item Reward function: For any $h\in[H]$, $\forall  s_h\in\cS,a_h\in\cA_h, \cR_h(s_h,a_h)=\cR_h^\cP(s_h^1,a_{1,h})/H$.
    \item Communication cost function: For any $h\in[H]$, $\forall  z_h^a\in\cZ_h^a, \cK_{h}(z_h^a)=\mathds{1}[z_h^a\neq\{m_h\}]$. It means that the communication cost is $1$ unless there is no additional sharing.
    \item We restrict the communication strategy to only use $c_h$ as input. And for any $t\in[H-1]$, we remove $a_{3,t}$ in $\tau_{h^-},\tau_{h^+}$ for all $h>t$.
\end{itemize} 
We first verify that such a construction satisfies Assumptions \ref{gamma observability}, \ref{limited communication strategy}, \ref{useless action}, and \ref{weak gamma observability}, but is not QC.
\begin{itemize} 
    \item $\cL$ satisfies Assumptions \ref{gamma observability}, \ref{weak gamma observability}, because both agent 2 and agent 3 have  individual $\gamma$-observability with $\gamma=\alpha$. That is, for any $b_1,b_2\in\Delta(\cS),i=2,3$, we have
    \begin{align*}
        &\norm{\OO_{i,h}^\top(b_1-b_2)}_1=\sum_{o_{i,h}\in\cO_{h}}|\sum_{s_h\in\cS}(b_1(s_h)-b_2(s_h))\OO_{i,h}(o_{i,h}\given s_h)|\\
        &\quad=\sum_{o_{i,h}\in\cO_{h}}|\sum_{s_h\in\cS}(b_1(s_h)-b_2(s_h))(\frac{1-\alpha}{|\cS|}+\alpha\mathds{1}[o_{i,h}=s_h])|\\
        &\quad=\sum_{o_{i,h}\in\cO_{h}}\alpha|b_1(o_{i,h})-b_2(o_{i,h})|=\alpha\norm{b_1-b_2}_1.
    \end{align*}
    \item $\cL$ satisfies Assumption \ref{limited communication strategy},  because we restrict that the communication strategy can only use $c_h$ as input.
    \item $\cL$ satisfies Assumption \ref{useless action} since the control actions $a_{3,t}$ (for all $t\in[H-1]$) do not influence the underlying state, and we remove $a_{3,t}$ from $\tau_h$ for any $h>t$. 
    \item $\cL$ is not QC, since for any $h\in[H-1],$ agent $(2,h)$ influences the state $s_{h+1}$ and then influences $o_{3,h+1}$ and $\tau_{3,h+1}$. However, agent $(3,h+1)$ does not know the information of agent $(2,h)$, i.e., $\sigma(\tau_{2,h})\nsubseteq 
\sigma(\tau_{3,h+1})$.
\end{itemize}

In this LTC problem $\cL$, let $(g_{1:H}^{a,\ast},g_{1:H}^{m,\ast})$ be any $\epsilon$-team  optimal strategy, with  $\epsilon\in[0,1/4)$. If any agent shares any information through additional sharing, i.e., $\exists h\in[H], \PP(z_h^a\neq \{m_h\}\given g_{1:H}^{a,\ast},g_{1:H}^{m,\ast})>0$, we then choose the $h$ to be the minimal one, i.e.,  $h=\min\{h'\in[H]\given \PP(z_{h'}^a\neq \{m_{h'}\}\given g_{1:H}^{a,\ast},g_{1:H}^{m,\ast})>0\}$ being the first time the additional sharing occurs. 
This means that at this timestep $h$, there is no observation or action in $c_{h^-}$ (almost surely), since baseline sharing is null.
Then, there exists an agent $i\in[2]$ such that $\PP(z_{i,h}^a\neq \{m_{i,h}\}\given g_{1:H}^{a,\ast},g_{1:H}^{m,\ast})>0$.

From the construction of $\cL$, we know that agent $i$ chooses $m_{i,h}$ based on $c_{h^-}$. It means it will always share, i.e.,  $\PP(z_{i,h}^a\neq \{m_{i,h}\}\given g_{1:H}^{a,\ast},g_{1:H}^{m,\ast})=1$. Therefore, $\PP(\kappa_h=1\given g_{1:H}^{a,\ast},g_{1:H}^{m,\ast})=1$, and $J_{\cL}(g_{1:H}^{a,\ast},g_{1:H}^{m,\ast})=\EE[\sum_{t=1}^H r_t-\kappa_t\given g_{1:H}^{a,\ast},g_{1:H}^{m,\ast}]\le \EE[\sum_{t=1}^H r_t\given g_{1:H}^{a,\ast},g_{1:H}^{m,\ast}]-\EE[\kappa_h\given g_{1:H}^{a,\ast},g_{1:H}^{m,\ast}]\le H\cdot \frac{1}{2H}-1\le -\frac{1}{2}$. Note that the rewards in $\cP$ are bounded by $[0,\frac{1}{2}]$, and the rewards in $\cL$ are  bounded by $[0,\frac{1}{2H}]$. Hence, $(g_{1:H}^{a,\ast},g_{1:H}^{m,\ast})$ is not an $\epsilon$-team-optimal  for any $\epsilon\in[0,1/4)$. 

Therefore, any $\epsilon$-team-optimal strategy yields no additional sharing. Then, any $(g_{1:H}^{a,\ast},g_{1:H}^{m,\ast})$ being an $\frac{\epsilon}{H}$-team-optimal strategy of $\cL$ will directly give an $\epsilon$-optimal strategy of $\cP$ as $\{g_{1,h}^{a,\ast}\}_{h\in [H]}$, since when there is no sharing, the decision process is only controlled by agent $1$.  From Proposition \ref{prop: epsilon-additive}, we complete the proof.
\end{proof}

\subsection{Proof of \Cref{lemma: limited communication}}\label{sec:proof of lemma: limited communication}

\begin{proof}
    We prove this result by showing a reduction from the Team Decision problem \cite{teamdecisionhardness}.
    \begin{definition}[Team decision problem (TDP)]
        Given finite sets $Y_1,Y_2,U_1,U_2$, a rational probability mass function $p:Y_1\times Y_2\rightarrow \mathbb{Q}$, and an integer cost function $c:Y_1\times Y_2\times U_1\times U_2\rightarrow \mathbb{N}$, find decision rules $\gamma_i:Y_i\rightarrow U_i, i=1,2$ that minimize the  expected cost
    \begin{equation}
            J(\gamma_1,\gamma_2)=\sum_{y_1\in Y_1,y_2\in Y_2} c(y_1,y_2,\gamma_1(y_1),\gamma_2(y_2))p(y_1,y_2).
        \end{equation}
    \end{definition}
    We show the \texttt{NP-hardness} of solving LTC from the problem of TDP, even with $|U_1|=|U_2|=2$ \cite{teamdecisionhardness}.  Given any  TDP $\cT\cD=(\tilde{Y}_1,\tilde{Y}_2,\tilde{U}_1,\tilde{U}_2,\tilde{c},\tilde{p}, \tilde{J})$ with $|\tilde{U}_1|=|\tilde{U}_2|=2$, let $\tilde{U}_1=\{1,2\},\tilde{U}_2=\{1,2\}$, then we can construct an  $H=3$ and $2$-agent LTC $\cL$ with two parameters $\alpha_1\in\RR, \alpha_2\in(0,1)$ (to be specified later) such that: 
    \begin{itemize}
        \item Number of agents: $n=2$. 
        \item Underlying state: $\cS=[2]^4$. For each $s_1\in \cS$, we can split $s_1$ into 4 parts as $s_1=(s_1^1, s_1^2, s_1^3, s_1^4)$, where $s_1^1,s_1^2,s_1^3, s_1^4\in [2]$.  Similarly, $s_2,s_3\in\cS$ can be split in the same way.
        \item Initial state distribution: $\forall s_1\in \cS, \mu_1(s_1)=\frac{1}{16}$.
        \item Control action space: For the first $2$ timesteps,  $\forall i=1,2,\cA_{i,1}=\cA_{i,2}=\{\emptyset\}$; for $h=3,\cA_{i,3}=[2]$.  
        \item Transition: 
        $\forall s\in \cS, a_1\in \cA_1,a_2\in\cA_2,a_3\in\cA_3, \TT_1(s\given s,a_1)=\TT_2(s\given s,a_2)=\TT_3(s\given s,a_3)=1$. Note that under the transition dynamics above, $s_1=s_2=s_3$ always holds, for any $s_1\in\cS$.
        \item Observation space: $\cO_{1,1}=\cO_{2,1}=\cO_{1,2}=\cO_{2,2}=[2]\times \cS$, $ \cO_{1,3}=\tilde{Y_1}\times\cS,\cO_{2,3}=\tilde{Y_2}\times\cS$. Hence, for each $i\in[2],h\in[2], o_{i,h}\in\cO_{i,h}$, we can split $o_{i,h}$ into 2 parts as $o_{i,h}=(o_{i,h}^1,o_{i,h}^2)$, where $o_{i,h}^1\in[2],o_{i,h}^2\in\cS$. For each $i\in[n], o_{i,3}\in\cO_{i,3}$, we can similarly  split $o_{i,3}$ into 2 parts as $o_{i,3}=(o_{i,3}^1,o_{i,3}^2)$, where $o_{i,3}^1\in\tilde{Y}_i,o_{i,3}^2\in\cS$. 
        \item Baseline sharing: null.
        \item Communication action space: For $i\in[2],h\in\{1,2\}, \cM_{i,h}=\{0,1\}^{h}$; $\phi_{i,h}$ is defined as:  $\forall h\in\{1,2\}, \phi_{i,h}(p_{i,h^-},m_{i,h})=\{o_{i,k}\given \forall k\le h, \text{the } k\text{-th}\text{ digit of $m_{i,h}$ is }1 \text{ and }o_{i,k}\in p_{i,h^-}\}$. For $h=3$, $\cM_{i,3}=\{1,2\}$, and if $m_{i,3}=1$, then $\phi_{i,3}(p_{i,3^-},m_{i,3})=\{o_{i,1},o_{i,3},m_{i,3}\}$; if $m_{i,3}=2$, then $\phi_{i,3}(p_{i,3^-},m_{i,3})=\{o_{i,2},o_{i,3},m_{i,3}\}$. 
        \item Emission: 
        For any $i\in[2],h\in[2],s_h\in\cS,o_{i,h}\in\cO_{i,h},\OO_{h}(o_h\given s_h)=\Pi_{i=1}^2\OO_{i,h}(o_{i,h}\given s_h)$ and $\OO_{i,h}(o_{i,h}\given s_h)$ is defined as:
        \begin{equation*}
            \OO_{i,h}(o_{i,h}\given s_h)=\begin{cases}
                \frac{1-\alpha_2}{16}&o_{i,h}^1=s_h^{i+2h-2},o_{i,h}^2\neq s_h\\
                \frac{1-\alpha_2}{16}+\alpha_2&o_{i,h}^1=s_h^{i+2h-2},o_{i,h}^2=s_h\\
                0& \text{o.w.}
            \end{cases}.
        \end{equation*}
        For $i\in[2],h=3$, $s_3\in\cS,o_{3}\in\cO_{3}, \OO_3(o_3\given s_3)=\tilde{p}(o_{1,3}^1,o_{2,3}^1)\Pi_{i=1}^2\OO_{i,3}^2(o_{i,3}^2\given s_3)$, where
                \begin{align*}
            \OO_{i,3}^2(o_{i,3}^2\given s_3)&=\begin{cases}
                \frac{1-\alpha_2}{16}&o_{i,3}^2\neq s_3\\
                \frac{1-\alpha_2}{16}+\alpha_2&o_{i,3}^2=s_3\\
            \end{cases}.
        \end{align*}
    \end{itemize}
    The reward functions are defined as: 
    \begin{align*}
        &\cR_{1}(s_1,a_1)=\cR_{2}(s_2,a_2)=0,~~ \forall s_1,s_2\in \cS, a_1\in \cA_1,a_2\in\cA_2;\\
        &\cR_{3}(s_3,a_3)=
        \begin{cases}
            1&\text{if ($a_{1,3}=s_3^2$ or $a_{1,3}=s_3^4$) and ($a_{2,3}=s_3^1$ or $a_{2,3}=s_3^3$)}\\
            0&\text{o.w.}\\
        \end{cases}. 
    \end{align*} 
    The communication cost functions are defined as:  
    \begin{align*}
        &\forall h\in\{1,2\},z_h^a\in\cZ_h^a,\qquad \cK_{h}(z_h^a)=1\text{ if $z_h^a\neq\{m_{1,h},m_{2,h}\}$, else }0;\\
        &\cK_{3}(z_3^a)=\begin{cases}
            \tilde{c}(o_{1,3}^1
            ,o_{2,3}^1,1,1)/\alpha_1&\text{if $\{o_{1,1},o_{2,1}\}\subseteq z_3^a$ and $\{o_{1,2},o_{2,2}\}\cap z_3^a=\emptyset$}\\
            \tilde{c}(o_{1,3}^1
            ,o_{2,3}^1,2,1)/\alpha_1&\text{if $\{o_{1,2},o_{2,1}\}\subseteq z_3^a$ and $\{o_{1,1},o_{2,2}\}\cap z_3^a=\emptyset$}\\
            \tilde{c}(o_{1,3}^1
            ,o_{2,3}^1,1,2)/\alpha_1&\text{if $\{o_{1,1},o_{2,2}\}\subseteq z_3^a$ and $\{o_{1,2},o_{2,1}\}\cap z_3^a=\emptyset$}\\
            \tilde{c}(o_{1,3}^1
            ,o_{2,3}^1,2,2)/\alpha_1&\text{if $\{o_{1,2},o_{2,2}\}\subseteq z_3^a$ and $\{o_{1,1},o_{2,1}\}\cap z_3^a=\emptyset$}\\
            0&\text{o.w.}
        \end{cases},
    \end{align*}
where we let $\alpha_0=\max_{y_1,y_2,u_1,u_2}\tilde{c}(y_1,y_2,u_1,u_2)$, and set $\alpha_1=2\alpha_0, \alpha_2\in(0,1)$. 
Note that without loss of optimality, we suppose $\alpha_1>0$, since if  $\alpha_1=0$, then $\tilde{c}(y_1,y_2,u_1,u_2)=0, \forall y_1\in \tilde{Y}_1,y_2\in\tilde{Y_2},u_1\in\tilde{U}_1,u_2\in\tilde{U}_2$,  which is a trivial instance that cannot be the one that leads to the hardness in \cite{teamdecisionhardness}. Hence, $0\leq\kappa_3=\cK_3(z_3^a)\leq \frac{1}{2}$ for any $z_3^a\in\cZ_3^a$ always holds. Also, we  remove $a_{i,t}$ in $\tau_{h^-}$ and $\tau_{h^+}$ for any $t\in[2],i\in[2], h>t$. 
Under such a construction, $\cL$ satisfies  the following conditions: 
\begin{itemize}
    \item $\cL$ is QC: For all $i_1,i_2\in[2],h_1,h_2\in[3]$ with $i_1\neq i_2$, agent $(i_1,h_1)$ does not influence $(i_2,h_2)$ because agent $(i_1,h_1)$ cannot influence the observations of agent $(i_2,h_2)$, and the baseline sharing is null. For the same agent with $i_1=i_2$, the information-inclusion assumption in \Cref{assumption:information evolution} (e) ensures that $\cL$ is QC. 
    \item $\cL$ satisfies Assumptions \ref{gamma observability} and \ref{weak gamma observability}: We prove this by showing that each agent $i\in[2]$ satisfies $\gamma$-observability. For any $i\in[2],h\in[2], b_1,b_2\in\Delta(\cS)$, let  
    \begin{align*}
        &\norm{\OO_{i,h}^\top (b_1-b_2)}_1
        =\sum_{o_{i,h}^1\in[2]}\sum_{o_{i,h}^2\in \cS}\bigg|\sum_{s_h\in\cS}(b_1(s_h)-b_2(s_h))\OO_{i,h}((o_{i,h}^1,o_{i,h}^2)\given s_h)\bigg|\\
        &\quad\ge\sum_{o_{i,h}^2\in \cS}\bigg|\sum_{o_{i,h}^1\in[2]}\sum_{s_h\in\cS}(b_1(s_h)-b_2(s_h))\OO_{i,h}((o_{i,h}^1,o_{i,h}^2)\given s_h)\bigg|\\
        &\quad= \sum_{o_{i,h}^2\in \cS}\bigg|\sum_{s_h\in\cS}\sum_{o_{i,h}^1\in[2]}(b_1(s_h)-b_2(s_h))\mathbbm{1}[o_{i,h}^1=s_h^{i+2h-2}](\frac{1-\alpha_2}{16}+\alpha_2\mathbbm{1}[o_{i,h}^2=s_h])\bigg|\\
        &\quad=\sum_{o_{i,h}^2\in \cS}\alpha_2|b_1(o_{i,h}^2)-b_2(o_{i,h}^2)|=\alpha_2\norm{b_1-b_2}_1.
    \end{align*}
    For any $i\in[2], h=3$, the proof is similar, where we  replace $o_{i,h}^1\in [2]$ with $o_{i,h}^1\in \tilde{Y}_i$ for $h=3$.
    \item $\cL$ satisfies Assumption \ref{useless action}, because the joint control action histories  $a_{1:3}$ do not influence the underlying states,  and we restrict the communication and control strategies to not use them as input.    
    \item $\cL$ does not satisfy \Cref{limited communication strategy}, since here we allow communication strategies $g_{i,h}^m$ to take $\tau_{i,h^-}$ as input for any $i\in[2],h\in[3]$. 
 \end{itemize} 
We will show below that computing a team-optimal strategy of $\cL$  can give us a team-optimal strategy of $\cT\cD$. 
Given $(g_{1:3}^{a,\ast},g_{1:3}^{m,\ast})$ being a team-optimal strategy of $\cL$,
firstly, it will have no additional sharing at timesteps $h=1,2$ under $(g_{1:3}^{a,\ast},g_{1:3}^{m,\ast})$, namely, for $h=1,2, \PP(z_h^a=\{m_{1,h},m_{2,h}\}\given g_{1:3}^{a,\ast},g_{1:3}^{m,\ast})=1$. If not, then for any $i\in[2]$, consider $g_{i,1}^{m,'}, g_{i,2}^{m,'}, g_{i,3}^{a,\ast}$ defined as 
\begin{align*}
    \forall \tau_{i,1^-},\tau_{i,2^-},\tau_{i,3^+}, g_{i,1}^{m,'}(\tau_{i,1^-})=(0),\quad g_{i,2}^{m,'}(\tau_{i,2^-})=(0,0),\quad g_{i,3^+}^{a,'}(\tau_{i,3^+})=\begin{cases}
        o_{3-i, 1}^1&\text{if $o_{3-i,1}\in \tau_{i,3^+}$}\\
        o_{3-i, 2}^1&\text{o.w.}
    \end{cases},
\end{align*}
and replace the $g_{1:2,1}^{m,\ast}, g_{1:2,2}^{m,\ast}, g_{1:2,3}^{a,\ast}$ by $g_{1:2,1}^{m,'}, g_{1:2,2}^{m,'}, g_{1:2,3}^{a,'}$ in $g_{1:3}^{a,\ast},g_{1:3}^{m,\ast}$ to get $(g_{1:3}^{a,'},g_{1:3}^{m,'})$. It is easy to verify that $J_{\cL}(g_{1:3}^{a,\ast},g_{1:3}^{m,\ast})<J_{\cL}(g_{1:3}^{a,'},g_{1:3}^{m,'})$, since   $g_{1:3}^{a,'},g_{1:3}^{m,'}$ can guarantee $r_3=1$ always holds; if there is no additional sharing under $g_{1:3}^{a,'},g_{1:3}^{m,'}$ at first two timesteps, these two strategies have the same communication cost, otherwise $(g_{1:3}^{a,'},g_{1:3}^{m,'})$ has $\sum_{h=1}^3\kappa_h=\kappa_3\le \frac{1}{2}$ but $(g_{1:3}^{a,\ast},g_{1:3}^{m,\ast})$ has $\sum_{h=1}^3\kappa_h\ge 1$.  This leads to the contradiction that $(g_{1:3}^{a,\ast},g_{1:3}^{m,\ast})$ is a team-optimal strategy. 

Also, if we replace the $g_{1:2,3}^{a,\ast}$ by $g_{1:2,3}^{a,'}$, the communication cost does not change and the reward can achieve optimal, i.e., $r_3=1$ always holds. Thus, without loss of generality, we can assume that $g_{1:2,3}^{a,\ast}=g_{1:2,3}^{a,'}$, since  otherwise we can do the replacement and it is still a team-optimal strategy. 

Therefore, $J_{\cL}(g_{1:3}^{a,\ast},g_{1:3}^{m,\ast})=\EE[\sum_{h=1}^3r_h-\kappa_h\given g_{1:3}^{a,\ast},g_{1:3}^{m,\ast}]=1-\EE[\kappa_3\given g_{1:3}^{a,\ast},g_{1:3}^{m,\ast}]=1-\frac{1}{\alpha_1}\EE[\tilde{c}(o_{1,3}^1,o_{2,3}^1,m_{1,3},m_{2,3})]$, where $m_{1,3}=g_{1,3}^{m,\ast}(\{\tau_{1,3^-}\}),m_{2,3}=g_{2,3}^{m,\ast}(\{\tau_{2,3^-}\})$, which means $(g_{1:3}^{a,\ast},g_{1:3}^{m,\ast})$ can minimize $\kappa_3$ through choosing $m_3$ properly, if there is no additional sharing at first two timesteps, i.e., $\tau_{i,3^-}=\{o_{i,1},o_{i,2},o_{i,3},m_{1:2}\}$. By  construction, $\kappa_{3}$  only depends on $o_3$ and $m_3$ and is irrelevant of $\{o_{1,1},o_{1,2},o_{1,3}^2\}$, and $m_{i,1}=(0),m_{i,2}=(0,0),\forall i\in[2]$ always hold. Hence,   $\{o_{1,1},o_{1,2},o_{1,3}^2\}$ are useless information for agent $1$ to choose $m_{1,3}$ and minimize $\EE[\kappa_3]$.  Therefore, not using them in $g_{1,3}^{m,\ast}$ does not lose any optimality. Hence, we can consider the $g_{1,3}^{m,\ast}$ that only has $o_{1,3}^1$ as input. In the same way,   we can consider the $g_{2,3}^{m,\ast}$ that only has $o_{2,3}^1$ as input. Therefore, $J_{\cL}(g_{1:3}^{a,\ast},g_{1:3}^{m,\ast})=1-\sum_{o_{1,3}^1,o_{1,3}^2}\frac{1}{\alpha_1}\tilde{c}(o_{1,3}^1,o_{2,3}^1, g_{1,3}^{m,\ast}( o_{1,3}^1), g_{2,3}^{m,\ast}(o_{2,3}^1))\tilde{p}(o_{1,3}^1,o_{2,3}^1)$. Further, we can leverage $\tilde{\gamma}^\ast_1=g_{1,3}^{m,\ast}, \tilde{\gamma}^\ast_2=g_{2,3}^{m,\ast}$ to minimize  the expected cost $\tilde{J}$ of the TDP. Therefore, 
from the \texttt{NP-hardness} of TDPs  \cite[Corollary 4.1]{teamdecisionhardness}, we complete our proof. 
\end{proof}

\subsection{Proof of \Cref{lemma: useless action}}
\begin{proof}[Proof of Lemma \ref{lemma: useless action}]
    We prove this result by showing a reduction from  the Team Decision problem. Given any TDP $\cT\cD=(\tilde{Y}_1,\tilde{Y}_2,\tilde{U}_1,\tilde{U}_2,\tilde{c},\tilde{p},\tilde{J})$ with $|\tilde{U}_1|=|\tilde{U}_2|=2$, let $\tilde{U}_1=\{1,2\},\tilde{U}_2=\{1,2\}$, then we can construct an $H=3$ and $2$-agent LTC $\cL$ with two parameters $\alpha_1\in\RR, \alpha_2\in(0,1)$ (to be specified later) such that:
\begin{itemize}
        \item Underlying state: $\cS=[2]^2$. For each  $s_1\in \cS$, we can split $s_1$ into 2 parts as $s_1=(s_1^1,s_1^2)$, where $s_1^1,s_1^2\in [2]$.  Similarly, $s_2,s_3\in\cS$ can be split in the same way.
        \item Initial state distribution: $\forall s_1\in \cS, \mu_1(s_1)=\frac{1}{4}$.
        \item Control action space: For any $i\in[2], h=1$, $\cA_{i,1}=\{\emptyset\}$; for $h=2$, $\cA_{i,2}=\{(0,x),(x,0)\given x\in[2]\}$; We can write $a_{i,2}=(a_{i,2}^1,a_{i,2}^2), a_{i,2}^1,a_{i,2}^2\in\{0,1,2\}$; for $h=3, \cA_{i,3}=[2]$
        \item Transition: 
        $\forall s\in \cS, a_1\in \cA_1,a_2\in\cA_2,a_3\in\cA_3, \TT_1(s\given s,a_1)=\TT_2(s\given s,a_2)=\TT_3(s\given s,a_3)=1$. Note that under the transition dynamics above, $s_1=s_2=s_3=s_4$ always holds, for any $s_1\in\cS$.
        \item Observation space: $\cO_{1,1}=\cO_{2,1}=[2]\times \cS, \cO_{1,2}=\tilde{Y_1}\times\cS,\cO_{2,2}=\tilde{Y_2}\times\cS, \cO_{1,3}=\cO_{2,3}=\cS$. For each $i\in[2], o_{i,1}\in\cO_{i,1}$, we can split $o_{i,1}$ into 2 parts as $o_{i,1}=(o_{i,1}^1,o_{i,1}^2)$, where $o_{i,1}^1\in[2],o_{i,1}^2\in\cS$. For each $i\in[2], o_{i,2}\in\cO_{i,2}$, similarly, we can split $o_{i,2}$ into 2 parts as $o_{i,2}=(o_{i,2}^1,o_{i,2}^2)$, where $o_{i,2}^1\in\tilde{Y}_i,o_{i,2}^2\in\cS$. 
        \item The baseline sharing is null.
        \item Communication action space: For $i\in[2],h\in\{1,2\}, \cM_{i,h}=\{0,1\}^{2h-1}$ and $\phi_{i,h}$ is defined as $\phi_{i,h}(p_{i,h^-},m_{i,h})=\{o_{i,k}\in p_{i,h^-}\given \forall k\le h, \text{the~}(2k-1)\text{-th digit of $m_{i,h}$ is }1\}\cup \{a_{i,k}\in p_{i,h^-}\given \forall k\le h-1, \text{the~}2k\text{-th digit of $m_{i,h}$ is }1\}\cup\{m_{i,h}\}$; For $h=3, \cM_{i,3}=\{1,2\}, \forall p_{i,3^-}\in \cP_{i,3^-}$, $\phi_{i,3}(p_{i,3^-},1)=\{o_{i,2}, a_{i,2}, m_{i,3}\}$ and $\phi_{i,3}(p_{i,3^-},2)=\{o_{i,2},a_{i,2},o_{i,3},m_{i,3}\}$.
        \item Emission: 
        For $h=1,\forall s_1\in\cS,o_{i,1}\in\cO_{i,1},\OO_{1}(o_1\given s_1)=\Pi_{j=1}^2\OO_{j,1}(o_{j,1}\given s_1)$ and $\forall i\in[2], \OO_{i,1}(o_{i,1}\given s_1)$ is defined as:
        \begin{equation*}
            \OO_{i,h}(o_{i,h}\given s_h)=\begin{cases}
                \frac{1-\alpha_2}{4}&o_{i,h}^1=s_h^{i+2h-2},o_{i,h}^2\neq s_h\\
                \frac{1-\alpha_2}{4}+\alpha_2&o_{i,h}^1=s_h^{i+2h-2},o_{i,h}^2=s_h\\
                0& \text{o.w.}
            \end{cases},
        \end{equation*}
        for $h=2, \forall s_2\in\cS,o_2\in\cO_{2}, \OO_2(o_2\given s_2)=\tilde{p}(o_{1,2}^1,o_{2,2}^1)\Pi_{j=1}^2\OO_{j,2}^2(o_{j,2}^2\given s_2)$, and $\forall i\in[2], \OO_{i,2}^2(o_{i,2}\given s_2)$ is defined as:
                \begin{align*}
            \OO_{i,2}^2(o_2^2\given s_2)=\begin{cases}
                \frac{1-\alpha_2}{4}&o_{i,2}^2\neq s_2\\
                \frac{1-\alpha_2}{4}+\alpha_2&o_{i,2}^2=s_2\\
                0&\text{o.w.}
            \end{cases},
        \end{align*}
        for $h=3, \forall s_3\in\cS,o_{i,3}\in\cO_{i,3},\OO_{3}(o_3\given s_3)=\Pi_{j=1}^2\OO_{j,3}(o_{j,3}\given s_3)$, and $\forall i\in[2], \OO_{i,3}(o_{i,3}\given s_3)$ is defined as:
                \begin{equation*}
            \OO_{i,3}(o_{i,3}\given s_3)=\begin{cases}
                \frac{1-\alpha_2}{4}&o_{i,3}\neq s_3\\
                \frac{1-\alpha_2}{4}+\alpha_2&o_{i,3}=s_3\\
                0&\text{o.w.}
            \end{cases}.
        \end{equation*}
        \item Reward functions:    
        \begin{align*}
        &\cR_{1}(s_1,a_1)=\cR_{2}(s_2,a_2)=0,~~ \forall s_1,s_2\in \cS, a_1\in \cA_1,a_2\in\cA_2;\\
        &\cR_{3}(s_3,a_3)=
        \begin{cases}
            1&\text{if $a_{1,3}=s_3^2$ and $a_{2,3}=s_3^1$}\\
            0&\text{o.w.}\\
        \end{cases}.
    \end{align*} 
    \item Communication cost functions:
    \begin{align*}
        &\forall h\in[2],z_h^a\in\cZ_h^a,\cK_{h}(z_h^a)=1\qquad \text{ if $z_h^a\neq\{m_{h}\}$,~}  \text{else }0;\\
        &\cK_{3}(z_3^a)=\begin{cases}
            \tilde{c}(o_{1,2}^1
            ,o_{2,2}^1,1,1)/\alpha_1&\text{if } a_{1,2},a_{2,2}\in z_3^a, a_{1,2}^1=0,a_{2,2}^1=0\\
            \tilde{c}(o_{1,2}^1
            ,o_{2,2}^1,2,1)/\alpha_1&\text{if } a_{1,2},a_{2,2}\in z_3^a, a_{1,2}^2=0,a_{2,2}^1=0\\
            \tilde{c}(o_{1,2}^1,o_{2,2}^1,1,2)/\alpha_1&\text{if } a_{1,2},a_{2,2}\in z_3^a, a_{1,2}^1=0,a_{2,2}^2=0\\
            \tilde{c}(o_{1,2}^1,o_{2,2}^1,2,2)/\alpha_1&\text{if } a_{1,2},a_{2,2}\in z_3^a, a_{1,2}^2=0,a_{2,2}^2=0\\
            0&\text{o.w.}
        \end{cases}.
    \end{align*}
    \end{itemize}

Let $\alpha_0=\max_{y_1,y_2,u_1,u_2}\tilde{c}(y_1,y_2,u_1,u_2)$, which is supposed to be positive without loss of optimality (see a discussion in \S\ref{sec:proof of lemma: limited communication}). 
We set $\alpha_1=2\alpha_0$,  and hence for any $z_3^a\in\cZ_3^a$, $\cK_3(z_3^a)\leq \frac{1}{2}$  always holds. We set $\alpha_2\in(0,1)$. Also, we restrict agents to decide their communication strategies only based on their common information.
Under such a construction, $\cL$ satisfies  the following conditions: 
\begin{itemize}
    \item $\cL$ is QC: For any $i_1,i_2\in[2],h_1,h_2\in[3]$, agent $(i_1,h_1)$ does not influence $(i_2,h_2)$, because agent $(i_1,h_1)$ cannot influence the observation of agent $(i_2,h_2)$, and the baseline sharing is null. 
    \item $\cL$ satisfies Assumptions \ref{gamma observability} and \ref{weak gamma observability}: We prove this by showing that each agent $i\in[2]$ satisfies $\gamma$-observability. For any $i\in[2],h=1, b_1,b_2\in\Delta(\cS)$, we have  
    \begin{align*}
        &\norm{\OO_{i,h}^\top (b_1-b_2)}_1
        =\sum_{o_{i,h}^1\in[2]}\sum_{o_{i,h}^2\in \cS}\bigg|\sum_{s_h\in\cS}(b_1(s_h)-b_2(s_h))\OO_{i,h}((o_{i,h}^1,o_{i,h}^2)\given s_h)\bigg|\\
        &\quad\ge\sum_{o_{i,h}^2\in \cS}\bigg|\sum_{o_{i,h}^1\in[2]}\sum_{s_h\in\cS}(b_1(s_h)-b_2(s_h))\OO_{i,h}((o_{i,h}^1,o_{i,h}^2)\given s_h)\bigg|\\
        &\quad= \sum_{o_{i,h}^2\in \cS}\bigg|\sum_{s_h\in\cS}\sum_{o_{i,h}^1\in[2]}(b_1(s_h)-b_2(s_h))\mathbbm{1}[o_{i,h}^1=s_h^{i+2h-2}](\frac{1-\alpha_2}{4}+\alpha_2\mathbbm{1}[o_{i,h}^2=s_h])\bigg|\\
        &\quad= \sum_{o_{i,h}^2\in \cS}\bigg|\frac{1-\alpha_2}{4}(\sum_{s_h\in\cS}(b_1(s_h)-b_2(s_h)))+\alpha_2(b_1(o_{i,h}^2)-b_2(o_{i,h}^2))\bigg|\\
        &\quad=\sum_{o_{i,h}^2\in \cS}\alpha_2|b_1(o_{i,h}^2)-b_2(o_{i,h}^2)|=\alpha_2\norm{b_1-b_2}_1.
    \end{align*}
    For any $i\in[2], h=2,3$, the proof is similar, by replacing $o_{i,h}^1\in [2]$ with $o_{i,h}^1\in \tilde{Y}_i$ for $h=2$ and replacing the space $o_{i,h}^1\in [2]$ with $\{\emptyset\}$  for $h=3$.
    \item $\cL$ satisfies Assumption \ref{limited communication strategy}, since we restrict agents to decide their communication actions only based on the common information.    
    \item $\cL$ does not satisfy \Cref{useless action}, since the joint control action histories $a_{1:3}$ do not influence the underlying states, but they are not removed in the information. 
 \end{itemize}

Now, we show that any team-optimal strategy of $\cL$ will give us the decision rules $\gamma_1,\gamma_2$ to solve $\cT\cD$. 
Given $(g_{1:3}^{a,\ast},g_{1:3}^{m,\ast})$ being a team-optimal strategy of $\cL$,  we can construct $(g_{1:3}^{a,'}, g_{1:3}^{m,'})$ as 
\begin{align*}
    &\forall i\in[2],\forall \tau_{i,1^-},\tau_{i,2^-}, \tau_{i,3^-}, g_{i,1}^{m,'}(\tau_{i,1^-})=(0),\quad g_{i,2}^{m,'}(\tau_{i,2^-})=(0,0,0), g_{i,3}^{m,'}(\tau_{i,3^-})=1\\
    &\forall i\in[2],\forall \tau_{i,2^+}, g_{i,2}^{a,'}(\tau_{i,2^+})=\begin{cases}
        (0,o_{i,1})&\text{if the first term of $[g_{i,2}^{a,'}(\tau_{i,2^+})]$ is 0}\\
         (o_{i,1},0)&\text{o.w.}
    \end{cases}\\
    &\forall i\in[2],\forall \tau_{i,3^+}, g_{i,3}^{a,'}(\tau_{i,3^+})=\begin{cases}
        a_{3-i,2}^2&\text{if $a_{3-i,2}^1=0$}\\
        a_{3-i,2}^1&\text{o.w.}
    \end{cases}.
\end{align*}
One can verify that $J_\cL(g_{1:3}^{a,\ast},g_{1:3}^{m,\ast})\le J_\cL(g_{1:3}^{a,'},g_{1:3}^{m,'})$. This is because, $(g_{1:3}^{a,'},g_{1:3}^{m,'})$ can always achieve $r_3=1$, which means $\EE[\sum_{h=1}^3r_h\given (g_{1:3}^{a,'},g_{1:3}^{m,'})]\ge \EE[\sum_{h=1}^3r_h\given g_{1:3}^{a,\ast},g_{1:3}^{m,\ast}]$. Also, at timesteps $h=1,2$,  $(g_{1:3}^{a,'},g_{1:3}^{m,'})$ has no additional sharing; if $(g_{1:3}^{a,\ast},g_{1:3}^{m,\ast})$ has additional sharing, then the communication cost $\sum_{h=1}^3\kappa_h$ is at least 1; if $(g_{1:3}^{a,\ast},g_{1:3}^{m,\ast})$ has no additional sharing, then $(g_{1:3}^{a,\ast},g_{1:3}^{m,\ast})$ and $(g_{1:3}^{a,'},g_{1:3}^{m,'})$ have the same communication cost of $\sum_{h=1}^3\kappa_h=\kappa_3\le \frac{1}{2}$. Therefore, $\EE[\sum_{h=1}^3\kappa_h\given (g_{1:3}^{a,'},g_{1:3}^{m,'})]\le \EE[\sum_{h=1}^3\kappa_h\given g_{1:3}^{a,\ast},g_{1:3}^{m,\ast}]$, and thus $J_\cL(g_{1:3}^{a,\ast},g_{1:3}^{m,\ast})\le J_\cL(g_{1:3}^{a,'},g_{1:3}^{m,'})$. 

From the above, we know that $(g_{1:3}^{a,'},g_{1:3}^{m,'})$ is also a  team-optimal strategy. Let $U_1=f_1(a_{1,2}):=2-\mathds{1}[a_{1,2}^1=0], U_2=f_2(a_{2,2}):=2-\mathds{1}[a_{2,2}^1=0]$, then $J_{\cL}(g_{1:3}^{a,'},g_{1:3}^{m,'})=\EE[\sum_{h=1}^3r_h-\kappa_h\given g_{1:3}^{a,'},g_{1:3}^{m,'}]=1-\EE[\kappa_3\given g_{1:3}^{a,'},g_{1:3}^{m,'}]=1-\frac{1}{\alpha_1}\EE[\tilde{c}(o_{1,2}^1,o_{2,2}^1,U_1,U_2)]$, where $U_1=f_1(a_{1,2})=f_1(g_{1,2}^{a,'}(\tau_{1,2^+})),U_2=f_2(a_{2,2})=f_2(g_{2,2}^{a,'}(\tau_{2,2^+}))$, which means  that  $(g_{1,2}^{a,'},g_{2,2}^{a,'})$ can minimize $\kappa_3$ through properly choosing $a_2$ (and $U_{1:2}$) if there is no additional sharing in the first two timesteps, i.e., $\tau_{i,2^+}=\{o_{i,1},a_{i,1},o_{i,2},m_{1:2}\}$. By construction, $\kappa_{3}$ only depends on $o_{1,2}^1,o_{2,2}^1$,  and $U_{1:2}$, and is irrelevant of $\{o_{1,1},a_{1,1}, o_{1,2}^2\}$, and $m_{i,1}=(0),m_{i,2}=(0,0,0),\forall i\in[2]$ always hold. Hence,   $\{o_{1,1},o_{1,2},o_{1,3}^2\}$ are useless information for agent $1$ to choose $a_{1,2}$ (and $U_1$),  and to minimize  $\EE[\kappa_3]$.  Therefore, not using them in $g_{1,2}^{a,'}$ does not lose any optimality. Hence, we can consider the $g_{1,2}^{a,'}$ that only has $o_{1,2}^1$ as input. In the same way,   we can consider the $g_{2,2}^{a,'}$ that only has $o_{2,2}^1$ as input. Therefore, $J_{\cL}(g_{1:3}^{a,'},g_{1:3}^{m,'})=1-\sum_{o_{1,3}^1,o_{1,3}^2}\frac{1}{\alpha_1}\tilde{c}(o_{1,2}^1,o_{2,2}^1, f_1(g_{1,2}^{a,'}( o_{1,2}^1)), f_2(g_{2,2}^{a,'}(o_{2,2}^1)))\tilde{p}(o_{1,2}^1,o_{2,2}^1)$. Further, we can leverage $\tilde{\gamma}^\ast_1=f_1\circ g_{1,2}^{a,'}, \tilde{\gamma}^\ast_2=f_2\circ g_{2,2}^{a,'}$ to minimize  the expected cost $\tilde{J}$ of the TDP, where $\circ$ denotes function composition. Therefore, 
from the \texttt{NP-hardness} of TDPs  \cite[Corollary 4.1]{teamdecisionhardness}, we complete our proof. 
\end{proof}

\subsection{Proof of \Cref{lemma: weak gamma observability}}

\begin{proof}
    We prove this by showing a reduction from the hardness of finding an $\epsilon$-optimal strategy for   POMDPs (\Cref{prop: epsilon-additive}). Given any POMDP $\cP=(\cS^{\cP},\cA^{\cP},\cO^{\cP},\{\OO_h^{\cP}\}_{h\in[H^\cP]},\{\TT_h^{\cP}\}_{h\in[H^\cP]},\{\cR_h^{\cP}\}_{h\in[H^\cP]}, \mu_1^{\cP})$ with  rewards bounded in $[0,\frac{1}{2}]$, we can construct an LTC $\cL$ with 2 agents as follows:
    \begin{itemize}
    \item Number of agents: $n=2$; length of episode: $H=H^\cP$.
        \item $\cS=\cS^\cP\times[2];$ for any $s\in \cS$, we can split state as two parts: $s=(s^1,s^2), s^1\in\cS, s^2\in[2]$. 
        \item Initial state distribution: $\forall s_1\in \cS, \mu_1(s_1)=\frac{\mu_1^\cP(s_1^1)}{2}$.
        \item Control action space: For any $h\in[H]$,  $\cA_{1,h}=\cA^\cP_h,\cA_{2,h}=[2]$.
        \item Transition: For any $h\in[H]$,   $\forall s_h,s_{h+1}\in \cS,a_h\in \cA_h, \TT_h(s_{h+1}\given s_h,a_h)=\TT^\cP_h(s_{h+1}^1\given s_h^1,a_{1,h})\mathds{1}[s_{h+1}^2=a_{2,h}]$.
        \item Observation space: For any $h\in[H]$, $\cO_{1,h}=\cO^\cP,\cO_{2,h}=\cS$.
        \item Emission: For any $h\in[H],\forall o_h\in\cO_h,s_h\in\cS$, $\OO_h(o_h\given s_h)=\OO_h^\cP(o_{1,h}\given s_h^1)\mathds{1}[o_{2,h}=s_h]$.
        \item Reward functions: For any $ h\in[H]$, $\forall s_h\in\cS,a_h\in\cA_h,\cR_{h}(s_h,a_h)=\cR_h^\cP(s_h^1,a_{1,h})/H$.
        \item Baseline sharing: For any $h\in[H],z_h^b=\{o_{1,h},a_{1,h-1}\}$.
        \item Communication action space: For any $h\in[H],\cM_{1,h}=\{\emptyset\}, \cM_{2,h}=\{0,1\}^{2h-1}$. For any $p_{1,h^-}\in\cP_{1,h^-}, p_{2,h^-}\in \cP_{2,h^-}, m_h\in\cM_h, \phi_{1,h}(p_{1,h^-},m_{1,h})=\{m_{1,h}\}, \phi_{2,h}(p_{2,h^-},m_{2,h})=\{o_{2,k}\given \forall k\le h, \text{ the }2k-1\text{-th digit of $m_{2,h}$ is 1 and } o_{2,k}\in p_{2,h^-}\}\cup\{a_{2,k}\given \forall k< h, \text{ the }2k\text{-th digit of $m_{2,h}$ is 1 and } a_{2,k}\in p_{2,h^-}\}\cup\{m_{2,h}\}.$
        \item Communication cost functions: For any $h\in[H], z_h^a\in\cZ_h^a,\cK_{h}(z_h^a)=\mathds{1}[z_h^a\neq \{m_h\}]$, which means that the communication cost is 1 unless there is no additional sharing.
        \item We restrict the communication strategy to only use $c_h$ as input. 
    \end{itemize}
We first verify that $\cL$ is QC and satisfies Assumptions \ref{gamma observability}, \ref{limited communication strategy}, and \ref{useless action}.
    \begin{itemize}
        \item {$\cL$ is QC:} For any $\forall h_1<h_2\le H$, agent $(2,h_1)$ does not influence agent $(1,h_2)$ under baseline sharing,  since agent $(2,h_1)$ does not influence $s_h^1,\forall h\in[H]$, then does not influence $o_{1,h},\forall h\in[H]$. Also, agent 2 shares nothing via baseline sharing. Therefore, agent $(2,h_1)$ does not influence  agent $(1,h_2)$. 
        For any $h_1<h_2\le H$, under baseline sharing, $p_{1,h_1^-}=\emptyset$. Then, we have  $\sigma(\tau_{1,h_1^-})\subseteq \sigma(c_{h_1^-})\subseteq\sigma(c_{h_2^-})\subseteq \sigma(\tau_{2,h_2^-})$.
        \item {$\cL$ satisfies Assumption \ref{gamma observability}:} For any $h\in[H], b_1,b_2\in \Delta(\cS)$, $\OO_h$ satisfies that
        \begin{align*}
        &\norm{\OO_{h}^\top (b_1-b_2)}_1
        =\sum_{o_{1,h}\in\cO^\cP}\sum_{o_{2,h}\in \cS}\bigg|\sum_{s_h\in\cS}(b_1(s_h)-b_2(s_h))\OO_{h}((o_{1,h},o_{2,h})\given s_h)\bigg|\\
        &\quad\ge \sum_{o_{2,h}\in \cS}\bigg|\sum_{o_{1,h}\in\cO^\cP}\sum_{s_h\in\cS}(b_1(s_h)-b_2(s_h))\OO_{1,h}(o_{1,h}\given s_h)\OO_{2,h}(o_{2,h}\given s_h)\bigg|\\
        &\quad= \sum_{o_{2,h}\in \cS}\bigg|\sum_{s_h\in\cS}(b_1(s_h)-b_2(s_h))\OO_{2,h}(o_{2,h}\given s_h)\sum_{o_{1,h}\in\cO^\cP}\OO_{1,h}(o_{1,h}\given s_h)\bigg|\\
        &\quad= \sum_{o_{2,h}\in \cS}\bigg|\sum_{s_h\in\cS}(b_1(s_h)-b_2(s_h))\mathds{1}[o_{2,h}=s_h]\bigg|\\
        &\quad= \sum_{o_{2,h}\in \cS}|b_1(o_{2,h})-b_2(o_{2,h})|=\norm{b_1-b_2}_1,
        \end{align*}
        showing that $\gamma$-observability is satisfied with $\gamma=1$. 
        \item {$\cL$ satisfies Assumption \ref{limited communication strategy}:} For any $h\in[H]$, we restricted that  each agent $i$ decides $m_{i,h}$ based on $c_h$ only. 
        \item {$\cL$ satisfies Assumption \ref{useless action}:} For any $h\in[H-1], a_{1,h}$ influences $s_{h+1}^1$, $a_{2,h}$ influences $s_{h+1}^2$. 
        \item $\cL$ does not satisfy \Cref{weak gamma observability}. This is because for any $h\in[H]$ and fixed  $s_h^1\in \cS^\cP$, we can define $s^1=(s_h^1,1), s^2=(s_h^1,2)$, and $b_1=\delta_{s^1}\in\Delta(\cS), b_2=\delta_{s^2}\in\Delta(\cS)$, where $\delta_{s^1},\delta_{s^2}$ denote the distributions over $\cS$ only with mass on $s^1,s^2$, respectively. Then, one can verify that $b_1\neq b_2$ but $\OO_{-1,h}^\top b_1=\OO_{-1,h}^\top b_2$.
    \end{itemize}
    Consider the communication strategy $g_{1:H}^{m,'}$ that yields no additional sharing, namely, $\forall h\in[H], \tau_{2,h^-}\in\cT_{2,h^-}, g_{2,h}^{m,'}(\tau_{2,h^-})=\bm{0}_h$.    For any $\epsilon$-optimal strategy $(g_{1:H}^{a,\ast},g_{1:H}^{m,\ast})$ with $\epsilon\in[0,\frac{1}{4})$, we claim that $(g_{1:H}^{a,\ast},g_{1:H}^{m,'})$ is also an $\epsilon$-optimal strategy. 
    
    This is because, comparing to strategy $(g_{1:H}^{a,\ast},g_{1:H}^{m,'})$, for the trajectories where $(g_{1:H}^{a,\ast},g_{1:H}^{m,\ast})$ leads to no  additional sharing, these two strategies output the same actions and gain the same total  return; for the other trajectories where $(g_{1:H}^{a,\ast},g_{1:H}^{m,\ast})$ leads to some additional sharing, i.e., $z_h\neq \{m_h\}$ for some $h\in[H]$, then it has the return of $\sum_{t=1}^H(r_t-\kappa_t)\le \sum_{t=1}^H(r_t)-\kappa_h\le \frac{1}{2H}\cdot H-1<0$, which is less than when replacing it by $(g_{1:H}^{a,\ast},g_{1:H}^{m,'})$ with return $\sum_{t=1}^H(r_t-\kappa_t)= \sum_{t=1}^Hr_t\ge 0$. 
    
    Therefore, $J_\cL(g_{1:H}^{a,\ast},g_{1:H}^{m,'})\ge J_\cL(g_{1:H}^{a,\ast},g_{1:H}^{m,\ast})$, and $(g_{1:H}^{a,\ast},g_{1:H}^{m,'})$ is also an $\epsilon$-optimal strategy with $\epsilon\in[0,\frac{1}{4})$.

    Meanwhile, any $(g_{1:H}^{a,\ast},g_{1:H}^{m,'})$ being an $\frac{\epsilon}{H}$-team-optimal strategy of $\cL$ will directly give an $\epsilon$-team-optimal strategy of $\cP$ as $\{g_{1,h}^{a,\ast}\}_{h\in [H]}$, since when there is no sharing, the decision process is only controlled by agent $1$.  From Proposition \ref{prop: epsilon-additive}, we complete the proof.
\end{proof}

\begin{remark}
    For each construction in the proofs of \Cref{lemma: nonqc_hardness,lemma: limited communication,lemma: useless action,lemma: weak gamma observability}, if we slightly change the construction to make it satisfy the corresponding assumption, then the associated  hardness can be avoided. For example, if we require the constructed instance in the proof of \Cref{lemma: limited communication} to satisfy \Cref{limited communication strategy}, i.e., the communication strategy $g_{1:3}^a$ only takes common information as input. Then, finding the optimal communication strategy at timestep $h=3$ can no longer be reduced to TDPs, and such LTC problems can be computed via \Cref{main algorithm} in polynomial time. 
\end{remark}

\section{Deferred Details of  \S\ref{sec:positive_results}}
\label{proof details sec 4}

\subsection{Proof of \Cref{prop: equivalence of LTC and Dec-POMDP}}\label{sec:equivalence_appendix}

\begin{proof}
    Given any strategy $(g_{1:H}^{a},g_{1:H}^{m}), g_{1:H}^{a}\in\cG_{1:H}^a, g_{1:H}^m\in\cG_{1:H}^m$, we can define $\tilde{g}_{1:\tilde{H}}=(g^m_1,g^a_1,\cdots,g^m_H,g^a_H)$. From the construction of $\cD_\cL$, comparing to applying $(g_{1:H}^{a},g_{1:H}^{m})$ in $\cL$, if we apply $\tilde{g}_{1:\tilde{H}}$ in $\cD_\cL$, then it is easy to verify that $\forall h\in [H], \tilde{r}_{2h-1}=-\kappa_h, \tilde{r}_{2h}=r_h$ always holds. Therefore, $J_{\cD_\cL}(\tilde{g}_{1:\tilde{H}})=J_\cL(g_{1:H}^{a},g_{1:H}^{m})$. In the same way, for any strategy $\tilde{g}_{1:\tilde{H}}$, we can define $\tilde{g}_{1:\tilde{H}}=(g^m_1,g^a_1,\cdots,g^m_H,g^a_H)$, and verify that $J_\cL(g_{1:H}^{a},g_{1:H}^{m})=J_{\cD_\cL}(\tilde{g}_{1:\tilde{H}})$.
\end{proof}

\subsection{Proof of \Cref{reformulate QC}}\label{sec:reformulate_QC_appendix}

\begin{proof} 
    Firstly, we prove the QC case. To show $\cD_\cL$ is QC, we  need to prove that $\forall i_1,i_2\in[n], h_1,h_2\in[\tilde{H}]$, if agent $(i_1,h_1)$ influences agent $(i_2,h_2)$ with $h_1<h_2$, then $\sigma(\tilde{\tau}_{i_1,h_1})\subseteq \sigma(\tilde{\tau}_{i_2,h_2})$, where we use $\tilde\tau_{i,h}$ to denote the  information available to agent $(i,h)$ in $\cD_\cL$.
    We prove it by considering the following cases:
    \begin{itemize}
        \item If $h_1=2t_1-1$ with $t_1\in[H]$, by the construction of $\cD_\cL$ and Assumption \ref{limited communication strategy}, we have       $\tilde{\tau}_{i_1,h_1}=\tilde{c}_{h_1}=c_{t_1^-}\subseteq \tilde{\tau}_{i_2,h_2}$, since common information accumulates over time by definition, and will always be included in the available information $\tilde\tau_{i,h}$ in later steps with $h>h_1$. Thus,  $\sigma(\tilde{\tau}_{i_1,h_1})\subseteq \sigma(\tilde{\tau}_{i_2,h_2})$.
            \item If $h_1=2t_1,h_2=2t_2$ with $ t_1,t_2\in[H]$, then $\tilde{\tau}_{i_1,h_1}=\tau_{i_1,t_1^+}=\tau_{i_1,t_1^-}\cup z_{t_1}^a$ by definition.  Consider agent $(i_1,t_1)$ and $(i_2,t_2)$ in $\cL$. {From \Cref{lemma: action influence},  we know that  if any agent in a  control step influences another agent in a  control step in $\cD_\cL$, then the former also influences the latter in $\cL$ under  baseline sharing, and thus $\sigma(\tau_{i_1,t_1^-})\subseteq \sigma(\tau_{i_2,t_2^-})$ since $\cL$ is QC. This further implies $\sigma(\tau_{i_1,t_1^-})\subseteq \sigma(\tau_{i_2,t_2^+})$ since $\sigma(\tau_{i_2,t_2^-})\subseteq \sigma(\tau_{i_2,t_2^+})$.} 
       Also, $z_{t_1}^a\subseteq c_{t_1^+}\subseteq c_{t_2^+}\subseteq \tau_{i_2,t_2^+}=\tilde \tau_{i_2,h_2}$ by the accumulation of $c_{h^+}$ over time. Thus, we have  $\sigma(\tilde{\tau}_{i_1,h_1})\subseteq \sigma(\tilde{\tau}_{i_2,h_2})$.  
        \item If $h_1=2t_1,h_2=2t_2-1, t_1,t_2\in[H]$, then $\tilde{\tau}_{i_2,h_2}=\tilde{c}_{h_2}$. Hence, $\exists i_3\in[n],i_3\neq i_1$, such that $ \tilde{\tau}_{i_2,h_2}\subseteq \tilde{c}_{h_2+1}\subseteq\tilde{\tau}_{i_3,h_2+1}$. Since agent $(i_1,h_1)$ influences agent $(i_2,h_2)$, we   know that agent $(i_1,h_1)$ also influences agent $(i_3,h_2+1)$ in $\cD_\cL$. 
        From \Cref{lemma: action influence}, similarly as above, we know that $\sigma(\tau_{i_1,t_1^-})\subseteq \sigma(\tau_{i_3,t_2^-})$ since $\cL$ is QC. From Assumption \ref{assumption:sigma_include} and $i_1\neq i_3$,  we know that $\sigma(\tau_{i_1,t_1^-})\subseteq \sigma(c_{t_2^-})=\sigma(\tilde{\tau}_{i_2,h_2})$. Also, it holds that $\tau_{i_1,t_1^+}\backslash\tau_{i_1,t_1^-}\subseteq c_{t_1^+}\subseteq c_{t_2^-}$.  Hence, we have $\sigma(\tilde{\tau}_{i_1,h_1})=\sigma(\tau_{i_1,t_1^+})\subseteq \sigma(c_{t_2^-})=\sigma(\tilde{\tau}_{i_2,h_2})$. 
    \end{itemize}
    Secondly, we prove the sQC case. In $\cD_\cL$, for any $i_1,i_2\in[n], h_1,h_2\in[\tilde{H}]$, suppose agent $(i_1,h_1)$ influences $(i_2,h_2)$. 
    From the proof  above, we know that $\sigma(\tilde{\tau}_{i_1,h_1})\subseteq \sigma(\tilde{\tau}_{i_2,h_2})$. We only need to prove $\sigma(\tilde{a}_{i_1,h_1})\subseteq \sigma(\tilde{\tau}_{i_2,h_2})$.
    \begin{itemize}
        \item If $h_1=2t_1-1$ with $t_1\in[H]$, then we know that $\tilde{a}_{i_1,h_1}=m_{i_1,t_1}$. From Assumption \ref{assumption:information evolution}, we know that $m_{i_1,t_1}\subseteq z_{i_1,t_1}^a$. Then,  we get $\sigma(\tilde{a}_{i_1,h_1})\subseteq \sigma(\tilde{z}_{h_1+1})\subseteq \sigma(\tilde{c}_{h_2})\subseteq \sigma(\tilde{\tau}_{i_2,h_2})$. 
        \item If  $h_1=2t_1,h_2=2t_2$ with $ t_1,t_2\in[H]$, then from \Cref{lemma: action influence}, we know that $\sigma(\tilde{a}_{i_1,h_1})\subseteq \sigma(\tilde{\tau}_{i_2,h_2})$.
        \item If $h_1=2t_1,h_2=2t_2-1$ with $t_1,t_2\in[H]$, then $\tilde{\tau}_{i_2,h_2}=\tilde{c}_{h_2}$. Hence, $\exists i_3\in[n],i_3\neq i_1$, such that $ \tilde{\tau}_{i_2,h_2}\subseteq \tilde{c}_{h_2+1}\subseteq\tilde{\tau}_{i_3,h_2+1}$. Since agent $(i_1,h_1)$ influences $(i_2,h_2)$, we thus  know that agent $(i_1,h_1)$ also influences agent $(i_3,h_2+1)$. 
        Since $\cL$ is sQC, we know that $\sigma(a_{i_1,t_1})\subseteq \sigma(\tau_{i_3,t_2^-})$  in $\cL$ by \Cref{lemma: action influence}. From Assumption \ref{assumption:sigma_include} and $i_1\neq i_3$,  we know that $\sigma(\tilde{a}_{i_1,h_1})=\sigma(a_{i_1,t_1})\subseteq \sigma(c_{t_2^-})=\sigma(\tilde{\tau}_{i_2,h_2})$.
    \end{itemize}
    This completes the proof.
\end{proof}

\subsection{Proof of Lemma \ref{lemma:QC to sQC}}
\begin{proof}
From the construction of $\cD_\cL^\dag$, since $\cD_\cL^\dag$ requires agents to share more than $\cD_\cL$, it is easy to observe that $\forall h\in[\tilde{H}],i\in[n], \tilde{c}_h\subseteq \Breve{c}_h, \tilde{\tau}_{i,h}\subseteq \Breve{\tau}_{i,h}$.\\
Let $i_1,i_2\in[n],h_1,h_2\in[\tilde{H}], h_1<h_2$, and agent $(i_1,h_1)$ influences agent $(i_2,h_2)$ in $\cD_\cL^\dag$. 
\begin{itemize} 
    \item If $h_1=2t_1-1$ with $t_1\in[H]$, then $h_1$ is a communication step. Hence,  $\Breve{\tau}_{i_1,h_1}=\Breve{c}_{h_1}\subseteq \Breve{c}_{h_2}$ and $\tilde{a}_{i_1,h_1}=m_{i_1,t_1}\subseteq \Breve{c}_{h_1+1}\subseteq \Breve{c}_{h_2}$ from Assumption \ref{assumption:information evolution}. Therefore, we have $\sigma(\Breve{\tau}_{i_1,h_1})\cup \sigma(\Breve{a}_{i_1,h_1})\subseteq \sigma(\Breve{c}_{h_2})\subseteq \sigma(\Breve{\tau}_{i_2,h_2})$.
    \item If $h_1=2t_1, h_2=2t_2-1$ with $t_1,t_2\in[H]$, then $\Breve{\tau}_{i_2,h_2}=\Breve{c}_{h_2}$. If agent $(i_1,h_1)$ does not influence $(i_2,h_2)$ in $\cD_\cL$, but agent $(i_1,h_1)$ influences $(i_2,h_2)$ in $\cD_\cL^\dag$, then it means $\Breve{a}_{i_1,h_1}\in \Breve{\tau}_{i_2,h_2}$ but $\tilde{a}_{i_1,h_1}\notin \tilde{\tau}_{i_2,h_2}$. This can only happen when $\sigma(\tilde{\tau}_{i_1,h_1})\subseteq \sigma(\tilde{c}_{h_2})\subseteq \sigma(\Breve{c}_{h_2})$, and $\tilde{a}_{i_1,h_1}\subseteq \Breve{c}_{h_2}$. Also, from the construction of $\cD_\cL^\dag$, we know that $\Breve{\tau}_{i_1,h_1}\backslash \tilde{\tau}_{i_1,h_1}\subseteq \Breve{c}_{h_1}$. Therefore, we have $\sigma(\Breve{\tau}_{i_1,h_1})\cup \sigma(\tilde{a}_{i_1,h_1})\subseteq \sigma(\Breve{c}_{h_2})\subseteq \sigma(\Breve{\tau}_{i_2,h_2})$.
    
    If agent $(i_1,h_1)$ influences agent $(i_2,h_2)$ in $\cD_\cL$, then we claim that $\tilde{a}_{i_1,h_1}$ influences $\tilde{s}_{h_1+1}$ in $\cD_\cL$, since otherwise, $\tilde{a}_{i_1,h_1}$ does not influence any $\tilde{s}_{h},\tilde{o}_h$, and $\tilde{a}_h$ with $h>h_1$; also, from Assumption \ref{useless action}, $\tilde{a}_{i_1,h_1}$ is removed from $\tilde{\tau}_h, \forall h>h_1$, and thus agent $(i_1,h_1)$ cannot influence agent $(i_2,h_2)$. Meanwhile, 
    from the QC IS of $\cD_\cL$, we know that $\sigma(\tilde{\tau}_{i_1,h_1})\subseteq \sigma(\tilde{\tau}_{i_2,h_2})=\sigma(\tilde{c}_{h_2})$. Then, from the construction of $\cD_\cL^\dag$, we know that $\tilde{a}_{i_1,h_1}$ is added in $\Breve{c}_{h_2}$, i.e.,  $\tilde{a}_{i_1,h_1}\in \Breve{c}_{h_2}$. Still, due to $\Breve{\tau}_{i_1,h_1}\backslash \tilde{\tau}_{i_1,h_1}\subseteq \Breve{c}_{h_1}$, we further have  $\sigma(\Breve{\tau}_{i_1,h_1})\cup \sigma(\tilde{a}_{i_1,h_1})\subseteq \sigma(\Breve{\tau}_{i_2,h_2})$. 
    \item If $h_1=2t_1, h_2=2t_2$ with $t_1,t_2\in[H]$. If agent $(i_1,h_1)$ does not influence $(i_2,h_2)$ in $\cD_\cL$, then it means that adding  $\tilde{a}_{i_1,h_1}$ in $\Breve{c}_{h_2}$ via strict expansion leads to such influence. Then, it must hold that $\sigma(\tilde{\tau}_{i_1,h_1})\subseteq \sigma(\tilde{c}_{h_2})\subseteq \sigma(\Breve{c}_{h_2})$, and $\tilde{a}_{i_1,h_1}\subseteq \Breve{c}_{h_2}$. We can conclude $\sigma(\Breve{\tau}_{i_1,h_1})\cup \sigma(\tilde{a}_{i_1,h_1})\subseteq \sigma(\Breve{c}_{h_2})\subseteq \sigma(\Breve{\tau}_{i_2,h_2})$.\\      
   If agent $(i_1,h_1)$ influences $(i_2,h_2)$ in $\cD_\cL$, then we know that $\tilde{a}_{i_1,h_1}$ influences $\tilde{s}_{h_1+1}$ in $\cD_\cL${, similarly as shown in the case for $h_2=2t_2-1$.
    Therefore, from Assumption \ref{weak gamma observability}, we know that there exists some $i_3\neq i_1$ such that agent $\tilde{a}_{i_1,h_1}$ influences agent $(i_3,h_1+1)$ in $\cD_\cL$. Then, from the QC IS of $\cD_\cL$, we know that $\sigma(\tilde{\tau}_{i_1,h_1})\subseteq \sigma(\tilde{\tau}_{i_3,h_1+1})$. Meanwhile, from Assumption \ref{assumption:information evolution} (e), we know that $\tilde{\tau}_{i_3,h_1+1}\subseteq \tilde{\tau}_{i_3,h_2}$. Therefore, we can get $\sigma(\tilde{\tau}_{i_1,h_1})\subseteq\sigma(\tilde{\tau}_{i_3,h_2})$ and further $\sigma(\tilde{\tau}_{i_1,h_1})\subseteq \sigma(\tilde{c}_{h_2})$ due to Assumption \ref{assumption:sigma_include} and $i_3\neq i_1$. Then, from the construction of $\cD_\cL^\dag$, we know that $\tilde{a}_{i_1,h_1}$ is added in $\Breve{c}_{h_2}$. Together with  $\Breve{\tau}_{i_1,h_1}\backslash\tilde{\tau}_{i_1,h_1}\subseteq \Breve{c}_{h_1}\subseteq \Breve{c}_{h_2}$, we can conclude that $\sigma(\Breve{\tau}_{i_1,h_1})\cup\sigma(\Breve{a}_{i_1,h_1})\subseteq\sigma(\Breve{\tau}_{h_2}).$
   }
\end{itemize}
This completes the proof. 
\end{proof}

\subsection{Proof of Theorem \ref{theorem: sQC to QC}}
\begin{proof}

Firstly, we claim that given any strategy $\Breve{g}_{1:\Breve{H}}$ and $\tilde{g}_{1:\tilde{H}}=\varphi(\Breve{g}_{1:\Breve{H}},\cD_\cL)$, $J_{\cD_\cL^\dag}(\Breve{g}_{1:\Breve{H}})=J_{\cD_\cL}(\tilde{g}_{1:\tilde{H}})$, where the function $\varphi$ is given by \Cref{algorithm varphi}. It suffices to prove that $\tilde{g}_{i,h}(\tilde{\tau}_{i,h})=\Breve{g}_{i,h}(\Breve{\tau}_{i,h})$ always holds for any $\tilde \tau_{i,h}$. Namely, for any $i\in[n], h\in[\tilde{H}]$, and  $\tilde{\tau}_{i,h}$,  \Cref{algorithm varphi} can compute the associated $\Breve{\tau}_{i,h}$ from the expansion in \Cref{construction:QC to sQC}, and use it as the input of $\Breve{g}_{i,h}$ (Line $11$ of \Cref{algorithm varphi}).  
Let $\Breve{\tau}_{i,h}'$ be the information constructed by \Cref{algorithm varphi}, which is the  
input of $\Breve{g}_{i,h}$ used in Line 11 of the algorithm. From the construction of the algorithm and the expansion, we have $\tilde{\tau}_{i,h}\subseteq \Breve{\tau}_{i,h}', \tilde{\tau}_{i,h}\subseteq \Breve{\tau}_{i,h}$;  $\Breve{\tau}_{i,h}\backslash \tilde{\tau}_{i,h}$ and $\Breve{\tau}_{i,h}'\backslash \tilde{\tau}_{i,h}$ only consist of some actions. For any action $\Breve{a}_{j,t}$ with $j\in[n],t<h$, it will be added in $\Breve{\tau}_{i,h}'\backslash \tilde{\tau}_{i,h}$ if and only if it lies in $\Breve{\tau}_{i,h}\backslash \tilde{\tau}_{i,h}$, since we use the same condition in construction in the algorithm as that in the expansion (i.e., \Cref{construction:QC to sQC}).

Since $\cD_\cL^\dag$ has larger strategy spaces, we have$\max_{\tilde{g}_{1:\tilde{H}}\in\tilde{G}_{1:\tilde{H}}}J_{\cD_\cL}(\tilde{g}_{1:\tilde{H}})\le \max_{\Breve{g}_{1:\Breve{H}}\in\Breve{G}_{1:\Breve{H}}}J_{\cD_\cL^\dag}(\Breve{g}_{1:\Breve{H}})$. Let $\Breve{g}^\ast_{1:\Breve{H}}$ be the strategy satisfying $J_{\cD_\cL^\dag}(\Breve{g}^\ast_{1:\Breve{H}})\ge \max_{\Breve{g}_{1:\Breve{H}}\in\Breve{G}_{1:\Breve{H}}}J_{\cD_\cL^\dag}(\Breve{g}_{1:\Breve{H}})-\epsilon$. Then, we have $J_{\cD_\cL}(\varphi(\Breve{g}_{1:\Breve{H}}^\ast,\cD_\cL))=J_{\cD_\cL^\dag}(\Breve{g}^\ast_{1:\Breve{H}})\ge \max_{\Breve{g}_{1:\Breve{H}}\in\Breve{G}_{1:\Breve{H}}}J_{\cD_\cL^\dag}(\Breve{g}_{1:\Breve{H}})-\epsilon\ge \max_{\tilde{g}_{1:\tilde{H}}\in\tilde{G}_{1:\tilde{H}}}J_{\cD_\cL}(\tilde{g}_{1:\tilde{H}})-\epsilon$. Thus,  $\varphi(\Breve{g}_{1:\Breve{H}}^\ast,\cD_\cL)$ is an $\epsilon$-team-optimal strategy of $\cD_\cL$.
\end{proof}

\begin{lemma}
\label{lemma: equivalence alg 3-4}
  For any given strategy $\Breve{g}_{1:\Breve{H}}\in \Breve{\cG}_{1:\Breve{H}}$, implementing Algorithm \ref{algorithm Implement varphi} in $\cD_\cL$ is equivalent to implementing $\varphi(\Breve{g}_{1:\Breve{H}},\cD_\cL)$ in $\cD_\cL$.
\end{lemma}
\begin{proof}
    As shown in Theorem \ref{theorem: sQC to QC}, Algorithm \ref{algorithm varphi} can compute the associated $\Breve{\tau}_{i,h}$ and use it as the input of $\Breve{g}_{i,h}$ (Line 12). Therefore, it suffices to prove that Algorithm \ref{algorithm Implement varphi} can also compute the associated $\Breve{\tau}_{i,h}$ and use it as the input of $\Breve{g}_{i,h}$ (Line 6), i.e., Algorithm \ref{algorithm recover} can output the associated $\Breve{\tau}_{i,h}$ from $\tilde{\tau}_{i,h}$ and $\Breve{g}_{1:h-1}$. We prove this by induction.\\
    Firstly, when $h=1$, it holds for any $i\in[n]$ such that $\tilde{\tau}_{i,1}=\Breve{\tau}_{i,1}$. In Algorithm \ref{algorithm recover}, when $h=1$, it will never enter the for loop, and thus the output is $\tilde{\tau}_{i,1}=\Breve{\tau}_{i,1}$.\\
    Secondly, we assume for any $h<t$, the hypothesis holds. Then, for $h=t$, given any $\tilde{\tau}_{i,t}\in\tilde{\cT}_{i,t}$ and $\Breve{g}_{1:t-1}$, let $\Breve{\tau}_{i,t}'$ be the output of Algorithm \ref{algorithm recover}. For any $j\in[n], h'<t$, if it holds that $\sigma(\tilde{\tau}_{j,h'})\subseteq \sigma(\tilde{c}_t)$ in $\cD_\cL$ and $\tilde{a}_{j,h'}\notin \tilde{\tau}_{i,t}$, then it can compute the associated $\Breve{\tau}_{j,h'}$ from induction hypothesis (Lines 5-6), compute the exact $\Breve{a}_{j,h'}$ based on $\Breve{g}_{j,h'}$ (Line 7), and add it into $\Breve{\tau}_{i,h}'$ (Line 8).
    Therefore, we know that $\Breve{\tau}_{i,t}'=\tilde{\tau}_{i,t}\cup \{\tilde{a}_{j,h'}\given \sigma(\tilde{\tau}_{j,h'})\subseteq \sigma(\tilde{c}_t) \text{ and } \tilde{a}_{j,h'}\notin \tilde{\tau}_{i,t}\}=\Breve{\tau}_{i,t}$. 
    By induction, we complete the proof.
\end{proof}
\begin{remark}
    The difference between Algorithms  \ref{algorithm varphi} and \ref{algorithm Implement varphi} lies as follows. Given any $\Breve{g}_{1:\Breve{H}}$ and $\cD_\cL$, Algorithm \ref{algorithm varphi} needs to recover the output $\tilde{a}_{i,h}$ of $\tilde{g}_{i,h}$ under all possible input $\tilde{\tau}_{i,h}\in\tilde{\cT}_{i,h}, i\in[n], h\in[\tilde{H}]$, where the cardinality of $\tilde{\cT}_{i,h}$ could be exponentially large. Thus, Algorithm \ref{algorithm varphi} may suffer from computational intractability. However, Algorithm \ref{algorithm Implement varphi} only requires recovering the output $\tilde{g}_{i,h}$ under the specific $\tilde{\tau}_{i,h}$ that occurred in the trajectory, which can be implemented in polynomial time.  
\end{remark}

\subsection{Proof of Theorem \ref{theorem: SI=QC}}
\begin{proof}

To prove that $\cD_\cL^\dag$ has SI-CIBs, it suffices to prove that for any  $h=2,\cdots,\Breve{H}$, fix any $h_1\in[h-1], i_1\in[n]$, and for any $\Breve{g}_{1:h-1}\in\Breve{\cG}_{1:h-1},\Breve{g}_{i_1,h_1}'\in\Breve{\cG}_{i_1,h_1}$, let $\Breve{g}_{h_1}':=(\Breve{g}_{1,h_1},\cdots,\Breve{g}_{i_1,h_1}',\cdots,\Breve{g}_{n,h_1})$ and $\Breve{g}_{1:h-1}':=(\Breve{g}_1,\cdots,\Breve{g}_{h_1}',\cdots,\Breve{g}_{h-1})$. If $\Breve{c}_h$ is reachable under both  $\Breve{g}_{1:h-1}$ and $\Breve{g}_{1:h-1}'$, then the following holds%
\begin{equation}
    \PP_h^{\cD_\cL^\dag}(\Breve{s}_h,\Breve{p}_h\given \Breve{c}_h,\Breve{g}_{1:h-1})=\PP_h^{\cD_\cL^\dag}(\Breve{s}_h,\Breve{p}_h\given \Breve{c}_h,\Breve{g}'_{1:h-1}).
\end{equation}
We prove this result  case by case as follows: 
\begin{itemize}
    \item 
If there exists some $i_3\neq i_1$ such that $\sigma(\Breve{\tau}_{i_1,h_1})\subseteq \sigma(\Breve{\tau}_{i_3,h})$  and $ \sigma(\Breve{a}_{i_1,h_1})\subseteq \sigma(\Breve{\tau}_{i_3,h})$, then from Assumption \ref{assumption:sigma_include}, we know that $\sigma(\Breve{\tau}_{i_1,h_1})\subseteq \sigma(\Breve{c}_h),\sigma(\Breve{a}_{i_1,h_1})\subseteq \sigma(\Breve{c}_h)$. Therefore, there exist deterministic measurable functions $\alpha_1,\alpha_2$ such that $\Breve{\tau}_{i_1,h_1}=\alpha_1(\Breve{c}_h), \Breve{a}_{i_1,h_1}=\alpha_2(\Breve{c}_h)$, and further it holds that 
\begin{align*}
&\PP_h^{\cD_\cL^\dag}(\Breve{s}_h,\Breve{p}_h\given \Breve{c}_h,\Breve{g}_{1:h-1})=\PP_h^{\cD_\cL^\dag}(\Breve{s}_h,\Breve{p}_h\given \alpha_1(\Breve{c}_h), \alpha_2(\Breve{c}_h), \Breve{c}_h,\Breve{g}_{1:h-1})=\PP_h^{\cD_\cL^\dag}(\Breve{s}_h,\Breve{p}_h\given \Breve{c}_h, \Breve{g}_{1:h-1}').  
\end{align*}
The last equality is due to the fact that both the input and output of $\Breve{g}_{i_1,h_1}$ are conditioned on.
\item If for any $i_2\neq i_1$, either $\sigma(\Breve{\tau}_{i_1,h_1})\nsubseteq \sigma(\Breve{\tau}_{i_2,h})$ or $\sigma(\Breve{a}_{i_1,h_1})\nsubseteq \sigma(\Breve{\tau}_{i_2,h})$, then agent $(i_1,h_1)$ does not influence any agent $(i_2,h)$ with $i_2\neq i_1$ in $\cD_\cL^\dag$, since otherwise, due to the sQC IS of $\cD_\cL^\dag$, it must hold that $\sigma(\Breve{\tau}_{i_1,h_1})\subseteq \sigma(\Breve{\tau}_{i_2,h})$ and $\sigma(\Breve{a}_{i_1,h_1})\subseteq \sigma(\Breve{\tau}_{i_2,h})$. Moreover, we claim that such an $h_1$ has to be even, since otherwise, the agent will be at a {communication}  step, and we must have $\Breve{\tau}_{i_1,h_1}=\Breve{c}_{h_1}\subseteq \Breve{c}_{h}\subseteq \Breve{\tau}_{i_2,h}$ by Assumption \ref{limited communication strategy}  and the reformulation in \Cref{LTC to Dec-POMDP}, and $\Breve{a}_{i_1,h_1}=m_{i_1,\frac{h_1+1}{2}}\in z_{\frac{h_1+1}{2}}^a=\Breve{z}_{h_1+1}\subseteq \Breve{c}_h\subseteq \Breve{\tau}_{i_2,h}$ by Assumption \ref{assumption:information evolution} (b), which violates the premise of this case. Let $k_1:=h_1/2$. 
Now, we claim that agent $(i_1,h_1)$ does not influence the state $\Breve{s}_h$ nor the information 
$\Breve{\tau}_{i_1,h}$. We prove this case by case as follows:
\begin{itemize}
\item Suppose   $h$ is even. If agent $(i_1,h_1)$ influences $\Breve{s}_{h_1+1}$, then from Assumption \ref{weak gamma observability}, there exists some $i_3\neq i_1$ such that agent $(i_1,h_1)$ influences  $\Breve{o}_{i_3,h_1+1}$. However, from Assumption \ref{assumption:information evolution}  (e), we know that $\Breve{o}_{i_3,h_1+1}=o_{i_3,\frac{h_1}{2}+1}\in\tau_{i_3,\frac{h_1}{2}+1}\subseteq\tau_{i_3,\frac{h}{2}}=\tilde{\tau}_{i_3,h}\subseteq\Breve{\tau}_{i_3,h}$, which means  that agent $(i_1,h_1)$ influences agent $(i_3,h)$, leading to a contradiction. { Therefore, we know that agent $(i_1,h_1)$ does not influence $\Breve{s}_{h_1+1}$, and thus for any $i_2\in[n]$, it does not influence $\Breve{o}_{i_2,h_1+1}$. Also, from Assumption \ref{useless action}, we know that $\Breve{a}_{i_1,h_1}\notin \Breve{\tau}_{i_2, h_1+1}$.  Therefore, agent $(i_1,h_1)$ does not influence  $\Breve{\tau}_{i_2,h_1+1}$ nor  $\Breve{a}_{i_2,h_1+1}$. By recursion, we know that agent $(i_1,h_1)$ does not influence $\Breve{s}_{h'}$ nor $\Breve{\tau}_{i_2,h'}$ $\forall~i_2\in[n],h'>h$. }
    \item Suppose $h$ is odd, then $\Breve{p}_{h}=\emptyset$ by \Cref{LTC to Dec-POMDP}.  If agent $(i_1,h_1)$ influences $\Breve{s}_h$ in $\cD_\cL^\dag$, then agent $(i_1,h_1)$ influences $\tilde{s}_{h}$ in $\cD_\cL$, since strict expansion does not change system dynamics. This implies that agent $(i_1,h_1)$ influences $s_{h_1+1}$ in $\cL$, since otherwise, as argued above, it will not influence the later state $\Breve{s}_h$. Moreover, from Assumption \ref{weak gamma observability}, we know that agent $(i_1,h_1)$ also influences $\tilde{o}_{-i_1,h}$, i.e., there   must exist some $i_3\neq i_1$ such that agent $(i_1,h_1)$ influences $\tilde{o}_{i_3,h}$ in $\cD_\cL$. From Assumption \ref{assumption:information evolution} (e), it holds that $\tilde{o}_{i_3,h}\in \tilde{\tau}_{i_3,h+1}$. 
    Therefore, agent $(i_1,h_1)$ influences agent $(i_3,h+1)$ in  $\cD_\cL$. From \Cref{lemma: action influence}, we know that $\sigma(\tau_{i_1,k_1^-})\subseteq\sigma(\tau_{i_3,k^-})$ in $\cL$, where $k:=(h+1)/2$. Furthermore, from Assumption \ref{assumption:sigma_include} and $i_3\neq i_1$, it holds that $\sigma(\tau_{i_1,k_1^-})\subseteq\sigma(c_{k^-})$. Also, from {\Cref{LTC to Dec-POMDP}}, it holds that  $\tilde{\tau}_{i_1,h_1}=\tau_{i_1,k_1^+}=\tau_{i_1,k_1^-}\cup z_{k_1}^a$ and $z_{k_1}^a=\tilde{z}_{h_1}\subseteq \tilde{c}_h$. Then, we have $\sigma(\tilde{\tau}_{i_1,h_1})\subseteq\sigma(\tilde{c}_{h})=\sigma(\tilde{\tau}_{i_3,h})$. Based on the strict expansion from $\cD_\cL$ to $\cD_\cL^\dag$, we can get $\Breve{\tau}_{i_1,h_1}\backslash\tilde{\tau}_{i_1,h_1}\subseteq \Breve{c}_{h_1}\subseteq \Breve{\tau}_{i_3,h}$  and $\Breve{a}_{i_1,h_1}\in\Breve{c}_h$. 
    Then, it holds that $\sigma(\Breve{\tau}_{i_1,h_1})\subseteq\sigma(\Breve{\tau}_{i_3,h}), \sigma(\Breve{a}_{i_1,h_1})\subseteq\sigma(\Breve{\tau}_{i_3,h})$, which leads to a contradiction to the premise that for any $i_2\neq i_1$, either $\sigma(\Breve{\tau}_{i_1,h_1})\nsubseteq \sigma(\Breve{\tau}_{i_2,h})$ or $\sigma(\Breve{a}_{i_1,h_1})\nsubseteq \sigma(\Breve{\tau}_{i_2,h})$.    
    Hence, we know that agent $(i_1,h_1)$ does not influence the state $\Breve{s}_h$.   Additionally, for any $i_2\neq i_1$, since agent $(i_1,h_1)$ does not influence  agent $(i_2,h)$ in  $\cD_\cL^\dag$, and   $\Breve{\tau}_{i_1,h}=\Breve{c}_h=\Breve{\tau}_{i_2,h}$, then  we know that agent $(i_1,h_1)$ does not influence $\Breve{\tau}_{i_1,h}$. 
\end{itemize}
Combining the two cases above, we know that agent $(i_1,h_1)$ does not influence  $\Breve{s}_h$, $\Breve{\tau}_{i,h}, \forall i\in[n]$, and $\Breve{c}_h=\cap_{i=1}^n\Breve{\tau}_{i,h}$ in $\cD_\cL^\dag$, yielding 
    \begin{align*}
        &\PP_h^{\cD_\cL^\dag}(\Breve{s}_h,\Breve{p}_h\given \Breve{c}_h,\Breve{g}_{1:h-1})=\PP_h^{\cD_\cL^\dag}(\Breve{s}_h, \Breve{p}_h,\Breve{c}_h\given \Breve{c}_h,\Breve{g}_{1:h-1})=\PP_h^{\cD_\cL^\dag}(\Breve{s}_h,\Breve{\tau}_{h}\given \Breve{c}_h,\Breve{g}_{1:h-1})=\PP_h^{\cD_\cL^\dag}(\Breve{s}_h,\{\Breve{\tau}_{i,h}\}_{i\in[n]}\given \Breve{c}_h,\Breve{g}_{1:h-1})\\
        &\qquad
        =\frac{\PP_h^{\cD_\cL^\dag}(\Breve{s}_h,\{\Breve{\tau}_{i,h}\}_{i\in[n]},\Breve{c}_h\given\Breve{g}_{1:h-1})}{\PP_h^{\cD_\cL^\dag}(\Breve{s}_h,\Breve{c}_h\given\Breve{g}_{1:h-1})}=\PP_h^{\cD_\cL^\dag}(\Breve{s}_h,\{\Breve{\tau}_{i,h}\}_{i\in[n]}\given \Breve{c}_h,\Breve{g}_{1:h-1}')
    =\PP_h^{\cD_\cL^\dag}(\Breve{s}_h,\Breve{p}_h\given \Breve{c}_h,\Breve{g}_{1:h-1}').
    \end{align*}
\end{itemize} 
This completes the proof.  
\end{proof}

\subsection{Proof of \Cref{theorem: refinement}}
\begin{proof}
    Firstly, from the construction of $\cD_\cL'$ and the strategy space $\overline{\cG}_{1:\overline{H}}$, we know that for any $h\in[H], i\in[n], \overline{\cC}_{2h-1}=\Breve{\cC}_{2h-1}, \overline{\cA}_{i, 2h-1}=\Breve{\cA}_{i,2h-1}, \overline{\cT}_{i,2h}=\Breve{\cT}_{i,2h}, \overline{\cA}_{i,2h}=\Breve{\cA}_{i,2h}$. Therefore, $\overline{\cG}_{1:\overline{H}}=\Breve{\cG}_{1:\Breve{H}}$, and finding a team-optimal strategy  of $\cD_\cL'$ in space $\overline{\cG}_{1:\overline{H}}$ is equivalent to finding that of $\cD_\cL^\dag$  in the  space $\Breve{\cG}_{1:\Breve{H}}$ by definition.
    
    Secondly, we will prove that the Dec-POMDP $\cD_\cL'$ satisfies the information evolution rules. From Assumption \ref{assumption:information evolution}, it holds that, for any $i\in[n], h\in[\overline{H}]$, if $h=2t-1$ with $t\in[H]$, then
    \begin{align*}
        \tilde{z}_{h}=\chi_{t}(\tilde{p}_{h-1},\tilde{a}_{h-1},\tilde{o}_h),\qquad 
        p_{i,h}=\xi_{i,t}(\tilde{p}_{i,h-1},\tilde{a}_{i,h-1},\tilde{o}_{i,h});
    \end{align*}
    if $h=2t$ with $t\in[H]$, then 
    \begin{align*}
        \tilde{z}_h=\phi_t(p_{h-1},\tilde{a}_{h-1}),\qquad \tilde{p}_{i,h}=p_{i,h-1}\backslash \phi_{i,t}(p_{i,h-1},\tilde{a}_{i,h-1}),
    \end{align*}
    where $\chi_t, \xi_{i,t}$ are fixed transformations and $\phi_{h},\phi_{i,h}$ are additional-sharing functions. Recall that we defined $p_{i,2t-1}=p_{i,t^-}$ for any $i\in[n],t\in[H]$. From the expansion (\Cref{construction:QC to sQC}), we know that $\Breve{z}_h=\Breve{c}_h\backslash\Breve{c}_{h-1}$.  Also, from the refinement, we know that $\forall i\in[n], h\in[\overline{H}], \overline{z}_{h}=\Breve{z}_{h}, \overline{c}_h=\Breve{c}_h, \overline{a}_{i,h}=\Breve{a}_{i,h}=\tilde{a}_{i,h}, \overline{o}_{i,h}=\Breve{o}_{i,h}=\tilde{o}_{i,h}$, and $\forall t\in[H], \overline{p}_{i, 2t-1}=p_{i,t^-}, \overline{p}_{i,2t}=\tilde{p}_{i,2t}$.

    Then, we can  construct $\{\overline{\chi}_{h}\}_{h\in[\overline{H}]}, \{\overline{\xi}_{i,h}\}_{i\in[n],h\in[\overline{H}]}$ accordingly as follows:
\begin{itemize}
       \item If $h=2t-1$ with $t\in[H]$, we define $\overline{\chi}_h,\{\overline{\xi}_{i,h}\}_{i\in[n]}$ as
       \begin{align*}
           &\forall i\in[n], \overline{\xi}_{i,h}:=\xi_{i,t}, \quad\forall \overline{p}_{h-1}\in\overline{\cP}_{h-1},\overline{a}_{h-1}\in\overline{\cA}_{h-1}, \overline{o}_{h}\in\overline{\cO}_{h}, \\
           &\overline{\chi}_h(\overline{p}_{h-1},\overline{a}_{h-1},\overline{o}_h):=\chi_t(\overline{p}_{h-1},\overline{a}_{h-1},\overline{o}_h)\cup \varrho_h^1\backslash\varrho_h^2,\text{ where}\\
           &\varrho_h^1:=\{\tilde{a}_{j,h-1}\given j\in[n],\sigma(\tilde{\tau}_{j,h-1})\subseteq \sigma(\tilde{c}_h), \tilde{a}_{j,h-1}\text{~influences }\tilde{s}_{h}\}\backslash\chi_t(\overline{p}_{h-1},\overline{a}_{h-1},\overline{o}_h)\text{ if $h>1$, otherwise $\emptyset.$}\\
           &\varrho_h^2:=\{\tilde{a}_{j,h_0}\given j\in[n], h_0<h-1, \sigma(\tilde{a}_{j,h_0})\subseteq \sigma(\tilde{c}_{h-1}),\tilde{a}_{j,h_0} \text{ influences } \tilde{s}_{h_0+1}\}\cap \chi_t(\overline{p}_{h-1},\overline{a}_{h-1},\overline{o}_h).
       \end{align*}
       Note that there exist  some functions $f_h^1,f_h^2$ such that $\varrho_h^1=f_h^1(\overline{p}_{h-1},\overline{a}_{h-1},\overline{o}_h),\varrho_h^2=f_h^2(\overline{p}_{h-1},\overline{a}_{h-1},\overline{o}_h)$. This is because, $\varrho_h^1\subseteq \{\overline{a}_{i,h-1}\}_{i\in[n]}$; which element $\overline{a}_{j,h-1}$ is in $\varrho_h^1$ is based on whether $\sigma(\tilde{\tau}_{j,h-1})\subseteq \sigma(\tilde{c}_h)$ and whether $ \tilde{a}_{j,h-1}\text{~influences }\tilde{s}_{h}$, which is the property of the problem $\cD_\cL$. Similarly, $\varrho_h^2\subseteq \chi_t(\overline{p}_{h-1},\overline{a}_{h-1},\overline{o}_h)$, and which element $\overline{a}_{j,h_0}$ is in $\varrho_h^2$ is based on whether $\sigma(\tilde{a}_{j,h_0})\subseteq \sigma(\tilde{c}_{h-1}),\tilde{a}_{j,h_0} \text{ influences } \tilde{s}_{h_0+1}$.
        Now we will show that $\overline{\chi}_h,\{\overline{\xi}_{i,h}\}_{i\in[n]}$ satisfy  
    \begin{align*}
        \overline{c}_{h}=\overline{c}_{h-1}\cup \overline{z}_{h},~~ \overline{z}_{h}=\overline{\chi}_{h}(\overline{p}_{h-1},\overline{a}_{h-1},\overline{o}_{h}),\qquad 
        \text{for each }i\in[n],~~ \overline{p}_{i,h}=\overline{\xi}_{i,h}(\overline{p}_{i,h-1},\overline{a}_{i,h-1},\overline{o}_{i,h}).
    \end{align*}
    For functions $\{\overline{\xi}_{i,h}\}_{i\in[n]}$, based on the construction of  $\cD_\cL'$ from $\cD_\cL$, we know that $\overline{p}_{i,h-1}=\tilde{p}_{i,h-1},\overline{a}_{i,h-1}=\tilde{a}_{i,h-1}, \overline{o}_{i,h-1}=\tilde{o}_{i,h-1}$, and $\overline{p}_{i,h}=p_{i,t^-}=p_{i,h}$. 

    For function $\overline{\chi}_h$, it is easy to verify that $\overline{z}_h\backslash\tilde{z}_h=(\overline{c}_h\backslash\tilde{c}_h)\backslash(\overline{c}_{h-1}\backslash\tilde{c}_{h-1})$ and $\tilde{z}_h\backslash\overline{z}_h=(\overline{c}_{h-1}\backslash\tilde{c}_{h-1})\backslash (\overline{c}_h\backslash\tilde{c}_h)$. Together with the fact that $\overline{z}_h=\tilde{z}_h\cup (\overline{z}_h\backslash\tilde{z}_h)\backslash(\tilde{z}_h\backslash\overline{z}_h)$, it suffices to show that $\varrho_h^1=\overline{z}_h\backslash\tilde{z}_h,\varrho_h^2=\tilde{z}_h\backslash\overline{z}_h$.  From the expansion, we know that for any $h'\in[\overline{H}],\overline{c}_{h'}\backslash\tilde{c}_{h'}$ only consists of some actions at the even timesteps, since the actions at odd timesteps cannot influence the underlying state. Therefore,  $\overline{z}_h\backslash\tilde{z}_h$ and $\tilde{z}_h\backslash\overline{z}_h$ only consist of some actions at the even timesteps. 

    For any $\tilde{a}_{i,2t_1},i\in[n],t_1<t,$ if $\tilde{a}_{i,2t_1}$ influences $\tilde{s}_{2t_1+1}$, then from Assumption \ref{weak gamma observability}, there exists $i_2\neq i$ such that $\tilde{a}_{i,2t_1}$ influences $\tilde{o}_{i_2,2t_1+1}$ and thus influences $\tilde{\tau}_{i_2,2t_1+1}$ due to Assumption \ref{assumption:information evolution} (e). From the QC IS of $\cD_\cL$, we know that $\sigma(\tilde{\tau}_{i,2t_1})\subseteq\sigma(\tilde{\tau}_{i_2,2t_1+1})$, and thus $\sigma(\tilde{\tau}_{i,2t_1})\subseteq\sigma(\tilde{c}_{2t_1+1})$ due to Assumption \ref{assumption:sigma_include} and $i_2\neq i$. Therefore, from \Cref{construction:QC to sQC}, we know that $\tilde{a}_{i,2t_1}\in\Breve{c}_{2t_1+1}$. This means, if any action $\tilde{a}_{i,2t_1}$ influences underlying state $\tilde{s}_{2t_1+1}$, we have $\tilde{a}_{i,2t_1}\in \Breve{c}_{2t_1+1}=\overline{c}_{2t_1+1}$, since it will be added  in $\Breve{c}_{2t_1+1}$ via expansion. 

    Therefore, if any action $\overline{a}_{i,h_1}=\tilde{a}_{i,h_1}$ with $i\in[n],h_1<h$ is in $(\overline{c}_h\backslash\tilde{c}_h)\backslash(\overline{c}_{h-1}\backslash\tilde{c}_{h-1})$, then it can only happen if $h_1=h-1$. Also, for any $i\in[n]$, if $\overline{a}_{i,h-1}\in \varrho_h^1$, then $\sigma(\tilde{\tau}_{i,h-1})\subseteq \sigma(\tilde{c}_h), \tilde{a}_{i,h-1}\text{~influences }\tilde{s}_{h}$, and furthermore,  $ \overline{a}_{i,h-1}\notin\chi_t(\overline{p}_{h-1},\overline{a}_{h-1},\overline{o}_h)$. Then, it means $\overline{a}_{i,h-1}\notin\tilde{c}_h$.  Therefore, $\overline{a}_{i,h-1}\in (\overline{c}_h\backslash\tilde{c}_h)\backslash(\overline{c}_{h-1}\backslash\tilde{c}_{h-1})$ and we proved that $ \varrho_h^1\subseteq \overline{z}_h\backslash\tilde{z}_h$. Also, for any  $i\in[n]$, $\overline{a}_{i,h-1}\in \overline{z}_{h}\backslash\tilde{z}_h$ only if $\overline{a}_{i,h-1}$ is added via expansion and $\overline{a}_{i,h-1}\notin \chi_t(\overline{p}_{h-1},\overline{a}_{h-1},\overline{o}_h)$. Then, this can only happen if $\sigma(\tilde{\tau}_{i,h-1})\subseteq \sigma(\tilde{c}_h)$ and $ \tilde{a}_{i,h-1}\text{~influences }\tilde{s}_{h}$. Therefore, we proved that $\overline{z}_{h}\backslash\tilde{z}_h\subseteq \varrho_h^1$. Combining the two parts we obtain  $\varrho_h^1=\overline{z}_{h}\backslash\tilde{z}_h$.

    If any action $\overline{a}_{i,h_1}$ with $i\in[n],h_1<h$ is in $\tilde{z}_h\backslash\overline{z}_h=(\overline{c}_{h-1}\backslash\tilde{c}_{h-1})\backslash (\overline{c}_h\backslash\tilde{c}_h)$,  then we will know that $\overline{a}_{i,h_1}\in \overline{c}_{h-1}$, then $h_1+1\le h-1$ and $\sigma(\overline{a}_{i,h_1})\subseteq \sigma(\overline{c}_{h-1})$. Also, from the proof above, we know that $\overline{a}_{i,h_1}$ influences the state $\overline{s}_{h_1+1}$, then together with $\overline{a}_{i,h_1}\in\tilde{z}_h$, we have $\overline{a}_{i,h_1}\in \{\tilde{a}_{j,h_0}\given j\in[n], h_0<h-1, \sigma(\tilde{a}_{j,h_0})\subseteq \sigma(\tilde{c}_{h-1}),\tilde{a}_{j,h_1} \text{ influences } \tilde{s}_{h_1+1}\}\cap\tilde{z}_h$. Therefore, $\tilde{z}_h\backslash\overline{z}_h\subseteq \varrho_h^2$. Meanwhile, for any $i\in[n],h_1<h$, if $\overline{a}_{i,h_1}\in\varrho_h^2$, then it holds $\overline{a}_{i,h_1}\in\tilde{z}_h$, $h_1<h-1$, and $\overline{a}_{i,h_1}$ influences $\overline{s}_{h_1+1}$. Then, from above, we know that $\overline{a}_{i,h_1}=\Breve{a}_{i,h_1}$ will be added in $\Breve{c}_{h_1+1}=\overline{c}_{h_1+1}\subseteq \overline{c}_{h-1}$ since $h_1<h-1$. Then, $\overline{a}_{i,h_1}\notin \overline{z}_h=\overline{c}_h\backslash\overline{c}_{h-1}$. Therefore, $\overline{a}_{i,h_1}\in \tilde{z}_h\backslash\overline{z}_h$, and then $\varrho_h^2\subseteq \tilde{z}_h\backslash\overline{z}_h$. Combining the two parts we obtain  $\varrho_h^2=\tilde{z}_h\backslash\overline{z}_{h}$.
    \item If $h=2t$ with $t\in[H]$, we define $\overline{\chi}_h,\{\overline{\xi}_{i,h}\}_{i\in[n]}$ as
       \begin{align*}
           &\forall i\in[n], \overline{p}_{i,h-1}\in\overline{\cP}_{i,h-1},\overline{a}_{i,h-1}\in \overline{\cA}_{i,h-1},\overline{o}_{i,h}\in \overline{\cO}_{i,h}, \overline{\xi}_{i,h}(\overline{p}_{i,h-1},\overline{a}_{i,h-1},\overline{o}_{i,h})=\overline{p}_{i,h-1}\backslash\phi_{i,t}(\overline{p}_{i,h-1}, \overline{a}_{i,h-1})\\
           &\forall \overline{p}_{h-1}\in\overline{\cP}_{h-1},\overline{a}_{h-1}\in \overline{\cA}_{h-1},\overline{o}_{h}\in \overline{\cO}_{h},\overline{\chi}_h(\overline{p}_{h-1},\overline{a}_{h-1},\overline{o}_{h})=\phi_t(\overline{p}_{h-1},\overline{a}_{h-1}). 
       \end{align*}
       Now we will show that $\overline{\chi}_h,\{\overline{\xi}_{i,h}\}_{i\in[n]}$ satisfy  
    \begin{align*}
        \overline{c}_{h}=\overline{c}_{h-1}\cup \overline{z}_{h},~~ \overline{z}_{h}=\overline{\chi}_{h}(\overline{p}_{h-1},\overline{a}_{h-1},\overline{o}_{h}),\qquad 
        \text{for each }i\in[n],~~ \overline{p}_{i,h}=\overline{\xi}_{i,h}(\overline{p}_{i,h-1},\overline{a}_{i,h-1},\overline{o}_{i,h}).
    \end{align*}
     For functions $\{\overline{\xi}_{i,h}\}_{i\in[n]}$, based on the construction of $\cD_\cL'$ from $\cD_\cL$, we know that $\overline{p}_{i,h-1}=p_{i,t^-}=p_{i,h-1},\overline{a}_{i,h-1}=\tilde{a}_{i,h-1}$ and $\overline{p}_{i,h}=\tilde{p}_{i,h}$. 

     For function $\overline{\chi}_h$, we know that $\overline{p}_{h-1}=p_{t^-}=p_{h-1}, \overline{a}_{h-1}=\tilde{a}_{h-1}$, so it suffices to show that $\tilde{z}_h=\overline{z}_h$. As shown above, $\tilde{z}_h\backslash\overline{z}_h$ and $\overline{z}_h\backslash\tilde{z}_h$ only consist of some actions at the even timesteps. Moreover, for any action $\tilde{a}_{i,2t_1}$ with $i\in[n], t_1<t$ that influences $\tilde{s}_{2t_1+1}$, it will be added in $\overline{c}_{2t_1+1}$  via expansion (if $\tilde{a}_{i,2t_1}\notin \tilde{c}_{2t_1+1}$); if $\tilde{a}_{i,2t_1}$ does not influence $\tilde{s}_{2t_1+1}$, then it will never be added via expansion. Therefore, $\overline{z}_h\backslash\tilde{z}_h=\emptyset$. Also, if some action $\tilde{a}_{i,2t_1}\in \tilde{z}_h\backslash\overline{z}_h$, then we know that $\tilde{a}_{i,2t_1}\in \tilde{z}_h=\phi_t(\overline{p}_{h-1},\tilde{a}_{h-1})$, and further we know that $\tilde{a}_{i,2t_1}\in \overline{p}_{h-1}=p_{i,t^-}$. However, if $\tilde{a}_{i,2t_1}$ influences $\tilde{s}_{2t_1+1}$, then it will be added in $\overline{c}_{2t_1+1}$ and cannot lie in $\overline{p}_{h-1}$; if $\tilde{a}_{i,2t_1}$ does not influence $\tilde{s}_{2t_1+1}$, then from Assumption \ref{useless action}, $\tilde{a}_{i,2t_1}\notin p_{i,t^-}$. Therefore, we know that $\tilde{z}_h\backslash\overline{z}_h=\emptyset$, and get $\tilde{z}_h=\overline{z}_h$.
\end{itemize}
    
    Thirdly, we prove that such a Dec-POMDP $\cD_\cL'$ has SI-CIBs  with respect to the strategy space $\overline{\cG}_{1:\overline{H}}$. This is equivalent to that
    for any $h\in[2:\overline{H}], \overline{s}_h\in\overline{\cS}, \overline{p}_h\in\overline{\cP}_h, \overline{c}_h\in \overline{\cC}_h, i_1\in[n],h_1<h, \overline{g}_{1:h-1},\overline{g}_{i_1,h_1}'\in \overline{\cG}_{i_1:h_1}$, let $\overline{g}_{h_1}':=(\overline{g}_{1,h_1},\cdots,\overline{g}_{i_1,h_1}',\cdots,\overline{g}_{n,h_1})$ and $\overline{g}_{1:h-1}':=(\overline{g}_1,\cdots,\overline{g}_{h_1}',\cdots,\overline{g}_{h-1})$. If $\overline{c}_h$ is reachable from both $\overline{g}_{1:h-1}$ and $\overline{g}_{1:h-1}'$, it holds that 
    \begin{equation}
        \PP_h^{\cD_\cL'}(\overline{s}_h,\overline{p}_h\given \overline{c}_h,\overline{g}_{1:h-1})=\PP_h^{\cD_\cL'}(\overline{s}_h,\overline{p}_h\given \overline{c}_h,\overline{g}'_{1:h-1}).
    \end{equation}
    We prove this case by case. If $h=2t$ with $t\in[H]$, then from the result of Theorem \ref{theorem: SI=QC}, it holds  that
    \begin{equation*}
        \PP_h^{\cD_\cL'}(\overline{s}_h,\overline{p}_h\given \overline{c}_h,\overline{g}_{1:h-1})=\PP_h^{\cD_\cL^\dag}(\overline{s}_h,\overline{p}_h\given \overline{c}_h,\overline{g}_{1:h-1})=\PP_h^{\cD_\cL^\dag}(\overline{s}_h,\overline{p}_h\given \overline{c}_h,\overline{g}_{1:h-1}')=\PP_h^{\cD_\cL'}(\overline{s}_h,\overline{p}_h\given \overline{c}_h,\overline{g}'_{1:h-1}).
    \end{equation*}
    Therefore, now we consider the case that $h=2t-1$ with $t\in[H]$.\\
    {
    Suppose $h_1$ is odd, which means that $\overline{a}_{h_1}$ corresponds to the communication action in $\cL$.  Then, it holds that $\overline{c}_{h_1}\subseteq \overline{c}_h, \overline{a}_{i_1,h_1}=m_{i_1,\frac{h_1+1}{2}}\in \overline{c}_h$, then
    \begin{align*}
        &\PP_h^{\cD_\cL'}(\overline{s}_h,\overline{p}_h\given \overline{c}_h,\overline{g}_{1:h-1})=\PP_h^{\cD_\cL'}(\overline{s}_h,\overline{p}_h\given \overline{c}_{h_1}, \overline{a}_{i_1,h_1}, \overline{c}_h,\overline{g}_{1:h-1})\\
        &\quad=\PP_h^{\cD_\cL'}(\overline{s}_h,\overline{p}_h\given \overline{c}_{h_1}, \overline{a}_{i_1,h_1}, \overline{c}_h,\overline{g}_{1:h-1}\backslash \overline{g}_{i_1,h_1})=\PP_h^{\cD_\cL'}(\overline{s}_h,\overline{p}_h\given \overline{c}_h,\overline{g}_{1:h-1}'),
    \end{align*}
    where the second equality is because the input and output of $\overline{g}_{i_1,h_1}$ are $\overline{c}_{h_1}$ and $\overline{a}_{i_1,h_1}$.\\
    Suppose $h_1$ is even, which means that $h_1$ is a  control timestep, and  let $t_1=\frac{h_1}{2}$. If $\sigma(\overline{\tau}_{i_1,h_1})\subseteq \sigma(\overline{c}_h)$ and $\sigma(\overline{a}_{i_1,h_1})\subseteq \sigma(\overline{c}_h)$, then  there exist deterministic measurable functions $\overline{\alpha}_1,\overline{\alpha}_2$ such that $\overline{\tau}_{i_1,h_1}=\overline{\alpha}_1(\overline{c}_h), \overline{a}_{i_1,h_1}=\overline{\alpha}_2(\overline{c}_h)$, and further it holds that 
\begin{align*}
&\PP_h^{\cD_\cL'}(\overline{s}_h,\overline{p}_h\given \overline{c}_h,\overline{g}_{1:h-1})=\PP_h^{\cD_\cL'}(\overline{s}_h,\overline{p}_h\given \overline{\alpha}_1(\overline{c}_h), \overline{\alpha}_2(\overline{c}_h), \overline{c}_h,\overline{g}_{1:h-1})\\
    &\quad=\PP_h^{\cD_\cL'}(\overline{s}_h,\overline{p}_h\given \overline{\tau}_{i_1,h_1},\overline{a}_{i_1,h_1},\overline{c}_h, \overline{g}_{1:h-1})=\PP_h^{\cD_\cL'}(\overline{s}_h,\overline{p}_h\given \overline{\tau}_{i_1,h_1},\overline{a}_{i_1,h_1},\overline{c}_h, \overline{g}_{1:h-1}').  
\end{align*}
    If $\sigma(\overline{\tau}_{i_1,h_1})\nsubseteq \sigma(\overline{c}_h)$ or $\sigma(\overline{a}_{i_1,h_1})\nsubseteq \sigma(\overline{c}_h)$, then  since $h_1$ is even, $\overline{\tau}_{i_1,h_1}=\Breve{\tau}_{i_1,h_1}, \overline{a}_{i_1,h_1}=\Breve{a}_{i_1,h_1}$. Also, we know that $\overline{c}_{h}=\Breve{c}_h$, then it holds that $\sigma(\Breve{\tau}_{i_1,h_1})\nsubseteq \sigma(\Breve{c}_h)$ or $\sigma(\Breve{a}_{i_1,h_1})\nsubseteq \sigma(\Breve{c}_h)$. 
    Firstly, from the sQC of $\cD_\cL^\dag$, we know that agent $(i_1,h_1)$ does not influence agent $(i_2,h)$ in $\cD_\cL^\dag$ for any $i_2\neq i_1$. Then, as shown in the proof of  Theorem \ref{theorem: SI=QC}, we know that agent  $(i_1,h_1)$ does not influence $\Breve{s}_{h_1+1}=\overline{s}_{h_1+1}$ and does not influence $\Breve{c}_{h'}=\overline{c}_{h'}$ for any $h'\le h$. Secondly,  from Assumption \ref{assumption:information evolution}, we know that for any $i\in[n], p_{i,(t_1+1)^-}=\xi_{i,t_1+1}(p_{i,t_1^+},a_{i,t_1},o_{i,t_1+1})$, where $\xi_{i,t_1+1}$ is a fixed transformation. Also, from Assumption \ref{useless action}, we know that $a_{i,t_1}\notin p_{i,(t_1+1)^-}$. Therefore, we can write $p_{i,(t_1+1)^-}=\xi_{i,t_1+1}(p_{i,t_1^+},o_{i,t_1+1})$. From the definition of refinement, we know that $\overline{p}_{i,h_1+1}=\xi_{i,t_1+1}(p_{i,t_1^+},\overline{o}_{i,h_1+1})$. Since  agent $(i_1,h_1)$ does not influence $\overline{s}_{h_1+1}$, it does not influence $\overline{o}_{i,h_1+1}$. Also, agent $(i_1,h_1)$ does not influence $\overline{c}_{h_1+1}$ and $p_{i,h_1^+}$ (which happens before choosing $a_{i_1,t_1}=\overline{a}_{i_1,h_1}$). Therefore, agent $(i_1,h_1)$ does not influence $\overline{\tau}_{i,h_1+1}$, and thus does not influence $\overline{a}_{i,h_1+1}$ for any $i\in[n]$.  Thirdly, we know that for any $i\in[n], \overline{p}_{i,h_1+2}=\overline{\xi}_{i,h_1+2}(\overline{p}_{i,h_1+1},\overline{a}_{i,h_1+1},\overline{o}_{i,h_1+2})$, where $\overline{o}_{i,h_1+2}=\emptyset$.  Since agent $(i_1,h_1)$ does not influence $\overline{p}_{i,h_1+1}$ and $\overline{a}_{i,h_1+1}$, it does not influence  $\overline{p}_{i,h_1+2}$. Also,  we know that it does not influence $\overline{\tau}_{i,h_1+2}$. Therefore, agent $(i_1,h_1)$ does not influence $\overline{\tau}_{i,h_1+2}$, and thus does not influence $\overline{a}_{i,h_1+2}$. This way, we know that agent $(i_1,h_1)$ does not influence $\overline{\tau}_{i,h'}$ for any $i\in[n],h_1<h'\le h$.\\
    Finally, since we proved that  agent $(i_1,h_1)$ does not influence  $\overline{s}_{h}$, $\overline{\tau}_{i,h},\forall i\in[n]$, and $\overline{c}_h=\cap_{i=1}^n\overline{\tau}_{i,h}$. Therefore, we have 
    \begin{align*}
        &\PP_h^{\cD_\cL'}(\overline{s}_h,\overline{p}_h\given \overline{c}_h,\overline{g}_{1:h-1})=\PP_h^{\cD_\cL'}(\overline{s}_h, \overline{p}_h,\overline{c}_h\given \overline{c}_h,\overline{g}_{1:h-1})=\PP_h^{\cD_\cL'}(\overline{s}_h,\overline{\tau}_{h}\given \overline{c}_h,\overline{g}_{1:h-1})\\
        &\quad =\PP_h^{\cD_\cL'}(\overline{s}_h,\{\overline{\tau}_{i,h}\}_{i\in[n]}\given \overline{c}_h,\overline{g}_{1:h-1})=\PP_h^{\cD_\cL'}(\overline{s}_h,\{\overline{\tau}_{i,h}\}_{i\in[n]}\given \overline{c}_h,\overline{g}_{1:h-1}')
    =\PP_h^{\cD_\cL'}(\overline{s}_h,\overline{p}_h\given \overline{c}_h,\overline{g}_{1:h-1}'),
    \end{align*}
    }
    which completes the proof.
\end{proof}

\subsection{Important Auxiliary Definitions} 
\begin{definition}[Full-history strategy]
    Given any Dec-POMDP $\cD_\cL'$ constructed from an LTC problem $\cL$ after reformulation, strict expansion, and refinement, we define the strategy $g_h:\prod_{t=1}^h\overline{\cO}_t\times\prod_{t=1}^{h-1}\overline{\cA}_t\rightarrow \overline{\cA}_{h}$ that uses the full joint observation-action history as a \emph{full-history strategy}, and define $\cG_h$ as the space of all such possible $g_h$, i.e., $\cG_h=\{g_h:\prod_{t=1}^h\overline{\cO}_t\times\prod_{t=1}^{h-1}\overline{\cA}_t\rightarrow \overline{\cA}_{h}\}$. 
    \label{def: fully history strategy}
\end{definition}

\begin{definition}[Perfect recall \cite{kuhn1953extensive}]\label{def:PR}
    We say that agent $i$ has perfect recall if $\forall h=2,\cdots, \overline{H}$, 
    it holds that $
        \overline{\tau}_{i,h-1}\cup\{\overline{a}_{i,h-1}\}\subseteq \overline{\tau}_{i,h}.$ 
    If for any $i\in[n]$, agent $i$ has perfect recall, we call that the Dec-POMDP has a perfect recall property. 
\end{definition}
The following definitions are important for solving the Dec-POMDP constructed from LTC.

\begin{definition}[Value function]
    For each $h\in[\overline{H}]$, given common information $\overline{c}_h$ and strategy $\overline{g}_{1:\overline{H}}\in\overline{\cG}_{1:\overline{H}}$, the value function conditioned on the common information is defined as:
    \begin{equation}
V^{\overline{g}_{1:\overline{H}},\cD_\cL'}_h(\overline{c}_h):=\EE_{\overline{g}_{1:\overline{H}}}^{\cD_\cL'}
\left[\sum_{h'=h}^{\overline{H}}\overline{\cR}_{h'}(\overline{s}_{h'},\overline{a}_{h'},\overline{p}_{h'})\bigggiven \overline{c}_h\right],
    \end{equation}
\end{definition}
\noindent where $\overline{\cR}_{h'}$ takes $\overline{s}_{h'},\overline{a}_{h'},\overline{p}_{h'}$ as input. 

\begin{definition}[Prescription and Q-value function]
\label{definition: prescription}
    \emph{Prescription} is an important concept in the common-information-based framework \cite{ashutosh2013team,ashutosh2013game}. For any $h\in[\overline{H}],i\in[n]$, the \emph{prescription} of agent $i$ at timestep $h$ is defined as $\gamma_{i,h}\in \Gamma_{i,h}$, where  $\Gamma_{i,h}:=\overline{\cA}_{i,h}=\cM_{i,\frac{h+1}{2}}$ if $h=2k-1,k\in[H]$, and $\Gamma_{i,h}:=\{\gamma_{i,h}:\overline{\cP}_{i,h}\rightarrow \overline{\cA}_{i,h}\}$ if $h=2k,k\in[H]$. We use $\gamma_h:=(\gamma_{1,h},\cdots,\gamma_{n,h})$ to denote the joint prescription and $\Gamma_h$ to denote the joint prescription space. The prescriptions are the marginalization of the  strategies $\overline{g}_{1:\overline{H}}$, i.e.,  $\forall k\in[H], \gamma_{i,2k-1}=\overline{g}_{i,2k-1}(\overline{c}_{2k-1}), \gamma_{i,2k}(\cdot)=\overline{g}_{i,2k}(\overline{c}_{2k},\cdot)$.  Then, for any $\overline{g}_{1:\overline{H}}$, we can define the Q-value function as
    \begin{equation}
    Q^{\overline{g}_{1:\overline{H}},\cD_\cL'}_h(\overline{c}_h,\gamma_h):=\EE_{\overline{g}_{1:\overline{H}}}^{\cD_\cL'}\left[\sum_{h'=h}^{\overline{H}}  \overline{\cR}_{h'}(\overline{s}_{h'},\overline{a}_{h'},\overline{p}_{h'})\bigggiven \overline{c}_h,\gamma_h\right].
    \end{equation}
\end{definition}

Note that for any $i\in[n]$, and any odd timestep $2k-1$ with $k\in[H]$, by Assumption \ref{limited communication strategy}, we consider the strategies that only take common information as input. Thus, the prescription $\gamma_{i,2k-1}=\overline{a}_{i,2k-1}$.  To unify the notation, we may also write it as $\gamma_{i,2k-1}(\cdot)=\overline{g}_{i,2k-1}(\overline{c}_{2k-1},\cdot)$, where $\cdot$ takes the value $\emptyset$ rather than $\overline{p}_{i,2k-1}$.

\begin{definition}[Expected approximate common-information model]
\label{definition: AIS}
We define an \emph{expected approximate common-information model} of $\cD_\cL'$ as
\begin{equation*}
    \cM:=\left(\{\hat{\cC}_h\}_{h\in[\overline{H}+1]},\{\hat{\phi}_h\}_{h\in[\overline{H}]},\{\PP^{\cM,z}_h\}_{h\in[\overline{H}]},\Gamma,\{\hat{\cR}_h^\cM\}_{h\in[\overline{H}]}\right)
\end{equation*}
where $\Gamma=\{\Gamma_h\}_{h\in[\overline{H}]}$ is the joint prescription space, $\hat{\cC}_h$ is the space of approximate common information at timestep $h$. $\PP^{\cM,z}_h:\hat{\cC}_h\times \Gamma_h\rightarrow \Delta(\overline{Z}_{h+1})$ gives the probability of $\overline{z}_{h+1}$ under $\hat{c}_h$ and $\gamma_h$. $\hat{\cR}_h^\cM:\hat{\cC}_h\times \Gamma_h\rightarrow [-1,1]$ gives the reward at timestep $h$ given $\hat{c}_h$ and $\gamma_h$.  Then, we call that $\cM$ is an \emph{$(\epsilon_r(\cM),\epsilon_z(\cM))$-expected approximate common-information model} of $\cD_\cL'$, if it has some compression function Compress$_h$ such that $\hat{c}_h=$Compress$_h(\overline{c}_h)$ for each  $h\in[\overline{H}+1]$, and satisfies the  following:  
\begin{itemize}
    \item There exists a transformation function $\hat{\phi}_h$ for all $h\in[\overline{H}]$ such that
    \begin{equation}
        \hat{c}_{h+1}=\hat{\phi}_{h+1}(\hat{c}_{h},\overline{z}_{h+1}), \qquad\text{ where $\overline{z}_{h+1}=\overline{c}_{h+1}\backslash \overline{c}_{h}$}.
        \label{AIS:evolution}
    \end{equation}
    \item For any $\overline{g}_{1:h-1}$ and any prescription $\gamma_h\in \Gamma_h$, it holds that 
    \begin{equation}
\EE^{\cD_\cL'}_{\overline{a}_{1:h-1},\overline{o}_{1:h}\sim \overline{g}_{1:h-1}}|\EE^{\cD_\cL'}[\overline{\cR}_h(\overline{s}_h,\overline{a}_h,\overline{p}_h)\given \overline{c}_h,\gamma_h]-\hat{\cR}^\cM_h(\hat{c}_h,\gamma_h)|\le \epsilon_r(\cM).
        \label{AIS: r close}
\end{equation}
    \item For any $\overline{g}_{1:h-1}$ and any prescription $\gamma_h\in \Gamma_h$, it holds that 
    \begin{equation}
\EE^{\cD_\cL'}_{\overline{a}_{1:h-1},\overline{o}_{1:h}\sim \overline{g}_{1:h-1}}\norm{\PP^{\cD_\cL'}_h(\cdot\given \overline{c}_h,\gamma_h)-\PP^{\cM,z}_h(\cdot\given \hat{c}_h,\gamma_h)}_1\le \epsilon_z(\cM).
        \label{AIS: z close}
    \end{equation}
\end{itemize}
\end{definition}

\begin{definition}[Value functions under $\cM$]\label{def:values_under_M}
Given any expected approximate common-information model $\cM$ of a Dec-POMDP $\cD_\cL'$, any strategy $\overline{g}_{1:\overline{H}}\in\overline{\cG}_{1:\overline{H}}$, and $h\in[\overline{H}]$, we define the value function and Q-value functions under $\cM$ as
\begin{align*}
&V_h^{\overline{g}_{1:\overline{H}},\cM}(\overline{c}_h)=\hat{\cR}_h^{\cM}(\text{Compress}_h(\overline{c}_h), \{\overline{g}_{j,h}(\overline{c}_h,\cdot)\}_{j\in[n]})+\EE^\cM[V_{h+1}^{\overline{g}_{1:\overline{H}},\cM}(\overline{c}_{h+1})\given \text{Compress}_h(\overline{c}_h),\{\overline{g}_{j,h}(\overline{c}_h,\cdot)\}_{j\in[n]}],\\
    &Q_h^{\overline{g}_{1:\overline{H}},\cM}(\overline{c}_h,\gamma_h)=\hat{\cR}_h^{\cM}(\text{Compress}_h(\overline{c}_h),\gamma_h)+\EE^\cM[V_{h+1}^{\overline{g}_{1:\overline{H}},\cM}(\overline{c}_{h+1})\given \text{Compress}_h(\overline{c}_h),\gamma_h],\\
    &Q_h^{\ast,\cM}(\overline{c}_h,\gamma_h)=\max_{\overline{g}_{1:\overline{H}}\in\overline{\cG}_{1:\overline{H}}}Q_h^{\overline{g}_{1:\overline{H}},\cM}(\overline{c}_h,\gamma_h).
\end{align*}
\end{definition}

  \begin{definition}[Model-belief consistency]\label{def:consistency}
   We say the {expected approximate} common-information model $\cM$ of $\cD_\cL'$ is \emph{consistent with} some approximate common-information-based beliefs  {$\{\PP_{h}^{\cM, c}(\overline{s}_h, \overline{p}_h\given \hat{c}_h)\}_{h\in [\overline{H}]}$},  if it satisfies the following:  for all  $h\in [H]$,  
    \begin{equation}
    \label{consis:t}
        \begin{aligned}
&\PP_{2h-1}^{\cM, z}(\overline{z}_{2h}\given\hat{c}_{2h-1}, {\gamma}_{2h-1}) = \sum_{\overline{s}_{2h-1}}\sum_{ \substack{\overline{p}_{2h-1}:\\ \overline{\chi}_{2h}(\overline{p}_{2h-1}, \gamma_{2h-1}, \overline{o}_{2h}=\emptyset) = \overline{z}_{2h}} } \PP_{2h-1}^{\cM, c}(\overline{s}_{2h-1}, \overline{p}_{2h-1}\given\hat{c}_{2h-1}),\\
&\PP_{2h}^{\cM, z}(\overline{z}_{2h+1}\given\hat{c}_{2h}, {\gamma}_{2h}) = \sum_{\overline{s}_{2h}}\sum_{ \substack{\overline{p}_{2h}, \overline{a}_{2h}, \overline{o}_{2h+1}:\\ \overline{\chi}_{2h+1}(\overline{p}_{2h}, \overline{a}_{2h}, \overline{o}_{2h+1}) = \overline{z}_{2h+1}}} \Big(\PP_{2h}^{\cM, c}(\overline{s}_{2h}, \overline{p}_{2h}\given\hat{c}_{2h})\mathds{1}[\overline{a}_{2h}=\gamma_{2h}(\overline{p}_{2h})]\\
&\qquad\qquad\qquad\qquad\qquad\qquad\sum_{\overline{s}_{2h+1}}\overline{\TT}_{2h}(\overline{s}_{2h+1}\given \overline{s}_{2h}, \overline{a}_{2h})]\overline{\OO}_{2h+1}(\overline{o}_{2h+1}\given \overline{s}_{2h+1})\Big),
\end{aligned}
\end{equation}
    \begin{equation}
        \begin{aligned}
&\hat{\cR}_{2h-1}^{\cM}(\hat{c}_{2h-1}, \gamma_{2h-1})=\sum_{\overline{s}_{2h-1}}\sum_{\overline{p}_{2h-1}}\PP_{2h-1}^{\cM, c}(\overline{s}_{2h-1}, \overline{p}_{2h-1}\given\hat{c}_{2h-1})\overline{\cR}_{2h-1}(\overline{s}_{2h-1}, \gamma_{2h-1}, \overline{p}_{2h-1}),\\ 
&\hat{\cR}_{2h}^{\cM}(\hat{c}_{2h}, \gamma_{2h})=\sum_{\overline{s}_{2h}, \overline{p}_{2h},\overline{a}_{2h}}\PP_{2h}^{\cM, c}(\overline{s}_{2h}, \overline{p}_{2h}\given\hat{c}_{2h}) \mathds{1}[\overline{a}_{2h}=\gamma_{2h}(\overline{p}_{2h})]\overline{\cR}_{2h}(\overline{s}_{2h}, \overline{a}_{2h}, \overline{p}_{2h}). 
\label{consis:r}
\end{aligned}
\end{equation}
\end{definition}

\begin{definition}[{Strategy-dependent approximate common-information model}]\label{def:simulation_main} Given an expected approximate common-information model $\Tilde \cM$ of $\cD_\cL'$ (as in Definition \ref{definition: AIS}) and $\overline{H}$ {joint} strategies $g^{1:\overline{H}}$, where each $g^h\in \cG_{1:\overline{H}}$  for $h\in [\overline{H}]$ (as in \Cref{def: fully history strategy}),  
 {we say $\Tilde \cM$ is a    \emph{strategy-dependent  expected approximate {common-information} model}}, denoted as  $\Tilde{\cM}(g^{1:\overline{H}})$,  if it is consistent with {the}  \emph{strategy-dependent} beliefs $\{\PP_{h}^{g^h, \cD_{\cL}'}(\overline{s}_h, \overline{p}_h\given \hat{c}_h)\}_{h\in [\overline{H}]}$ (as per Definition \ref{def:consistency}).  
\end{definition}

\begin{definition}[{Length of approximate common information}]\label{def:L_main}
Given the compression functions $\{\text{Compress}_h\}_{h\in[\overline{H}]}$,  we define an integer $\hat{L}> 0$ as the minimum length such that there exists a mapping $\hat{f}_h:\overline{\cO}_{\max\{1,h-\hat{L}\}: h}\times\overline{\cA}_{\max\{1,h-\hat{L}\}: h-1}\rightarrow\hat{\cC}_h$ that satisfies: for each $h\in [\overline{H}+1]$ and joint history $\{\overline{o}_{1:h}, \overline{a}_{1:h-1}\}$, we have $\hat{f}_h(x_h)=\hat{c}_h$, where $x_{h}=\{\overline{o}_{\max\{h-\hat{L},1\}},\overline{a}_{\max\{h-\hat{L}, 1\}}, \overline{o}_{\max\{h-\hat{L}, 1\} + 1}, \cdots, \overline{a}_{h-1}, \overline{o}_{h}\}$.
\end{definition}

\subsection{Main Results for Planning in QC LTCs}\label{sec:planning_QC_LTC}
\begin{theorem}[Full version of \Cref{theorem: planning}]
    Given any QC LTC problem $\cL$ satisfying Assumptions \ref{gamma observability}, \ref{limited communication strategy}, \ref{useless action}, and \ref{weak gamma observability}, we can construct a Dec-POMDP problem $\cD_\cL'$ with SI-CIBs  such that for any $\epsilon>0$, solving an $\epsilon$-team-optimal strategy in $\cD_\cL'$ can give us an $\epsilon$-team-optimal strategy of $\cL$, and the following holds. Fix $\epsilon_r, \epsilon_z>0$ and given any $(\epsilon_r, \epsilon_z)$-expected approximate common-information model $\cM$ for $\cD_\cL'$ that satisfies  Assumption \ref{assu: one_step_tract}, there exists an algorithm that can compute a  $(2\overline{H}\epsilon_r + \overline{H}^2\epsilon_z)$-team-optimal strategy for the original LTC problem $\cL$ with time complexity $\max_{h\in[\overline{H}]}|\hat{\cC}_h|\cdot \texttt{poly}(|\overline{\cS}|, |\overline{\cA}_h|, |\overline{\cP}_h|, \overline{H})$. In particular, for any fixed $\epsilon>0$, if $\cL$ has a baseline sharing protocol as one of the examples in \S \ref{sec: examples of QC}, one can construct such an   $\cM$ and apply \Cref{main algorithm} to compute an $\epsilon$-team-optimal strategy for $\cL$ with the following complexities:  
    \begin{itemize}
        \item \textbf{Examples 1, 3, 5, 6:} $\texttt{poly}(\max_{h\in[\overline{H}]}(|\overline{\cO}_h||\overline{\cA}_h|)^{C\gamma^{-4}\log(\frac{|\overline{\cS}|}{\epsilon})}, |\overline{\cS}|,
        \overline{H},\frac{1}{\epsilon})$;  
        \item \textbf{Examples 2, 4, 7, 8:} $\texttt{poly}(\max_{h\in[\overline{H}]}(|\overline{\cO}_h||\overline{\cA}_h|)^{C\gamma^{-4}\log(\frac{|\overline{\cS}|}{\epsilon})+2d}, |\overline{\cS}|,
        \overline{H},\frac{1}{\epsilon})$,
    \end{itemize}
    for some universal constant $C>0$. Recall that $\gamma$ is the constant in \Cref{gamma observability}. And $d$ is the delayed step of sharing,  which is a constant as stated in \S\ref{sec: examples of QC}. Note that if $d=\texttt{poly}\log H$, the complexity is still quasi-polynomial.
    \label{thm: full planning}
\end{theorem}

\begin{proof}
    We divide the proof into the following three {\bf Parts}. \\
    
    \noindent\textbf{Part I:} Given any QC LTC problem $\cL$ satisfying Assumptions \ref{gamma observability}, \ref{limited communication strategy}, \ref{useless action}, and \ref{weak gamma observability}, we can construct a Dec-POMDP problem $\cD_\cL'$ with SI-CIBs  such that finding an $\epsilon$-team-optimal strategy can give us an $\epsilon$-team-optimal strategy of $\cL$, as shown in \Cref{main algorithm}.\\
    
    \noindent\textbf{Part II:} Given any $(\epsilon_r,\epsilon_z)$-expected approximate common-information model $\cM$ of the Dec-POMDP $\cD_\cL'$, we aim to show that there exists an algorithm, Algorithm \ref{algorithm under AIS}, that can output an $\epsilon$-team-optimal strategy of $\cD_\cL'$ with $\epsilon=2\overline{H}\epsilon_r + \overline{H}^2\epsilon_z$. Then, together with the result in \textbf{Part I}, we can further conclude that the output of \Cref{main algorithm} is an $\epsilon$-team-optimal strategy of $\cL$. 
    
    First,  we need to prove that solving $\cM$ can get the $\epsilon$-team-optimal strategy of $\cD_\cL'$. We prove the following 2 lemmas first.
    \begin{lemma}
    \label{lemma: V closed}
        For any strategy $\overline{g}_{1:\overline{H}}\in\overline{\cG}_{1:\overline{H}}$, and $h\in[\overline{H}]$, we have
        \begin{equation}
            \EE_{\overline{g}_{1:\overline{H}}}^{\cD_\cL'}[|V_h^{\overline{g}_{1:\overline{H}},\cD_\cL'}(\overline{c}_h)-V_h^{\overline{g}_{1:\overline{H}},\cM}(\overline{c}_h)|]\le (\overline{H}-h+1)\epsilon_r+\frac{(\overline{H}-h+1)(\overline{H}-h)}{2}\epsilon_z.
        \end{equation}
    \end{lemma}
    \begin{proof}
        We prove it by induction. For $h=\overline{H}+1$, we have $V_h^{\overline{g}_{1:\overline{H}},\cD_\cL'}(\overline{c}_h)=V_h^{\overline{g}_{1:\overline{H}},\cM}(\overline{c}_h)=0$. \\
        For the step $h\le \overline{H}$, we have
        \begin{align*}
            &\EE_{\overline{g}_{1:\overline{H}}}^{\cD_\cL'}[|V_h^{\overline{g}_{1:\overline{H}},\cD_\cL'}(\overline{c}_h)-V_h^{\overline{g}_{1:\overline{H}},\cM}(\overline{c}_h)|]\\
            \le &\EE_{\overline{g}_{1:\overline{H}}}^{\cD_\cL'}\left[|\EE^{\cD_\cL}[\overline{\cR}_h(\overline{s}_h,\overline{a}_h,\overline{p}_h)\given \overline{c}_h,\{\overline{g}_{j,h}(\overline{c}_h,\cdot)\}_{j\in[n]}]-\hat{\cR}_h^\cM(\hat{c}_h,\{\overline{g}_{j,h}(\overline{c}_h,\cdot)\}_{j\in[n]})|\right]\\
            &+\EE_{\overline{g}_{1:\overline{H}}}^{\cD_\cL'}\left[|\EE_{\overline{z}_{h+1}\sim \PP^{\cD_\cL'}_h(\overline{c}_h,\{\overline{g}_{j,h}(\overline{c}_h,\cdot)\}_{j\in[n]})}[V_h^{\overline{g}_{1:\overline{H}},\cD_\cL'}(\overline{c}_h\cup \overline{z}_{h+1})]-\EE_{\overline{z}_{h+1}\sim \PP^{\cM,z}_h(\cdot\given\hat{c}_h,\{\overline{g}_{j,h}(\overline{c}_h,\cdot)\}_{j\in[n]})}[V_h^{\overline{g}_{1:\overline{H}},\cM}(\overline{c}_h\cup \overline{z}_{h+1})]|\right]\\
            \le& \epsilon_r+(\overline{H}-h)\EE^{\cD_\cL'}_{\overline{a}_{1:h-1},\overline{o}_{1:h}\sim \overline{g}_{1:h-1}}\norm{\PP^{\cD_\cL'}_h(\cdot\given \overline{c}_h,\gamma_h)-\PP^{\cM,z}_h(\cdot\given \hat{c}_h,\gamma_h)}_1+\EE^{\cD_\cL'}_{\overline{a}_{1:h-1},\overline{o}_{1:h}\sim \overline{g}_{1:h-1}}\left[|V_{h+1}^{\overline{g}_{1:\overline{H}},\cD_\cL'}(\overline{c}_{h+1})-V_{h+1}^{\overline{g}_{1:\overline{H}},\cM}(\overline{c}_{h+1})|\right]\\
            \le &\epsilon_r+(\overline{H}-h)\epsilon_z+(\overline{H}-h)\epsilon_r+\frac{(\overline{H}-h)(\overline{H}-h-1)}{2}\epsilon_z\\
            \le & (\overline{H}-h+1)\epsilon_r+\frac{(\overline{H}-h)(\overline{H}-h+1)}{2}\epsilon_z.
        \end{align*}

        In the third line of this proof, we had $\overline{z}_{h+1}\sim \PP^{\cD_\cL'}_h(\cdot\given \overline{c}_h,\{\overline{g}_{j,h}( \overline{c}_h,\cdot)\}_{j\in[n]})$, where $\overline{z}_{h+1}$ is generated as
        \begin{align*}
        &\PP^{\cD_\cL'}_h(\overline{z}_{h+1}\given \overline{c}_h,\gamma_h)\\
        =&\sum_{\overline{s}_h\in\overline{\cS},\overline{p}_h\in\overline{\cP}_h}\PP^{\cD_\cL'}_h(\overline{s}_h,\overline{p}_h\given \overline{c}_h)\sum_{\overline{s}_{h+1}\in\overline{\cS},\overline{o}_{h+1}\in\overline{\cO}_{h+1}}\overline{\TT}_{h}(\overline{s}_{h+1}\given \overline{s}_h,
        \overline{a}_h)\overline{\OO}_{h+1}(\overline{o}_{h+1}\given \overline{s}_{h+1})\mathds{1}[\overline{\chi}_{h+1}(\overline{p}_h,\overline{a}_h,\overline{o}_{h+1})],
        \end{align*}
        with $\gamma_h=\{\overline{g}_{j,h}( \overline{c}_h,\cdot)\}_{j\in[n]}, \overline{a}_h=\gamma_h(\overline{p}_h)$ if $h=2k$, and $\overline{a}_h=\gamma_h$ otherwise, for  $k\in[H]$.  
    \end{proof}
    \begin{lemma}
    \label{lemma: algorithm output optimum}
        Let $\hat{g}_{1:\overline{H}}^\ast\in \overline{\cG}_{1:\overline{H}}$ be the strategy output by Algorithm \ref{algorithm under AIS}, then for any $h\in[\overline{H}],\overline{c}_h\in\overline{\cC}_h, \overline{g}_{1:\overline{H}}\in\overline{\cG}_{1:\overline{H}}$, it holds that
        \begin{equation}
            V_h^{\overline{g}_{1:\overline{H}},\cM}(\overline{c}_h)\le V_h^{\hat{g}_{1:\overline{H}}^\ast,\cM}(\overline{c}_h).
        \end{equation}
    \end{lemma}
    \begin{proof}
        We prove it by induction. For $h=\overline{H}+1$, we have $V_h^{\overline{g}_{1:\overline{H}},\cM}(\overline{c}_h)=V_h^{\hat{g}_{1:\overline{H}}^\ast,\cM}(\overline{c}_h)=0$. 
        For the timestep $h\le \overline{H}$, we have
        \begin{align*}
             &V_h^{\overline{g}_{1:\overline{H}},\cM}(\overline{c}_h)=\EE^\cM[\hat{r}_h^\cM(\hat{c}_h,\{\overline{g}_{j,h}(\overline{c}_h,\cdot)\}_{j\in[n]})+V_{h+1}^{\overline{g}_{1:\overline{H}},\cM}(\overline{c}_{h+1})\given \hat{c}_h,\overline{g}_{h}]\\
             &\quad\le\EE^\cM[\hat{r}_h^\cM(\hat{c}_h,\{\overline{g}_{j,h}(\overline{c}_h,\cdot)\}_{j\in[n]})+V_{h+1}^{\hat{g}_{1:\overline{H}}^\ast,\cM}(\overline{c}_{h+1})\given \hat{c}_h,\overline{g}_{h}]\\
             &\quad=Q_h^{\hat{g}_{1:\overline{H}}^\ast,\cM}(\overline{c}_h,\{\overline{g}_{j,h}(\overline{c}_h,\cdot)\}_{j\in[n]})\\
            &\quad \le Q_h^{\hat{g}_{1:\overline{H}}^\ast,\cM}(\overline{c}_h,\{\hat{g}_{j,h}(\overline{c}_h,\cdot)\}_{j\in[n]})\\
             &\quad=V_h^{\hat{g}_{1:\overline{H}}^\ast,\cM}(\overline{c}_h).
        \end{align*}
        For the first inequality, we use the induction hypothesis. For the second inequality, we use the property of $\argmax$ in algorithm and $V_h^{\hat{g}_{1:\overline{H}}^\ast,\cM}(\overline{c}_h)=V_h^{\hat{g}_{1:\overline{H}}^\ast,\cM}(\hat{c}_h)$. By induction, we complete the proof.
    \end{proof}
    We now go back to the proof of the theorem. Let $\hat{g}_{1:\overline{H}}^\ast$ be the solution output by Algorithm \ref{algorithm under AIS}, then for any $\overline{g}_{1:\overline{H}}\in\overline{\cG}_{1:\overline{H}}, h\in[\overline{H}],\overline{c}_h\in\overline{\cC}_h$, we have
    \begin{equation}
    \begin{aligned}
    &\EE^{\cD_\cL'}_{\overline{g}_{1:\overline{H}}}\left[V_h^{\overline{g}_{1:\overline{H}},\cD_\cL'}(\overline{c}_h)-V_h^{\hat{g}_{1:\overline{H}}^\ast,\cD_\cL'}(\overline{c}_h)\right]\\
    &\quad=\EE^{\cD_\cL'}_{\overline{g}_{1:\overline{H}}}\left[\left(V_h^{\overline{g}_{1:\overline{H}},\cD_\cL'}(\overline{c}_h)-V_h^{\hat{g}_{1:\overline{H}}^\ast,\cM}(\overline{c}_h)\right)+\left(V_h^{\hat{g}_{1:\overline{H}}^\ast,\cM}(\overline{c}_h)-V_h^{\hat{g}_{1:\overline{H}}^\ast,\cD_\cL'}(\overline{c}_h)\right)\right]\\
    &\quad\le \EE^{\cD_\cL'}_{\overline{g}_{1:\overline{H}}}\left[\left(V_h^{\overline{g}_{1:\overline{H}},\cD_\cL'}(\overline{c}_h)-V_h^{\overline{g}_{1:\overline{H}},\cM}(\overline{c}_h)\right)+\left(V_h^{\hat{g}_{1:\overline{H}}^\ast,\cM}(\overline{c}_h)-V_h^{\hat{g}_{1:\overline{H}}^\ast,\cD_\cL'}(\overline{c}_h)\right)\right]\\
    &\quad\le  (\overline{H}-h+1)\epsilon_r+\frac{(\overline{H}-h)(\overline{H}-h+1)}{2}\epsilon_z+(\overline{H}-h+1)\epsilon_r+\frac{(\overline{H}-h)(\overline{H}-h+1)}{2}\epsilon_z\\
    &\quad= 2(\overline{H}-h+1)\epsilon_r+(\overline{H}-h)(\overline{H}-h+1)\epsilon_z.
    \end{aligned}
    \label{eq: algorithm 6 output gap}
    \end{equation} 
    For the first inequality, we use Lemma \ref{lemma: algorithm output optimum}. For the second inequality, we use Lemma \ref{lemma: V closed}. Then, we can apply $h=1$ to \Cref{eq: algorithm 6 output gap} and have  $J_{\cD_\cL'}(\overline{g}_{1:\overline{H}})\le J_{\cD_\cL'}(\hat{g}_{1:\overline{H}}^\ast)+2\overline{H}\epsilon_r+\overline{H}^2\epsilon_z$. Note that if $\cM$ (the model constructed in the Line \ref{line: construct AIS} of \Cref{main algorithm}) satisfies \Cref{assu: one_step_tract}, then \Cref{main algorithm} has time complexity $\max_{h\in[\overline{H}]}|\hat{\cC}_h|\cdot \texttt{poly}(|\overline{\cS}|, \max_{h\in[\overline{H}]}|\overline{\cA}_h|, \max_{h\in[\overline{H}]}|\overline{\cP}_h|, \overline{H})$. This is because there are at most $\overline{H}\max_{h\in[\overline{H}]}|\hat{\cC}_h|$ iterations in \Cref{algorithm under AIS} (called in Line \ref{line: dp in planning} of \Cref{main algorithm}), and for each iteration, it has $\texttt{poly}(|\overline{\cS}|, \max_{h\in[\overline{H}]}|\overline{\cA}_h|, \max_{h\in[\overline{H}]}||\overline{\cP}_h|)$ time complexity due to \Cref{assu: one_step_tract}. 
    This completes the proof of {\bf Part II}.\\
    
    \noindent\textbf{Part III:} If the baseline sharing of $\cL$ is one of the 8 examples in \S\ref{sec: examples of QC}, we can construct an expected approximate common-information model $\cM$ of $\cD_\cL'$ with $\epsilon_r=\epsilon_z=\epsilon$, and $\cM$ satisfies \Cref{assu: one_step_tract} and is consistent with some belief. 
    
    We first prove the following lemmas, which are useful for bounding the errors $\epsilon_r, \epsilon_z$.
    \begin{lemma}\label{lemma: eps_r,eps_z to belief}
        Given any belief $\{\PP_h^{\cM,c}(\overline{s}_h,\overline{p}_h\given \hat{c}_h)\}_{h\in[\overline{H}]}$ such that an expected approximate common-information model $\cM$ is consistent with such belief, it holds that for any $h\in[\overline{H}],\overline{\cC}_h,\gamma_h\in \Gamma_h$:
        \begin{align*}
            \norm{\PP_h^{\cD_\cL'}(\cdot \given \overline{c}_h, \gamma_h)-\PP_h^{\cM,z}(\cdot\given \hat{c}_h,\gamma_h)}_1&\le \norm{\PP_h^{\cD_\cL'}(\cdot,\cdot\given \overline{c}_h)-\PP_h^{\cM,c}(\cdot,\cdot\given \hat{c}_h)}_1,\\
            |\EE^{\cD_\cL'}[\cR_h(\overline{s}_h,\overline{a}_h,\overline{p}_h)\given \overline{c}_h,\gamma_h]-\hat{\cR}_h^{\cM}(\hat{c}_h,\gamma_h)|&\le \norm{\PP_h^{\cD_\cL'}(\cdot,\cdot\given \overline{c}_h)-\PP_h^{\cM,c}(\cdot,\cdot\given \hat{c}_h)}_1,
        \end{align*}
        where  $\hat{c}_h=$Compress$_h(\overline{c}_h)$.
    \end{lemma}
    \begin{proof}
        Adapted 
        from Lemma 4   in \cite{liu2023tractable} by changing the reward function of $r_{i,h}(s_h,a_h)$ to  $\cR_h(\overline{s}_h,\overline{a}_h,\overline{p}_h)$. Note that the latter can still be evaluated given the common-information-based belief, $\PP_h^{\cD_\cL'}(\overline{s}_h,\overline{p}_h\given \overline{c}_h)$.
    \end{proof}
    Then, we define the belief states following the notation in \cite{noah2022gamma,liu2023tractable} as $\overline{\bb}_1(\emptyset)=\mu_1$, $\overline{\bb}_h(\overline{o}_{1:h},\overline{a}_{1:h-1})=\PP_h^{\cD_\cL'}(\overline{s}_h=\cdot\given \overline{o}_{1:h},\overline{a}_{1:h-1}), \overline{\bb}_h(\overline{o}_{1:h-1},\overline{a}_{1:h-1})=\PP_h^{\cD_\cL'}(\overline{s}_h=\cdot\given \overline{o}_{1:h-1},\overline{a}_{1:h-1})$, where $\overline{\bb}\in\Delta(\cS)$.
    Also, for any $L\ge1$, we define the approximate beliefs of states using the most recent $L$-step history as: 
    \begin{align*}
        \overline{\bb}_{h}'(\overline{o}_{h-L+1:h},\overline{a}_{h-L:h-1})&=\PP_h^{\cD_\cL'}(\overline{s}_{h}=\cdot\given \overline{s}_{h-L}\sim \text{Unif}(\overline{\cS}),\overline{o}_{h-L+1:h},\overline{a}_{h-L:h-1})\\
            \overline{\bb}_h'(\overline{o}_{h-L+1:h-1},\overline{a}_{h-L:h-1})&=\PP_h^{\cD_\cL'}(\overline{s}_h=\cdot\given \overline{s}_{h-L}\sim \text{Unif}(\overline{\cS}),\overline{o}_{h-L+1:h-1},\overline{a}_{h-L:h-1}).
    \end{align*}
    Moreover, for any set $N\subseteq [n]$, we define $\overline{o}_{N,h}=\{\overline{o}_{i,h}\}_{i\in N}$. We can also define the beliefs of states given historical observations and actions as:  for any $N\subseteq[n]$, %
    \begin{align*}
        \overline{\bb}_h(\overline{o}_{1:h-1},\overline{a}_{1:h-1},\overline{o}_{N, h})&=\PP_h^{\cD_\cL'}(\overline{s}_{h}=\cdot\given \overline{o}_{1:h-1},\overline{a}_{1:h-1},\overline{o}_{N, h})\\
        \overline{\bb}_{h}'(\overline{o}_{h-L+1:h-1},\overline{a}_{h-L:h-1},\overline{o}_{N, h})&=\PP_h^{\cD_\cL'}(\overline{s}_{h}=\cdot\given \overline{s}_{h-L}\sim \text{Unif}(\overline{\cS}),\overline{o}_{h-L+1:h-1},\overline{a}_{h-L:h-1},\overline{o}_{N, h}). 
    \end{align*}
    Let $S=|\cS|$ be the cardinality of the state space $\cS$, we have the following lemma. 
    \begin{lemma}\label{lemma: AIS belief close}
        There is a constant $C\ge 1$ such that the following holds. Given any LTC problem $\cL$ satisfying Assumption \ref{gamma observability}, and let $\cD_\cL'$ be the Dec-POMDP after reformulation, strict expansion,  and refinement. Let $\epsilon\ge 0$, fix a strategy $\overline{g}_{1:\overline{H}}\in\overline{\cG}_{1:\overline{H}}$ and indices $1\le h-L<h-1\le \overline{H}$. If $L\ge C\gamma^{-4}\log(\frac{|\overline{\cS}|}{\epsilon})$, then the following set of inequalities hold %
        \begin{align}
            &\EE_{\overline{o}_{1:h},\overline{a}_{1:h-1}\sim \overline{g}_{1:\overline{H}}}\norm{\overline{\bb}_h(\overline{o}_{1:h},\overline{a}_{1:h-1})-\overline{\bb}_h'(\overline{o}_{h-L+1:h},\overline{a}_{h-L:h-1})}_1\le \epsilon\label{belief difference 1}\\
            &\EE_{\overline{o}_{1:h-1},\overline{a}_{1:h-1}\sim \overline{g}_{1:\overline{H}}}\norm{\overline{\bb}_h(\overline{o}_{1:h-1},\overline{a}_{1:h-1})-\overline{\bb}_h'(\overline{o}_{h-L+1:h-1},\overline{a}_{h-L:h-1})}_1\le \epsilon\label{belief difference 2}\\
            &\EE_{\overline{o}_{1:h},\overline{a}_{1:h-1}\sim \overline{g}_{1:\overline{H}}}\norm{\overline{\bb}_h(\overline{o}_{1:h-1},\overline{a}_{1:h-1},\overline{o}_{N,h})-\overline{\bb}_h'(\overline{o}_{h-L+1:h-1},\overline{a}_{h-L:h-1},\overline{o}_{N,h})}_1\le \epsilon\label{belief difference 2.1}.
        \end{align}
    \end{lemma}
    \begin{proof}
        Given any LTC problem $\cL$, we can construct a Dec-POMDP $\check{\cD}$ such that the transition and emission functions of $\check{\cD}$ are the same as $\cL$, and $\check{\cD}$ has a fully-sharing information structure, which means that it shares all the $o_{1:h-1},a_{1:h}$ as common information at timestep $h$. Since $\cD_\cL'$ is constructed from $\cL$ after reformulation, expansion,  and refinement, for any $h\in[\overline{H}]$, then for any $t\in[H]$, $\overline{o}_{2t}=\emptyset$ always holds  and $\overline{a}_{2t-1}$ does not influence the underlying states. Therefore, we have 
    \begin{align*}
            \overline{\bb}_h(\overline{o}_{1:h},\overline{a}_{1:h-1})=\bb_{\lfloor\frac{h+1}{2}\rfloor}(o_{1:\lfloor \frac{h+1}{2}\rfloor},a_{1:\lfloor \frac{h-1}{2}\rfloor})=\check{\bb}_{\lfloor\frac{h+1}{2}\rfloor}(\check{o}_{1:\lfloor \frac{h+1}{2}\rfloor},\check{a}_{1:\lfloor \frac{h-1}{2}\rfloor})\\
            \overline{\bb}_h(\overline{o}_{1:h-1},\overline{a}_{1:h-1})=\bb_{\lfloor\frac{h+1}{2}\rfloor}(o_{1:\lfloor\frac{h}{2}\rfloor},a_{1:\lfloor\frac{h-1}{2}\rfloor})=\check{\bb}_{\lfloor\frac{h+1}{2}\rfloor}( \check{o}_{1:\lfloor \frac{h}{2}\rfloor},\check{a}_{1:\lfloor \frac{h-1}{2}\rfloor}),
        \end{align*}
        where $\bb$ and $\check{\bb}$ are the beliefs of  states similarly defined as $\overline{\bb}$, but with respect to the models $\cL$ and $\check{\cD}$, respectively. 
        For the approximate beliefs of the state, similarly, for any $h\in[\overline{H}]$, we have
        \begin{align*}
            \overline{\bb}'_{h}(\overline{o}_{h-L+1:h-1},\overline{a}_{h-L:h-1})=\bb'_{\lfloor\frac{h+1}{2}\rfloor}(o_{\lfloor\frac{h-L+2}{2}\rfloor:\lfloor\frac{h}{2}\rfloor}, a_{\lfloor \frac{h-L}{2}\rfloor: \lfloor \frac{h-1}{2}\rfloor})=\check{\bb}'_{\lfloor\frac{h+1}{2}\rfloor}( \check{o}_{\lfloor\frac{h-L+2}{2}\rfloor:\lfloor\frac{h}{2}\rfloor},\check{a}_{\lfloor \frac{h-L}{2}\rfloor: \lfloor \frac{h-1}{2}\rfloor})\\
            \overline{\bb}'_{h}(\overline{o}_{h-L+1:h},\overline{a}_{h-L:h-1})=\bb'_{\lfloor\frac{h+1}{2}\rfloor}( o_{\lfloor\frac{h-L+2}{2}\rfloor:\lfloor\frac{h+1}{2}\rfloor},a_{\lfloor \frac{h-L}{2}\rfloor: \lfloor \frac{h-1}{2}\rfloor})=\check{\bb}'_{\lfloor\frac{h+1}{2}\rfloor}( \check{o}_{\lfloor\frac{h-L+2}{2}\rfloor:\lfloor\frac{h+1}{2}\rfloor},\check{a}_{\lfloor \frac{h-L}{2}\rfloor: \lfloor \frac{h-1}{2}\rfloor}),
        \end{align*}
        where $\bb'$ and $\check{\bb}'$ are the belief states similarly defined as $\overline{\bb}'$ with respect to models $\cL$ and $\check{\cD}$.
        Also, since for any $t\in[H], \overline{a}_{2t-1}$ are communication actions,  $\overline{o}_{2t}=\emptyset$ is null, and $\overline{s}_{2t-1}=\overline{s}_{2t}$ always holds. Then, we can write \Cref{belief difference 1} and  \Cref{belief difference 2} as %
        \begin{align}
            \EE_{\{\overline{o}_{2t-1}\}_{t=1}^{\lfloor\frac{h+1}{2}\rfloor},\{\overline{a}_{2t}\}_{t=1}^{\lfloor\frac{h-1}{2}\rfloor}\sim \overline{g}_{1:\overline{H}}}\norm{\overline{\bb}_h(\overline{o}_{1:h},\overline{a}_{1:h-1})-\overline{\bb}_h'(\overline{o}_{h-L+1:h},\overline{a}_{h-L:h-1})}_1&\le \epsilon \label{belief difference 3}\\
            \EE_{\{\overline{o}_{2t-1}\}_{t=1}^{\lfloor\frac{h}{2}\rfloor},\{\overline{a}_{2t}\}_{t=1}^{\lfloor\frac{h-1}{2}\rfloor}\sim \overline{g}_{1:\overline{H}}}\norm{\overline{\bb}_h(\overline{o}_{1:h-1},\overline{a}_{1:h-1})-\overline{\bb}_h'(\overline{o}_{h-L+1:h-1},\overline{a}_{h-L:h-1})}_1 &\le \epsilon\label{belief difference 4}.
        \end{align}
        Since $\check{\cD}$ has a fully-sharing IS, 
          given any strategy $\overline{g}_{1:\overline{H}}$, we can construct a strategy $\check{g}_{1:H}$ such that, for any $\overline{a}_{1:h-1},\overline{o}_{1:h}$
        \begin{align*}
            \PP^{\cD_\cL'}(\{\overline{o}_{2t-1}\}_{t=1}^{\lfloor\frac{h+1}{2}\rfloor},\{\overline{a}_{2t}\}_{t=1}^{\lfloor\frac{h-1}{2}\rfloor}\given \overline{g}_{1:\overline{H}})=\PP^{\check{\cD}}(\check{o}_{1:\lfloor\frac{h+1}{2}\rfloor},\check{a}_{1:\lfloor\frac{h-1}{2}\rfloor}\given \check{g}_{1:H}).
        \end{align*}
        Since $\check{\cD}$ satisfies Assumption \ref{gamma observability}, we can apply Theorem 10 in \cite{liu2023tractable} with $\check{g}_{1:H}$ to get the result that there is a constant $C_0\ge 1$ such that if $L'\ge C_0\gamma^{-4}\log(\frac{S}{\epsilon})$, the following holds
        \begin{align}
            \EE_{\check{o}_{1:\lfloor\frac{h+1}{2}\rfloor},\check{a}_{1:\lfloor\frac{h-1}{2}\rfloor}\sim \check{g}_{1:H}}\norm{\check{\bb}_{\lfloor\frac{h+1}{2}\rfloor}( \check{o}_{1:\lfloor \frac{h+1}{2}\rfloor},\check{a}_{1:\lfloor \frac{h-1}{2}\rfloor})-\check{\bb}'_{\lfloor\frac{h+1}{2}\rfloor}( \check{o}_{\lfloor\frac{h+1}{2}\rfloor-L'+1:\lfloor\frac{h+1}{2}\rfloor},\check{a}_{\lfloor \frac{h}{2}\rfloor-L': \lfloor \frac{h-1}{2}\rfloor})}_1\le \epsilon\\
            \EE_{\check{o}_{1:\lfloor\frac{h}{2}\rfloor},\check{a}_{1:\lfloor\frac{h-1}{2}\rfloor}\sim \check{g}_{1:H}}\norm{\check{\bb}_{\lfloor\frac{h+1}{2}\rfloor}(\check{o}_{1:\lfloor \frac{h}{2}\rfloor},\check{a}_{1:\lfloor \frac{h-1}{2}\rfloor})-\check{\bb}'_{\lfloor\frac{h+1}{2}\rfloor}( \check{o}_{\lfloor\frac{h+1}{2}\rfloor-L'+1:\lfloor\frac{h}{2}\rfloor},\check{a}_{\lfloor \frac{h}{2}\rfloor-L': \lfloor \frac{h-1}{2}\rfloor})}_1\le \epsilon.
        \end{align}
        We choose $C=3C_0, L=2L'+1$. If $L\ge C\gamma^{-4}\log(\frac{|\cS|}{\epsilon})$, we have $L'\ge C_0\gamma^{-4}\log(\frac{|\overline{\cS}|}{\epsilon})$. Therefore, we can get  \Cref{belief difference 3} and \Cref{belief difference 4}.
         
        For \Cref{belief difference 2.1}, we can apply Equation  (E.10) of Theorem 10 in \cite{liu2023tractable} with a slight change as 
        \begin{equation}
            \EE_{\check{o}_{1:h},\check{a}_{1:h-1}\sim \check{g}_{1:H}}^{\check{\cD}}\norm{\check{\bb}_h(\check{o}_{1:h-1}, \check{a}_{1:h-1}, \check{o}_{N,h})-\check{\bb}_h'(\check{o}_{h-L+1:h-1},\check{a}_{h-L:h-1}, \check{o}_{N,h})}_1\le \epsilon.
            \label{belief difference new 2.1}
        \end{equation}
        This still holds we change the posterior update $F^q(P:\check{o}_{1,h})$ to $F^q(P:\check{o}_{N,h})$,  when applying Lemma 12 in the proof of Theorem 10 in \cite{liu2023tractable}. Therefore, we can use the same arguments to prove \Cref{belief difference 2.1} from \Cref{belief difference new 2.1} as above, which  completes the proof.
    \end{proof}
    
    Then, we discuss the 8 examples of QC LTC given in \S\ref{sec: examples of QC} case by case. For each of them, we compress the common information as $\{\hat{c}_h\}_{h\in[\overline{H}]}$ using a finite-memory truncation, 
    construct beliefs $\{\PP_h^{\cM,c}(\overline{s}_h,\overline{p}_h\given \hat{c}_h)\}_{h\in[\overline{H}]}$, and then  
    construct an expected approximate common-information model $\cM$ according to $\{\hat{c}_h\}_{h\in[\overline{H}]}$, $\{\PP_h^{\cM,c}(\overline{s}_h,\overline{p}_h\given \hat{c}_h)\}_{h\in[\overline{H}]}$, and \Cref{def:consistency}.   Note that after reformulation, strict expansion, and refinement, {\bf Example 5} will be the same as {\bf Example 1}, and
    {\bf Example 7} will be the same as {\bf Example 2}; {\bf Example 6} is similar  to \textbf{Example 1}, and \textbf{Example 8} is similar to  \textbf{Example 2}. 
    Hence, we categorize the examples into 6 {\bf Types}, and verify that the constructed $\cM$ is an $(\epsilon_r, \epsilon_z)$-expected approximate common-information model with $\epsilon_r=\epsilon_z=\epsilon$ and satisfies \Cref{assu: one_step_tract}.\\

    \noindent\textbf{Type 1:} Baseline sharing of $\cL$ is one of {\bf Examples 1} and \textbf{5} in \S\ref{sec: examples of QC}.
    Then, common information should be that for any $t\in[H], \overline{c}_{2t-1}=\{\overline{o}_{1:2t-2},\overline{a}_{1:2t-2}\}, \overline{c}_{2t}=\{\overline{o}_{1:2t-2},\overline{a}_{1:2t-1},\overline{o}_{N,2t-1}\}, N\subseteq [n]$, where $N$ is the set of agents that choose to share their observations through additional sharing, and $N$ can be inferred from $\overline{c}_{2t}$. Then, we have that
    $\PP_{2t-1}^{\cD_\cL'}(\overline{s}_{2t-1},\overline{p}_{2t-1}\given \overline c_{2t-1})=\overline{\bb}_{2t-1}(\overline{o}_{1:2t-2},\overline{a}_{1:2t-2})(\overline{s}_{2t-1})\overline{\OO}_{2t-1}(\overline{o}_{2t-1}\given \overline{s}_{2t-1})$. Fix compression  length $L>0$, for timestep $2t-1$, we define the approximate common information as $\hat{c}_{2t-1}=\{\overline{o}_{2t-L:2t-2},\overline{a}_{2t-1-L:2t-2}\}$, and the approximate common information conditioned belief as $\PP_{2t-1}^{\cM,c}(\overline{s}_{2t-1},\overline{p}_{2t-1}\given \hat{c}_{2t-1})=\overline{\bb}_{2t-1}'(\overline{o}_{2t-L:2t-2},\overline{a}_{2t-1-L:2t-2})(\overline{s}_{2t-1})\overline{\OO}_{2t-1}(\overline{o}_{2t-1}\given \overline{s}_{2t-1})$. Also, we have $\PP_{2t}^{\cD_\cL'}(\overline{s}_{2t},\overline{p}_{2t}\given \overline{c}_{2t})=\overline{\bb}_{2t-1}(\overline{o}_{1:2t-2},\overline{a}_{1:2t-2},\overline{o}_{N,2t-1})(\overline{s}_{2t-1})\PP_{2t-1}(\overline{o}_{-N,2t-1}\given \overline{s}_{2t-1},\overline{o}_{N,2t-1})$, where $\PP_{2t-1}(\overline{o}_{-N,2t-1}\given \overline{s}_{2t-1},\overline{o}_{N,2t-1})=\frac{\overline{\OO}_{2t-1}(\overline{o}_{N,2t-1},\overline{o}_{-N,2t-1}\given \overline{s}_{2t-1})}{\sum_{\overline{o}_{-N,2t-1}'}\overline{\OO}_{2t-1}(\overline{o}_{N,2t-1},\overline{o}_{-N,2t-1}'\given \overline{s}_{2t-1})}$, we used the belief at timestep $2t-1$ since $\overline{s}_{2t}=\overline{s}_{2t-1}$, and we noted that $\overline{p}_{2t}=\overline{o}_{-N,2t-1}$. For timestep $2t$, we define the approximate common information a $\hat{c}_{2t}=\{\overline{o}_{2t-L:2t-2},\overline{a}_{2t-L-1:2t-1},\overline{o}_{N,2t-1}\}$, and the common information conditioned belief as  $\PP_{2t}^{\cM,c}(\overline{s}_{2t},\overline{p}_{2t}\given \hat{c}_{2t})=\overline{\bb}_{2t-1}'(\overline{o}_{2t-L:2t-2},\overline{a}_{2t-1-L:2t-2},\overline{o}_{N,2t-1})(\overline{s}_{2t-1})\PP_{2t-1}(\overline{o}_{-N,2t-1}\given \overline{s}_{2t-1},\overline{o}_{N,2t-1})$. After that, we can construct $\{\PP_h^{\cM,z}\}_{h\in[\overline{H}]}, \{\hat{\cR}_h\}_{h\in[\overline{H}]}$ based on $\cD_\cL'$ and $\{\PP_h^{\cM,c}\}_{h\in[\overline{H}]}$.

    Now, we need to verify that Definition \ref{definition: AIS} is satisfied.
        \begin{itemize}
            \item The $\{\hat{c}_h\}_{h\in[\overline{H}]}$ satisfied \Cref{AIS:evolution}, since for any $h\in[\overline{H}]$, $\hat{c}_{h+1}\subseteq \hat{c}_h\cup \overline{z}_h$.
            \item Note that for any $\overline{c}_{2t-1}$ and the corresponding $\hat{c}_{2t-1}$ constructed above, according to Lemma \ref{lemma: eps_r,eps_z to belief}, we have:
            \begin{align*}
                &\norm{\PP_{2t-1}^{\cD_\cL'}(\cdot,\cdot\given \overline{c}_{2t-1})-\PP_{2t-1}^{\cM,c}(\cdot,\cdot\given \hat{c}_{2t-1})}_1\\
                &\quad=\sum_{\overline{s}_{2t-1},\overline{o}_{2t-1}}|\overline{\bb}_{2t-1}(\overline{o}_{1:2t-2},\overline{a}_{1:2t-2})(\overline{s}_{2t-1})\overline{\OO}_{2t-1}(\overline{o}_{2t-1}\given  \overline{s}_{2t-1})\\&\qquad\qquad-\overline{\bb}_{2t-1}'(\overline{o}_{2t-L:2t-1},\overline{a}_{2t-1-L:2t-2})(\overline{s}_{2t-1})\overline{\OO}_{2t-1}(\overline{o}_{2t-1}\given \overline{s}_{2t-1})|\\
                &\quad\le\norm{\overline{\bb}_{2t-1}(\overline{o}_{1:2t-2},\overline{a}_{1:2t-2})-\overline{\bb}_{2t-1}'(\overline{o}_{2t-L:2t-1},\overline{a}_{2t-1-L:2t-2})}_1.
            \end{align*} For any $\overline{c}_{2t}$ and the corresponding $\hat{c}_{2t}$ constructed above, according to Lemma \ref{lemma: eps_r,eps_z to belief}, we have: 
            \begin{align*}
                &\norm{\PP_{2t}^{\cD_\cL'}(\cdot,\cdot\given \overline{c}_{2t})-\PP_{2t}^{\cM,c}(\cdot,\cdot\given \hat{c}_{2t})}_1\\
                &\quad=\sum_{\overline{s}_{2t-1},\overline{o}_{-N,2t-1}}|\overline{\bb}_{2t-1}(\overline{o}_{1:2t-2},\overline{a}_{1:2t-2},\overline{o}_{N,2t-1})(\overline{s}_{2t-1})\PP_{2t-1}(\overline{o}_{-N,2t-1}\given \overline{s}_{2t-1},\overline{o}_{N,2t-1})\\
                &\qquad\qquad-\overline{\bb}_{2t-1}'(\overline{o}_{2t-L:2t-2},\overline{a}_{2t-1-L:2t-2},\overline{o}_{N,2t-1})(\overline{s}_{2t-1})\PP_{2t-1}(\overline{o}_{-N,2t-1}\given \overline{s}_{2t-1},\overline{o}_{N,2t-1})|\\
                &\quad\le\norm{\overline{\bb}_{2t-1}(\overline{o}_{1:2t-2},\overline{a}_{1:2t-2},\overline{o}_{N,2t-1})-\overline{\bb}_{2t-1}'(\overline{o}_{2t-L:2t-2},\overline{a}_{2t-1-L:2t-2},\overline{o}_{N,2t-1})}_1.
            \end{align*}
            If we choose $L\ge C\gamma^{-4}\log(\frac{|\overline{\cS}|}{\epsilon})$, then from \Cref{lemma: AIS belief close}, we have that for any $h\in [\overline{H}]$
        \begin{align*}
        \EE_{\overline{o}_{1:h},\overline{a}_{1:h-1}\sim \overline{g}_{1:\overline{H}}}\norm{\PP_{h}^{\cD_\cL'}(\cdot,\cdot\given \overline{c}_h)-\PP_h^{\cM,c}(\cdot,\cdot\given \hat{c}_h)}_1\le \epsilon. 
        \end{align*}
            \item Based on \Cref{lemma: one_step_tract}, $\cM$ satisfies \Cref{assu: one_step_tract} since: \textbf{Turn-based structures} in \S\ref{one-step tractability condition}  is satisfied if the additional condition \ref{cond:tract_1}) in \S\ref{sec: examples of QC} holds; \textbf{Factorized structures} in \S\ref{one-step tractability condition}  is satisfied if the additional condition \ref{cond:tract_2}) in \S\ref{sec: examples of QC} holds; \textbf{Nested private information} in \S\ref{one-step tractability condition}  is satisfied if the additional condition \ref{cond:tract_3}) in \S\ref{sec: examples of QC} holds.
             
        \end{itemize}
        Therefore, based on Lemma \ref{lemma: eps_r,eps_z to belief}, such a model is an $(\epsilon_r,\epsilon_z)$-expected approximate common-information model with $\epsilon_r=\epsilon_z=\epsilon$ and satisfies \Cref{assu: one_step_tract}.\\ 
        
        \noindent\textbf{Type 2:} Baseline sharing of $\cL$ is {\bf Example 3} in \S\ref{sec: examples of QC}. Then, common information should be that for any $t\in[H], \overline{c}_{2t-1}=\{\overline{o}_{1:2t-2},\overline{a}_{1:2t-2},\overline{o}_{1,2t-1}\}, \overline{c}_{2t}=\{\overline{o}_{1:2t-2},\overline{a}_{1:2t-1},\overline{o}_{N,2t-1}\}, N\subseteq [n], 1\in N$. Here $N$ is the same as defined in {\bf Type 1}, but it must satisfy that $1\in N$. Similarly as {\bf Type 1}, we construct $\hat{c}_{2t-1}=\{\overline{o}_{2t-L:2t-2},\overline{a}_{2t-L-1:2t-2},\overline{o}_{1,2t-1}\}, \hat{c}_{2t}=\{\overline{o}_{2t-L:2t-2},\overline{a}_{2t-L-1:2t-1},\overline{o}_{N,2t-1}\}$, and approximate common information conditioned belief as 
        $\PP_{2t-1}^{\cM,c}(\overline{s}_{2t-1},\overline{p}_{2t-1}\given \hat{c}_{2t-1})=\overline{\bb}_{2t-1}'(\overline{o}_{2t-L:2t-2},\overline{a}_{2t-1-L:2t-2},\overline{o}_{1,2t-1})(\overline{s}_{2t-1})\PP_{2t-1}(\overline{o}_{-1,2t-1}\given \overline{s}_{2t-1},\overline{o}_{1,2t-1}), \PP_{2t}^{\cM,c}(\overline{s}_{2t},\overline{p}_{2t}\given \hat{c}_{2t})=\overline{\bb}_{2t-1}'(\overline{o}_{2t-L:2t-2},\overline{a}_{2t-1-L:2t-2},\overline{o}_{N,2t-1})(\overline{s}_{2t-1})\PP_{2t-1}(\overline{o}_{-N,2t-1}\given \overline{s}_{2t-1},\overline{o}_{N,2t-1})$, where we used the belief at timestep $2t-1$, similarly as \textbf{Type 1}. After that, we can construct $\{\PP_h^{\cM,z}\}_{h\in[\overline{H}]}, \{\hat{\cR}_h\}_{h\in[\overline{H}]}$ based on $\cD_\cL'$ and    $\{\PP_h^{\cM,c}\}_{h\in[\overline{H}]}$. Now, we need to verify Definition \ref{definition: AIS} is satisfied.
        \begin{itemize}
            \item The $\{\hat{c}_h\}_{h\in[\overline{H}]}$ satisfies \Cref{AIS:evolution}, since for any $h\in[\overline{H}]$, $\hat{c}_{h+1}\subseteq \hat{c}_h\cup \overline{z}_h$.
            \item Note that for any $\overline{c}_{2t-1}$ and the corresponding $\hat{c}_{2t-1}$ constructed above, according to Lemma \ref{lemma: eps_r,eps_z to belief}, we have:
                \begin{align*}
                &\norm{\PP_{2t-1}^{\cD_\cL'}(\cdot,\cdot\given \overline{c}_{2t-1})-\PP_{2t-1}^{\cM,c}(\cdot,\cdot\given \hat{c}_{2t-1})}_1\\
                &\quad=\sum_{\overline{s}_{2t-1},\overline{o}_{-1,2t-1}}|\overline{\bb}_{2t-1}(\overline{o}_{1:2t-2},\overline{a}_{1:2t-1},\overline{o}_{1,2t-1})(\overline{s}_{2t-1})\PP_{2t-1}(\overline{o}_{-1,2t-1}\given \overline{s}_{2t-1},\overline{o}_{1,2t-1})\\
               &\qquad\qquad \quad -\overline{\bb}_{2t-1}'(\overline{o}_{2t-L:2t-2},\overline{a}_{2t-1-L:2t-2},\overline{o}_{1,2t-1})(\overline{s}_{2t-1})\PP_{2t-1}(\overline{o}_{-1,2t-1}\given \overline{s}_{2t-1},\overline{o}_{1,2t-1})|\\
                &\quad\le\norm{\overline{\bb}_{2t-1}(\overline{o}_{1:2t-2},\overline{a}_{1:2t-1},\overline{o}_{1,2t-1})-\overline{\bb}_{2t-1}'(\overline{o}_{2t-L:2t-2},\overline{a}_{2t-1-L:2t-2},\overline{o}_{1,2t-1})}_1.
            \end{align*}
            For any $\overline{c}_{2t}$ and the corresponding $\hat{c}_{2t}$ constructed above, according to Lemma \ref{lemma: eps_r,eps_z to belief}, we have: 
            \begin{align*}
                &\norm{\PP_{2t}^{\cD_\cL'}(\cdot,\cdot\given \overline{c}_{2t})-\PP_{2t}^{\cM,c}(\cdot,\cdot\given \hat{c}_{2t})}_1\\
                &\quad=\sum_{\overline{s}_{2t-1},\overline{o}_{-N,2t-1}}|\overline{\bb}_{2t-1}(\overline{o}_{1:2t-2},\overline{a}_{1:2t-2},\overline{o}_{N,2t-1})(\overline{s}_{2t-1})\PP_{2t-1}(\overline{o}_{-N,2t-1}\given \overline{s}_{2t-1},\overline{o}_{N,2t-1})\\
               &\qquad\qquad -\overline{\bb}_{2t-1}'(\overline{o}_{2t-L:2t-2},\overline{a}_{2t-1-L:2t-2},\overline{o}_{N,2
               t-1})(\overline{s}_{2t-1})\PP_{2t-1}(\overline{o}_{-N,2t-1}\given \overline{s}_{2t-1},\overline{o}_{N,2t-1})|\\
                &\quad\le\norm{\overline{\bb}_{2t-1}(\overline{o}_{1:2t-2},\overline{a}_{1:2t-2},\overline{o}_{N,2t-1})-\overline{\bb}_{2t-1}'(\overline{o}_{2t-L:2t-2},\overline{a}_{2t-1-L:2t-2},\overline{o}_{N,2t-1})}_1.
            \end{align*}        
            If we choose $L\ge C\gamma^{-4}\log(\frac{|\overline{\cS}|}{\epsilon})$, then from Lemma \ref{lemma: AIS belief close}, we have that   for any $h\in [\overline{H}]$ 
        \begin{align*}
        \EE_{\overline{o}_{1:h},\overline{a}_{1:h-1}\sim \overline{g}_{1:\overline{H}}}\norm{\PP_{h}^{\cD_\cL'}(\cdot,\cdot\given \overline{c}_h)-\PP_h^{\cM,c}(\cdot,\cdot\given \hat{c}_h)}_1\le \epsilon. 
        \end{align*}
            \item {Based on \Cref{lemma: one_step_tract}, $\cM$ satisfies \Cref{assu: one_step_tract},  since the condition \textbf{Nested private information} in \S\ref{one-step tractability condition}  is satisfied.}
        \end{itemize}
        Therefore, based on Lemma \ref{lemma: eps_r,eps_z to belief}, such a model is an $(\epsilon_r,\epsilon_z)$-expected approximate common-information model with $\epsilon_r=\epsilon_z=\epsilon$  and satisfies \Cref{assu: one_step_tract}.\\ 

        \noindent\textbf{Type 3:} Baseline sharing of $\cL$ is one of {\bf Examples 2}  and \textbf{7} in \S\ref{sec: examples of QC}. 
        Then, the common information should be that, for any $h\in [\overline{H}],\overline{c}_{h}=\{\overline{o}_{1:h-2d},\overline{a}_{1,1:h-1},\{\overline{a}_{-1,2t-1}\}_{t=1}^{\lfloor\frac{h}{2}\rfloor}, \overline{o}_{1,h-2d+1:h},\overline{o}_{M}\}$, where $M\subseteq\{(i,t)\given 1<i\le n, h-2d+1\le t\le h\}$, $\overline{o}_M=\{\overline{o}_{i,t}\given (i,t)\in M\}$, and the $-1$ index means all the agents except agent 1. The corresponding private information is defined as $\overline{p}_h=\{\overline{o}_{i,t}\given 1<i\le n,h-2d<t\le h, (i,t)\notin M\}$. Actually, $\overline{o}_M$ is the observation shared by the additional sharing in $\cL$.  Denote $f_{\tau, h-2d}=\{\overline{o}_{1:h-2d},\overline{a}_{1,1:h-2d-1},\{\overline{a}_{-1,2t-1}\}_{t=1}^{\lfloor\frac{h-2d}{2}\rfloor}\}, f_a=\{\overline{a}_{1,h-2d: h-1}, \{\overline{a}_{-1,2t-1}\}_{t=\lfloor \frac{h-2d}{2}\rfloor+1}^{\lfloor\frac{h}{2}\rfloor}\}, f_o=\{\overline{o}_{1,h-2d+1:h},\overline{o}_M\}$, 
        we can then compute the common-information-based belief as
        \begin{align*}
            \PP_h^{\cD_\cL'}(\overline{s}_h,\overline{p}_h\given\overline{c}_h)=&\sum_{\overline{s}_{h-2d}}\PP_h^{\cD_\cL'}(\overline{s}_h,\overline{p}_h\given\overline{s}_{h-2d},f_a,f_o)\PP_h^{\cD_\cL'}(\overline{s}_{h-2d}\given f_{\tau,h-2d},f_a, f_o)\\
            =&\sum_{\overline{s}_{h-2d}}\PP_h^{\cD_\cL'}(\overline{s}_h,\overline{p}_h\given\overline{s}_{h-2d},f_a,f_o)\frac{\PP_h^{\cD_\cL'}(\overline{s}_{h-2d},f_a,f_o\given f_{\tau,h-2d})}{\sum_{\overline{s}_{h-2d}'}\PP_h^{\cD_\cL'}(\overline{s}_{h-2d}',f_a,f_o\given f_{\tau,h-2d})}.
        \end{align*}
        Denote the probability $P_h(f_o\given \overline{s}_{h-2d},f_a):=\Pi_{t=1}^{2d}\PP_h^{\cD_\cL'}(\overline{o}_{1,h-2d+t},\overline{o}_{M_{h-2d+t}}\given \overline{s}_{h-2d},\overline{a}_{1,h-2d:h-2d+t})$, where $M_{h-2d+t}=\{(i,h-2d+t)\given (i,h-2d+t)\in M\}$ denotes the set of observations at timestep $h-2d+t$ and shared through additional sharing. With such notation, we have
        \begin{align*}
            \PP_h^{\cD_\cL'}(\overline{s}_{h-2d}\given f_{\tau,h-2d},f_a,f_o)=&\frac{\overline{\bb}_{h-2d}(\overline{o}_{1:h-2d},\overline{a}_{1:h-2d-1})(\overline{s}_{h-2d})P_h(f_o\given \overline{s}_{h-2d},f_a)}{\sum_{\overline{s}_{h-2d}'}\overline{\bb}_{h-2d}(\overline{o}_{1:h-2d},\overline{a}_{1:h-2d-1})(\overline{s}_{h-2d}')P_h(f_o\given \overline{s}_{h-2d}',f_a)}\\
            =&F^{P_h(\cdot\given \cdot, f_a)}(\overline{\bb}_{h-2d}(\overline{o}_{1:h-2d},\overline{a}_{1:h-2d-1});f_o)(\overline{s}_{h-2d}),
        \end{align*}
        where $F^{P_h(\cdot\given \cdot, f_a)}(\cdot; f_o):\Delta(\cS)\rightarrow \Delta(\cS)$ is the posterior belief update function, which is defined as: $F^{P_h(\cdot\given \cdot, f_a)}(\overline{\bb}; f_o)(\overline{s})=\frac{\overline{\bb}(\overline{s})P_h(f_o\given \overline{s}, f_a)}{\sum_{\overline{s}'\in\overline{\cS}}\overline{\bb}(\overline{s}')P_h(f_o\given \overline{s}', f_a)}$ for any $\overline{s}\in\overline{\cS},\overline{\bb}\in\Delta(\overline{\cS})$.
        
        Then, we can define the approximate common information as  $\hat{c}_h:=\{\overline{o}_{h-2d-L+1:h-2d},\overline{o}_{1,h-2d+1:h},\overline{a}_{1,h-2d-L:h-1},\{\overline{a}_{-1,2t-1}\}_{t=\lfloor \frac{h-2d+1}{2}\rfloor}^{\lfloor\frac{h}{2}\rfloor}, \overline{o}_M\}$ and the corresponding approximate common information conditioned belief as 
        \begin{align*}
            \PP_h^{\cM,c}(\overline{s}_h,\overline{p}_h\given \hat{c}_h)=\sum_{\overline{s}_{h-2d}}\PP_h^{\cD_\cL'}(\overline{s}_h,\overline{p}_h\given \overline{s}_{h-2d},f_a,f_o)F^{P_h(\cdot\given \cdot, f_a)}(\overline{\bb}_{h-2d}'(\overline{o}_{h-2d-L+1:h-2d},\overline{a}_{h-2d-L:h-2d-1});f_o)(\overline{s}_{h-2d}).
        \end{align*}
         After that, we can construct $\{\PP_h^{\cM,z}\}_{h\in[\overline{H}]}, \{\hat{\cR}_h\}_{h\in[\overline{H}]}$ based on $\cD_\cL'$ and  $\{\PP_h^{\cM,c}\}_{h\in[\overline{H}]}$. Now we verify that Definition \ref{definition: AIS} is satisfied. 
        \begin{itemize}
            \item By definition, the $\{\hat{c}_h\}_{h\in[\overline{H}]}$ satisfies  \Cref{AIS:evolution}. 
            \item For any $\overline{c}_h$ and the corresponding $\hat{c}_h$ constructed above, according to Lemma \ref{lemma: eps_r,eps_z to belief}, we have:
            \begin{align*}
                &\norm{\PP_{h}^{\cD_\cL'}(\cdot,\cdot\given \overline{c}_h)-\PP_h^{\cM,c}(\cdot,\cdot\given \hat{c}_h)}_1\\
                &\quad\le \norm{F^{P_h(\cdot\given \cdot, f_a)}(\overline{\bb}_{h-2d}(\overline{o}_{1:h-2d},\overline{a}_{1:h-2d-1});f_o)-F^{P_h(\cdot\given \cdot, f_a)}(\overline{\bb}_{h-2d}'(\overline{o}_{h-2d-L+1:h-2d},\overline{a}_{h-2d-L:h-2d-1});f_o)}_1. 
            \end{align*}
            If we choose $L\ge C\gamma^{-4}\log(\frac{|\overline{\cS}|}{\epsilon})$, then for any strategy $\overline{g}_{1:\overline{H}}$, by taking expectations over $f_{\tau,h-2d},f_a,f_o$, from Lemma \ref{lemma: AIS belief close} and Lemma 12 in \cite{liu2023tractable}, we have that  for any $h\in [\overline{H}]$
            \begin{align*}
                \EE_{\overline{o}_{1:h},\overline{a}_{1:h-1}\sim \overline{g}_{1:\overline{H}}}\norm{\PP_{h}^{\cD_\cL'}(\cdot,\cdot\given \overline{c}_h)-\PP_h^{\cM,c}(\cdot,\cdot\given \hat{c}_h)}_1\le \epsilon. 
            \end{align*}
            \item {Based on \Cref{lemma: one_step_tract}, $\cM$ satisfies \Cref{assu: one_step_tract},  since the condition \textbf{Turn-based structures} in \S\ref{one-step tractability condition}  is satisfied.}
        \end{itemize}
        Therefore, based on Lemma \ref{lemma: eps_r,eps_z to belief}, such a model is an $(\epsilon_r,\epsilon_z)$-expected approximate common-information model with $\epsilon_r=\epsilon_z=\epsilon$  and satisfies \Cref{assu: one_step_tract}.\\ 

        \noindent\textbf{Type 4:} Baseline sharing of $\cL$ is  {\bf Example 4} in \S\ref{sec: examples of QC}. Then, the common information should be  that for any 
        $h\in[\overline{H}]$,  $\overline{c}_h=\{\overline{o}_{1:h-2d},\{\overline{a}_{2t-1}\}_{t=1}^{\lfloor\frac{h}{2}\rfloor},\overline{o}_M\}$, where $M\subseteq\{(i,t)\given i\in[n],h-2d+1\le t\le h \}$. Then, we denote $f_{\tau,h-2d}=\{\overline{o}_{1:h-2d},\{\overline{a}_{2t-1}\}_{t=1}^{\lfloor\frac{h}{2}\rfloor}\}, f_o=\{\overline{o}_M\}$. We can compute the common-information-based belief as
        \begin{align*}
            \PP_h^{\cD_\cL'}(\overline{s}_h,\overline{p}_h\given\overline{c}_h)=&\sum_{\overline{s}_{h-2d}}\PP_h^{\cD_\cL'}(\overline{s}_h,\overline{p}_h\given\overline{s}_{h-2d},f_o)\PP_h^{\cD_\cL'}(\overline{s}_{h-2d}\given f_{\tau,h-2d}, f_o)\\
            =&\sum_{\overline{s}_{h-2d}}\PP_h^{\cD_\cL'}(\overline{s}_h,\overline{p}_h\given\overline{s}_{h-2d},f_o)\frac{\PP_h^{\cD_\cL'}(\overline{s}_{h-2d},f_o\given f_{\tau,h-2d})}{\sum_{\overline{s}_{h-2d}'}\PP_h^{\cD_\cL'}(\overline{s}_{h-2d}', f_o\given f_{\tau,h-2d})}.
        \end{align*}
        Denote the probability
        $P_h(f_o\given \overline{s}_{h-2d}):=\Pi_{t=1}^{2d}\PP_h^{\cD_\cL'}(\overline{o}_{1,h-2d+t},\overline{o}_{M_{h-2d+t}}\given \overline{s}_{h-2d})$, where  $M_{h-2d+t}=\{(i,h-2d+t)\given (i,h-2d+t)\in M\}$ denotes the set of observations at timestep $h-2d+t$ and shared through additional sharing. Since the actions do not influence the underlying states, here we use the belief notation $\overline{\bb}_k(\overline{o}_{1:k}), \overline{\bb}_k(\overline{o}_{k-L:k})$, $\forall k\in[\overline{H}],L<k$.  With such notation, we have 
        \begin{align*}
            \PP_h^{\cD_\cL'}(\overline{s}_{h-2d}\given f_{\tau,h-2d},f_o)=&\frac{\overline{\bb}_{h-2d}(\overline{o}_{1:h-2d})(\overline{s}_{h-2d})P_h(f_o\given \overline{s}_{h-2d})}{\sum_{\overline{s}_{h-2d}'}\overline{\bb}_{h-2d}(\overline{o}_{1:h-2d})(\overline{s}_{h-2d}')P_h(f_o\given \overline{s}_{h-2d}')} =F^{P_h(\cdot\given \cdot)}(\overline{\bb}_{h-2d}(\overline{o}_{1:h-2d});f_o)(\overline{s}_{h-2d}),
        \end{align*}
        where $F^{P_h(\cdot\given \cdot)}(\cdot; f_o):\Delta(\cS)\rightarrow \Delta(\cS)$ is the posterior belief update function, as discussed in {\bf Type 3}.\\
        Then, we define the approximate common information as $\hat{c}_h:=\{\overline{o}_{h-2d-L+1:h-2d}, \{\overline{a}_{2t-1}\}_{t=\lfloor \frac{h-2d+1}{2}\rfloor}^{\lfloor \frac{h}{2}\rfloor}, \overline{o}_M\}$ and the corresponding approximate common information conditioned belief as 
        \begin{align*}
            \PP_h^{\cM,c}(\overline{s}_h,\overline{p}_h\given \hat{c}_h)=\sum_{\overline{s}_{h-2d}}\PP_h^{\cD_\cL'}(\overline{s}_h,\overline{p}_h\given \overline{s}_{h-2d}, f_o)F^{P_h(\cdot\given \cdot)}(\overline{\bb}_{h-2d}'(\overline{o}_{h-2d-L+1:h-2d});f_o)(\overline{s}_{h-2d}).
        \end{align*}
         Furthermore,  we can construct $\{\PP_h^{\cM,z}\}_{h\in[\overline{H}]}, \{\hat{\cR}_h\}_{h\in[\overline{H}]}$ based on $\cD_\cL'$ and  $\{\PP_h^{\cM,c}\}_{h\in[\overline{H}]}$. Now we verify that Definition \ref{definition: AIS} is satisfied. 
        \begin{itemize}
            \item By definition, the $\{\hat{c}_h\}_{h\in[\overline{H}]}$ satisfies  \Cref{AIS:evolution}. 
            \item For any $\overline{c}_h$ and the corresponding $\hat{c}_h$ constructed above, according to Lemma \ref{lemma: eps_r,eps_z to belief}, we have:
            \begin{align*}
                &\norm{\PP_{h}^{\cD_\cL'}(\cdot,\cdot\given \overline{c}_h)-\PP_h^{\cM,c}(\cdot,\cdot\given \hat{c}_h)}_1\\
                &\quad\le \norm{F^{P(\cdot\given \cdot)}(\overline{\bb}_{h-2d}(\overline{o}_{1:h-2d});f_o)-F^{P(\cdot\given \cdot)}(\overline{\bb}_{h-2d}'(\overline{o}_{h-2d-L+1:h-2d});f_o)}_1.
            \end{align*}
            If we choose $L\ge C\gamma^{-4}\log(\frac{|\overline{\cS}|}{\epsilon})$, then for any strategy $\overline{g}_{1:\overline{H}}$, by taking expectations over $f_{\tau,h-2d}, f_o$, from Lemma \ref{lemma: AIS belief close} and Lemma 12 in \cite{liu2023tractable}, we have that  for any $h\in [\overline{H}]$
            \begin{align*}
                \EE_{\overline{o}_{1:h},\overline{a}_{1:h-1}\sim \overline{g}_{1:\overline{H}}}\norm{\PP_{h}^{\cD_\cL'}(\cdot,\cdot\given \overline{c}_h)-\PP_h^{\cM,c}(\cdot,\cdot\given \hat{c}_h)}_1\le \epsilon. 
            \end{align*}
            \item Based on \Cref{lemma: one_step_tract}, $\cM$ satisfies \Cref{assu: one_step_tract},  since the condition of \textbf{Turn-based structures} in \S\ref{one-step tractability condition} is satisfied.
        \end{itemize}
        Therefore, based on Lemma \ref{lemma: eps_r,eps_z to belief}, such a model is an $(\epsilon_r,\epsilon_z)$-expected approximate common-information model with $\epsilon_r=\epsilon_z=\epsilon$  and satisfies \Cref{assu: one_step_tract}.\\ 
        
        \noindent\textbf{Type 5:} Baseline sharing of $\cL$ is {\bf Example 6}. Note that after reformulation, expansion, and refinement, the common information in \textbf{Example 6} is the same as that in \textbf{Example 1}. The only difference is in private information: for any $t\in[H-1]$, $i\in[n]$, $\overline{a}_{i,2t}\in \overline{p}_{i,2t+1}$ always holds, and $\overline{a}_{i,2t}\in \overline{p}_{i,2t+2}$ may happen. Meanwhile, $\overline{a}_{i,2t}\in \overline{c}_{2t+1}\subseteq \overline{c}_{2t+2}$ always holds. Therefore, we can use the same expected approximate common-information model as \textbf{Type 1}: $\forall t\in[H], \hat{c}_{2t-1}=\{\overline{o}_{2t-L:2t-2},\overline{a}_{2t-1-L:2t-2}\}, \hat{c}_{2t}=\{\overline{o}_{2t-L:2t-2},\overline{a}_{2t-1-L:2t-1},\overline{o}_{N,2t-1}\}$. We construct the approximate common information conditioned belief as $\PP_{2t-1}^{\cM,c}(\overline{s}_{2t-1},\overline{p}_{2t-1}\given \hat{c}_{2t-1})=\overline{\bb}_{2t-1}'(\overline{o}_{2t-L:2t-2},\overline{a}_{2t-1-L:2t-2})(\overline{s}_{2t-1})\overline{\OO}_{2t-1}(\overline{o}_{2t-1}\given \overline{s}_{2t-1})\Psi_{2t-1}^1(\overline{p}_{2t-1},\hat{c}_{2t-1}),\PP_{2t}^{\cM,c}(\overline{s}_{2t},\overline{p}_{2t}\given \hat{c}_{2t})=\overline{\bb}_{2t-1}'(\overline{o}_{2t-L:2t-2},\overline{a}_{2t-1-L:2t-2},\overline{o}_{N,2t-1})(\overline{s}_{2t-1})\PP_{2t-1}(\overline{o}_{-N,2t-1}\given \overline{s}_{2t-1},\overline{o}_{N,2t-1})\Psi_{2t}^1(\overline{p}_{2t},\hat{c}_{2t})$. Note that here we use the belief at timestep $2t-1$ to construct $\PP_{2t}^{\cM,c}$ similarly as \textbf{Type 1}, and the functions $\{\Psi^1_h\}_{h\in[\overline{H}]}$ are defined as:  $\forall t\in[H], \overline{p}_{2t-1},\hat{c}_{2t-1},\overline{p}_{2t}$ and $\hat{c}_{2t}$, 
        \begin{align*}
            &\Psi_{2t-1}^1(\overline{p}_{2t-1},\hat{c}_{2t-1})=\begin{cases}
                1&\text{the value of $\overline{a}_{2t-1}$ in $\overline{p}_{2t-1}$ is the same as the value of that in $\hat{c}_{2t-1}$}\\
                0&\text{o.w.}
            \end{cases}\\
            &\Psi_{2t}^1(\overline{p}_{2t},\hat{c}_{2t})=\begin{cases}
                1&\text{for any $i\in[n]$ such that the random variable $\overline{a}_{i,2t-2}\in \overline{p}_{2t}$, }\\
                &\text{the value of $\overline{a}_{i,2t-2}$ in $\overline{p}_{2t}$ is the same as that in $\hat{c}_{2t}$ }\\
                0&\text{o.w.}
            \end{cases}. 
        \end{align*}
        One can verify that Definition \ref{definition: AIS} is satisfied as in \textbf{Type 1}. This is because, for any $h\in[\overline{H}]$, compared to \textbf{Type 1}, the only difference is that there are only some actions in $\overline{p}_h$, and such actions will also appear in $\hat{c}_{h}$. Therefore, if the values of such actions in $\overline{p}_h$ are consistent with those in $\hat{c}_h$, which is ensured by the functions $\{\Psi^1_h\}_{h\in[\overline{H}]}$, then we can leverage the validation in \textbf{Type 1}. \\

        \noindent\textbf{Type 6:} Baseline sharing of $\cL$ is {\bf Example 8}. Note that after reformulation, expansion, and refinement, the common information in \textbf{Example 8} is the same as that in \textbf{Example 2}. The only difference is in the private information: for any $t\in[H-1], \overline{a}_{1,2t}\in \overline{p}_{1,2t+1}$ always holds, and $\overline{a}_{1,2t}\in \overline{p}_{1,2t+2}$ may happen. Meanwhile, $\overline{a}_{1,2t}\in \overline{c}_{2t+1}\subseteq \overline{c}_{2t+2}$ always holds. Therefore, we can use the same expected approximate common-information model as \textbf{Type 3}:  $\forall h\in[\overline{H}], \hat{c}_h:=\{\overline{o}_{h-2d-L+1:h-2d},\overline{o}_{1,h-2d+1:h},\overline{a}_{1,h-2d-L:h-1},\{\overline{a}_{-1,2t-1}\}_{t=\lfloor \frac{h-2d+1}{2}\rfloor}^{\lfloor\frac{h}{2}\rfloor}, \overline{o}_M\}$, and construct beliefs as:
        \begin{align*}
            \PP_h^{\cM,c}(\overline{s}_h,\overline{p}_h\given\hat{c}_h)=\sum_{\overline{s}_{h-2d}}\PP_h^{\cD_\cL'}(\overline{s}_h,\overline{p}_h\given \overline{s}_{h-2d},f_a,f_o)F^{P_h(\cdot\given \cdot, f_a)}(\overline{\bb}_{h-2d}'(\overline{o}_{h-2d-L:h-2d},\overline{a}_{h-2d-L:h-2d-1},);f_o)(\overline{s}_{h-2d})\Psi_h^2(\overline{p}_h,\hat{c}_h),
        \end{align*}
        where the functions $\{\Psi_h^2\}_{h\in[\overline{H}]}$ are defined as:  $\forall t\in[H], \overline{p}_{2t-1},\hat{c}_{2t-1},\overline{p}_{2t}$ and $\hat{c}_{2t}$ 
        \begin{align*}
            &\Psi_{2t-1}^2(\overline{p}_{2t-1},\hat{c}_{2t-1})=\begin{cases}
                1&\text{the value of $\overline{a}_{1,2t-1}$ in $\overline{p}_{2t-1}$ is the same as the value of that in $\hat{c}_{2t-1}$}\\
                0&\text{o.w.}
            \end{cases}\\
            &\Psi_{2t}^2(\overline{p}_{2t},\hat{c}_{2t})=\begin{cases}
                1&\text{random variable $\overline{a}_{1,2t-2}\notin \overline{p}_{2t}$ or the value of $\overline{a}_{1,2t-2}$ in $\overline{p}_{2t}$ is the same as that in $\hat{c}_{2t}$ }\\
                0&\text{o.w.}
            \end{cases}. 
        \end{align*} 
        One can verify that Definition \ref{definition: AIS} is satisfied as in \textbf{Type 3}. This is because, for any $h\in[\overline{H}]$, compared to \textbf{Type 3}, the only difference is that there are only some actions of agent 1  in $\overline{p}_h$, but such actions will also appear in $\hat{c}_{h}$. Therefore, if the values of such actions in $\overline{p}_h$ are consistent with those in $\hat{c}_h$, which is ensured by the functions $\{\Psi^2_h\}_{h\in[\overline{H}]}$, then we can leverage the validation in \textbf{Type 3}. 
    
        Lastly, we can apply the result of {\bf Part II} and obtain the time complexity as follows,  and complete {\bf Part III}: 
    \begin{itemize}
        \item \textbf{Examples 1, 3, 5, 6:} Since the compressed common information satisfies $\max_{h\in[\overline{H}]}|\hat{\cC}_h|=\max_{h\in[\overline{H}]}(|\overline{\cO}_h||\overline{\cA}_h|)^{C\gamma^{-4}\log(\frac{|\overline{\cS}|}{\epsilon})}$, the complexity is $\texttt{poly}(\max_{h\in[\overline{H}]}(|\overline{\cO}_h||\overline{\cA}_h|)^{C\gamma^{-4}\log(\frac{|\overline{\cS}|}{\epsilon})}, |\overline{\cS}|,
        \overline{H},\frac{1}{\epsilon})$ ;  
        \item \textbf{Examples 2, 4, 7, 8:} Since the compressed common information satisfies $\max_{h\in[\overline{H}]}|\hat{\cC}_h|=\max_{h\in[\overline{H}]}(|\overline{\cO}_h||\overline{\cA}_h|)^{C\gamma^{-4}\log(\frac{|\overline{\cS}|}{\epsilon})+2d}$, the complexity is $\texttt{poly}(\max_{h\in[\overline{H}]}(|\overline{\cO}_h||\overline{\cA}_h|)^{C\gamma^{-4}\log(\frac{|\overline{\cS}|}{\epsilon})+2d}, |\overline{\cS}|,
        \overline{H},\frac{1}{\epsilon})$.
    \end{itemize}
        Combining {\bf Parts I, II, III}, we complete the proof.
\end{proof}

\subsection{Main Results for Learning in QC LTCs}\label{main_results_append_learning}

\begin{theorem}[Full version of \Cref{thm:learning}]

Given any QC LTC problem $\cL$ satisfying Assumptions  \ref{gamma observability}, \ref{limited communication strategy}, \ref{useless action},  and \ref{weak gamma observability}, we can construct a Dec-POMDP problem $\cD_\cL'$ with SI-CIBs. 
Moreover, given any compression functions $\{\text{Compress}_h\}_{h\in[\overline{H}]}$, evolution rules $\{\hat{\phi}_h\}_{h\in[\overline{H}]}$ of the compressed common information $\{\hat{c}_h\in\hat{\cC}_h\}_{h\in[\overline{H}]}$, $\epsilon\in(0,1),\delta\in(0,1)$, and $\hat{L}$ as defined in \Cref{def:L_main},  
we can apply \Cref{main learning algorithm} with a universal constant $C$ as chosen in \cite[Theorem 8]{liu2023tractable}. If the $K=2\overline{H}|\overline{\cS}|$ learned  expected approximate common-information models $\{\hat{\cM}(g^{1:\overline{H},j})\}_{j\in[K]}$ all satisfy Assumption \ref{assu: one_step_tract}, then 
an $\epsilon_0$-team-optimal strategy for $\cL$ can be learned with probability $1-\delta$, 
with time and sample complexities polynomial in the parameters of $\{\hat{\cM}(g^{1:\overline{H},j})\}_{j\in[K]}$, 
where $\epsilon_0$ is defined as
\begin{align*}
\epsilon_0:=\min_{j\in[K]}~~  2\overline{H}\epsilon_r(\hat{\cM}(g^{1:\overline{H},j})) +\overline{H}^2\epsilon_z(\hat{\cM}(g^{1:\overline{H},j}))+\frac{\epsilon}{2}.
\end{align*}

Specifically, if $\cL$ has a baseline sharing protocol as one of the examples in \S \ref{sec: examples of QC}, then given any $\epsilon\in(0,1),\delta\in(0,1)$, we can construct compression functions $\{\text{Compress}_h\}_{h\in[\overline{H}]}$ and  evolution rules $\{\hat{\phi}_h\}_{h\in[\overline{H}]}$, such that \Cref{main learning algorithm} can learn an $\epsilon$-team-optimal strategy of $\cL$ with probability $1-\delta$, with the following time and sample complexities:

    \begin{itemize}
        \item \textbf{Examples 1, 3, 5, 6:} $\texttt{poly}\left(\max_{h\in[\overline{H}]}(|\overline{\cO}_h|| \overline{\cA}_h|)^{C\gamma^{-4}\log(\frac{|\overline{\cS}|}{\epsilon})}, |\overline{\cS}|, \overline{H},\frac{1}{\epsilon},\log(\frac{1}{\delta})\right)$;  
        \item \textbf{Examples 2, 4, 7, 8:} $\texttt{poly}\left(\max_{h\in[\overline{H}]}(|\overline{\cO}_h||\overline{\cA}_h|)^{C\gamma^{-4}\log(\frac{|\cS|}{\epsilon})+2d}, |\overline{\cS}|,\overline{H},\frac{1}{\epsilon},\log(\frac{1}{\delta})\right)$.
    \end{itemize}
    \label{theorem: full learning}
\end{theorem} 

\begin{proof}We divide the proof into the following three {\bf Parts}.\\

   \noindent \textbf{Part I:} Given any LTC problem $\cL$ satisfying the  assumptions in the theorem, we can construct a $\cD_\cL'$ by conducting the reformulation, strict expansion, and refinement for $\cL$. According to Theorem \ref{theorem: refinement}, we know that $\cD_\cL'$ has SI-CIBs with respect to the strategy space  $\overline{\cG}_{1:\overline{H}}$.\\
    
     \noindent\textbf{Part II:}  
     Given any LTC problem $\cL$, we can apply Algorithm $\ref{main learning algorithm}$ for learning in such a problem. For any $j\in[K]$, let $\overline{g}_{1:\overline{H}}^{j,\ast}$ be the output of \Cref{algorithm under AIS} (Line \ref{line: dp in learning} of \Cref{main learning algorithm}), we can guarantee that 
     $J_{\cD_\cL'}(\overline{g}_{1:\overline{H}}^{j,\ast})\ge \max_{\overline{g}_{1:\overline{H}}\in\overline{\cG}_{1:\overline{H}}}J_{\cD_\cL'}(\overline{g}_{1:\overline{H}})-2 \overline{H}\epsilon_r(\hat{\cM}(g^{1:\overline{H},j}))-\overline{H}^2\epsilon_z(\hat{\cM}(g^{1:\overline{H},j}))$ according to the result of \textbf{Part II} of the proof of \Cref{thm: full planning}. Meanwhile, let $\hat{j}$ be the index of the strategy selected by Algorithm Pos-Dec (Line \ref{line: pos}), such that  $\overline{g}_{1:\overline{H}}^{\hat{j},\ast}$ is the output of Pos-Dec.  Then, we can adapt \cite[Lemma 21]{liu2023tractable} to the team setting and guarantee that 
     \begin{align*}
         \max_{\overline{g}_{1:\overline{H}}\in\overline{\cG}_{1:\overline{H}}} J_{\cD_\cL'}(\overline{g}_{1:\overline{H}})-J_{\cD_\cL'}(\overline{g}_{1:\overline{H}}^{\hat{j},\ast})\le \min_{j\in[K]}\left(\max_{\overline{g}_{1:\overline{H}}\in\overline{\cG}_{1:\overline{H}}} J_{\cD_\cL'}(\overline{g}_{1:\overline{H}})-J_{\cD_\cL'}(\overline{g}_{1:\overline{H}}^{j,\ast})\right)+\frac{\epsilon}{2} 
     \end{align*}
     with probability $1-\delta_3\ge 1-\delta$. Combining these two results, we can guarantee that $\max_{\overline{g}_{1:\overline{H}}\in\overline{\cG}_{1:\overline{H}}} J_{\cD_\cL'}(\overline{g}_{1:\overline{H}})-J_{\cD_\cL'}(\hat{g}_{1:\overline{H}}^{\hat{j},\ast})\le \epsilon_0$ with probability $1-\delta$. 
         Also, together with \Cref{prop: equivalence of LTC and Dec-POMDP} and \Cref{theorem: sQC to QC}, with probability $1-\delta$,  we conclude that the output of \Cref{main learning algorithm} is an $\epsilon_0$-team-optimal strategy of $\cL$. 
         
         Meanwhile, we analyze the sample and time complexities of \Cref{main learning algorithm}. For Algorithm BaSeCAMP (Line \ref{line: sample strategy}), from \cite[Lemma 19]{liu2023tractable}, it has time and sample complexities $(\max_{h\in[\overline{H}]}|\overline{\cO}_h||\overline{\cA_h}|)^{\hat{L}}\log(\frac{1}{\delta_2})$; for Algorithm LEE (Line \ref{line: LEE}), it has sample complexity $\overline{H}N_0$ and time complexity $\texttt{poly}(\max_{h\in[\overline{H}]}|\hat{\cC}_h|,\max_{h\in[\overline{H}]}|\overline{\cP}_h|,\max_{h\in[\overline{H}]}|\overline{\cO}_h|,\max_{h\in[\overline{H}]}|\Gamma_h|,\max_{h\in[\overline{H}]}|\overline{\cZ}_h|,\max_{h\in[\overline{H}]}|\overline{\cA}_h|)$; for Algorithm Pos-Dec (Line \ref{line: pos}), it has sample and time complexities $KN_2$; for Algorithm \ref{algorithm under AIS} (Line \ref{line: dp in learning}), as long as the models $\{\hat{\cM}(g^{1:\overline{H},j})\}_{j\in[K]}$ satisfy   \Cref{assu: one_step_tract}, it has time complexity $\max_{h\in[\overline{H}]}|\hat{\cC}_h|\cdot \texttt{poly}(|\overline{\cS}|, \max_{h\in[\overline{H}]}|\overline{\cA}_h|, \max_{h\in[\overline{H}]}||\overline{\cP}_h|, \overline{H})$.  Therefore, \Cref{main learning algorithm} has time and sample complexities polynomial in the parameters of $\{\hat{\cM}(g^{1:\overline{H},j})\}_{j\in[K]}$.  \\

     \noindent\textbf{Part III:} If the baseline sharing of $\cL$ is one of the examples in \S\ref{sec: examples of QC}, we can construct the compressed common information with length $L\ge 2C\frac{\log(\overline{H}|\overline{\cS}|\max_{h\in[\overline{H}]}|\overline{\cO}_h|/(\epsilon_1\gamma))}{\gamma^4}$ as follows, where $\gamma$ is the constant in \Cref{gamma observability} and $\epsilon_1$ is defined in \Cref{main learning algorithm}. 
     \begin{itemize}
         \item \textbf{Examples 1, 5, 6:} For any $t\in[H]$, $\hat{c}_{2t-1}=\{\overline{o}_{2t-L:2t-2},\overline{a}_{2t-1-L:2t-2}\}, \hat{c}_{2t}=\{\overline{o}_{2t-L:2t-2},\overline{a}_{2t-L-1:2t-1},\overline{o}_{N,2t-1}\}$ with $N\subseteq[n]$. Here $\hat{L}=L$.
         \item \textbf{Example 3:} For any $t\in[H], \hat{c}_{2t-1}=\{\overline{o}_{1,2t-L:2t-2},\overline{a}_{2t-L-1:2t-2},\overline{o}_{1,2t-1}\}, \hat{c}_{2t}=\{\overline{o}_{2t-L:2t-2},\overline{a}_{2t-L-1:2t-1},\overline{o}_{N,2t-1}\}$ with $N\subseteq[n]$ and $1\in N$. Here $\hat{L}=L$.
         \item \textbf{Examples 2, 7, 8:} For any $h\in[\overline{H}], \hat{c}_h:=\{\overline{o}_{h-2d-L+1:h-2d},\overline{o}_{1,h-2d+1:h},\overline{a}_{1,h-2d-L:h-1},\{\overline{a}_{-1,2t-1}\}_{t=\lfloor \frac{h-2d+1}{2}\rfloor}^{\lfloor\frac{h}{2}\rfloor}, \overline{o}_M\}$ with $M\subseteq\{(i,t)\given 1<i\le n, h-2d+1\le t\le h\}$, and recall that $\overline{o}_M=\{o_{i,t}\given (i,t)\in M\}$. Here $\hat{L}=L+2d$.
         \item \textbf{Example 4:} For any $h\in[\overline{H}], \hat{c}_h:=\{\overline{o}_{h-L-2d+1:h-2d},\{\overline{a}_{2t-1}\}_{t=\lfloor \frac{h-2d+1}{2}\rfloor}^{\lfloor\frac{h}{2}\rfloor},\overline{o}_M\}$ with $M\subseteq\{(i,t)\given 1\le i\le n, h-2d+1\le t\le h\}$, and recall that $\overline{o}_M=\{o_{i,t}\given (i,t)\in M\}$. Here $\hat{L}=L+2d$.
     \end{itemize}
     The compression functions $\{\text{Compress}_h\}_{h\in[\overline{H}]}$ and evolution rules $\{
     \hat{\phi}_h\}_{h\in[\overline{H}]}$ of the compressed common information can be constructed correspondingly.
     
     Meanwhile, we can slightly change  Theorem 8 in \cite{liu2023tractable} to obtain that, for any $j\in[K]$, Algorithm LEE (Line \ref{line: LEE} of \Cref{main learning algorithm}) can guarantee that the output $\hat{\cM}(g^{1:\overline{H},j})$ has an approximation error as 
    \begin{align}
        \forall \overline{g}_{1:\overline{H}}\in \overline{\cG}_{1:\overline{H}}, |V_0^{\overline{g}_{1:\overline{H}},\cD_\cL'}(\emptyset)-V_0^{\overline{g}_{1:\overline{H}},\hat{\cM}(g^{1:\overline{H},j})}(\emptyset)|\le \overline{H}\epsilon_r(\tilde{\cM}(g^{1:\overline{H},j})) +\frac{1}{2}\overline{H}^2\epsilon_z(\tilde{\cM}(g^{1:\overline{H},j}))+\epsilon_{apx}^j
        \label{eq: difference}
    \end{align}
    with probability $1-\delta_1$, with %
    \begin{align*}
        \epsilon_{apx}^j&:=\theta_1+2\max_{h\in[\overline{H}]}|\overline{\cA}_h|\max_{h\in[\overline{H}]}|\overline{\cP}_h|\frac{\zeta_1}{\zeta_2}+\max_{h\in[\overline{H}]}|\overline{\cA}_h|\max_{h\in[\overline{H}]}|\overline{\cP}_h|\theta_2+\frac{\max_{h\in[\overline{H}]}|\overline{\cA}_h|^{2\hat{L}}\max_{h\in[\overline{H}]}|\overline{\cO}_h|^{\hat{L}}}{\phi}\\
        &\qquad +\max_{h\in[\overline{H}]}\max_{\overline{g}\in\overline{\cG}_{1:\overline{H}}}\mathds{1}[h>\hat{L}]\cdot2\cdot d_{\cS,h-\hat{L}}^{\overline{g},\cD_\cL'}(\cU_{\phi,h-\hat{L}}^{\cD_\cL'}(g^{h,j})),
    \end{align*}
    where for any $g\in\cG_{1:\overline{H}},h\in[\overline{H}]$, we define $d_{\cS,h}^{g,\cD_\cL'}(\overline{s}):=\PP_h^{g,\cD_\cL'}(\overline{s}_h=\overline{s})$, and for any set $X\subseteq \overline{\cS}, d_{\cS,h}^{g,\cD_\cL'}(X)=\sum_{\overline{s}\in X}d_{\cS,h}^{g,\cD_\cL'}(\overline{s}), \cU_{\phi,h}^{\cD_\cL'}(g):=\{\overline{s}\in\overline{\cS}\given d_{\cS,h}^{g,\cD_\cL'}(\overline{s})<\phi\}$. 
    Note that $\tilde{\cM}$  here is the expected approximate common-information model with components $\{\hat{c}_h\}_{h\in[\overline{H}]},\{\hat{\phi}_h\}_{h\in[\overline{H}]}, \Gamma$, and $\forall j\in[K], \tilde{\cM}(g^{1:\overline{H},j})$ are  as defined in \Cref{def:simulation_main}. 
    Under the parameters specified in Line \ref{line: parameters} of \Cref{main learning algorithm}, we have
    \begin{align*}
        \epsilon_{apx}^j\le 4\epsilon_1+\max_{h\in[\overline{H}]}\max_{\overline{g}\in\overline{\cG}_{1:\overline{H}}}\mathds{1}[h>\hat{L}]\cdot2\cdot d_{\cS,h-\hat{L}}^{\overline{g},\cD_\cL'}(\cU_{\phi,h-\hat{L}}^{\cD_\cL'}(g^{h,j})).
    \end{align*}
    
    Also, we can prove the following lemma.
    \begin{lemma}\label{lemma:belief_contract_learning}
        Given any $\hat{L}>0$  and parameters $K,\alpha,\beta, \epsilon_1, N_0,N_1, O, S$ specified in \Cref{main learning algorithm}, and let $\{g^{1:\overline{H},j}\}_{j\in[K]}$ be the output of the algorithm BaSeCAMP($\hat{L},N_0,N_1,\alpha,\beta, K$) (Line \ref{line: sample strategy} in \Cref{main learning algorithm}). As long as $\hat{L}\ge C\frac{\log(\overline{H}SO/(\epsilon_1\gamma))}{\gamma^4}$,  where $C$ is a large enough constant as chosen in \cite[Theorem 8]{liu2023tractable}, then with probability at least $1-\delta_2$, there exists at least one $j^\ast\in[K]$ such that for any strategy $\overline{g}\in\overline{\cG}_{1:\overline{H}}, g\in \cG_{1:\overline{H}}, h\in[\overline{H}], N\subseteq [n]$
        \begin{align*}
            &\EE_{\overline{g}}\norm{\overline{\bb}_h(\overline{o}_{1:h},\overline{a}_{1:h-1})-\overline{\bb}_h^{g}(\overline{o}_{h-\hat{L}+1:h},\overline{a}_{h-\hat{L}:h-1})}_1\le \epsilon_1+\mathds{1}[h>\hat{L}]\cdot6\cdot d_{\cS,h-\hat{L}}^{\overline{g},\cD_\cL'}(\cU_{\phi,h-\hat{L}}^{\cD_\cL'}(g))\\
            &\EE_{\overline{g}}\norm{\overline{\bb}_h(\overline{o}_{1:h-1},\overline{a}_{1:h-1})-\overline{\bb}_h^{g}(\overline{o}_{h-\hat{L}+1:h-1},\overline{a}_{h-\hat{L}:h-1})}_1\le \epsilon_1+\mathds{1}[h>\hat{L}]\cdot6\cdot d_{\cS,h-\hat{L}}^{\overline{g},\cD_\cL'}(\cU_{\phi,h-\hat{L}}^{\cD_\cL'}(g))\\
            &\EE_{\overline{g}}\norm{\overline{\bb}_h(\overline{o}_{1:h-1},\overline{a}_{1:h-1},\overline{o}_{N,h})-\overline{\bb}_h^{g}(\overline{o}_{h-\hat{L}+1:h-1},\overline{a}_{h-\hat{L}:h-1},\overline{o}_{N,h})}_1\le \epsilon_1+\mathds{1}[h>\hat{L}]\cdot6\cdot d_{\cS,h-\hat{L}}^{\overline{g},\cD_\cL'}(\cU_{\phi,h-\hat{L}}^{\cD_\cL'}(g)),\\
            &d_{\cS,h-\hat{L}}^{\overline{g},\cD_\cL'}(\cU_{\phi,h-\hat{L}}^{\cD_\cL'}(g^{h,j^\ast}))<\epsilon_1,
        \end{align*}where $\overline{\bb}_h^{g}(\cdot):=\overline{\bb}_h^{apx,\cD_\cL'}(\cdot,d_{\cS,h-\hat{L}}^{g,\cD_\cL'})$, and $\overline{\bb}_h^{apx,\cD_\cL'}$ is defined as follows:  $\forall D\in \Delta(\overline{\cS}), N\subseteq[n]$
           \label{lemma: learning belief contract}
    \begin{align*}
         &\overline{\bb}_h^{apx,\cD_\cL'}(\overline{o}_{h-\hat{L}+1:h-1}, \overline{a}_{h-\hat{L}:h-1},D)=\PP(\overline{s}_{h}=\cdot\given \overline{s}_{h-\hat{L}}\sim D,\overline{o}_{h-\hat{L}+1:h-1},\overline{a}_{h-\hat{L}:h-1})\\
           &\overline{\bb}_h^{apx,\cD_\cL'}(\overline{o}_{h-\hat{L}+1:h},\overline{a}_{h-\hat{L}:h-1},D)=\PP(\overline{s}_h=\cdot\given \overline{s}_{h-\hat{L}}\sim D,\overline{o}_{h-\hat{L}+1:h},\overline{a}_{h-\hat{L}:h-1})\\
           &\overline{\bb}_h^{apx,\cD_\cL'}(\overline{o}_{h-\hat{L}+1:h-1},\overline{a}_{h-\hat{L}:h-1},\overline{o}_{N,h}, D)=\PP(\overline{s}_h=\cdot\given \overline{s}_{h-\hat{L}}\sim D,\overline{o}_{h-\hat{L}+1:h-1},\overline{a}_{h-\hat{L}:h-1},\overline{o}_{N,h}).
    \end{align*}
    \end{lemma} 
    \begin{proof}
        We can adapt Lemma 18 and Lemma 19 in \cite{liu2023tractable} to our setting. Note that $\cD_\cL'$ only has $\gamma$-observability at odd timesteps $h=2t-1, t\in[H]$. 
        Therefore, it holds that $\overline{\bb}_{h+1}(\overline{o}_{1:h+1},\overline{a}_{1:h})=\overline{\bb}_h(\overline{o}_{1:h},\overline{a}_{1:h-1})$  and $\overline{\bb}_{h+1}^{g}(\overline{o}_{h-\hat{L}+1:h+1},\overline{a}_{h-\hat{L}:h})=\overline{\bb}_h^{g}(\overline{o}_{h-\hat{L}+1:h},\overline{a}_{h-\hat{L}:h-1})$,  since the communication action at timestep $h=2t-1$ (i.e., $\overline{a}_h=m_t$) does not affect the state, and $\overline{s}_{h+1}=\overline{s}_h,\overline{o}_{h+1}=\emptyset$. A similar relationship of the beliefs also holds for the other two types of beliefs.   We can thus apply the corollary.
    \end{proof}
    
    Therefore, based on Lemma \ref{lemma: learning belief contract}, we have $\max_{\overline{g}_{1:\overline{H}}\in \overline{\cG}_{1:\overline{H}}}J_{\cD_\cL'}(\overline{g}_{1:\overline{H}})-J_{\cD_\cL'}(\overline{g}_{1:\overline{H}}^{j^\ast,\ast})\le 2\big(\overline{H}\epsilon_r(\tilde{\cM}(g^{1:\overline{H},j^\ast})) +\frac{1}{2}\overline{H}^2\epsilon_z(\tilde{\cM}(g^{1:\overline{H},j^\ast}))+\frac{6\epsilon}{200(\overline{H}+1)^2}\big)$, where $j^\ast$ is the $j$ selected in Lemma \ref{lemma: learning belief contract}, and $\overline{g}_{1:\overline{H}}^{j^\ast,\ast}$ is the strategy computed in the $j^\ast$-th iteration of Line \ref{line: dp in learning} of \Cref{main learning algorithm}.
    
    Now, we validate for each example in \S\ref{sec: examples of QC} that, under the compressed common information, $\epsilon_r(\tilde{\cM}(g^{1:\overline{H},j^\ast}))$ and $\epsilon_z(\tilde{\cM}(g^{1:\overline{H},j^\ast}))$ in \Cref{eq: difference} 
    can be made small. Note that after reformulation, expansion, and refinement, \textbf{Examples 1} and \textbf{5} are the same, and \textbf{Examples 2} and \textbf{7} are the same. 
    Therefore, all the examples in \S\ref{sec: examples of QC} can be classified into the following 6 {\bf Types}. 
    Finally, from  Lemma \ref{lemma: eps_r,eps_z to belief}, for any $h\in[\overline{H}]$, we have 
    \begin{align}
            &|\EE^{\cD_\cL'}[\cR_h(\overline{s}_h,\overline{a}_h,\overline{p}_h)\given \overline{c}_h,\gamma_h]-\hat{\cR}_h^{\tilde{\cM}(g^{1:\overline{H},j^\ast})}(\hat{c}_h,\gamma_h)|\le \norm{\PP_h^{\cD_\cL'}(\cdot,\cdot\given \overline{c}_h)-\PP_h^{\tilde{\cM}(g^{1:\overline{H},j^\ast}),c}(\cdot,\cdot\given \hat{c}_h)}_1\label{equ:eps_to_z_1}\\
            &\norm{\PP_h^{\cD_\cL'}(\cdot\given \overline{c}_h,\gamma_h)-\PP_h^{\tilde{\cM}(g^{1:\overline{H},j^\ast}),z}(\cdot\given \hat{c}_h,\gamma_h)}_1\le \norm{\PP_h^{\cD_\cL'}(\cdot,\cdot\given \overline{c}_h)-\PP_h^{\tilde{\cM}(g^{1:\overline{H},j^\ast}),c}(\cdot,\cdot\given \hat{c}_h)}_1.\label{equ:eps_to_z_2}
        \end{align} 
        
\noindent\textbf{Type 1:} The baseline sharing is either   \textbf{Example 1} or \textbf{Example 5}. 
        Consider any $h\in[\overline{H}]$. If $h=2t-1,t\in[H]$, it holds that 
        \begin{align*}
            &\PP_h^{\tilde{\cM}(g^{1:\overline{H},j^\ast}),c}(\overline{s}_h,\overline{p}_h\given \hat{c}_h)=\overline{\bb}_h^{g^{h,j^\ast}}(\overline{o}_{h-\hat{L}+1:h-1},\overline{a}_{h-\hat{L}:h-1})(\overline{s}_h)\overline{\OO}_h(\overline{o}_h\given \overline{s}_h),\\
            &\PP_h^{\cD_\cL'}(\overline{s}_h,\overline{p}_h\given \overline{c}_h)=\overline{\bb}_h(\overline{o}_{1:h-1},\overline{a}_{1:h-1})(\overline{s}_h)\overline{\OO}_h(\overline{o}_h\given \overline{s}_h). 
        \end{align*} 
        Then, from \Cref{lemma:belief_contract_learning} and \Cref{equ:eps_to_z_1}, we have
        \begin{align*}
            &\max_{\overline{g}\in\overline{\cG}_{1:\overline{H}}}\EE^{\cD_\cL'}_{\overline{o}_{1:h},\overline{a}_{1:h-1}\sim \overline{g}}|\EE^{\cD_\cL'}[\cR_h(\overline{s}_h,\overline{a}_h,\overline{p}_h)\given \overline{c}_h,\gamma_h]-\hat{\cR}_h^{\tilde{\cM}(g^{1:\overline{H},j^\ast})}(\hat{c}_h,\gamma_h)|\\
            &\quad\le \max_{\overline{g}\in\overline{\cG}_{1:\overline{H}}}\EE^{\cD_\cL'}_{\overline{o}_{1:h-1},\overline{a}_{1:h-1}\sim \overline{g}}\norm{\overline{\bb}_h(\overline{o}_{1:h-1},\overline{a}_{1:h-1})-\overline{\bb}_h^{g^{h,j^\ast}}(\overline{o}_{h-\hat{L}+1:h-1},\overline{a}_{h-\hat{L}:h-1})}_1\le 7\epsilon_1.
        \end{align*}
        Similarly,  $\max_{\overline{g}\in\overline{\cG}_{1:\overline{H}}}\EE^{\cD_\cL'}_{\overline{o}_{1:h},\overline{a}_{1:h-1}\sim \overline{g}}\norm{\PP_h^{\cD_\cL'}(\cdot\given \overline{c}_h,\gamma_h)-\PP_h^{\tilde{\cM}(g^{1:\overline{H},j^\ast}),z}(\cdot\given \hat{c}_h,\gamma_h)}_1\le 7\epsilon_1.$

        If $h=2t,t\in[H]$, then it holds that
              \begin{align*}
            &\PP_h^{\tilde{\cM}(g^{1:\overline{H},j^\ast}),c}(\overline{s}_h,\overline{p}_h\given \hat{c}_h)=
        \overline{\bb}_{h-1}^{g^{h,j^\ast}}(\overline{o}_{h-\hat{L}+1:h-2},\overline{a}_{h-\hat{L}:h-2},\overline{o}_{N,h-1})(\overline{s}_h)\PP_{h-1}(\overline{o}_{-N,h-1}\given \overline{s}_{h-1},\overline{o}_{N,h-1})\\
            &\PP_h^{\cD_\cL'}(\overline{s}_h,\overline{p}_h\given \overline{c}_h)=\overline{\bb}_{h-1}(\overline{o}_{1:h-2},\overline{a}_{1:h-2},\overline{o}_{N,h-1})(\overline{s}_h)\PP_{h-1}(\overline{o}_{-N,h-1}\given \overline{s}_{h-1},\overline{o}_{N,h-1}).
        \end{align*} 
        Note that we use the belief at timestep $h-1$ since $\overline{s}_h=\overline{s}_h$, and it holds that $\overline{p}_{h}=\overline{o}_{N,h}$, where $N\subseteq[n]$ can be inferred by $\overline{a}_{h-1}$, which lies in $\hat{c}_h$ and $\overline{c}_h$. 
        Then, from \Cref{lemma:belief_contract_learning} and \Cref{equ:eps_to_z_1}, we have 
        \begin{align*}
            &\max_{\overline{g}\in\overline{\cG}_{1:\overline{H}}}\EE^{\cD_\cL'}_{\overline{o}_{1:h},\overline{a}_{1:h-1}\sim \overline{g}}|\EE^{\cD_\cL'}[\cR_h(\overline{s}_h,\overline{a}_h,\overline{p}_h)\given \overline{c}_h,\gamma_h]-\hat{\cR}_h^{\tilde{\cM}(g^{1:\overline{H},j^\ast})}(\hat{c}_h,\gamma_h)|\\
            &\quad\le \max_{\overline{g}\in\overline{\cG}_{1:\overline{H}}}\EE^{\cD_\cL'}_{\overline{o}_{1:h-1},\overline{a}_{1:h-2}\sim \overline{g}}\norm{\overline{\bb}_{h-1}(\overline{o}_{1:h-2},\overline{a}_{1:h-2},\overline{o}_{N,h-1})-\overline{\bb}_{h-1}^{g^{h,j^\ast}}(\overline{o}_{h-\hat{L}+1:h-2},\overline{a}_{h-\hat{L}:h-2},\overline{o}_{N,h-1})}_1\le 7\epsilon_1. 
        \end{align*}   
        Similarly, $\max_{\overline{g}\in\overline{\cG}_{1:\overline{H}}}\EE^{\cD_\cL'}_{\overline{o}_{1:h},\overline{a}_{1:h-1}\sim \overline{g}}\norm{\PP_h^{\cD_\cL'}(\cdot\given \overline{c}_h,\gamma_h)-\PP_h^{\tilde{\cM}(g^{1:\overline{H},j^\ast}),z}(\cdot\given \hat{c}_h,\gamma_h)}_1\le 7\epsilon_1.$\\
        
        \noindent\textbf{Type 2:} The baseline sharing is \textbf{Example 3}. 
        For any $h\in[\overline{H}]$, if $h=2t-1$ with $ t\in[H]$,
              \begin{align*}
            &\PP_h^{\tilde{\cM}(g^{1:\overline{H},j^\ast}),c}(\overline{s}_h,\overline{p}_h\given \hat{c}_h)=\overline{\bb}_h^{g^{h,j^\ast}}(\overline{o}_{h-\hat{L}+1:h-1},\overline{a}_{h-\hat{L}:h-1},\overline{o}_{1,h})(\overline{s}_h)\PP_{2t-1}(\overline{o}_{-1,2t}\given \overline{s}_{2t-1},\overline{o}_{1,2t-1})\\
            &\PP_h^{\cD_\cL'}(\overline{s}_h,\overline{p}_h\given \overline{c}_h)=\overline{\bb}_h(\overline{o}_{1:h-1},\overline{a}_{1:h-1},\overline{o}_{1,h})(\overline{s}_h)\PP_{2t-1}(\overline{o}_{-1,2t}\given \overline{s}_{2t-1},\overline{o}_{1,2t-1}).
        \end{align*}  
        Then, from \Cref{lemma:belief_contract_learning} and \Cref{equ:eps_to_z_1}, we have
         \begin{align*}
            &\max_{\overline{g}\in\overline{\cG}_{1:\overline{H}}}\EE^{\cD_\cL'}_{\overline{o}_{1:h},\overline{a}_{1:h-1}\sim \overline{g}}|\EE^{\cD_\cL'}[\cR_h(\overline{s}_h,\overline{a}_h,\overline{p}_h)\given \overline{c}_h,\gamma_h]-\hat{\cR}_h^{\tilde{\cM}(g^{1:\overline{H},j^\ast})}(\hat{c}_h,\gamma_h)|\\
            &\quad \le \max_{\overline{g}\in\overline{\cG}_{1:\overline{H}}}\EE^{\cD_\cL'}_{\overline{o}_{1:h},\overline{a}_{1:h-1}\sim \overline{g}}\norm{\overline{\bb}_h(\overline{o}_{1:h-1},\overline{a}_{1:h-1},\overline{o}_{1,h})-\overline{\bb}_h^{g^{h,j^\ast}}(\overline{o}_{h-\hat{L}+1:h-1},\overline{a}_{h-\hat{L}:h-1},,\overline{o}_{1,h})}_1\le 7\epsilon_1.
        \end{align*}
        Similarly, $\max_{\overline{g}\in\overline{\cG}_{1:\overline{H}}}\EE^{\cD_\cL'}_{\overline{o}_{1:h},\overline{a}_{1:h-1}\sim \overline{g}}\norm{\PP_h^{\cD_\cL'}(\cdot\given \overline{c}_h,\gamma_h)-\PP_h^{\tilde{\cM}(g^{1:\overline{H},j^\ast}),z}(\cdot\given \hat{c}_h,\gamma_h)}_1\le 7\epsilon_1. $

        If $h=2t$ with $ t\in[H]$, then it holds that
              \begin{align*}
            &\PP_h^{\tilde{\cM}(g^{1:\overline{H},j^\ast}),c}(\overline{s}_h,\overline{p}_h\given \hat{c}_h)=\overline{\bb}_{h-1}^{g^{h,j^\ast}}(\overline{o}_{h-\hat{L}+1:h-2},\overline{a}_{h-\hat{L}:h-2},\overline{o}_{N,h-1})(\overline{s}_h)\PP_{h-1}(\overline{o}_{-N,h-1}\given \overline{s}_{h-1},\overline{o}_{N,h-1})\\
            &\PP_h^{\cD_\cL'}(\overline{s}_h,\overline{p}_h\given \overline{c}_h)=\overline{\bb}_{h-1}(\overline{o}_{1:h-2},\overline{a}_{1:h-2},\overline{o}_{N,h-1})(\overline{s}_h)\PP_{h-1}(\overline{o}_{-N,h-1}\given \overline{s}_{h-1},\overline{o}_{N,h-1}).
        \end{align*} 
        Note that  we use the belief at timestep $h-1$ since $\overline{s}_h=\overline{s}_h$, and it holds that $1\in N$, and $N$ can be inferred by $\overline{a}_{h-1}$, which lies in $\hat{c}_h$ and $\overline{c}_h$.
        Then, from \Cref{lemma:belief_contract_learning} and \Cref{equ:eps_to_z_1}, we have 
        \begin{align*}
            &\max_{\overline{g}\in\overline{\cG}_{1:\overline{H}}}\EE^{\cD_\cL'}_{\overline{o}_{1:h},\overline{a}_{1:h-1}\sim \overline{g}}|\EE^{\cD_\cL'}[\cR_h(\overline{s}_h,\overline{a}_h,\overline{p}_h)\given \overline{c}_h,\gamma_h]-\hat{\cR}_h^{\tilde{\cM}(g^{1:\overline{H},j^\ast})}(\hat{c}_h,\gamma_h)|\\
            &\quad\le \max_{\overline{g}\in\overline{\cG}_{1:\overline{H}}}\EE^{\cD_\cL'}_{\overline{o}_{1:h-1},\overline{a}_{1:h-2}\sim \overline{g}}\norm{\overline{\bb}_{h-1}(\overline{o}_{1:h-2},\overline{a}_{1:h-2},\overline{o}_{N,h-1})-\overline{\bb}_{h-1}^{g^{h,j^\ast}}(\overline{o}_{h-\hat{L}+1:h-2},\overline{a}_{h-\hat{L}:h-2},\overline{o}_{N,h-1})}_1\le 7\epsilon_1. 
        \end{align*}   
        Similarly, $\max_{\overline{g}\in\overline{\cG}_{1:\overline{H}}}\EE^{\cD_\cL'}_{\overline{o}_{1:h},\overline{a}_{1:h-1}\sim \overline{g}}\norm{\PP_h^{\cD_\cL'}(\cdot\given \overline{c}_h,\gamma_h)-\PP_h^{\tilde{\cM}(g^{1:\overline{H},j^\ast}),z}(\cdot\given \hat{c}_h,\gamma_h)}_1\le 7\epsilon_1.$\\
        
        \noindent\textbf{Type 3:} The baseline sharing is either   \textbf{Example 2} or \textbf{Example 7}. 
        For any $h\in[\overline{H}]$, 
        \begin{align*}
            &\PP_h^{\cD_\cL'}(\overline{s}_{h},\overline{p}_h\given \overline{c}_h)=\sum_{\overline{s}_{h-2d}}\PP_h^{\cD_\cL'}(\overline{s}_h,\overline{p}_h\given \overline{s}_{h-2d},f_a,f_o)F^{P_h(\cdot\given \cdot, f_a)}(\overline{\bb}_{h-2d}(\overline{o}_{1:h-2d},\overline{a}_{1:h-2d-1});f_o)(\overline{s}_{h-2d}),\\
            &\PP_h^{\tilde{\cM}(g^{1:\overline{H},j^\ast}),c}(\overline{s}_{h},\overline{p}_h\given \hat{c}_h)=\sum_{\overline{s}_{h-2d}}\PP_h^{\cD_\cL'}(\overline{s}_h,\overline{p}_h\given \overline{s}_{h-2d},f_a,f_o)F^{P_h(\cdot\given \cdot, f_a)}(\overline{\bb}_{h-2d}^{g^{h,j^\ast}}(\overline{o}_{h-\hat{L}+1:h-2d},\overline{a}_{h-\hat{L}:h-2d-1});f_o)(\overline{s}_{h-2d}),
        \end{align*} 
        where $f_{\tau, h-2d}=\{\overline{o}_{1:h-2d},\overline{a}_{1,1:h-2d-1},\{\overline{a}_{-1,2t-1}\}_{t=1}^{\lfloor\frac{h-2d}{2}\rfloor}\}, f_a=\{\overline{a}_{1,h-2d: h-1}, \{\overline{a}_{-1,2t-1}\}_{t=1\lfloor \frac{h-2d+1}{2}\rfloor}^{\frac{h}{2}}\}, f_o=\{\overline{o}_{1,h-2d+1:h},\overline{o}_M\}, F^{P_h(\cdot\given \cdot, f_a)}(\cdot; f_o):\Delta(\cS)\rightarrow \Delta(\cS)$ is the posterior belief update function (as introduced in \textbf{Type 3} in the proof of  \Cref{thm: full planning}).
        Then, from \Cref{lemma:belief_contract_learning} and \Cref{equ:eps_to_z_1}, we have 
        \begin{align*}
            &\max_{\overline{g}\in\overline{\cG}_{1:\overline{H}}}\EE^{\cD_\cL'}_{\overline{o}_{1:h},\overline{a}_{1:h-1}\sim \overline{g}}|\EE^{\cD_\cL'}[\cR_h(\overline{s}_h,\overline{a}_h,\overline{p}_h)\given \overline{c}_h,\gamma_h]-\hat{\cR}_h^{\tilde{\cM}(g^{1:\overline{H},j^\ast})}(\hat{c}_h,\gamma_h)|\\
            &\qquad\le \max_{\overline{g}\in\overline{\cG}_{1:\overline{H}}}\EE^{\cD_\cL'}_{\overline{o}_{1:h},\overline{a}_{1:h-1}\sim \overline{g}}\norm{\sum_{\overline{s}_{h-2d}}\PP_h^{\cD_\cL'}(\overline{s}_h,\overline{p}_h\given \overline{s}_{h-2d},f_a,f_o)F^{P_h(\cdot\given \cdot, f_a)}(\overline{\bb}_{h-2d}(\overline{o}_{1:h-2d},\overline{a}_{1:h-2d-1});f_o)\\
            &\qquad\qquad-\sum_{\overline{s}_{h-2d}}\PP_h^{\cD_\cL'}(\overline{s}_h,\overline{p}_h\given \overline{s}_{h-2d},f_a,f_o)F^{P_h(\cdot\given \cdot, f_a)}(\overline{\bb}_{h-2d}^{g^{h,j^\ast}}(\overline{o}_{h-\hat{L}:h-2d},\overline{a}_{h-\hat{L}:h-2d-1});f_o)(\overline{s}_{h-2d})}_1\\
            &\qquad\le \max_{\overline{g}\in\overline{\cG}_{1:\overline{H}}}\EE^{\cD_\cL'}_{\overline{o}_{1:h},\overline{a}_{1:h-1}\sim \overline{g}}\norm{\overline{\bb}_{h-2d}(\overline{o}_{1:h-2d},\overline{a}_{1:h-1-2d})-\overline{\bb}_{h-2d}^{g^{h,j^\ast}}(\overline{o}_{h-\hat{L}+1:h-2d},\overline{a}_{h-\hat{L}:h-2d-1})}_1\le 7\epsilon_1. 
        \end{align*}        
        Similarly, $\max_{\overline{g}\in\overline{\cG}_{1:\overline{H}}}\EE^{\cD_\cL'}_{\overline{o}_{1:h},\overline{a}_{1:h-1}\sim \overline{g}}\norm{\PP_h^{\cD_\cL'}(\cdot\given \overline{c}_h,\gamma_h)-\PP_h^{\tilde{\cM}(g^{1:\overline{H},j^\ast}),z}(\cdot\given \hat{c}_h,\gamma_h)}_1\le 7\epsilon_1.$\\
        
        \noindent\textbf{Type 4:} The baseline sharing is \textbf{Example 4}.
        For any $h\in[\overline{H}]$,
        \begin{align*}
            &\PP_h^{\cD_\cL'}(\overline{s}_h,\overline{p}_h\given \overline{c}_h)=\sum_{\overline{s}_{h-2d}}\PP_h^{\cD_\cL'}(\overline{s}_h,\overline{p}_h\given \overline{s}_{h-2d},f_o)F^{P_h(\cdot\given \cdot)}(\overline{\bb}_{h-2d}(\overline{o}_{1:h-2d},\overline{a}_{1:h-2d-1});f_o)(\overline{s}_{h-2d}),\\
            &\PP_h^{\tilde{\cM}(g^{1:\overline{H},j^\ast}),c}(\overline{s}_h,\overline{p}_h\given \hat{c}_{h})=\sum_{\overline{s}_{h-2d}}\PP_h^{\cD_\cL'}(\overline{s}_h,\overline{p}_h\given \overline{s}_{h-2d},f_o)F^{P_h(\cdot\given \cdot)}(\overline{\bb}_{h-2d}^{g^{h,j^\ast}}(\overline{o}_{h-\hat{L}+1:h-2d},\overline{a}_{h-\hat{L}:h-2d-1});f_o)(\overline{s}_{h-2d}),
        \end{align*} 
        where $f_{\tau,h-2d}=\{\overline{o}_{1:h-2d},\{\overline{a}_{2t-1}\}_{t=1}^{\lfloor\frac{h}{2}\rfloor}\}, f_o=\{\overline{o}_M\}, F^{P_h(\cdot\given \cdot)}(\cdot; f_o):\Delta(\cS)\rightarrow \Delta(\cS)$ is the posterior belief update function (as introduced in \textbf{Type 4} in the proof  of  \Cref{thm: full planning})). Recall that $M\subseteq \{(i,t)\given i\in[n], h-2d+1\le t\le h\}$. 
        Then, from \Cref{lemma:belief_contract_learning} and \Cref{equ:eps_to_z_1}, we have 
        \begin{align*}
            &\max_{\overline{g}\in\overline{\cG}_{1:\overline{H}}}\EE^{\cD_\cL'}_{\overline{o}_{1:h},\overline{a}_{1:h-1}\sim \overline{g}}|\EE^{\cD_\cL'}[\cR_h(\overline{s}_h,\overline{a}_h,\overline{p}_h)\given \overline{c}_h,\gamma_h]-\hat{\cR}_h^{\tilde{\cM}(g^{1:\overline{H},j^\ast})}(\hat{c}_h,\gamma_h)|\\
            &\qquad\le \max_{\overline{g}\in\overline{\cG}_{1:\overline{H}}}\EE^{\cD_\cL'}_{\overline{o}_{1:h},\overline{a}_{1:h-1}\sim \overline{g}}\norm{\sum_{\overline{s}_{h-2d}}\PP_h^{\cD_\cL'}(\overline{s}_h,\overline{p}_h\given \overline{s}_{h-2d},f_a,f_o)F^{P_h(\cdot\given \cdot)}(\overline{\bb}_{h-2d}(\overline{o}_{1:h-2d},\overline{a}_{1:h-2d-1});f_o)\\
            &\qquad\qquad-\sum_{\overline{s}_{h-2d}}\PP_h^{\cD_\cL'}(\overline{s}_h,\overline{p}_h\given \overline{s}_{h-2d},f_a,f_o)F^{P_h(\cdot\given \cdot)}(\overline{\bb}_{h-2d}^{g^{h,j^\ast}}(\overline{o}_{h-\hat{L}+1:h-2d},\overline{a}_{h-\hat{L}:h-2d-1});f_o)(\overline{s}_{h-2d})}_1\\
            &\qquad\le \max_{\overline{g}\in\overline{\cG}_{1:\overline{H}}}\EE^{\cD_\cL'}_{\overline{o}_{1:h},\overline{a}_{1:h-1}\sim \overline{g}}\norm{\overline{\bb}_{h-2d}(\overline{o}_{1:h-2d},\overline{a}_{1:h-1-2d})-\overline{\bb}_{h-2d}^{g^{h,j^\ast}}(\overline{o}_{h-\hat{L}+1:h-2d},\overline{a}_{h-\hat{L}:h-2d-1})}_1\le 7\epsilon_1. 
        \end{align*}
        Similarly, $\max_{\overline{g}\in\overline{\cG}_{1:\overline{H}}}\EE^{\cD_\cL'}_{\overline{o}_{1:h},\overline{a}_{1:h-1}\sim \overline{g}}\norm{\PP_h^{\cD_\cL'}(\cdot\given \overline{c}_h,\gamma_h)-\PP_h^{\tilde{\cM}(g^{1:\overline{H},j^\ast}),z}(\cdot\given \hat{c}_h,\gamma_h)}_1\le 7\epsilon_1.$\\
        
    \noindent\textbf{Type 5:} The baseline sharing is \textbf{Example 6}. 
    Consider any $h\in[\overline{H}]$. If $h=2t-1,t\in[H]$, it holds that 
        \begin{align*}
            &\PP_h^{\tilde{\cM}(g^{1:\overline{H},j^\ast}),c}(\overline{s}_h,\overline{p}_h\given \hat{c}_h)=\overline{\bb}_h^{g^{h,j^\ast}}(\overline{o}_{h-\hat{L}+1:h-1},\overline{a}_{h-\hat{L}:h-1})(\overline{s}_h)\overline{\OO}_h(\overline{o}_h\given \overline{s_h})\Psi_{h}^1(\overline{p}_{h},\hat{c}_h),\\
            &\PP_h^{\cD_\cL'}(\overline{s}_h,\overline{p}_h\given \overline{c}_h)=\overline{\bb}_h(\overline{o}_{1:h-1},\overline{a}_{1:h-1})(\overline{s}_h)\overline{\OO}_h(\overline{o}_h\given \overline{s_h})\Psi_{h}^1(\overline{p}_{h},\overline{c}_h);
        \end{align*} 
        if $h=2t,t\in[H]$, then it holds that for any $t\in[H]$, 
              \begin{align*}
            &\PP_h^{\tilde{\cM}(g^{1:\overline{H},j^\ast}),c}(\overline{s}_h,\overline{p}_h\given \hat{c}_h)=\overline{\bb}_{h-1}^{g^{h,j^\ast}}(\overline{o}_{h-\hat{L}+1:h-2},\overline{a}_{h-\hat{L}:h-2},\overline{o}_{N,h-1})(\overline{s}_h)\PP_{h-1}(\overline{o}_{-N,h-1}\given \overline{s}_{h-1},\overline{o}_{N,h-1})\Psi_{h}^1(\overline{p}_{h},\hat{c}_h)\\
            &\PP_h^{\cD_\cL'}(\overline{s}_h,\overline{p}_h\given \overline{c}_h)=\overline{\bb}_{h-1}(\overline{o}_{1:h-2},\overline{a}_{1:h-2},\overline{o}_{N,h-1})(\overline{s}_h)\PP_{h-1}(\overline{o}_{-N,h-1}\given \overline{s}_{h-1},\overline{o}_{N,h-1})\Psi_{h}^1(\overline{p}_{h},\overline{c}_h),
        \end{align*} 
        where  we extend the definition of $\Psi_h^1$ in \textbf{Type 5} in the proof of \Cref{thm: full planning}, by replacing the input $\hat{c}_h$ therein  by $\overline{c}_h, \forall h\in[\overline{H}]$.     
        Then, similar  to \textbf{Type 1}, we can verify that for any $h\in[\overline{H}]$
        \begin{align*}
            &\max_{\overline{g}\in\overline{\cG}_{1:\overline{H}}}\EE^{\cD_\cL'}_{\overline{o}_{1:h},\overline{a}_{1:h-1}\sim \overline{g}}|\EE^{\cD_\cL'}[\cR_h(\overline{s}_h,\overline{a}_h,\overline{p}_h)\given \overline{c}_h,\gamma_h]-\hat{\cR}_h^{\tilde{\cM}(g^{1:\overline{H},j^\ast})}(\hat{c}_h,\gamma_h)|\le 7\epsilon_1\\      &\max_{\overline{g}\in\overline{\cG}_{1:\overline{H}}}\EE^{\cD_\cL'}_{\overline{o}_{1:h},\overline{a}_{1:h-1}\sim \overline{g}}\norm{\PP_h^{\cD_\cL'}(\cdot\given \overline{c}_h,\gamma_h)-\PP_h^{\tilde{\cM}(g^{1:\overline{H},j^\ast}),z}(\cdot\given \hat{c}_h,\gamma_h)}_1\le 7\epsilon_1. 
        \end{align*}
        
        \noindent\textbf{Type 6:} The baseline sharing is \textbf{Example 8}. 
        For any $h\in[\overline{H}]$,
        \begin{align*}
            &\PP_h^{\cD_\cL'}(\overline{s}_{h},\overline{p}_h\given f_{\tau,h-2d},f_a,f_o)=\sum_{\overline{s}_{h-2d}}\PP_h^{\cD_\cL'}(\overline{s}_h,\overline{p}_h\given \overline{s}_{h-2d},f_a,f_o)F^{P_h(\cdot\given \cdot, f_a)}(\overline{\bb}_{h-2d}(\overline{o}_{1:h-2d},\overline{a}_{1:h-2d-1});f_o)(\overline{s}_{h-2d})\Psi_{h}^2(\overline{p}_{h},\overline{c}_{h}),\\
            &\PP_h^{\tilde{\cM}(g^{1:\overline{H},j^\ast}),c}(\overline{s}_{h},\overline{p}_h\given f_{\tau,h-2d},f_a,f_o)\\
            &\quad=\sum_{\overline{s}_{h-2d}}\PP_h^{\cD_\cL'}(\overline{s}_h,\overline{p}_h\given \overline{s}_{h-2d},f_a,f_o)F^{P_h(\cdot\given \cdot, f_a)}(\overline{\bb}_{h-2d}^{g^{h,j^\ast}}(\overline{o}_{h-\hat{L}+1:h-2d},\overline{a}_{h-\hat{L}:h-2d-1});f_o)(\overline{s}_{h-2d})\Psi_{h}^2(\overline{p}_{h},\hat{c}_{h}),
        \end{align*} 
        where $f_a,f_o,f_{\tau,h-2d}, F^{P_h(\cdot\given \cdot, f_a)}$ are  defined in \textbf{Type 3}, and 
        we extend the definition of $\Psi_h^2$ in \textbf{Type 6} of the proof of \Cref{thm: full planning},  by replacing $\hat{c}_h$ by $\overline{c}_h$ as input, $ \forall h\in[\overline{H}]$. 
       Then, similar to \textbf{Type 3}, we can verify that for any $h\in[\overline{H}]$ 
        \begin{align*}
            &\max_{\overline{g}\in\overline{\cG}_{1:\overline{H}}}\EE^{\cD_\cL'}_{\overline{o}_{1:h},\overline{a}_{1:h-1}\sim \overline{g}}|\EE^{\cD_\cL'}[\cR_h(\overline{s}_h,\overline{a}_h,\overline{p}_h)\given \overline{c}_h,\gamma_h]-\hat{\cR}_h^{\tilde{\cM}(g^{1:\overline{H},j^\ast})}(\hat{c}_h,\gamma_h)|\le 7\epsilon_1\\      &\max_{\overline{g}\in\overline{\cG}_{1:\overline{H}}}\EE^{\cD_\cL'}_{\overline{o}_{1:h},\overline{a}_{1:h-1}\sim \overline{g}}\norm{\PP_h^{\cD_\cL'}(\cdot\given \overline{c}_h,\gamma_h)-\PP_h^{\tilde{\cM}(g^{1:\overline{H},j^\ast}),z}(\cdot\given \hat{c}_h,\gamma_h)}_1\le 7\epsilon_1. 
        \end{align*}
    Therefore, for any example in \S\ref{sec: examples of QC}, with probability $1-\delta_2$, we know that there exists some $j^\ast\in[K]$ such that
    \begin{align*}
        \epsilon_r(\tilde{\cM}(g^{1:\overline{H},j^\ast})) \le 7\epsilon_1,~~ \epsilon_z(\tilde{\cM}(g^{1:\overline{H},j^\ast}))\le  7\epsilon_1,~~ 
        \epsilon_{apx}^{j^\ast}\le 4\epsilon_1+\max_{h\in[\overline{H}]}\max_{\overline{g}\in\overline{\cG}_{1:\overline{H}}}\mathds{1}[h>\hat{L}]\cdot2\cdot d_{\cS,h-\hat{L}}^{\overline{g},\cD_\cL'}(\cU_{\phi,h-\hat{L}}^{\cD_\cL'}(g^{h,j^\ast}))\le 6\epsilon_1. 
    \end{align*} 
    Then, combined with the result in \textbf{Part II},  we have that with probability $1-\delta_1-\delta_2$, $\max_{\overline{g}_{1:\overline{H}}\in\overline{\cG}_{1:\overline{H}}} J_{\cD_\cL'}(\overline{g}_{1:\overline{H}})-J_{\cD_\cL'}(\overline{g}_{1:\overline{H}}^{j^\ast,\ast})\le 2(\frac{7}{2}\overline{H}^2+\overline{H}+6)\epsilon_1$, which further leads to 
    \begin{align*}
    &\max_{\overline{g}_{1:\overline{H}}\in\overline{\cG}_{1:\overline{H}}} J_{\cD_\cL'}(\overline{g}_{1:\overline{H}})-J_{\cD_\cL'}(\overline{g}_{1:\overline{H}}^{\hat{j},\ast})\le \min_{j\in[K]}\left(\max_{\overline{g}_{1:\overline{H}}\in\overline{\cG}_{1:\overline{H}}} J_{\cD_\cL'}(\overline{g}_{1:\overline{H}})-J_{\cD_\cL'}(\overline{g}_{1:\overline{H}}^{j,\ast})\right)+\frac{\epsilon}{2}\\
    &\qquad\le\max_{\overline{g}_{1:\overline{H}}\in\overline{\cG}_{1:\overline{H}}} J_{\cD_\cL'}(\overline{g}_{1:\overline{H}})-J_{\cD_\cL'}(\overline{g}_{1:\overline{H}}^{j^\ast,\ast})+\frac{\epsilon}{2}\le \epsilon,
    \end{align*}     
    with probability $1-\delta_1-\delta_2-\delta_3\ge 1-\delta$, with the choice of $\epsilon_1$ in \Cref{main learning algorithm}.

    {Finally, to apply the results in {\bf Part II} above, it remains to examine the satisfaction of  \Cref{assu: one_step_tract} for the learned models $\{\hat{\cM}(g^{1:\overline{H},j})\}_{j\in[K]}$. To this end,  we will instantiate \cite[Algorithm 5]{liu2023tractable} (Line \ref{line: LEE} of \Cref{main learning algorithm}) as follows:
    \begin{itemize}
        \item If $\cL$ has a baseline sharing protocol  as one of  \textbf{Examples 1, 5, 6} with the additional condition \ref{cond:tract_2}) in \S\ref{sec: examples of QC}, then we replace the Equation  (B.1) of \cite[Algorithm 5]{liu2023tractable} by: for any $i\in[n]$
    \begin{equation}
        \begin{aligned}          
\text{if $h=2t-1$, }\PP_h^{\hat{\cM}(g^{1:\overline{H}}), z}(\overline{z}_{i,h+1}\given\hat{c}_{i,h}, {\gamma}_{i,h}) &\leftarrow \sum_{\overline{p}_{i,h}}\mathds{1}[\overline{z}_{i,h+1}=\overline{\chi}_{i,h+1}(\overline{p}_{i,h},\overline{\gamma}_{i,h},\overline{o}_{i,h+1}=\emptyset)]\PP_h^{\hat{\cM}(g^{1:\overline{H}})}(\overline{p}_{i,h}\given\hat{c}_{i,h}),\\
\text{if $h=2t$, }\PP_h^{\hat{\cM}(g^{1:\overline{H}}), z}(\overline{z}_{i,h+1}\given\hat{c}_{i,h}, {\gamma}_{i,h}) &\leftarrow \sum_{\overline{p}_{i,h},\overline{a}_{i,h},\overline{o}_{i,h+1}}\mathds{1}[\overline{z}_{i,h+1}=\overline{\chi}_{i,h+1}(\overline{p}_{i,h},\overline{a}_{i,h},\overline{o}_{i,h+1})]\PP_h^{\hat{\cM}(g^{1:\overline{H}})}(\overline{p}_{i,h}\given\hat{c}_{i,h})\\
&\qquad\qquad \gamma_{i,h}(\overline{a}_{i,h}\given \overline{p}_{i,h})\PP_h^{\hat{\cM}(g^{1:\overline{H}})}(\overline{o}_{i,h+1}\given\hat{c}_{i,h},\overline{p}_{i,h},\overline{a}_{i,h}),
\end{aligned}
\end{equation}
and replace Equation (B.2) of \cite[Algorithm 5]{liu2023tractable} by: for any $i\in[n]$
    \begin{equation}
        \begin{aligned}
&\text{if $h=2t-1$, }\hat{\cR}^{\hat{\cM}(g^{1:\overline{H}})}_{i,h}(\hat{c}_{i,h}, {\gamma}_{i,h}) \leftarrow \sum_{\overline{p}_{i,h}}\PP_h^{\hat{\cM}(g^{1:\overline{H}})}(\overline{p}_{i,h}\given\hat{c}_{i,h})\hat{\cR}_{i,h}^{\hat{\cM}(g^{1:\overline{H}})}(\hat{c}_{i,h},\overline{p}_{i,h},\gamma_{i,h}),\\
&\text{if $h=2t$, }\hat{\cR}^{\hat{\cM}(g^{1:\overline{H}})}_{i,h}(\hat{c}_{i,h}, {\gamma}_{i,h}) \leftarrow \sum_{\overline{p}_{i,h},\overline{a}_{i,h}}\PP_h^{\hat{\cM}(g^{1:\overline{H}})}(\overline{p}_{i,h}\given\hat{c}_{i,h})\hat{\cR}_{i,h}^{\hat{\cM}(g^{1:\overline{H}})}(\hat{c}_{i,h},\overline{p}_{i,h},\overline{a}_{i,h})\gamma_{i,h}(\overline{a}_{i,h}\given\overline{p}_{i,h}),
\end{aligned}
\end{equation}
where 
$\{\overline{\chi}_{i,h}\}_{h\in[\overline{H}]}$ can be constructed based on $\{\chi_{i,t}\}_{t\in[H]}$ and $\{\phi_{i,t}\}_{t\in[H]}$ similarly as the proof of \Cref{theorem: refinement} as follows:
\begin{align*}
           &\forall \overline{p}_{i,h-1}\in\overline{\cP}_{i,h-1},\overline{a}_{i,h-1}\in\overline{\cA}_{i,h-1}, \overline{o}_{i,h}\in\overline{\cO}_{i,h}, \text{ if $h$ is even, then }\overline{\chi}_{i,h}(\overline{p}_{i,h-1},\overline{a}_{i,h-1},\overline{o}_{i,h})=\phi_{i,\frac{h}{2}}(\overline{p}_{i,h-1},\overline{a}_{i,h-1})\\
           &\text{if $h$ is odd, then }\overline{\chi}_{i,h}(\overline{p}_{i,h-1},\overline{a}_{i,h-1},\overline{o}_{i,h})=\chi_{i,\frac{h+1}{2}}(\overline{p}_{i,h-1},\overline{a}_{i,h-1},\overline{o}_{i,h})\cup \varrho_{i,h}^1\backslash\varrho_{i,h}^2,\text{ where}\\
           &\varrho_{i,h}^1:=\{\tilde{a}_{i,h-1}\given \forall \sigma(\tilde{\tau}_{i,h-1})\subseteq \sigma(\tilde{c}_h), \tilde{a}_{i,h-1}\text{~influences }\tilde{s}_{h}\}\backslash\chi_{i,\frac{h+1}{2}}(\overline{p}_{i,h-1},\overline{a}_{i,h-1},\overline{o}_{i,h})\text{ if $h>1$, otherwise $\emptyset.$}\\
           &\varrho_{i,h}^2:=\{\tilde{a}_{i,h_0}\given \forall h_0<h-1, \sigma(\tilde{a}_{i,h_0})\subseteq \sigma(\tilde{c}_{h-1}),\tilde{a}_{i,h_0} \text{ influences } \tilde{s}_{h_0+1}\}\cap \chi_{i,\frac{h+1}{2}}(\overline{p}_{i,h-1},\overline{a}_{i,h-1},\overline{o}_{i,h}). 
       \end{align*}
       \item If $\cL$ has a baseline sharing protocol as one of \textbf{Examples 2, 4, 7, 8}, or one of \textbf{Examples 1, 5, 6} with additional condition  \ref{cond:tract_1}) in \S\ref{sec: examples of QC}, then we replace Equation  (B.1) of \cite[Algorithm 5]{liu2023tractable} by: 
    \begin{equation}
        \begin{aligned}         
&\text{if $h=2t-1$, we do not make any change,}\\
&\text{if $h=2t$, }\PP_h^{\hat{\cM}(g^{1:\overline{H}}), z}(\overline{z}_{h+1}\given\hat{c}_{h}, {\gamma}_{ct(h),h}) \leftarrow \sum_{\overline{p}_{h},\overline{a}_{ct(h),h},\overline{o}_{h+1}}\mathds{1}[\overline{z}_{h+1}=\overline{\chi}_{h+1}(\overline{p}_{i,h},\overline{a}_{ct(h),h},\overline{o}_{h+1})]\PP_h^{\hat{\cM}(g^{1:\overline{H}})}(\overline{p}_{h}\given\hat{c}_{h})\\
&\qquad\qquad\qquad\qquad\qquad\qquad  \qquad\qquad \qquad \gamma_{ct(h),h}(\overline{a}_{ct(h),h}\given \overline{p}_{h})\PP_h^{\hat{\cM}(g^{1:\overline{H}})}(\overline{o}_{h+1}\given\hat{c}_{h},\overline{p}_{h},\overline{a}_{ct(h),h
}),
\end{aligned}
\end{equation}
and replace Equation (B.2) of \cite[Algorithm 5]{liu2023tractable} by:
    \begin{equation}
        \begin{aligned}
&\text{if $h=2t-1$, we do not make any change,}\\
&\text{if $h=2t$, }\hat{\cR}^{\hat{\cM}(g^{1:\overline{H}})}_{i,h}(\hat{c}_{h}, {\gamma}_{i,h}) \leftarrow \sum_{\overline{p}_h,\overline{a}_{i,h}}\PP_h^{\hat{\cM}(g^{1:\overline{H}})}(\overline{p}_{h}\given\hat{c}_{h})\hat{\cR}_{i,h}^{\hat{\cM}(g^{1:\overline{H}})}(\hat{c}_{h},\overline{p}_{h},\overline{a}_{i,h})\gamma_{i,h}(\overline{a}_{i,h}\given\overline{p}_{i,h}).
\end{aligned}
\end{equation}
\item If $\cL$ has a baseline sharing protocol as  \textbf{Example 3} or one of \textbf{Examples 1} and  {\bf 5} with additional condition \ref{cond:tract_3}) in \S\ref{sec: examples of QC}, then we do not make any change.
    \end{itemize}
    Then, the $\hat{\cM}=\{\hat{\cM}(g^{1:\overline{H},j})\}_{j=1}^K$ learned in Line \ref{line: LEE} of \Cref{main learning algorithm} satisfies that: for any $j\in [K]$, $\hat{\cM}(g^{1:\overline{H},j})$ satisfies the  \textbf{Factorized structures} condition in \S\ref{sec:one_step_examples} if
    the baseline sharing protocol of $\cL$ is one of \textbf{Examples 1, 5, 6} with additional condition \ref{cond:tract_2});  $\hat{\cM}(g^{1:\overline{H},j})$ satisfies the \textbf{Turn-based structures} condition in \S\ref{sec:one_step_examples}, if the baseline sharing protocol of $\cL$ is one of \textbf{Examples 2, 4, 7, 8} or \textbf{Examples 1, 5, 6} with additional condition \ref{cond:tract_1});
$\hat{\cM}(g^{1:\overline{H},j})$ satisfies the  \textbf{Nested private information} condition in \S\ref{sec:one_step_examples}, if the baseline sharing protocol of $\cL$ is \textbf{Example  3} or \textbf{Examples 1, 5} with additional condition \ref{cond:tract_3}). From Lemma \ref{lemma: one_step_tract}, we can conclude that Assumption \ref{assu: one_step_tract} holds, by noticing that in these examples,  $\max_{h\in[\overline{H}]}|\overline{\cP}_h|$ depends polynomially on  the parameters of the original  LTC problem $\cL$. Hence, \Cref{main learning algorithm} has  the following  sample and time complexities
    \begin{itemize}
        \item \textbf{Examples 1, 3, 5, 6:} $\texttt{poly}\left(\max_{h\in[\overline{H}]}(|\overline{\cO}_h|| \overline{\cA}_h|)^{C\gamma^{-4}\log(\frac{|\overline{\cS}|}{\epsilon})}, |\overline{\cS}|, \overline{H},\frac{1}{\epsilon},\log(\frac{1}{\delta})\right)$;  
        \item \textbf{Examples 2, 4, 7, 8:} $\texttt{poly}\left(\max_{h\in[\overline{H}]}(|\overline{\cO}_h||\overline{\cA}_h|)^{C\gamma^{-4}\log(\frac{|\cS|}{\epsilon})+2d}, |\overline{\cS}|,\overline{H},\frac{1}{\epsilon},\log(\frac{1}{\delta})\right)$, 
    \end{itemize}
    to achieve $\max_{\overline{g}_{1:\overline{H}}\in\overline{\cG}_{1:\overline{H}}} J_{\cD_\cL'}(\overline{g}_{1:\overline{H}})-J_{\cD_\cL'}(\overline{g}_{1:\overline{H}}^{\hat{j},\ast})\leq \epsilon$ with probability  at least $1-\delta$,  
    which completes the proof of \textbf{Part III}.
}
        Combining {\bf Parts I, II, III}, we complete the proof.
\end{proof}

\section{Deferred Details of  \S\ref{general Dec-POMDP}}
\label{proof details sec 5}
We first introduce the  notion of \emph{perfect recall} \cite{kuhn1953extensive}: 
\begin{definition}[Perfect recall]\label{def:PR}
    We say that agent $i$ has perfect recall if $\forall h=2,\cdots, \overline{H}$, 
    it holds that $
        \overline{\tau}_{i,h-1}\cup\{\overline{a}_{i,h-1}\}\subseteq \overline{\tau}_{i,h}.$ 
    If for any $i\in[n]$, agent $i$ has perfect recall, we call that the Dec-POMDP has a perfect recall property. 
\end{definition}

\subsection{Proof of Theorem \ref{theorem: SI=sQC general}}
\begin{proof} 
sQC $\Rightarrow$ SI-CIB:\\
Let $\cD$ be a Dec-POMDP with an sQC information structure that  satisfies Assumptions \ref{assumption:information evolution} (e), \ref{assumption:sigma_include}, \ref{useless action}, and \ref{weak gamma observability}. To prove that $\cD$ has SI-CIBs, it is sufficient to prove that for any $h=2,\cdots,\overline{H}$, fix any $h_1\in[h-1], i_1\in[n]$, and for any $\overline{g}_{1:h-1}\in \overline{\cG}_{1:h-1}, \overline{g}'_{i_1,h_1}\in \overline{\cG}_{i_1,h_1}$, let $\overline{g}_{h_1}':=(\overline{g}_{1,h_1},\cdots,\overline{g}_{i_1,h_1}',\cdots,\overline{g}_{n,h_1})$ and 
$\overline{g}_{1:h-1}':=(\overline{g}_{1},\cdots,\overline{g}_{h_1}',\cdots,\overline{g}_{h-1})$, for any $\overline{c}_h\in\overline{\cC}_h$ reachable under both $\overline{g}_{1:h-1},\overline{g}_{1:h-1}'$, 
the following holds:
\begin{equation}
\PP(\overline{s}_h,\overline{p}_h\given \overline{c}_h,\overline{g}_{1:h-1})=\PP(\overline{s}_h,\overline{p}_h\given \overline{c}_h,\overline{g}'_{1:h-1}),
\end{equation}
where, for notational simplicity, we omit both the superscript $\cD$ and the subscript $h$ for the beliefs throughout this proof, as they are clear from the context.  
We prove this case by case as follows: 
\begin{itemize}
    \item 
If there exists some $i_2\neq i_1$ such that
$\sigma(\overline{\tau}_{i_1,h_1})\cup \sigma(\overline{a}_{i_1,h_1})\subseteq \sigma(\overline{\tau}_{i_2,h})$, then from Assumption \ref{assumption:sigma_include}, we know that $\sigma(\overline{\tau}_{i_1,h_1})\cup \sigma(\overline{a}_{i_1,h_1})\subseteq \sigma(\overline{c}_h)$. Therefore, there exist deterministic measurable functions $\beta_1, \beta_2$ such that $\overline{\tau}_{i_1,h_1}=\beta_1(\overline{c}_h), \overline{a}_{i_1,h_1}=\beta_2(\overline{c}_h)$, and it holds that 
\begin{align*} 
\PP(\overline{s}_h,\overline{p}_h\given \overline{c}_h,\overline{g}_{1:h-1})&=\PP(\overline{s}_h,\overline{p}_h\given \beta_1(\overline{c}_h), \beta_2(\overline{c}_h), \overline{c}_h,\overline{g}_{1:h-1})=\PP(\overline{s}_h,\overline{p}_h\given \overline{\tau}_{i_1,h_1},\overline{a}_{i_1,h_1},\overline{c}_h, \overline{g}_{1:h-1}')\\
&=\PP(\overline{s}_h,\overline{p}_h\given \overline{c}_h,\overline{g}_{1:h-1}'). 
\end{align*}
\normalsize 
The last equality is because the input and output of $\overline{g}_{i_1,h_1}$ (and $\overline{g}_{i_1,h_1}'$) are $\overline{\tau}_{i_1,h_1}$ and $\overline{a}_{i_1,h_1}$, respectively, and they are both conditioned on.  
 \item If 
 for all $i_2\neq i_1$, 
 either  $\sigma(\overline{\tau}_{i_1,h_1})\nsubseteq \sigma(\overline{\tau}_{i_2,h})$ or $\sigma(\overline{a}_{i_1,h_1})\nsubseteq \sigma(\overline{\tau}_{i_2,h})$, then agent $(i_1,h_1)$ does not influence agent $(i_2,h)$ for any $i_2\neq i_1$,  since $\cD$ is sQC.
 
 Firstly, we claim that agent $(i_1,h_1)$ does not influence  $\overline{s}_{h_1+1}$:  if it influences, from Assumption \ref{weak gamma observability}, there exists some $i_3\neq i_1$ such that agent $(i_1,h_1)$ influences $\overline{o}_{i_3,h_1+1}$; however, from Assumption \ref{assumption:information evolution} (e), we know that  $\overline{o}_{i_3,h_1+1}\in \overline{\tau}_{i_3,h_1+1}\subseteq \overline{\tau}_{i_3,h}$; therefore, agent $(i_1,h_1)$ influences agent $(i_3,h)$, contradicting the premise above.  
 Similarly, we  have that agent $(i_1,h_1)$ does not influence $\overline{s}_{h_2}$  for any $h_2>h_1$.

 Secondly, we claim that agent $(i_1,h_1)$ does not influence $\overline{\tau}_{i_3,h_2}$, for any $i_3\in[n]$ and $h_2>h_1$. Since agent $(i_1,h_1)$ does not influence $\overline{s}_{h_1+1}$, then by Assumption \ref{useless action}, for any $h_2>h_1$, $\overline{a}_{i_1,h_1}\notin \overline{\tau}_{h_2}$ and agent   $(i_1,h_1)$ does not influence $\overline{o}_{i_3,h_1+1}$ for any $i_3\in[n]$, which implies that agent $(i_1,h_1)$  does not influence any element in $\overline{\tau}_{i_3,h_1+1}$ for any $i_3\in[n]$, either directly or indirectly. Since $\overline{\tau}_{i_3,h_1+1}$ is the input of agent $i_3$'s strategy at timestep $h_1+1$ to decide $\overline{a}_{i_3,h_1+1}$, agent $(i_1,h_1)$ thus does not influence $\overline{a}_{i_3,h_1+1}$ for any $i_3\in[n]$, either,  which, together with the fact that it does not influence $\overline{s}_{h_1+2}$ and thus not $\overline{o}_{i_3,h_1+2}$ for any $i_3\in[n]$, further implies that it does not influence any element in $\overline{\tau}_{i_3,h_1+2}$ for any $i_3\in[n]$. By recursion, agent $(i_1,h_1)$ does not influence  $\overline{\tau}_{i_3,h_2}$  for any $i_3\in[n]$ and $h_2>h_1$.
 
 Therefore, agent $(i_1,h_1)$ does not influence $\overline{c}_{h}=\cap_{i_3=1}^n \overline{\tau}_{i_3,h}$ (due to Assumption \ref{assumption:information evolution}) nor $\overline{p}_h=\overline{\tau}_h\backslash \overline{c}_h$, and it does not influence $\overline{s}_h$, either, which means that 
 \begin{align*}
\PP(\overline{s}_h,\overline{p}_h\given \overline{c}_h,\overline{g}_{1:h-1})=\PP(\overline{s}_h,\overline{p}_h\given \overline{c}_h,\overline{g}'_{1:h-1}). 
 \end{align*}
\end{itemize}
SI-CIB $\Rightarrow$ sQC:\\
Since $\cD$ has perfect recall and has SI-CIBs, we know that $\forall i\in[n], h\in[\overline{H}],\forall \overline{g}_{1:h-1},\overline{g}'_{1:h-1}\in \overline{\cG}_{1:h-1},\overline{c}_h\in\overline{\cC}_h,\overline{s}_h\in\overline{\cS},\overline{p}_h\in\overline{\cP}_h$, if $\overline{c}_h$ is reachable under both $\overline{g}_{1:h-1},\overline{g}_{1:h-1}'$, then the following holds:
\begin{equation*}
    \PP(\overline{s}_h,\overline{p}_h\given \overline{c}_h,\overline{g}_{1:h-1})=\PP(\overline{s}_h,\overline{p}_h\given \overline{c}_h,\overline{g}'_{1:h-1}).
\end{equation*}
Our goal is to prove that 
{if agent $(i_1,h_1)$ influences agent $(i_2,h_2)$ in the intrinsic model of the Dec-POMDP $\cD$,  
then under any strategy $\overline{g}_{1:\overline{H}}\in \overline{\cG}_{1:\overline{H}}$, $\sigma(\overline{\tau}_{i_1,h_1})\cup\sigma(\overline{a}_{i_1,h_1})\subseteq\sigma(\overline{\tau}_{i_2,h_2})$ holds.} 
{Note that throughout the proof, when it comes to $\sigma$-algebra inclusion, we meant it up to the null sets generated by $\overline{g}_{1:\overline{H}}$}.
We prove this by contradiction. If this is not true, then there exist $i_1,i_2\in[n],h_1,h_2\in[\overline{H}]$, such that agent $(i_1,h_1)$ influences agent $(i_2,h_2)$, but either  $\sigma(\overline{\tau}_{i_1,h_1})\nsubseteq\sigma(\overline{\tau}_{i_2,h_2})$ or $\sigma(\overline{a}_{i_1,h_1})\nsubseteq\sigma(\overline{\tau}_{i_2,h_2})$. 
First, we know that $i_2\neq i_1$, since otherwise it always holds that $\overline{\tau}_{i_1,h_1}\subseteq\overline{\tau}_{i_1,h_2}$ and $\overline{a}_{i_1,h_1}\in\overline{\tau}_{i_1,h_2}$, due to the perfect recall assumption.
Then, we 
discuss the following  different cases.  
\begin{itemize}
\item If $\sigma(\overline{a}_{i_1,h_1})\nsubseteq \sigma(\overline{\tau}_{i_2,h_2})$, then it implies that $\sigma(\overline{a}_{i_1,h_1})\nsubseteq \sigma(\overline{c}_{h_2})$ and $\overline{a}_{i_1,h_1}\notin \overline{c}_{h_2}$,  since $\overline{c}_{h_2}\subseteq \overline{\tau}_{i_2,h_2}$. 
{This further implies that there must exist a strategy  $\overline{g}_{1:h_2-1}$ and some realizations $\overline{c}_{h_2}\in\overline{\cC}_{h_2},  \overline{a}_{i_1,h_1}\neq\overline{a}_{i_1,h_1}'\in\overline{\cA}_{i_1,h_1}$ such that $\PP(\overline{c}_{h_2}\given \overline{g}_{1:h_2-1})>0, \PP(\overline{a}_{i_1,h_1}\given \overline{c}_{h_2},\overline{g}_{1:h_2-1})>0$, and $\PP(\overline{a}_{i_1,h_1}'\given \overline{c}_{h_2},\overline{g}_{1:h_2-1})>0$.  Hence, due to perfect recall, there must exist some realizations  $\overline{p}_{h_2}\in\overline{\cP}_{h_2}, \overline{s}_{h_2}\in\overline{\cS}$ such that $\overline{a}_{i_1,h_1}\in\overline{p}_{h_2}$, and $\PP(\overline{s}_{h_2},\overline{p}_{h_2}\given \overline{c}_{h_2}, \overline{g}_{1:h_2-1})>0$. Then, we define another strategy  $\overline{g}'_{i_1,h_1}$ as:
\begin{equation}\label{equ:reassign_g}
        \forall \overline{\tau}_{i_1,h_1}\in \overline{\cT}_{i_1,h_1},~~~ \overline{g}'_{i_1,h_1}(\overline{\tau}_{i_1,h_1})=\overline{a}_{i_1,h_1}',
\end{equation} 
and we let  $\overline{g}_{h_1}':=(\overline{g}_{1,h_1},\cdots,\overline{g}_{i_1,h_1}',\cdots,\overline{g}_{n,h_1})$ and 
$\overline{g}_{1:h_2-1}':=(\overline{g}_{1},\cdots,\overline{g}_{h_1}',\cdots,\overline{g}_{h_2-1})$. Now we claim that $\overline{c}_{h_2}$ has non-zero probability under $\overline{g}'_{1:h_2-1}$, i.e., $\overline{c}_{h_2}$ is reachable under $\overline{g}'_{1:h_2-1}$. Since $\overline{c}_{h_2}$ is reachable under $\overline{g}_{1:h_2-1}$ and  
$ \PP(\overline{a}_{i_1,h_1}'\given \overline{c}_{h_2}, \overline{g}_{1:h_2-1})>0$, we know that $\PP(\overline{a}_{i_1,h_1}', \overline{c}_{h_2}\given \overline{g}_{1:h_2-1})>0$. Since $\overline{g}_{1:h_2-1}'$ only differs from $\overline{g}_{1:h_2-1}$ in the strategy of agent $(i_1,h_1)$, and $\overline{g}_{i_1,h_1}'$ always chooses $\overline{a}_{i_1,h_1}'$, then it holds that $\PP(\overline{a}_{i_1,h_1}', \overline{c}_{h_2}\given \overline{g}_{1:h_2-1}')\ge\PP(\overline{a}_{i_1,h_1}', \overline{c}_{h_2}\given \overline{g}_{1:h_2-1})>0$, and thus $\PP(\overline{c}_{h_2}\given \overline{g}_{1:h_2-1}')>0$.
Meanwhile, due to \eqref{equ:reassign_g}, we have $
\PP(\overline{s}_{h_2},\overline{p}_{h_2}\given \overline{c}_{h_2},\overline{g}'_{1:h_2-1})=0
    \neq \PP(\overline{s}_{h_2},\overline{p}_{h_2}\given \overline{c}_{h_2},\overline{g}_{1:h_2-1}),$ 
which leads to a contradiction to the fact that $\cD$ has SI-CIBs. 
\item If  $\sigma(\overline{a}_{i_1,h_1})\subseteq \sigma(\overline{\tau}_{i_2, h_2})$, then we know that $\sigma(\overline{\tau}_{i_1,h_1})\nsubseteq \sigma(\overline{\tau}_{i_2,h_2})$. This further implies that there must exist a strategy $\overline{g}_{1:h_2-1}$ and some realizations $\overline{\tau}_{i_1,h_1}\neq \overline{\tau}_{i_1,h_1}'\in\overline{\cT}_{i_1,h_1}, \overline{\tau}_{i_2,h_2}\in\overline{\cT}_{i_2,h_2}$ such that $\PP(\overline{\tau}_{i_2,h_2}\given \overline{g}_{1:h_2-1})>0, \PP(\overline{\tau}_{i_1,h_1}\given \overline{\tau}_{i_2,h_2}, \overline{g}_{1:h_2-1})>0$,  and $\PP(\overline{\tau}'_{i_1,h_1}\given \overline{\tau}_{i_2,h_2}, \overline{g}_{1:h_2-1})>0$. {Then, due to perfect recall, there exist some realizations $\overline{s}_{h_2}\in \overline{\cS}, \overline{\tau}_{h_2}\in\overline{\cT}_{h_2}$ such that $\overline{\tau}_{i_1,h_1}\subseteq \overline{\tau}_{h_2}, \overline{\tau}_{i_2,h_2}\subseteq \overline{\tau}_{h_2}$ and $\PP(\overline{s}_{h_2},\overline{\tau}_{h_2}\given \overline{g}_{1:h_2-1})>0$. Let $\overline{c}_{h_2}\in\overline{\cC}_{h_2}, \overline{p}_{h_2}\in\overline{\cP}_{h_2}$ be the realizations such that $\overline{c}_{h_2}\subseteq \overline{\tau}_{h_2}$ and $\overline{p}_{h_2}\subseteq \overline{\tau}_{h_2}$, then it holds that $\PP(\overline{s}_{h_2},\overline{p}_{h_2},\overline{c}_{h_2}\given\overline{g}_{1:h_2-1})=\PP(\overline{s}_{h_2},\overline{\tau}_{ h_2}\given\overline{g}_{1:h_2-1})>0$. Thus, we know that $\PP(\overline{s}_{h_2},\overline{p}_{h_2}\given \overline{c}_{h_2},\overline{g}_{1:h_2-1})>0$, and that $\overline{c}_{h_2}$ is reachable under $\overline{g}_{1:h_2-1}$, i.e., $\PP(\overline{c}_{h_2}\given \overline{g}_{1:h_2-1})>0$.}
Meanwhile, since  $\sigma(\overline{a}_{i_1,h_1})\subseteq \sigma(\overline{\tau}_{i_2,h_2})$, we know that 
there exists a realization   $\overline{a}_{i_1,h_1}\in \overline{\cA}_{i_1,h_1}$ such that 
$\PP(\overline{a}_{i_1,h_1}\given \overline{\tau}_{i_2,h_2},\overline{g}^{\circ}_{1:h_2-1})=1$ holds for any strategy $\overline{g}^{\circ}_{1:h_2-1}\in\overline{\cG}_{1:h_2-1}$. By applying $\overline{g}^{\circ}_{1:h_2-1}=\overline{g}_{1:h_2-1}$, we obtain  $\PP(\overline{a}_{i_1,h_1}\given \overline{\tau}_{i_2,h_2},\overline{g}_{1:h_2-1})=1$. 
Consider another different action realization $\overline{a}_{i_1,h_1}'\neq \overline{a}_{i_1,h_1}$, and define a new strategy $\overline{g}'_{i_1,h_1}$ as
\begin{equation}
    \overline{g}'_{i_1,h_1}(\overline{\tau}_{i_1,h_1})=\overline{a}_{i_1,h_1}',~~~\overline{g}'_{i_1,h_1}(\overline{\tau}_{i_1,h_1}')=\overline{a}_{i_1,h_1}, 
    \label{equ:reassign_g2}
\end{equation}
and keep $g'_{i_1,h_1}(\overline{\tau}_{i_1,h_1}'')$ the same as $g_{i_1,h_1}(\overline{\tau}_{i_1,h_1}'')$ for any other realizations  $\overline{\tau}_{i_1,h_1}''\in\overline{\cT}_{i_1,h_1}$. We denote $\overline{g}_{h_1}':=(\overline{g}_{1,h_1},\cdots,\overline{g}_{i_1,h_1}',\cdots,\overline{g}_{n,h_1})$ and 
$\overline{g}_{1:h_2-1}':=(\overline{g}_{1},\cdots,\overline{g}_{h_1}',\cdots,\overline{g}_{h_2-1})$. Since $\overline{\tau}_{i_2,h_2}$ has non-zero probability under $\overline{g}_{1:h_2-1}$, and $\PP(\overline{\tau}_{i_1,h_1}'\given \overline{\tau}_{i_2,h_2}, \overline{g}_{1:h_2-1})>0$, it holds that $\PP(\overline{\tau}_{i_1,h_1}', \overline{\tau}_{i_2,h_2}\given \overline{g}_{1:h_2-1})>0$. Together with the fact that $\PP(\overline{a}_{i_1,h_1}\given \overline{\tau}_{i_2,h_2}, \overline{g}_{1:h_2-1})=1$, we can obtain that
$\PP(\overline{\tau}_{i_1,h_1}', \overline{a}_{i_1,h_1}, \overline{\tau}_{i_2,h_2}\given \overline{g}_{1:h_2-1})>0$. Since $\overline{g}_{1:h_2-1}'$ only differs from $\overline{g}_{1:h_2-1}$ in the strategy of agent $(i_1,h_1)$, and $\overline{g}_{i_1,h_1}'$ always chooses $\overline{a}_{i_1,h_1}$ when inputting $\overline{\tau}_{i_1,h_1}'$, then it holds that 
$\PP(\overline{c}_{h_2}\given \overline{g}_{1:h_2-1}')\ge \PP(\overline{\tau}_{i_2,h_2}\given \overline{g}_{1:h_2-1}')\ge \PP(\overline{\tau}_{i_1,h_1}',\overline{a}_{i_1,h_1}, \overline{\tau}_{i_2,h_2}\given \overline{g}_{1:h_2-1}')\ge \PP(\overline{\tau}_{i_1,h_1}', \overline{a}_{i_1,h_1}, \overline{\tau}_{i_2,h_2}\given \overline{g}_{1:h_2-1})>0$, which means $\overline{c}_{h_2}$ is reachable under $\overline{g}_{1:h_2-1}'$. Meanwhile, it holds that
\begin{equation}
\begin{aligned}    
&\PP(\overline{s}_{h_2},\overline{p}_{h_2}\given \overline{c}_{h_2},\overline{g}'_{1:h_2-1})=\frac{\PP(\overline{s}_{h_2},\overline{p}_{h_2},\overline{c}_{h_2}\given\overline{g}'_{1:h_2-1})}{\PP(\overline{c}_{h_2}\given\overline{g}'_{1:h_2-1})}=\frac{\PP(\overline{s}_{h_2},\overline{\tau}_{h_2}\given\overline{g}'_{1:h_2-1})}{\PP(\overline{c}_{h_2}\given\overline{g}'_{1:h_2-1})}=0
    \neq\PP(\overline{s}_{h_2},\overline{p}_{h_2}\given \overline{c}_{h_2},\overline{g}_{1:h_2-1}),
\end{aligned}
\label{equ:sQC->SI tau contradict}
\end{equation}
where the third equality is because  1) $\overline{\tau}_{i_1,h_1}\subseteq \overline{\tau}_{h_2}$, $\overline{a}_{i_1,h_1}\in \overline{\tau}_{h_2}$ since  $\PP(\overline{a}_{i_1,h_1}\given \overline{\tau}_{i_2,h_2})=1$, $\overline{\tau}_{i_2,h_2}\subseteq \overline{\tau}_{h_2}$, and perfect recall; 2) $\overline{a}_{i_1,h_1},\overline{\tau}_{i_1,h_1}$ can never be realized simultaneously under $\overline{g}'_{1:h_2-1}$ due to \eqref{equ:reassign_g2}. Therefore, \eqref{equ:sQC->SI tau contradict} leads to a contradiction to the fact that $\cD$ has SI-CIBs.
}
\end{itemize}
This  completes the proof.  
\end{proof}

\clearpage 
\section{Collection of Algorithm Pseudocodes}

Here we collect both our planning and learning algorithms as pseudocodes below. 

\begin{algorithm}
    \caption{Planning in QC LTC Problems}
    \label{main algorithm}
    \begin{algorithmic}[1]
    \REQUIRE LTC problem  $\cL$ 
    \STATE Reformulate $\cL$ as $\cD_\cL$ based on \Cref{LTC to Dec-POMDP} 
    \STATE Expand $\cD_\cL$ to $\cD_\cL^\dag$  based on \Cref{construction:QC to sQC} 
    \STATE Refine $\cD_\cL^\dag$ as $\cD_\cL'$ based on $\cL$ and \S\ref{sec: refinement}  
    \STATE Construct an expected approximate common-information model $\cM$ from $\cD_\cL'$ (see examples of such constructions in the proof of \Cref{theorem: planning}) \label{line: construct AIS}
    \STATE $\hat{g}_{1:\tilde{H}}^\ast\leftarrow$ Algorithm \ref{algorithm under AIS}($\cM$)\label{line: dp in planning}
    \STATE $\tilde{g}_{1:\tilde{H}}^\ast\leftarrow$ $\varphi(\hat{g}_{1:\tilde{H}}^\ast,\cD_\cL)$ 
    \STATE $g^{m,\ast}_{1:H}\leftarrow \{\tilde{g}_{1}^\ast,\tilde{g}_{3}^\ast,\cdots,\tilde{g}_{2H-1}^\ast\}$
    \STATE $g^{a,\ast}_{1:H}\leftarrow \{\tilde{g}_{2}^\ast,\tilde{g}_{4}^\ast,\cdots,\tilde{g}_{2H}^\ast\}$
    \RETURN $(g^{a,\ast}_{1:H},g^{m,\ast}_{1:H})$
    \end{algorithmic}
\end{algorithm}
\begin{algorithm}[!h]
    \caption{Learning in QC LTC Problems}
    \label{main learning algorithm}
    \begin{algorithmic}[1]
    \REQUIRE  LTC problem  $\cL$, compression functions $\{\text{Compress}_h\}_{h\in[\overline{H}]}$ and rules $\{\hat{\phi}_{h}\}_{h\in[\overline{H}]}$, length $\hat{L}$, accuracy level $\epsilon$, probability $\delta$, constant $C$ 
    \STATE Reformulate $\cL$ as $\cD_\cL$ based on \Cref{LTC to Dec-POMDP} 
    \STATE Expand $\cD_\cL$ to $\cD_\cL^\dag$  based on \Cref{construction:QC to sQC} 
    \STATE Refine $\cD_\cL^\dag$ as $\cD_\cL'$ based on $\cL$ and \S\ref{sec: refinement}  
    \STATE Construct $\{\hat{\cC}_h\}_{h\in[\overline{H}]}$ as $\forall h\in[\overline{H}], \hat{\cC}_h=\{\text{Compress}_h(\overline{c}_h)\given \overline{c}_h\in\overline{\cC}_h\}$
    \STATE Denote $S=|\overline\cS|,A=\max_{h\in[\overline{H}]}|\overline{\cA}_h|, O=\max_{h\in[\overline{H}]}|\overline{\cO}_h|, P=\max_{h\in[\overline{H}]}|\overline{\cP}_h|, \hat{C}=\max_{h\in[\overline{H}]}|\hat{\cC}_h|$,  and recall $\gamma$ is the parameter in \Cref{gamma observability} 
    \STATE Define $K=2\overline{H}S, \alpha=\frac{C\overline{H}^2\epsilon}{200(\overline{H}+1)^2}, \epsilon_1=\frac{\epsilon}{200(\overline{H}+1)^2}, \theta_1=\frac{\epsilon}{200(\overline{H}+1)^2O}, \theta_2=\frac{\epsilon}{200(\overline{H}+1)^2AP},  \phi=\frac{\epsilon_1\gamma^2}{C^2\overline{H}^8S^5O^4},\zeta_1=\min\left\{\frac{\epsilon\phi}{200(\overline{H}+1)^2A^{2\hat{L}}O^{\hat{L}}},\frac{\epsilon}{400(\overline{H}+1)^2AP}\right\}, \zeta_2=\zeta_1^2, \beta=\frac{\delta}{3}, N_0=\lceil\max\left\{\frac{C(P+\log\frac{4\overline{H}{\hat{C}}}{\beta}}{\zeta_1\theta_1^2},\frac{CA(O+\log\frac{4\overline{H}\hat{C}PA}{\beta})}{\zeta_2\theta_2^2}\right\}\rceil, N_1=\lceil(AO)^{\hat{L}}\log(\frac{1}{\beta})\rceil, N_2=\lceil 16C\frac{\overline{H}^2\log\frac{K^2n}{\beta}}{\epsilon^2}\rceil$\label{line: parameters}  
    \STATE Define $\hat{\cM}:=\{\hat{\cM}(g^{1:\overline{H},j})\}_{j=1}^K$ \label{line: hat M} 
    \STATE $\{g^{1:\overline{H},j}\}_{j=1}^K\leftarrow \text{BaSeCAMP}(\hat{L},N_0,N_1,\alpha,\beta,K)$ by calling Algorithm 3 of \cite{noahlearning}  under $\cD_\cL'$ \label{line: sample strategy}
    \FOR{$j=1$ to $K$}
    \STATE $\hat{\cM}(g^{1:\overline{H},j})\leftarrow \text{LEE}(g^{1:\overline{H},j},\{\hat{\cC}_h\}_{h\in[\overline{H}]},\{\hat{\phi}_h\}_{h\in[\overline{H}]}, \Gamma, \zeta_1,\zeta_2,\theta_1,\theta_2,\beta)$ by calling Algorithm 5 of \cite{liu2023tractable} under $\cD_\cL'$ \label{line: LEE}
    \STATE $\overline{g}_{1:\overline{H}}^{j,\ast}\leftarrow$ Algorithm \ref{algorithm under AIS}$(\hat{\cM}(g^{1:\overline{H},j}))$\label{line: dp in learning}
    \ENDFOR
    \STATE $\overline{g}_{1:\overline{H}}^\ast\leftarrow \text{Pos-Dec}(\{\overline{g}_{1:\overline{H}}^{j,\ast}\}_{j=1}^K, N_2)$ by calling Algorithm 8 of \cite{liu2023tractable} under $\cD_\cL'$ \label{line: pos}
    \STATE $\tilde{g}_{1:\tilde{H}}^\ast\leftarrow$ $\varphi(\overline{g}_{1:\overline{H}}^\ast,\cD_\cL)$ 
    \STATE $g^{m,\ast}_{1:H}\leftarrow \{\tilde{g}_{1}^\ast,\tilde{g}_{3}^\ast,\cdots,\tilde{g}_{2H-1}^\ast\}$
    \STATE $g^{a,\ast}_{1:H}\leftarrow \{\tilde{g}_{2}^\ast,\tilde{g}_{4}^\ast,\cdots,\tilde{g}_{2H}^\ast\}$
    \RETURN $(g^{a,\ast}_{1:H},g^{m,\ast}_{1:H})$
    \end{algorithmic}
\end{algorithm}
 
\begin{algorithm}[!h]
    \caption{Vanilla  Realization of $\varphi(\Breve{g}_{1:\Breve{H}},\cD_\cL)$}
    \label{algorithm varphi}
    \begin{algorithmic}[1]
    \REQUIRE Strategy $\Breve{g}_{1:\Breve{H}}$, QC Dec-POMDP $\cD_\cL$
    \STATE $\tilde{g}_{1:\Breve{H}}\leftarrow \emptyset$
    \FOR{$h_2=1$ to $\Breve{H}$, $i_2=1$ to $n$, $\tilde{\tau}_{i_2,h_2}\in \tilde{\cT}_{i_2,h_2}$}
    \STATE $\Breve{\tau}_{i_2,h_2}'\leftarrow \tilde{\tau}_{i_2,h_2}$
    \FOR{$h_1=1$ to $h_2-1$, $i_1=1$ to $n$}
\IF{$\sigma(\tilde{\tau}_{i_1,h_1})\subseteq \sigma(\tilde{c}_{h_2})$\footnotemark~in $\cD_\cL$ and $\tilde{a}_{i_1,h_1}$ influences $\tilde{s}_{h_1+1}$}
    \STATE Obtain the value of $\tilde{\tau}_{i_1,h_1}$ from that of $\tilde{c}_{h_2}$ (based on $\tilde{\tau}_{i_2,h_2}$)
    \STATE $\tilde{a}_{i_1,h_1}\leftarrow \tilde{g}_{i_1,h_1}(\tilde{\tau}_{i_1,h_1})$ \\
    \STATE $\Breve{\tau}_{i_2,h_2}'\leftarrow \Breve{\tau}_{i_2,h_2}'\cup \{\tilde{a}_{i_1,h_1}\}$
    \ENDIF
    \ENDFOR
    \STATE $\tilde{g}_{i_2,h_2}(\tilde{\tau}_{i_2,h_2})\leftarrow \Breve{g}_{i_2,h_2}(\Breve{\tau}_{i_2,h_2}')$
    \ENDFOR
    \RETURN $\tilde{g}_{1:\tilde{H}}$
    \end{algorithmic}
\end{algorithm}
\footnotetext{
Note that the inclusion of $\sigma$-algebras here does not rely on the realized values of $\tilde{\tau}_{i_1,h_1},\tilde{c}_{h_2}$, but relies on the information structure of $\cD_\cL$.}

\begin{algorithm}
    \caption{Efficient  Implementation of $\varphi(\Breve{g}_{1:\Breve{H}},\cD_\cL)$} 
    \label{algorithm Implement varphi}
    \begin{algorithmic}[1]
    \REQUIRE Strategy $\Breve{g}_{1:\Breve{H}}$, QC Dec-POMDP $\cD_\cL$
    \FOR{$h=1$ to $\Breve{H}$}
    \FOR{$i=1$ to $n$}
    \STATE Agent $i$ receives $\tilde{\tau}_{i,h}$ 
    \STATE $\Breve{\tau}_{i,h}\leftarrow$ Recover$(\tilde{\tau}_{i,h},\Breve{g}_{1:h-1},\cD_\cL)$ $\backslash\backslash$ Recursion of  Algorithm \ref{algorithm recover}
    \STATE Agent $i$ takes action  $\tilde{a}_{i,h}\leftarrow \Breve{g}_{i,h}(\Breve{\tau}_{i,h})$  
    \ENDFOR
    \ENDFOR
    \end{algorithmic}
\end{algorithm}
\begin{algorithm}[!h]
    \caption{Recover$(\tilde{\tau}_{i,h},\Breve{g}_{1:h-1},\cD_\cL)$}    \label{algorithm recover}
    \begin{algorithmic}[1]
        \REQUIRE Information $\tilde{\tau}_{i,h}$, Strategy $\Breve{g}_{1:h-1}$, QC Dec-POMDP $\cD_\cL$
        \STATE $\Breve{\tau}_{i,h}\leftarrow \tilde{\tau}_{i,h}$
        \FOR{$j=1$ to $n$, $h'=1$ to $h-1$}
        \IF{$\sigma(\tilde{\tau}_{j,h'})\subseteq \sigma(\tilde{c}_h)$ in $\cD_\cL$ and $\tilde{a}_{j,h'}\notin \tilde{\tau}_{i,h}$}
        \STATE Obtain the value of $\tilde{\tau}_{j,h'}$ from that of $\tilde{c}_{h}$ (based on $\tilde{\tau}_{i,h}$)
        \STATE $\Breve{\tau}_{j,h'}\leftarrow $ Recover$(\tilde{\tau}_{j,h'},\Breve{g}_{1:h'-1},\cD_\cL)$
        \STATE $\tilde{a}_{j,h'}\leftarrow \Breve{g}_{j,h'}(\Breve{\tau}_{j,h'})$ 
        \STATE $\Breve{\tau}_{i,h}\leftarrow \Breve{\tau}_{i,h}\cup \{\tilde{a}_{j,h'}\}$
        \ENDIF
        \ENDFOR
        \RETURN $\Breve{\tau}_{i,h}$
    \end{algorithmic}
\end{algorithm}
\begin{algorithm}[!t] 
    \caption{Planning in  Dec-POMDPs with Expected Approximate Common-information Model}
    \label{algorithm under AIS}
    \begin{algorithmic}[1]
    \REQUIRE Expected approximate common-information model $\cM$  
    \FOR{$\hat{c}_{\overline{H}+1}\in \hat{\cC}_{\overline{H}+1}$}
    \STATE $V^{\ast,\cM}_{\overline{H}+1}(\hat{c}_{\overline{H}+1})\leftarrow 0$
    \ENDFOR
    \FOR{$h=\overline{H}$ to $1$}
    \FOR{$\hat{c}_{h}\in \hat{\cC}_h$}
    \STATE Define $Q^{\ast,\cM}_{h}
    (\hat{c}_h,\gamma_{1,h},\cdots,\gamma_{n,h}):=\hat{\cR}_h^\cM(\hat{c}_h,\gamma_h)+\EE^\cM\left[V^{\ast,\cM}_{h+1}(\hat{c}_{h+1})\given \hat{c}_h,\gamma_h\right]$ \label{line:one_step_opt}
    \#\label{equ:one_step_opt}
\left(\hat{g}_{1,h}^\ast(\cdot\given \hat{c}_h,\cdot),\cdots,\hat{g}_{n,h}^\ast(\cdot\given \hat{c}_h,\cdot)\right)\leftarrow \argmax_{\gamma_{1:n,h}\in\Gamma_h}Q^{\ast,\cM}_{h}(\hat{c}_h,\gamma_{1,h},\cdots,\gamma_{n,h})
    \#
    \STATE $V^{\ast,\cM}_h(\hat{c}_h)\leftarrow \max_{\gamma_{1:n,h}\in\Gamma_h}Q^{\ast,\cM}_{h}(\hat{c}_h,\gamma_{1,h},\cdots,\gamma_{n,h})$
    \ENDFOR
    \ENDFOR
    \RETURN $\hat{g}^\ast_{1:\overline{H}}$
    \end{algorithmic}
\end{algorithm}

\section{Decentralized POMDPs (with Information Sharing)} 
\label{sec:Dec-POMDP definition}

A Dec-POMDP with $n$ agents and potential information sharing can be characterized  by a tuple
\begin{equation*}
\cD=\la H,\cS,\{\cA_{i,h}\}_{i\in[n],h\in[H]},\{\cO_{i,h}\}_{i\in[n],h\in[H]},\{\TT_{h}\}_{h\in[H]},\{\OO_h\}_{h\in[H]},\mu_1,\{\cR_{h}\}_{h\in[H]}\ra,
\end{equation*}
where $H$ denotes the length of each episode, $\cS$ denotes the state space, and $\cA_{i,h}$ denotes the \emph{control action} space of agent $i$ at timestep  $h$.  We denote by $s_h\in \cS$ the state and by $a_{i,h}$ the control action of agent $i$ at timestep $h$. We use $a_h:=(a_{1,h},\cdots,a_{n,h})\in \cA_h:=\cA_{1,h}\times\cA_{2,h}\times\cdots\cA_{n,h}$ to denote the joint control action for all the $n$ agents at timestep $h$, with $\cA_h$ denoting 
the joint control action space at timestep $h$. We denote $\TT=\{\TT_h\}_{h\in[H]}$ the collection of transition functions, where 
$\TT_h(\cdot\given s_h,a_h)\in \Delta(\cS)$ gives the transition probability to the next state $s_{h+1}$ when taking the joint control action $a_h$ at state $s_h$. We use $\mu_1\in\Delta(\cS)$ to denote the distribution of the initial state $s_1$.  We denote by $\cO_{i,h}$ the observation space and by $o_{i,h}\in \cO_{i,h}$ the observation of agent $i$ at timestep $h$. We use
$o_h:=(o_{1,h},o_{2,h},\cdots,o_{n,h})\in\cO_h:=\cO_{1,h}\times\cO_{2,h}\times\cdots\cO_{n,h}$ to denote the joint observation of all the $n$ agents at timestep $h$, with $\cO_h$ denoting 
the joint observation space at timestep $h$. We use $\{\OO_h\}_{h\in[H]}$ to denote the collection of emission matrices, where $o_h\sim\OO_h(\cdot\given s_h)\in \Delta(\cO_h)$ at timestep $h$ under state $s_h\in\cS$. For notational convenience, we adopt the matrix convention, where $\OO_h$ is a matrix with each row  $\OO_h(\cdot\given s_h)$ for all $s_h\in\cS$. Also, we denote by $\OO_{i,h}$ the marginalized emission for agent $i$ at timestep $h$. Finally, $\{\cR_{h}\}_{h\in[H]}$ is a collection of reward functions among all the  agents, where $\cR_{h}:\cS\times \cA_h\rightarrow [0,1]$.

At timestep $h$, 
each agent $i$ in the Dec-POMDP has access  to some
information $\tau_{i,h}$, a subset of historical joint observations and actions, namely, $\tau_{i,h}\subseteq\{o_1,a_1,o_2,\cdots, a_{h-1},o_h\}$, and the collection of all possible such available information is denoted by $\cT_{i,h}$. We use $\tau_{h}$ to denote the \emph{joint} available information at timestep $h$. 
Meanwhile, agents may \emph{share}  part of the history with each other. The \emph{common information}  $c_h=\cup_{t=1}^h z_t$ at timestep $h$ is thus a subset of the joint history $\tau_h$, where $z_h$ is the information shared at timestep $h$. We use $\cC_h$ to denote the collection of all possible $c_h$ at timestep $h$, and use $\cT_{i,h}$ to denote the collection of all possible $\tau_{i,h}$ of agent $i$ at timestep $h$. Besides the common information $c_h$, each agent also has her \emph{private information}  $p_{i,h}=\tau_{i,h}\backslash c_h$, where the collection of $p_{i,h}$ is denoted by $\cP_{i,h}$. We also denote by $p_h$ the \emph{joint} private information, and by $\cP_h$ the collection of all possible $p_h$ at timestep $h$. We refer to the above the \emph{state-space model} of the Dec-POMDP (with information sharing), which follows the models in \cite{ashutosh2013team,ashutosh2013game}. 

{Each agent $i$ at timestep $h$ chooses the control action $a_{i,h}$ based on some strategy $g_{i,h}:\cT_{i,h}\rightarrow \cA_{i,h}$. We denote by $g_h:=(g_{1,h},g_{2,h},\cdots,g_{n,h})$ the joint control strategy of all the agents, and by $g_{1:h}:=(g_{1},g_{2},\cdots,g_h), \forall h\in[H]$ the sequence of joint strategies from timestep $1$ to $h$. We use $\cG_{i,h}$ to denote the strategy space of $g_{i,h}$, and use $\cG_{h},\cG_{1:h}$ to denote joint strategy spaces, correspondingly.} 
Next, we introduce some background on the intrinsic model and information structure of Dec-POMDPs.

\subsection{Intrinsic Model}
\label{sec:Intrinsic model}
In an intrinsic model \cite{witsenhausen1975intrinsic}, we regard the agent $i$ at different timesteps as \emph{different agents}, and each agent only acts \emph{once}  throughout.  
Any Dec-POMDP $\cD$ with $n$ agents can be formulated within the  intrinsic-model framework, and can be characterized by a tuple
$\la (\Omega,\mathscr{F}),N,\{(\UU_l,\mathscr{U}_l)\}_{l=1}^N,\{(\II_l,\mathscr{I}_l)\}_{l=1}^N\ra$ \cite{aditya2012information}, where $(\Omega,\cF)$ is a measurable space of the environment, 
$N=n\times H$ is the number of agents in the intrinsic model. By a slight abuse of notation, we write $[N]:=[n]\times[H]$, and  write $l:=(i,h)\in[N]$ for notational convenience. This way, any agent $l\in[N]$ corresponds to an agent $i\in[n]$ at timestep $h\in[H]$ in the state-space model.
We denote by $\UU_l$ the measurable action space of agent $l$ and by $\mathscr{U}_l$ the $\sigma$-algebra over $\UU_l$. For $A\subseteq[N]$, let $\HH_A:=\Omega\times \prod_{l\in A}\UU_l$ and $\HH:=\HH_{[N]}$. For any $\sigma$-algebra $\mathscr{C}$ over $\HH_{A}$, let $\la \mathscr{C}\ra$ denote the cylindrical extension of $\mathscr{C}$ on $\HH$. 
Let $\mathscr{H}_A:=\la \mathscr{F}$\scalebox{0.8}{$\otimes$}$(\otimes_{l\in A}\mathscr{U}_l)\ra$ and $\mathscr{H}=\mathscr{H}_{[N]}$. 
We denote by $\II_l$ the space of \emph{information available} to agent $l$, and by $\mathscr I_l$ the $\sigma$-algebra over $\HH$. For $l\in[N]$, we denote by $I_l$ the information of agent $l$, and $U_l$ the action of agent $l$. 
The spaces and random variables of agent $l=(i,h)$ in the intrinsic model are related to those in the state-space model as follows:   
$\forall l=(i,h)\in[N], \UU_l=\cA_{i,h},\II_l=\cT_{i,h}, U_l=a_{i,h},I_l=\tau_{i,h}$.
\subsection{Information Structures of Dec-POMDPs}
\label{sec:Information structure}
 
 An important class of IS is the \emph{quasi-classical}  one, which is defined as follows \cite{witsenhausen1975intrinsic,aditya2012information,yuksel2023stochastic}. 

\begin{definition}[Quasi-classical Dec-POMDPs]
    We call a Dec-POMDP problem \emph{QC} if each agent in the intrinsic model knows the information available to the agents who influence her{, directly or indirectly}, i.e.,   $\forall l_1,l_2\in[N], l_1=(i_1,h_1),l_2=(i_2,h_2), i_1,i_2\in[n], h_1,h_2\in[H]$, if agent $l_1$ influences agent $l_2$, Then, $\mathscr{I}_{l_1}\subseteq \mathscr{I}_{l_2}$. 
\end{definition}

Furthermore, \emph{strictly} quasi-classical IS \cite{witsenhausen1975intrinsic,mahajan2010measure}, as a subclass of QC IS, is defined as follows. 

\begin{definition}[Strictly quasi-classical Dec-POMDPs]
We call a Dec-POMDP problem \emph{sQC} if each agent in the intrinsic model knows the information \emph{and} actions available to the agents who influence her, directly or indirectly. That is, $\forall l_1,l_2\in[N], l_1=(i_1,h_1),l_2=(i_2,h_2), i_1,i_2\in[n], h_1,h_2\in[H]$, if agent $l_1$ influences agent $l_2$, Then, $\mathscr{I}_{l_1}\cup\la \mathscr{U}_{l_1}\ra\subseteq \mathscr{I}_{l_2}$. 
\end{definition}

\subsection{Intrinsic Model of LTC Problems} 
\label{intrinsic of LTC}
Given any LTC $\cL$ of the state-space-model form defined in \S\ref{sec:ltc_formulation}, we define the intrinsic model of $\cL$ as a tuple   $\la (\Omega,\mathscr{F}), N,\{(\UU_l,\mathscr{U}_l)\}_{l=1}^N,\{(\MM_l,\mathscr{M}_l)\}_{l=1}^N,\{(\II_{l^-},\mathscr{I}_{l^-})\}_{l=1}^N,\{(\II_{l^+},\mathscr{I}_{l^+})\}_{l=1}^N\ra$, where $(\Omega,\mathscr{F})$ is the measure space representing all the uncertainty in the system;  $N=n\times H$ is the number of agents in the intrinsic model. By a slight abuse of notation, we write $[N]:=[n]\times[H]$, and  write $l:=(i,h)\in[N]$ for convenience. This way, any agent $l\in[N]$ corresponds to an agent $i\in[n]$ at timestep $h\in[H]$ in the state-space model, and we thus define $l^-:=(i,h^-)$ and $l^+:=(i,h^+)$ accordingly. 
    We denote by $\UU_l$ and $\MM_l$ the measurable control and  communication action spaces of agent $l$, and by $\mathscr{U}_l$ and $\mathscr{M}_l$ the $\sigma$-algebra over $\UU_l$ and $\MM_l$, respectively. 
    For any $A\subseteq[N]$, let $\HH_A:=\Omega\times \prod_{l\in A}(\UU_l\times \MM_l)$ and $\HH:=\HH_{[N]}$. For any $\sigma$-algebra $\mathscr{C}$ over $\HH_{A}$, let $\la \mathscr{C}\ra$ denote the cylindrical extension of $\mathscr{C}$ on $\HH$. 
    Let 
   $\mathscr{H}_A:=\la \mathscr{F}$\scalebox{0.8}{$\otimes$}$(\otimes_{l\in A}\mathscr{U}_l)$\scalebox{0.8}{$\otimes$}$(\otimes_{l\in A}\mathscr{M}_l)\ra$, $\mathscr{H}=\mathscr{H}_{[N]}$. 
    We denote by $\II_{l^-}$ and $\II_{l^+}$ the spaces of \emph{information available} to agent $l$ \emph{before} and \emph{after} additional sharing, respectively, 
    and by $\mathscr I_{l^-}\subseteq \mathscr{H}$ and $\mathscr I_{l^+}\subseteq \mathscr{H}$ the associated  $\sigma$-algebra.
The spaces and random variables of agent $l=(i,h)$ in the intrinsic model are related to those in the state-space model as follows:   
$\forall l=(i,h)\in[N], \UU_l=\cA_{i,h},\MM_l=\cM_{i,h},\II_{l^-}=\cT_{i,h^-},\II_{l^+}=\cT_{i,h^+}, U_l=a_{i,h},M_l=m_{i,h}, I_{l^-}=\tau_{i,h^-},I_{l^+}=\tau_{i,h^+}$. For notational convenience, for any random variable $B$ in LTC and the $\sigma$-algebra $\mathscr{B}$ generated by $B$, we overload $\sigma(B)$ to denote the cylindrical extension of $\mathscr{B}$ on $\HH$, i.e.,  $\sigma(B)=\la \mathscr{B}\ra$.

\section{Other Supplementary  Results}\label{sec:supp_mat}

\subsection{Optimality of Deterministic Strategies}\label{sec: proof details sec 2} 
We now show a supplementary result that for the formulated LTC problem, it does not lose optimality to consider \emph{deterministic} strategies as introduced in  \S\ref{problem formulation}. 
For any LTC problem $\cL$, consider generic, stochastic communication and control strategy spaces:  $\forall i\in[n],~h\in[H],~\cG_{i,h}^{m,S}:=\{ g_{i,h}^{m,S}:\Omega_{i,h}^m\times\cT_{i,h^-}\rightarrow \cM_{i,h}\},~ \cG_{i,h}^{a,S}:=\{ g_{i,h}^{a,S}:\Omega_{i,h}^a\times\cT_{i,h^+}\rightarrow\cA_{i,h}\}$, where $\{\Omega_{i,h}^m\}_{i\in[n],h\in[H]},\{\Omega_{i,h}^a\}_{i\in[n],h\in[H]}$ are the sets of random seeds which could be correlated to each other across agents and timesteps. 
Note that these strategy classes include those of {the  strategies} randomized over the action sets, i.e., $\forall i\in[n],~h\in[H],~\cG_{i,h}^{m,S}:=\{ g_{i,h}^{m,S}:\cT_{i,h^-}\rightarrow \Delta(\cM_{i,h})\},~ \cG_{i,h}^{a,S}:=\{ g_{i,h}^{a,S}:\cT_{i,h^+}\rightarrow\Delta(\cA_{i,h})\}$.  
Also, we denote by $\cG_h^{a,S}, \cG_h^{m,S}$ the \emph{joint} stochastic control and communication spaces at timestep $h$, respectively.  
Similarly, we define the objective under the stochastic  strategies as
\begin{equation*}
     \forall g_{1:H}^{a,S}\in \cG_{1:H}^{a,S}, g_{1:H}^{m,S}\in \cG_{1:H}^{m,S}, \qquad\qquad  J_{\cL}(g^{a,S}_{1:H},g^{m,S}_{1:H}):=\EE_{\cL}\left[\sum_{h=1}^H (r_h-\kappa_h)\bigggiven g^{a,S}_{1:H},g^{m,S}_{1:H}\right].
\end{equation*} 
\begin{lemma}\label{lemma:no_lose_optimality}
    It does not lose optimality to consider deterministic control and communication strategies in LTC. Namely, for any LTC problem $\cL$,
    \begin{equation}
        \max_{g_{1:H}^{a,S}\in \cG_{1:H}^{a,S}, g_{1:H}^{m,S}\in \cG_{1:H}^{m,S}}J_{\cL}(g^{a,S}_{1:H},g^{m,S}_{1:H})=\max_{g_{1:H}^a\in \cG_{1:H}^a, g_{1:H}^m\in \cG_{1:H}^m}J_{\cL}(g^a_{1:H},g^m_{1:H}).
        \label{eq: sto=det}
    \end{equation}
\end{lemma}
\begin{proof}
    For any $i\in[n],h\in[H]$, since space $\cG_{i,h}^{m,S}$ covers space $\cG_{i,h}^m$ and 
    space $\cG_{i,h}^{a,S}$ covers space $\cG_{i,h}^a$, we have  that 
    \begin{align}
        \label{eq: sto>det}
        \max_{g_{1:H}^{a,S}\in \cG_{1:H}^{a,S}, g_{1:H}^{m,S}\in \cG_{1:H}^{m,S}}J_{\cL}(g^{a,S}_{1:H},g^{m,S}_{1:H})\ge\max_{g_{1:H}^a\in \cG_{1:H}^a, g_{1:H}^m\in \cG_{1:H}^m}J_{\cL}(g^a_{1:H},g^m_{1:H}).
    \end{align}
    In the other direction, from the tower property, for any $g_{1:H}^{a,S}\in \cG_{1:H}^{a,S}, g_{1:H}^{m,S}\in \cG_{1:H}^{m,S}$
    \small
    \begin{align*}
    &J_{\cL}(g^{a,S}_{1:H},g^{m,S}_{1:H})=\EE_{\cL}\left[\sum_{h=1}^H (r_h-\kappa_h)\bigggiven g^{a,S}_{1:H},g^{m,S}_{1:H}\right]
    =\EE\left[\EE_{\cL}\big[\sum_{h=1}^H (r_h-\kappa_h)\bigggiven \{g^{a,S}_{i,h}[\omega_{i,h}^a]\}_{i\in[n],h\in[H]},\{g^{m,S}_{i,h}[\omega_{i,h}^m]\}_{i\in[n],h\in[H]}\big]\right],
    \end{align*} 
    \normalsize 
    where $\omega_{i,h}^a\in \Omega_{i,h}^a$, $\omega_{i,h}^m\in \Omega_{i,h}^m$ are random seeds, and $g_{i,h}^{a,S}[\omega_{i,h}^a]\in \cG_{i,h}^a,g^{a,S}_{i,h}[\omega_{i,h}^m]\in\cG_{i,h}^m$ are deterministic strategies defined  as
    \begin{align*}
        &\forall \tau_{i,h^-}\in\cT_{i,h^-},g^{m,S}_{i,h}[\omega_{i,h}^m](\tau_{i,h^-}):=g^{m,S}_{i,h}(\omega_{i,h}^m,\tau_{i,h^-}),\\
        &\forall \tau_{i,h^+}\in\cT_{i,h^+},g^{a,S}_{i,h}[\omega_{i,h}^a](\tau_{i,h^+}):=g^{a,S}_{i,h}(\omega_{i,h}^a,\tau_{i,h^+}).
    \end{align*}
    Therefore, we have 
    \begin{align*}
        &J_{\cL}(g^{a,S}_{1:H},g^{m,S}_{1:H})=\EE\left[\EE_{\cL}\big[\sum_{h=1}^H (r_h-\kappa_h)\bigggiven \{g^{a,S}_{i,h}[\omega_{i,h}^a]\},\{g^{m,S}_{i,h}[\omega_{i,h}^m]\}\big]\right]\\
        &\quad\le \max_{\{\omega_{i,h}^a\}_{i\in[n],h\in[H]},\{\omega_{i,h}^m\}_{i\in[n],h\in[H]}}\EE_{\cL}\left[\sum_{h=1}^H (r_h-\kappa_h)\bigggiven \{g^{a,S}_{i,h}[\omega_{i,h}^a]\},\{g^{m,S}_{i,h}[\omega_{i,h}^m]\}\right]\\
        &\quad\le \max_{g_{1:H}^a\in \cG_{1:H}^a, g_{1:H}^m\in \cG_{1:H}^m}J_{\cL}(g^a_{1:H},g^m_{1:H})
    \end{align*}
    holds for any $g_{1:H}^{a,S}\in \cG_{1:H}^{a,S}, g_{1:H}^{m,S}\in \cG_{1:H}^{m,S}$.  Hence, we further get
    \begin{align}
        \label{eq: sto<det}
        \max_{g_{1:H}^{a,S}\in \cG_{1:H}^{a,S}, g_{1:H}^{m,S}\in \cG_{1:H}^{m,S}}J_{\cL}(g^{a,S}_{1:H},g^{m,S}_{1:H})\le\max_{g_{1:H}^a\in \cG_{1:H}^a, g_{1:H}^m\in \cG_{1:H}^m}J_{\cL}(g^a_{1:H},g^m_{1:H}). 
    \end{align}
    Combining \Cref{eq: sto>det} and \Cref{eq: sto<det}, we prove  \Cref{eq: sto=det}.
\end{proof}

\subsection{Conditions  Leading to Assumption \ref{assu: one_step_tract}}\label{sec:one_step_examples} 

As a minimal requirement for computational tractability (for both Dec-POMDPs and LTCs), Assumption \ref{assu: one_step_tract} is needed for the one-step tractability of the team-decision problem involved in the value iteration in Algorithm \ref{algorithm under AIS}. We now adapt several such structural conditions from \cite{liu2023tractable} to the LTC setting,  which lead to this assumption and have been studied in the literature. Note that since we need to do planning in the approximate model $\cM$, which is oftentimes constructed based on the original problem $\cL$ and the approximate beliefs $\{\PP_{h}^{\cM,c}(\overline{s}_h,\overline{p}_h\given \hat{c}_h)\}_{h\in[\overline{H}]}$, we necessarily need conditions on models $\cL$, $\cM$, and the beliefs $\{\PP_{h}^{\cM,c}(\overline{s}_h,\overline{p}_h\given \hat{c}_h)\}_{h\in[\overline{H}]}$ that $\cM$ is consistent with, to ensure Assumption \ref{assu: one_step_tract} holds, for which we refer to as the \textbf{Part (1)} and \textbf{Part (2)} of the conditions below, respectively.

\begin{itemize}
\label{one-step tractability condition}
    \item \textbf{Turn-based structures.} \textbf{Part (1):} At each timestep $h\in[\overline{H}]$, there is only one agent, denoted as  $ct(h)\in[n]$, that can affect the state transition. More concretely, the transition dynamics take the forms of $\TT_h:\cS\times\cA_{ct(h),h}\rightarrow \Delta(\cS)$. Additionally, we assume the reward function admits an additive structure such that $\cR_h(s_h, a_h)=\sum_{i\in[n]} \cR_{i, h}(s_h, a_{i, h})$ for some functions $\{\cR_{i,h}\}_{i\in[n]}$. Meanwhile, since only agent $ct(h)$ takes the action, we assume the increment of the common information satisfies   $z_{h+1}^b=\chi_{h+1}(p_{h^+}, a_{ct(h),h}, o_{h+1})$.  \textbf{Part (2):} No additional requirement. Such a structure has been commonly studied in (fully observable) stochastic games and multi-agent RL \cite{filar2012competitive,bai2020provable}. 
    \item \textbf{Nested private information.} \textbf{Part (1):} No additional requirement. \textbf{Part (2): } At each timestep $h=2t$ with $t\in[H]$, all the agents form a \emph{hierarchy} according to the private information they possess, in the sense that  $\forall~i,j\in[n],j<i, \overline{p}_{j,h}=Y_h^{i,j}(\overline{p}_{i,h})$ for some function $Y_h^{i,j}$. More formally, the approximate belief satisfies that $\PP_h^{\cM,c}(\overline{p}_{j,h}=Y_h^{i,j}(\overline{p}_{i,h})\given \overline{p}_{i,h},\hat{c}_h)=1$. Such a structure has been investigated in  \cite{peralez2024solving} with heuristic search, and in \cite{liu2023tractable} with finite-time complexity analysis when there is no additional sharing to decide/learn. 
    \item \textbf{Factorized structures.} \textbf{Part (1):} At each timestep $h\in[\overline{H}]$, the state $s_h$ can be partitioned into $n$ local states, i.e., $s_h=(s_{1,h},s_{2,h},\cdots,s_{n,h}).$ Meanwhile, the transition kernel takes the product form of $\TT_h(s_{h+1}\given s_h,a_h)=\prod_{i=1}^n \TT_{i,h}(s_{i,h+1}\given s_{i,h}, a_{i,h})$, the emission also takes the product form of $\OO_h(o_h\given s_h)=\prod_{i=1}^n \OO_{i,h}(o_{i,h}\given s_{i,h})$, and the communication cost and reward functions can be decoupled into $n$ terms such that $\cK_h (z_h^a)=\sum_{i=1}^n\cK_{i,h}(z_{i,h}^a), \cR_h(s_h,a_h)=\sum_{i=1}^n\cR_{i,h}(s_{i,h}, a_{i,h})$. 
    \textbf{Part (2):} At each timestep $h\in[\overline{H}]$, the approximate common information is also factorized so that $\hat{c}_h=(\hat{c}_{1,h},\hat{c}_{2,h},\cdots,\hat{c}_{n,h})$ and its evolution satisfies that $\hat{c}_{i,h}=\hat{\phi}_{i,h}(\hat{c}_{i,h-1},\overline{z}_{i,h-1})$ for some function $\hat{\phi}_{i,h}$. Correspondingly, the approximate beliefs need to satisfy that $\PP_h^{\cM,c}(\overline{s}_h,\overline{p}_h\given \hat{c}_h)=\Pi_{i=1}^n\PP_{i,h}^{\cM,c}(\overline{s}_{i,h},\overline{p}_{i,h}\given \hat{c}_{i,h})$ for some functions $\{\PP_{i,h}^{\cM,c}\}_{i\in[n],h\in[\overline{H}]}$. 
    Such a structure, under general information sharing protocols, can lead to non-classical IS. In this case, it may  be viewed as  an example of     
    {non-classical} ISs where the agents have no incentive for {signaling} \cite[\S 3.8.3]{yuksel2023stochastic}.
\end{itemize}
\begin{lemma}
    Given any LTC problem $\cL$, let  $\cD_\cL'$ be the Dec-POMDP after reformulation, strict expansion, and refinement. For any $\cM$ to be the approximate model of $\cD_\cL'$ and $\{\PP_{h}^{\cM,c}\}_{h\in[\overline{H}]}$ to be the approximate belief, if they satisfy any of the $3$ conditions above, Then,  \Cref{equ:one_step_opt} in Algorithm \ref{algorithm under AIS} can be solved with time complexity  $\max_{h\in[\overline{H}]}\text{poly}~(|\overline{\cP}_h|,|\overline{\cA}_h|,|\overline{\cS}|)$. 
    \label{lemma: one_step_tract}
\end{lemma} 
\begin{proof} 
    For any $h\in[\overline{H}]$, if $h=2t-1,t\in[H]$, from Assumption \ref{limited communication strategy} and the construction of $\cD_\cL'$, since we need to find the optimal strategy of $\cD_\cL'$ in the spaces   $\overline{g}_h\in \overline{\cG}_h=\{g:\overline{\cC}_h\rightarrow \overline{\cA}_h\}$ for all $h\in[\overline{H}]$, where we recall that $\overline{\cA}_h=\cM_{t}$ is joint communication action space. Then, $\Gamma_h=\cM_t$ has cardinality $|\Gamma_h|=|\cM_t|$, and $\gamma_{1:n,h}^\ast$ can be computed as $\gamma_{1:n,h}^\ast\in \argmax_{\gamma_{1:n,h}\in \Gamma_h}Q_h^{\ast,\cM}(\hat{c}_h,\gamma_{1:n,h})$  
    by enumerating all possible $\gamma_{1:n,h}\in \Gamma_h$, with a complexity of $|\cM_t|$. 
    
    If $h=2t,t\in[H]$, we prove the result case by case: 
    \begin{itemize}
    \item \textbf{Nested private information:}
    We first define the $u_{i,h}\in\cU_{i,h}:=\{(\prod_{j=1}^i\overline{\cP}_{j,h})\times (\prod_{j=1}^{i-1}\overline{\cA}_{j,h})\rightarrow \overline{\cA}_{i,h}\}$ and slightly abuse the notation for $Q_h^{\ast,\cM}$ as follows
    \begin{align*}
        Q_h^{\ast,\cM}(\hat{c}_h,u_{1,h},\cdots, u_{n,h}):=&\sum_{\overline{s}_h,\overline{p}_h,\overline{a}_h,\overline{s}_{h+1},\overline{o}_{h+1}}\PP_h^{\cM,c}(\overline{s}_h,\overline{p}_h\given \hat{c}_h)\Pi_{i=1}^n\mathds{1}[\overline{a}_{i,h}=u_{i,h}(\overline{p}_{1:i,h},\overline{a}_{1:i-1,h})]\overline{\TT}_h(\overline{s}_{h+1}\given \overline{s}_h, \overline{a}_h)\\
        &\qquad\quad\overline{\OO}_{h+1}(\overline{o}_{h+1}\given \overline{s}_{h+1})[\overline{\cR}_h(\overline{s}_h,\overline{a}_h,\overline{p}_h)+V_{h+1}^{\ast,\cM}(\hat{c}_{h+1})]. 
    \end{align*}
Since the space of $\cU_{i,h}$ covers the space of $\Gamma_{i,h}$, Then, for the $u_{1:n,h}^\ast$ to be an optimal one that maximizes  the $Q_h^{\ast,\cM}$, we have
\begin{align*}
    Q_h^{\ast,\cM}(\hat{c}_h,u_{1,h}^\ast,\cdots,u_{n,h}^\ast)=\max_{\{u_{i,h}\in \cU_{i,h}\}_{i\in[n]}}Q_h^{\ast,\cM}(\hat{c}_h,u_{1,h},\cdots, u_{n,h})\ge \max_{\{\gamma_{i,h}\in \Gamma_{i,h}\}_{i\in[n]}}Q_h^{\ast,\cM}(\hat{c}_h,\gamma_{1,h},\cdots, \gamma_{n,h}).
\end{align*}
Meanwhile, due to the nested private information condition, for any $\overline{p}_{h}\in\overline{\cP}_h$, there must exist $\gamma_{1:n,h}'$ such that $\gamma_{1:n,h}'$ output the same actions as $u_{1:n,h}^\ast$ under $\overline{p}_h$. Therefore, we can conclude that 
\begin{align*}
    \max_{\{u_{i,h}\in \cU_{i,h}\}_{i\in[n]}}Q_h^{\ast,\cM}(\hat{c}_h,u_{1,h},\cdots, u_{n,h})=\max_{\{\gamma_{i,h}\in \Gamma_{i,h}\}_{i\in[n]}}Q_h^{\ast,\cM}(\hat{c}_h,\gamma_{1,h},\cdots, \gamma_{n,h}).
\end{align*}
Therefore, we can solve \Cref{equ:one_step_opt} and compute $\gamma_{1:n,h}^\ast$ from computing $u_{1:n,h}^\ast$, which  can be solved with complexity poly$(|\overline{\cP}_h|,|\overline{\cA}_h|,|\overline{\cS}|)$.
        
        \item \textbf{Turn-based structures:} For any $\gamma_{ct(h),h}\in \Gamma_{ct(h),h}, \gamma_{-ct(h),h}, \gamma_{-ct(h),h}'\in\Gamma_{-ct(h),h}$, where $ct(h)$ is the controller at timestep $h$, it holds that for any $\hat{c}_h$: 
        \begin{align*}
            &Q_h^{\ast,\cM}(\hat{c}_h, \gamma_{ct(h),h},\gamma_{-ct(h),h})\\
            =&\sum_{\overline{s}_h, \overline{p}_h, \overline{s}_{h+1},\overline{o}_{h+1}}\PP_h^{\cM,c}(\overline{s}_h,\overline{p}_h\given \hat{c}_h)\overline{\TT}_{h}(\overline{s}_{h+1}\given \overline{s}_{h}, \gamma_{ct(h),h}(\overline{p}_{ct(h),h}),\gamma_{-ct(h),h}(\overline{p}_{-ct(h),h}))\\
            &\overline{\OO}_{h+1}(\overline{o}_{h+1}\given \overline{s}_{h+1})[\overline{\cR}_h(\overline{s}_h,\gamma_{h}(\overline{p}_{h}))+V_{h+1}^{\ast, \cM}(\hat{c}_{h+1})]\\
            =&\sum_{\overline{s}_h, \overline{p}_h, \overline{s}_{h+1},\overline{o}_{h+1}}\PP_h^{\cM,c}(\overline{s}_h,\overline{p}_h\given \hat{c}_h)\overline{\TT}_{h}(\overline{s}_{h+1}\given \overline{s}_{h}, \gamma_{ct(h),h}(\overline{p}_{ct(h),h})\overline{\OO}_{h+1}(\overline{o}_{h+1}\given \overline{s}_{h+1})[\overline{\cR}_h(\overline{s}_h,\gamma_{ct(h),h}(\overline{p}_{ct(h),h}))+V_{h+1}^{\ast, \cM}(\hat{c}_{h+1})],\\
            &\qquad+\sum_{i\neq ct(h)}\sum_{\overline{s}_h, \overline{p}_h}\PP_h^{\cM,c}(\overline{s}_h,\overline{p}_h\given \hat{c}_h)\overline{\cR}_{i,h}(\overline{s}_h,\gamma_{i,h}(\overline{p}_{i,h})):=\sum_{i\in[n]}U_{i,h}(\hat{c}_h,\gamma_{i,h}),
        \end{align*}
    where the last step is due to the fact that $\hat{c}_{h+1}=\hat{\phi}_{h+1}(\hat{c}_h,\overline{z}_{h+1})$. Note that we can write $\overline{\cR}_h(\overline{s}_h,\overline{a}_h,\overline{p}_{h})$ as $\overline{\cR}_h(\overline{s}_h,\overline{a}_h)$ since $h$ is even. 
    Then,  \Cref{equ:one_step_opt} can be solved with respect to each individual $\gamma_{i,h}$ with total complexity poly$(|\overline{\cS}|,|\overline{\cP}_{h}|,|\overline{\cA}_{h}|)$.
\item \textbf{Factorized structures:} 
Note that for any $h\in[\overline{H}]$, we can write the reward function of $\cD_\cL'$ as $\overline{\cR}(\overline{s}_h,\overline{a}_h,\overline{p}_h)=\sum_{i=1}^n\overline{\cR}(\overline{s}_{i,h},\overline{a}_{i,h},\overline{p}_{i,h}),\forall \overline{s}_h,\overline{a}_h,\overline{p}_h$. Then,
for any $h\in[\overline{H}]$, $\hat{c}_h\in\hat{\cC}_h, \gamma_h\in \Gamma_h$,  we use backward induction to prove that, there exist  $n$ functions $\{F_{i,h}\}_{i\in[n]}$ such that
\begin{equation*}
    Q_h^{\ast,\cM}(\hat{c}_h,\gamma_h)=\sum_{i=1}^n F_{i,h}(\hat{c}_{i,h},\gamma_{i,h}).
\end{equation*}
It holds for $h=\overline{H}$ since 
$Q_{\overline{H}}^{\ast,\cM}(\hat{c}_{\overline{H}},\gamma_{\overline{H}})=\sum_{i=1}^n \sum_{\overline{s}_{i,\overline{H}},\overline{p}_{i,\overline{H}}}\PP_{i,\overline{H}}^{\cM,c}(\overline{s}_{i,\overline{H}},\overline{p}_{i,\overline{H}}\given \hat{c}_{i,\overline{H}})\overline{\cR}_{i,\overline{H}}(\overline{s}_{i,\overline{H}},\gamma_{i,h}(\overline{p}_{i,\overline{H}}),\overline{p}_{i,\overline{H}})$. 
For any $h\le \overline{H}-1$, it holds that
\begin{align*}
    Q_h^{\ast,\cM}(\hat{c}_h,\gamma_h)&=\sum_{\overline{s}_h,\overline{p}_h,\overline{s}_{h+1},\overline{o}_{h+1}}\PP_h^{\cM,c}(\overline{s}_h,\overline{p}_h\given \hat{c}_h)\overline{\TT}_h(\overline{s}_{h+1}\given \overline{s}_h,\gamma_h(\overline{p}_h))\overline{\OO}_{h+1}(\overline{o}_{h+1}\given \overline{s}_{h+1})\\
&\qquad\qquad\bigg[\sum_{i=1}^n\overline{\cR}_{i,h}(\overline{s}_{i,h},\gamma_{i,h}(\overline{p}_{i,h}),\overline{p}_{i,h})+F_{i,h+1}(\hat{c}_{i,h+1},\hat{g}_{i,h+1}^\ast(\hat{c}_{i,h+1}))\bigg]\\
    &\quad=\sum_{i=1}^n \sum_{\overline{s}_{i,h},\overline{p}_{i,h},\overline{s}_{i,h+1},\overline{o}_{i,h+1}}\PP_{i,h}^{\cM,c}(\overline{s}_{i,h},\overline{p}_{i,h}\given \hat{c}_{i,h})\overline{\TT}_{h}(\overline{s}_{i,h+1}\given \overline{s}_{i,h},\gamma_{i,h}(\overline{p}_{i,h}))\\
    &\qquad\qquad\overline{\OO}_{i,h+1}(\overline{o}_{i,h+1}\given \overline{s}_{i,h+1})[\overline{\cR}_{i,h}(\overline{s}_{i,h},\gamma_{i,h}(\overline{p}_{i,h}),\overline{p}_{i,h})+F_{i,h+1}(\hat{c}_{i,h+1},\hat{g}_{i,h+1}^\ast(\hat{c}_{i,h+1}))]\\
    &\quad=:\sum_{i=1}^n F_{i,h}(\hat{c}_{i,h},\gamma_{i,h}).
\end{align*}
Then, by induction, we know that it holds for any $h\in[\overline{H}]$. We can define $\hat{g}_{i,h}^\ast(\hat{c}_{h})\in \argmax_{\gamma_{i,h}\in\Gamma_{i,h}} F_{i,h}(\hat{c}_{i,h},\gamma_{i,h})$, and thus solve  \Cref{equ:one_step_opt} with complexity $\sum_{i=1}^n$ poly$(|\overline{\cS}_i|,|\overline{\cA}_{i,h}|,|\overline{\cP}_{i,h}|)$.
    \end{itemize}
    This completes the proof. 
\end{proof}

Note that, strictly speaking, the time-complexity given in \Cref{lemma: one_step_tract} does not satisfy Assumption \ref{assu: one_step_tract} yet, since $|\overline{\cP}_h|$ may not be necessarily small and polynomial in the LTC parameters $|\cS|,|\cO_h|,|\cA_h|,|\cM_h|, H$. For the examples in \S\ref{sec: examples of QC}, more specifically, one can show that $|\overline{\cP}_h|$ is indeed polynomial in these parameters (when viewing the delay $d$ of sharing, if it exists, as a constant), which led to the final quasi-polynomial complexity guarantees in \Cref{theorem: planning} and \Cref{thm:learning} (see their proofs for the formal arguments).

\section{Examples in the Venn Diagram \Cref{fig:Venn_DecPOMDP}}\label{sec:examples_venn_diag}

Here, we show some examples of the areas \ding{172}-\ding{176} in the Venn diagram in \Cref{fig:Venn_DecPOMDP}. Note that, unless otherwise noted, the diagram here concerns the common scenarios where  
the transition kernel $\{\TT_h\}_{h\in[H]}$ and observation emission $\{\OO_h\}_{h\in [H]}$ are  \emph{non-degenerate}, i.e., for any agent $i\in[n]$ and timestep $h\in[H]$, the change of state $s_{h}$ can influence the observation  $o_{i,h}$, and the action $a_{i,h}$ can influence the state $s_{h+1}$. 

\begin{itemize}
    \item \textbf{\ding{172}: Multi-agent MDP \cite{MMDP}  with historical states.} The Dec-POMDPs satisfying that for any $h\in[H],i\in[n], \cO_{i,h}=\cS, \OO_{i,h}(s\given s)=1, c_h=\{s_{1:h}\}, p_h=\emptyset$ lie in the area \ding{172}. 
    \item \textbf{\ding{173}: Uncontrolled state process without any historical information.} The Dec-POMDPs satisfying that for any $h\in[H], i\in[n], s_h,a_h,a_h', \TT_h(\cdot\given s_h,a_h)=\TT_h(\cdot\given s_h,a_h'), c_h=\emptyset, p_{i,h}=\{o_{i,h}\}$ lie in the area \ding{173}.
    \item \textbf{\ding{174}:} {\bf Dec-POMDPs with sQC information structure, perfect recall, and Assumptions \ref{useless action} and \ref{weak gamma observability}.} One-step delayed information sharing ({\bf Example 1} in \S\ref{sec: examples of QC}) lies in this area. 
    \item \textbf{\ding{175}: State controlled by one controller with no sharing and only observability of the controller.} We consider a Dec-POMDP $\cD$. The state dynamics are controlled by only one agent (for convenience, agent 1), and only agent 1 has observability, i.e., $\TT_h(\cdot\given s_h, a_{1,h}, a_{-1,h})=\TT_h(\cdot\given s_h, a_{1,h}, a_{-1,h}')$ for all $s_h, a_{1,h}, a_{-1,h}, a_{-1,h}'$, and $\cO_{-1,h}=\emptyset$. There is no information sharing, i.e., $c_h=\emptyset, p_{1,h}=\{o_{1:h},a_{1:h-1}\}, p_{j,h}=\{a_{j,1:h-1}\}, \forall j\neq 1$.  Then, $\forall j\neq 1, h_1<h_2\in [H]$, agent $(1,h_1)$ does not influence $(j,h_2)$,  since $\tau_{j,h_2}=\{a_{j,1:h_2-1}\}$ is not influenced by agent $(1,h_1)$. Therefore, $\cD$ is sQC and has perfect recall, but $\cD$ does not have SI-CIBs, since the underlying state $s_h$ is influenced by $g_{1,1:h-1}$. This is essentially  because $\cD$ does not satisfy Assumption \ref{weak gamma observability}. Thus, $\cD$ lies in the area \ding{175}.
    \item \textbf{\ding{176}: One-step delayed observation  sharing and two-step delayed action sharing.} The Dec-POMDPs satisfying that for any $h\in[H],i\in[n], c_h=\{o_{1:h-1},a_{1:h-2}\}, p_{i,h}=\{a_{i,h-1},o_{i,h}\}$ lie in the area \ding{176}. 
\end{itemize}

\clearpage
\section{Additional Figures}

The additional figures,  \Cref{fig:illustrate}  and \Cref{fig: time line},  illustrate the paradigm and the timeline of the LTC problems considered in this paper. 

\begin{figure}[!h]
    \centering
    \includegraphics[width=0.62\textwidth]{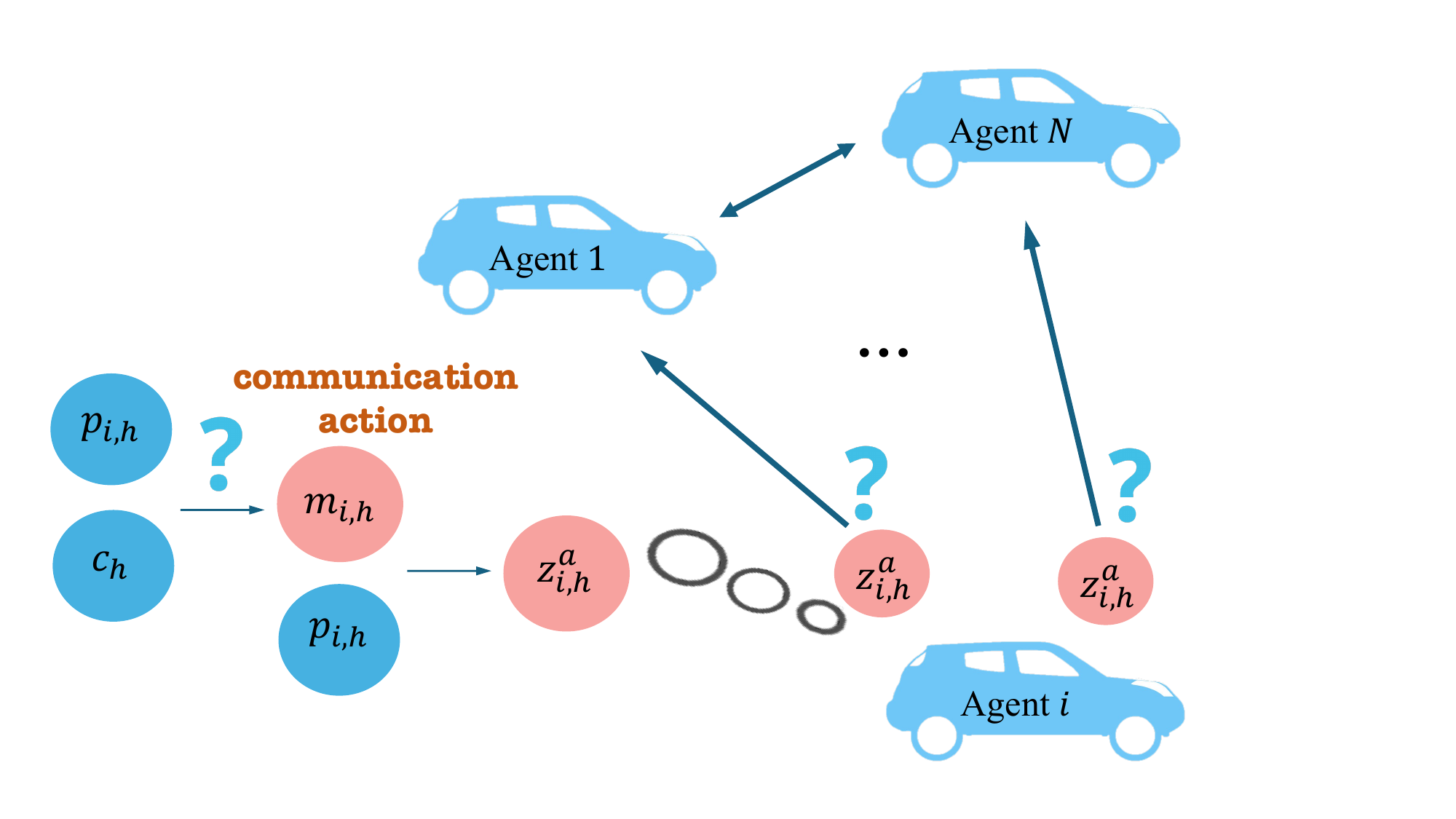}
    \caption{Illustrating the  learning-to-communicate problem considered in this  paper.}\label{fig:illustrate} 
    \vspace{50pt}
    \includegraphics[width=0.7\textwidth] {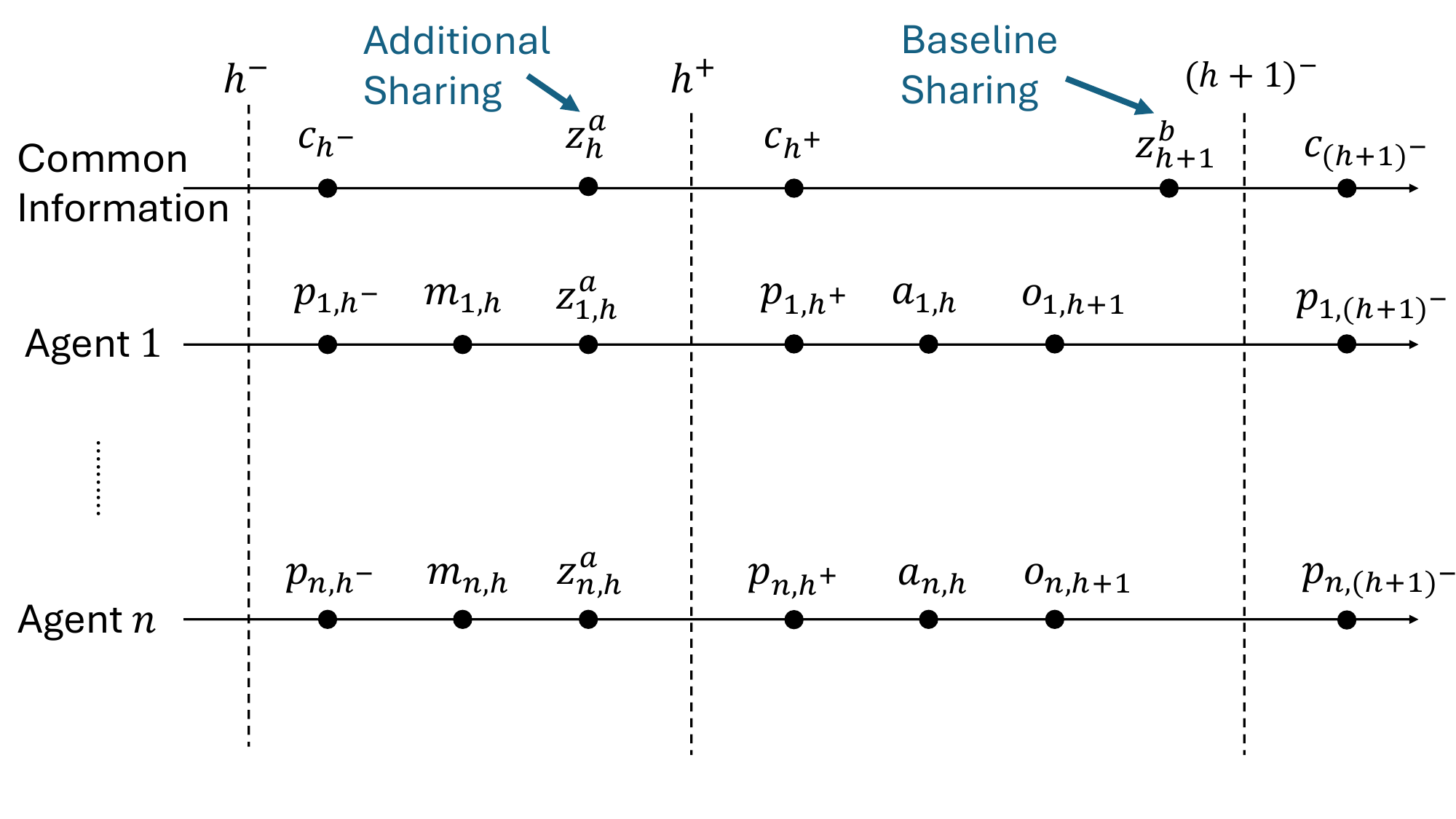}
    \caption{Timeline of the information sharing and information evolution protocols in the learning-to-communicate problem considered.}\label{fig: time line}
\end{figure}

\end{document}